\documentclass[11pt,a4paper,reqno]{amsart}%
\usepackage{amsthm,amsmath,amsfonts,amssymb,amsxtra,appendix,bookmark,dsfont,latexsym,color,epsfig}

\makeatletter
\usepackage{hyperref}
\usepackage{delarray,a4,color}
\usepackage{euscript}
\usepackage[latin1]{inputenc}

\numberwithin{equation}{section}

\newtheorem{theorem}{Theorem}[section]
\newtheorem{lemma}[theorem]{Lemma}
\newtheorem{corollary}[theorem]{Corollary}
\newtheorem{proposition}[theorem]{Proposition}

\theoremstyle{definition}
\newtheorem{definition}[theorem]{Definition}

\theoremstyle{remark}
\newtheorem{remark}[theorem]{Remark}

\numberwithin{equation}{section}

\usepackage{color}

\numberwithin{equation}{section}

\newcommand{\norm}[1]{\left\lVert #1 \right\rVert}

\newcommand{\bdm}{\begin{displaymath}}
\newcommand{\edm}{\end{displaymath}}
\newcommand{\bdn}{\begin{eqnarray}}
\newcommand{\edn}{\end{eqnarray}}
\newcommand{\bay}{\begin{array}{c}}
\newcommand{\eay}{\end{array}}
\newcommand{\ben}{\begin{enumerate}}
\newcommand{\een}{\end{enumerate}}
\newcommand{\beq}{\begin{equation}}
\newcommand{\eeq}{\end{equation}}

\newcommand{\tx}{\textstyle}

\newcommand{\eps}{\varepsilon}

\newcommand{\R}{\mathbb{R}}
\newcommand{\N}{\mathbb{N}}

\newcommand{\C}{\mathbb{C}}

\newcommand{\F}{\mathcal{F}}
\newcommand{\E}{\mathcal{E}}

\newcommand{\cS}{\mathcal{S}}

\newcommand{\PP}{\mathcal{P}}

\newcommand{\one}{{\ensuremath {\mathds 1} }}

\newcommand{\al}{\alpha}

\newcommand{\ep}{\varepsilon}

\newcommand{\Om}{\Omega}
\newcommand{\om}{\omega}

\newcommand{\half}{\tx{\frac{1}{2}}}

\newcommand{\supp}{\mathrm{supp}}
\newcommand{\wto}{\rightharpoonup}

\newcommand{\bral}{\left<}
\newcommand{\brar}{\right|}
\newcommand{\ketl}{\left|}
\newcommand{\ketr}{\right>}

\newcommand{\ZN}{\mathcal{Z}_{N}}

\newcommand{\MFf}{\F ^{\rm MF}}
\newcommand{\MFe}{F ^{\rm MF}}
\newcommand{\MFmin}{\mathcal{M} ^{\rm MF}}
\newcommand{\rhoMF}{\varrho ^{\rm MF}}

\newcommand{\MFEf}{\E ^{\rm MF}}
\newcommand{\MFEe}{E ^{\rm MF}}

\newcommand{\MFfo}{\E ^{\rm MF}}
\newcommand{\MFeo}{E ^{\rm MF}}
\newcommand{\MFfal}{\E ^{\rm MF}_{\al}}
\newcommand{\MFeal}{E ^{\rm MF}_{\al}}

\newcommand{\logal}{\log_{\al}}

\DeclareMathOperator{\Tr}{Tr}
\DeclareMathOperator{\tr}{Tr}

\def\geqslant{\ge}
\def\leqslant{\le}
\def\bq{\begin{eqnarray}}
\def\eq{\end{eqnarray}}
\def\bqq{\begin{eqnarray*}}
\def\eqq{\end{eqnarray*}}
\def\nn{\nonumber}

\def\eps{\varepsilon}
\def\wto{\rightharpoonup}
\def\gB {\mathfrak{B}}

\newcommand\1{{\ensuremath {\mathds 1} }}

\newcommand{\gammaP}{\gamma_{\Psi}}
\newcommand{\dM}{{\rm d}M}

\renewcommand{\epsilon}{\varepsilon}

\def\cF {\mathcal{F}}

\def\R {\mathbb{R}}
\def\C {\mathbb{C}}

\def\cS {\mathcal{S}}

\def\E {\mathcal{E}}
\def\cE {\mathcal{E}}

\def\F {\mathcal{F}}

\def\R {\mathbb{R}}
\def\C {\mathbb{C}}

\def\gS{\mathfrak{S}}

\def\cS {\mathcal{S}}

\def\E {\mathcal{E}}
\def\cM {\mathcal{M}}

\def\d{{\rm d}}

\def\gH{\mathfrak{H}}
\newcommand\ii{{\ensuremath {\infty}}}
\newcommand\pscal[1]{{\ensuremath{\left\langle #1 \right\rangle}}}

\renewcommand{\leq}{\leqslant}
\renewcommand{\geq}{\geqslant}

\newcommand{\cEH}{\ensuremath{\cE_{\text{\textnormal{H}}}}}

\newcommand{\EH}{\E_{\rm H}}
\newcommand{\eH}{e_{\rm H}}
\newcommand{\uH}{u_{\rm H}}

\newcommand{\ENLS}{\cE_{\rm nls}}
\newcommand{\eNLS}{e_{\rm nls}}
\newcommand{\uNLS}{u_{\rm nls}}

\newcommand{\mubf}{\boldsymbol{\mu}}
\newcommand{\nubf}{\boldsymbol{\nu}}
\newcommand{\mut}{\tilde{\boldsymbol{\mu}}}
\newcommand{\Gammat}{\tilde{\Gamma}}
\newcommand{\gammat}{\tilde{\gamma}}

\newcommand{\Pp}{P_{\perp}}
\newcommand{\wep}{w_{\varepsilon}}
\newcommand{\aep}{a_{\varepsilon}}
\newcommand{\Eep}{E ^{\varepsilon}}

\newcommand{\MNLS}{\cM_{\rm nls}}

\newcommand{\Fcl}{\F_{\rm cl}}
\newcommand{\Fcle}{F_{\rm cl}}
\newcommand{\mucl}{\mu_{\rm cl}}

\numberwithin{equation}{section}

\title[De Finetti theorems and Bose-Einstein condensation]{De Finetti theorems, mean-field limits and Bose-Einstein condensation}

\author[Nicolas Rougerie]{Nicolas Rougerie\\ 
L\lowercase{aboratoire de }P\lowercase{hysique et }M\lowercase{od\'elisation des }M\lowercase{ilieux} C\lowercase{ondens\'es}, U\lowercase{niversit\'e} G\lowercase{renoble 1} \& CNRS.}
\address{Universit\'e Grenoble 1 \& CNRS, LPMMC, UMR 5493, BP 166, 38042 Grenoble, France.}
\email{nicolas.rougerie@lpmmc.cnrs.fr}

\date{June 2015.}

\begin{document}

%


\maketitle

\bigskip

\begin{center}
Lecture notes for a course at the Ludwig-Maximilian Universit\"at, M\"unich in April 2015.\\
Translated from a ``Cours Peccot'' at the coll\`ege de France, February-April 2014. 
\end{center}

\bigskip

\begin{abstract}
These notes deal with the mean-field approximation for equilibrium states of $N$-body systems in classical and quantum statistical mechanics. A general strategy for the justification of effective models based on statistical independence assumptions is presented in details. The main tools are structure theorems \`a la de Finetti, describing the large $N$ limits of admissible states for these systems. These rely on the symmetry under exchange of particles, due to their indiscernability. Emphasis is put on quantum aspects, in particular the mean-field approximation for the ground states of large bosonic systems, in relation with the Bose-Einstein condensation phenomenon. Topics covered in details include: the structure of reduced density matrices for large bosonic systems, Fock-space localization methods, derivation of effective energy functionals of Hartree or non-linear Schr\"odinger type, starting from the many-body Schr\"odinger Hamiltonian.   
\end{abstract}

\newpage

$\phantom{o}$

\bigskip
\bigskip

\bigskip
\bigskip

\bigskip
\bigskip

\bigskip
\bigskip

\bigskip
\bigskip

\bigskip
\bigskip

\noindent\emph{Dedicated to my daughter C\'eleste. }

\newpage

$\phantom{o}$

\bigskip
\bigskip

\bigskip
\bigskip

\tableofcontents

\newpage

\section*{\textbf{Foreword}}\label{sec:pre intro}

The purpose of these notes is to present as exhaustively and pedagogically as possible some recent mathematical results bearing on the \emph{Bose-Einstein condensation} phenomenon observed in ultra-cold atomic gases. One of the numerous theoretical problems posed by these experiments is the understanding of the link between effective models, describing the physics with a remarkable precision, and the first principles of quantum mechanics. The process leading from the fundamental to the effective theories is often called a \emph{mean-field limit}, and it has motivated a very large number of investigations in theoretical and mathematical physics. In this course, we shall focus on one of the methods allowing to deal with mean-field limits, which is based on \emph{de Finetti theorems}. The emergence of the mean-field models will be interpreted as a fundamental consequence of the structure of physical states under consideration.  

This text will touch on subjects from mathematical analysis, probability theory, condensed matter physics, ultra-cold atoms physics, quantum statistical mechanics and quantum information. The emphasis will be on the author's speciality, namely the analytic aspects of the derivation of mean-field equilibrium states. The presentation will thus have a pronounced mathematical style, but readers should keep in mind the connection between the questions adressed here and cold atoms physics, in particular as regards the experiments leading to the observation of Bose-Einstein condensates in the mid 90's.

\subsection*{A few words about the experiments}

\smallskip

Bose-Einstein condensation (BEC) is at the heart of a rapidly expanding research field since the mid 90's. The extreme versatility of cold atoms experiments allows for a direct investigation of numerous questions of fundamental physics. For more thorough developments in this direction, I refer the reader to the literature, in particular to~\cite{Aftalion-06,BloDalZwe-08,DalGioPitStr-99,DalGerJuzOhb-11,LieSeiSolYng-05,PetSmi-01,PitStr-03,Fetter-09,Cooper-08} and references therein. French readers will find very accessible discussions in~\cite{Dalibard-01,CheDal-03,CohDalLal-05}.

\medskip

The first experimental realizations of BEC took place at the MIT and at Boulder, in the groups of W. Ketterle on the one hand and E. Cornell-C. Wieman on the other hand. The 2001 Nobel prize in physics was awarded to Cornell-Wieman-Ketterle for this remarkable achievement. The possibilities opened up by these experiments for investigating macroscopic quantum phenomena constitute a cornerstone of contemporary physics.

\medskip

A Bose-Einstein condensate is made of a large number of particles (alkali atoms usually) occupying the same quantum state. BEC thus requires that said particles be bosons, i.e. that they do not satisfy Pauli's exclusion principle which prevents mutliple occupancy of a single quantum state. 

This macroscopic occupancy of a unique low-energy quantum state appears only at very low temperatures. There exists a critical temperature $T_c$ for the existence of a condensate, and macroscopic occupancy occurs only for temperatures $T<T_c$. The existence of such a critical temperature was theoretically infered in works of Bose and Einstein~\cite{Bose-24,Einstein-24} in the 1920's. Important objections were however formulated:
\begin{enumerate}
\item The critical temperature $T_c$ is extremely low, unrealistically so as it seemed in the~1920's.
\item At such a temperature, all known materials should form a solid state, and not be gaseous as assumed in Bose and Einstein's papers. 
\item The argument of Bose and Einstein applies to an ideal gas, neglecting interactions between particles, which is a serious drawback.
\end{enumerate}

The first objection could be bypassed only in the 1990's with the advent of powerful techniques such as laser cooling\footnote{1997 Nobel prize in physics: S. Chu-W. Phillips-C. Cohen-Tannoudji.} and evaporative cooling. These allowed to reach temperatures in the micro-Kelvin range in quantum gases trapped by magneto-optics means. As for the second objection, the answer has to do with the diluteness of the samples: three-particles collisions necessary to initiate the formation of molecules, and ultimately of a solid phase, are extremely rare in the experiments. One thus has the possibility of observing a metastable gaseous phase during a sufficiently long time for a condensate to form.   

The third objection is of a more theoretical nature. Most of the material discussed in these notes is part of a research program (of many authors, see references) whose goal is to remove that objection. We will thus have the opportunity to discuss it a length in the sequel. 

Many agreeing observations have confirmed the experimental realization of BEC: imaging of the atoms' distribution in momentum/energy space, interference of condensates, superfluidity in trapped gases ... The importance thus acquired by the mathematical models used in the description of this phenomenon has motivated a vast literature devoted to their derivation and analysis. 

\subsection*{Some mathematical questions raised by experiments}

\smallskip

In the presence of BEC, the gas under consideration can be described by a single wave-functions $\psi:\R ^d \mapsto \C$, corresponding to the quantum state in which all particles reside. A system of $N$ quantum particles should normally be described by a $N$-particle wave-function $\Psi_N: \R ^{dN} \mapsto \C$. One thus has to understand how and why can this huge simplification be justified, that is how the collective behavior of the $N$ particles emerges. The investigation of the precision of this approximation, whose practical and theoretical consequences are fundamental, is a task of extreme importance for theoretical and mathematical physicists.  

One may ask the following questions:
\begin{enumerate}
\item Can one describe the ground state (that is the equlibrium state at zero temperature) of an interacting Bose system with a single wave-function~$\psi$~?
\item Start from a single wave-function and let the system evolve along the natural dynamics ($N$-body Schr\"odinger flow). Is the single wave-function description preserved by the dynamics ? 
\item Can one rigorously prove the existence of a critical temperature $T_c$ under which the finite temperature equilibrium states may be described by a single wave-function~?
\end{enumerate}

These questions are three aspects of the third objection mentioned in the preceding paragraph. We thus recall that the point is to understand the BEC phenomenon \emph{in the presence of interactions}. The case of an ideal gas is essentially trivial, at least for questions 1 and 2. 

One should keep in mind that, in the spirit of statistical mechanics, we aim at justifying the single wave-function description \emph{asymptotically in the limit of large particle numbers}, under appropriate assumptions on the model under consideration. Ideally, the assumptions should reduce to those ensuring that both the $N$-body model one starts from and the $1$-body model one arrives at are mathematically well-defined. Let us note that, for interacting quantum particles, the former is always linear, while the latter is always non-linear.

Many recent results presented in the sequel are generalizations to the quantum case of better known results of classical statistical mechanics. Related questions indeed occur also in this simpler context. For pedagogical reasons, some notions on mean-field limits in classical mechanics will thus be recalled in these notes.

\medskip

This course deals with question number 1, and we will be naturally lead to develop tools of intrinsic mathematical interest. One may use essentially two approaches:

\begin{itemize}
\item The first one exploits particular properties of certain important physical models. It thus applies differently to different models, and under often rather restrictive assumptions, in particular as regards the shape of the inter-particle interactions. A non-exhaustive list of references using such ideas is~\cite{SanSer-12,SanSer-13,RouSer-13,RouSerYng-13b,Serfaty-14} for the classical case, and~\cite{BenLie-83,LieYau-87,Seiringer-11,GreSei-12,SeiYngZag-12,LieSeiSolYng-05,LieSei-06} for the quantum case
\item The second one is the object of these lectures. It exploits properties of the set of admissible states, that is of $N$-body wave-functions $\Psi_N$. That we consider bosonic particles implies a fundamental symmetry property for these functions. This approach has the merit of being much more general than the first one, and in many cases to get pretty close to the ``ideally minimal assumptions'' for the validity of the mean-field approximation mentioned before.

A possible interpretation is to see the mean-field limit as a parameter regime where correlations between particles become negligible. We will use very strongly the key notion of bosonic symmetry. There will be many opportunities to discuss the literature in details, but let us mention immediately~\cite{MesSpo-82,CagLioMarPul-92,Kiessling-89,Kiessling-93,KieSpo-99,RouYng-14} and~\cite{FanSpoVer-80,FanVan-06,PetRagVer-89,RagWer-89,LewNamRou-13,LewNamRou-14} for applications of these ideas in classical and quantum mechanics respectively.  
\end{itemize}
The distinction between the two philosophies is of course somewhat artificial since one will often benefit from borrowing ideas to both, see e.g.~\cite{NamRouSei-15} for a recent example.

\medskip

To keep these notes reasonably short, question 2 will not be treated at all, although a vast literature exists, see e.g.~\cite{Hepp-74,GinVel-79,GinVel-79b,Spohn-80,BarGolMau-00,ElgErdSchYau-06,ElgSch-07,AmmNie-08,ErdSchYau-09,FroKnoSch-09,RodSch-09,BenOliSch-12,KnoPic-10,Pickl-11} and references therein, as well as the lecture notes~\cite{Golse-13,BenPorSch-15}. We note that the quantum de Finetti theorems that we will discuss in the sequel have recently proved useful in dealing with question 2, see~\cite{AmmNie-08,AmmNie-09,AmmNie-11,CheHaiPavSei-13,CheHaiPavSei-14}. The use of classical de Finetti theorems in a dynamic framework is older~\cite{Spohn-80,Spohn-81,Spohn-12}

\medskip

As for question 3, this is a famous open problem in mathematical physics. Very little is known at a satisfying level of mathematical rigor, but see however~\cite{SeiUel-09,BetUel-10}. We will touch on questions raised by taking temperature into account only very briefly in Appendix~\ref{sec:large T}, in a greatly simplified framework. This subject is further studied in the paper~\cite{LewNamRou-14b}.

\subsection*{Plan of the course}\mbox{}

\smallskip

These notes are organized as follows:

\begin{itemize}
\item A long introduction, Chapter~\ref{sec:intro}, recalls the basic formalism we shall need to give a precise formulation to the problems we are interested in. We will start with classical mechanics and then move on to quantum aspects. The question of justifying the mean-field approximation for equilibrium states of a given Hamiltonian will be formulated in both contexts. The proof stategy that will occupy the bulk of the notes will be described in a purely formal manner, so as to introduce as quickly as possible the de Finetti theorems that will be our main tools. Similarities between the classical and quantum frameworks are very strong. Differences show up essentially when discussing the proofs of the fundamental de Finetti theorems.  
\item Chapter~\ref{sec:class} is essentially independent from the rest of the notes. It contains the analysis of classical systems: proof of the classical de Finetti theorem (also called Hewitt-Savage theorem), application to equilibrium states of a classical Hamiltonian. The proof of the Hewitt-Savage theorem we will present, due to Diaconis and Freedman, is purely classical and does not generalize to the quantum case. 
\item We start adressing quantum aspects in Chapter~\ref{sec:quant}. Two versions (strong and weak) of the quantum de Finetti theorem are given without proofs, along with their direct applications to ``relatively simple'' bosonic systems in the mean-field regime. Section~\ref{sec:rel deF} contains a discussion of the various versions of the quantum de Finetti theorem and outlines the proof strategy that we shall follow.
\item Chapters~\ref{sec:deF finite dim} and~\ref{sec:locFock} contain the two main steps of the proof of the quantum de Finetti theorem we choosed to present: respectively ``explicit construction and estimates in finite dimension'' and ``generalization to infinite dimensions via localization in Fock space''. The proof should not be seen as a black box: not only the final result but also the intermediary constructions will be of use in the sequel.

\item Equiped with the results of the two previous chapters, we will be able to give in Chapter~\ref{sec:Hartree} the justification of the mean-field approximation for the ground state of an essentially generic bosonic system. Contrarily to the case treated in Chapter~\ref{sec:quant}, the quantum de Finetti theorem will not be sufficient in this case, and we will have to make use of some of the ingredients introduced in Chapter~\ref{sec:locFock}. 

\item The mean-field limit is not the only physically relevant one. In Chapter~\ref{sec:NLS} we will study a dilute regime in which the range of the interactions goes to $0$ when $N\to\infty$. In this case one obtains in the limit Gross-Pitaevskii (or non-linear Schr\"odinger) functionals with local non-linearities. We will present a strategy for the derivation of such objects based on the tools of Chapter~\ref{sec:deF finite dim}.  
\end{itemize}

The main body of the text is supplemented with two appendices containing each an unpublished note of Mathieu Lewin and the author.

\begin{itemize}
\item  Appendix~\ref{sec:class quant} shows how, in some particular cases, one may use the classical de Finetti theorem to deal with quantum problems. This strategy is less natural (and less efficient) than that presented in Chapters~\ref{sec:quant} and~\ref{sec:Hartree}, but it has a conceptual interest.   

\item Appendix~\ref{sec:large T} deviates from the main line of the course since the Hilbert spaces under consideration will be finite dimensional. In this context, combining a large temperature limit with a mean-field limit, one may obtain a theorem of semi-classical nature which gives examples of de Finetti measures not encountered previously. This will be the occasion to mention Berezin-Lieb inequalities and their link with the considerations of Chapter~\ref{sec:deF finite dim}.
\end{itemize}

\subsection*{Acknowledgements}

The motivation to write the french version of these notes came from the opportunity of presenting the material during a ``cours Peccot'' at the Coll\`ege de France in the spring of 2014. Warm thanks to the audience of these lectures for their interest and feedback, and to the Spartacus project (website: http://spartacus-idh.com) for their proposal to publish the french version.  

The motivation for translating the notes into english was provided by an invitation to give an (expanded version of) the course at the LMU-Munich in the spring of 2015. I thank in particular Thomas \O{}stergaard S\o{}rensen for the invitation and the practical organization of the lectures, as well as Heinz Siedentop, Martin Fraas, Sergey Morozov, Sven Bachmann and Peter M\"uller for making my stay in M\"unich very enjoyable. Feedback from the students following the course was also very useful in revising the notes to produce the current version.

I am of course indebted to my collaborators on the subjects of these notes: Mathieu Lewin and Phan Th\`anh Nam for the quantum aspects (more recently also Douglas Lundholm, Robert Seiringer and Dominique Spehner), Sylvia Serfaty and Jakob Yngvason for the classical aspects. Fruitful exchanges with colleagues, in particular Zied Ammari, Francis Nier, Patrick G\'erard, Isaac Kim, Jan-Philip Solovej, J\"urg Fr\"ohlich, Elliott Lieb, Eric Carlen, Antti Knowles, Benjamin Schlein, Maxime Hauray and Samir Salem were also very useful. 

Finally, I acknowledge financial support from the French ANR (Mathostaq project, ANR-13-JS01-0005-01)

\newpage

\section{\textbf{Introduction: Problems and Formalism}}\label{sec:intro}

Here we describe the mathematical objects that will be studied throughout the course. Our main object of interest will be many-body quantum mechanics, but the analogy with some questions of classical mechanics is instructive enough for us to also describe that formalism. Questions of units and dimensionality are systematically ignored to simplify notation.

\subsection{Statistical mechanics and mean-field approximation}\label{sec:forma class}

For pedagogical reasons we shall recall some notions on mean-field limits in classical mechanics before going to the quantum aspects, relevant to the BEC phenomenon. This paragraph aims at fixing notation and reviewing some basic concepts of statistical mechanics. We will limit ourselves to the description of the equilibrium states of a classical system. Dynamic aspects are voluntarily ignored, and the reader is refered to~\cite{Golse-13} for a review on these subjects.

\subsubsection*{Phase space.} The state of a classical particle is entirely determined by its position $x$ and its speed $v$ (or equivalently its momentum $p$). For a particle living in a domain $\Om\subset \R ^d$ we thus work in the phase space $\Om \times \R ^d$, the set of possible positions and momenta. For a $N$-particle system we work in $\Om ^N \times \R ^{dN}.$

\subsubsection*{Pure states.} We call a pure state one where the positions and momenta of all particles are known exactly. Equilibrium states at zero temperature for example are pure: in classical mechanics, uncertainty on the state of a system is only due to ``thermal noise''.

For a $N$-particle system, a pure state corresponds to a point 
$$(X;P) = (x_1,\ldots,x_N;p_1,\ldots,p_N) \in \Om ^N \times \R ^{dN}$$
in phase-space, where the pair $(x_i;p_i)$ gives the position and momentum of particle number~$i$. Having in mind the introduction of mixed states in the sequel, we will identify a pure state with a superposition of Dirac masses
\begin{equation}\label{eq:def etat pur class}
\mu_{X;P} = \sum_{\sigma \in \Sigma_N} \delta_{X_{\sigma};P_{\sigma}}. 
\end{equation}
This identification takes into account the fact that real particles are \emph{indistinguishable}. One can actually not attribute the pair $(x_i;p_i)$ of position/momentum to any one of the $N$ particles in particular, whence the sum over the permutation group $\Sigma_N$ in~\eqref{eq:def etat pur class}. Here and in the sequel our notation is
\begin{align}\label{eq:convention permutation}
X_{\sigma} &= (x_{\sigma(1)},\ldots,x_{\sigma(N)}) \nonumber
\\P_{\sigma} &= (p_{\sigma(1)},\ldots,p_{\sigma(N)}).
\end{align}
Saying that the system is in the state  $\mu_{X;P}$ means that one of the particles has position and momentum $(x_i;p_i),i=1 \ldots N$, but one cannot specify which because of indistinguishability.

\subsubsection*{Mixed states.} At non-zero temperature, that is when some thermal noise is present, one cannot determine with certainty the state of the system. One in fact looks for a statistical superposition of pure states, which corresponds to specifying the probability that the system is in a certain pure state. One then speaks of mixed states, which are the convex cominations of pure states, seen as Dirac masses as in Equation~\eqref{eq:def etat pur class}. The set of convex combinations of pure states of course corresponds to the set of symmetric probability measures on phase-space. A general $N$-particles mixed state is thus a probability measure $\mubf_N \in \PP_s (\Om ^{N}\times \R ^{dN})$ satisfying 
\begin{equation}\label{eq:class sym}
d\mubf_N (X;P) = d\mubf_N (X_\sigma;P_\sigma)
\end{equation}
for all permutations $\sigma\in \Sigma_N$. One interprets $\mubf_N (X;P)$ as the probability density that particle $i$ has position $x_i$ and momentum $p_i$, $i=1\ldots N$. Pure states of the form~\eqref{eq:def etat pur class} are a particular type of mixed states where the statistical uncertainty is reduced to zero, up to indistinguishability.

\subsubsection*{Free energy.} The energy of a classical system is specified by the choice of a Hamiltonian, a function over phase-space. In non-relativistic mechanics, the kinetic energy of a particle of momentum $p$ is always $m|p|^2 / 2$. Taking $m=1$ to simplify notation, we will consider an energy of the form    
\begin{equation}\label{eq:intro class hamil 1}
H_N (X;P) := \sum_{j=1} ^N \frac{|p_j| ^2}{2} + \sum_{j=1} ^N V(x_j) + \lambda \sum_{1\leq i<j\leq N} w (x_i-x_j) 
\end{equation}
where $V$ is an external potential (e.g. electrostatic) felt by all particles and $w$ a pair-interaction potential that we will assume symmetric,
$$ w (-x) = w (x).$$ 
The real parameter $\lambda$ sets the strength of interparticle interactions. We could of course add three-particles, four-particles etc ... interactions, but this is seldom required by the physics, and when it is, there is no additional conceptual difficulty.

The energy of a mixed state $\mubf_N \in \PP_s (\Om ^{N}\times \R ^{dN})$ is then given by 
\begin{equation}\label{eq:def ener func class}
\E [\mubf_N] := \int_{\Om ^N \times \R ^{dN}} H_N (X;P) d\mubf_N (X;P) 
\end{equation}
and by symmetry of the Hamiltonian this reduces to $H_N (X;P)$ for a pure state of the form~\eqref{eq:def etat pur class}. At zero temperature, equilibrium states are found by minimizing the energy functional~\eqref{eq:def ener func class}:
\begin{equation}\label{eq:def ener class}
E (N) = \inf\left\{  \E [\mubf_N], \mubf_N \in \PP_s (\Om ^{N}\times \R ^{dN})\right\rbrace 
\end{equation}
and the infimum (ground state energy) is of course equal to the minimum of the Hamiltonian $H_N$. It is attained by a pure state~\eqref{eq:def etat pur class} where $(X;P)$ is a minimum point for $H_N$ (in particular $P= (0,\ldots,0)$).

In presence of thermal noise, one must take the entropy
\begin{equation}\label{eq:def class entr}
S [\mubf_N] := - \int_{\Om ^N \times \R ^{dN}} d \mubf_N (X;P) \log (\mubf_N (X;P))  
\end{equation}
into account. This is a measure of the degree of uncertainty on the state of the system. Note for example that pure states have the lowest possible entropy: $S [\mubf_N] = -\infty$ if $\mubf_N$ if of the form~\eqref{eq:def etat pur class}. At temperature $T$, one finds the equilibrium state by minimizing the free-energy functional
\begin{align}\label{eq:def free ener func class}
\F [\mubf_N] &= \E [\mubf_N] - T S [\mubf_N]\nonumber 
\\&= \int_{\Om ^N \times \R ^{dN}} H_N (X;P) d\mubf_N (X;P) + T \int_{\Om ^N \times \R ^{dN}} d \mubf_N (X;P) \log (\mubf_N (X;P))
\end{align}
which amounts to saying that the more probable states must find a balance between having a low energy and having a large entropy. We shall denote
\begin{equation}\label{eq:def free ener class}
F (N) = \inf\left\{  \F [\mubf_N], \mubf_N \in \PP_s (\Om ^{N}\times \R ^{dN})\right\rbrace
\end{equation}
without specifying the temperature dependence. A minimizer must be a sufficiently regular probability so that (minus) the entropy is finite

\subsubsection*{Momentum minimization.} In the absence of a prescribed relation between the position and the momentum distribution of a classical state, the minimization in momentum variables of the above functionals is in fact trivial. A state minimizing ~\eqref{eq:def ener func class} is always of the form 
$$ \mubf_N = \delta_{P=0} \otimes \sum_{\sigma \in \Sigma_N} \delta_{X=X ^0 _\sigma}$$
where $X^0$ is a minimum point for $H_N(X;0,\ldots,0)$, i.e. particles are all at rest. We are thus reduced to looking for the minimum points  of $H_N (X;0,\ldots,0)$ as a function of $X$.

The minimization of~\eqref{eq:def free ener func class} leads to a Gaussian in momentum variables mutliplied by a Gibbs state in position variables\footnote{The \emph{partition functions} $Z_P$ and $Z_N$ normalize the state in $L ^1$.}
$$ \mubf_N = \frac{1}{Z_P} \exp \left(-\frac{1}{2T} \sum_{j=1} ^N |p_j| ^2\right) \otimes \frac{1}{Z_N} \exp \left(-\frac{1}{T} H_N (X;0,\ldots, 0)\right).$$
Momentum variables thus no longer intervene in the minimization of the functionals determining equilibrium states and they will be completely ignored in the sequel. We will keep the preceding notation for the minimization in position variables:
\begin{align}\label{eq:intro class hamil 2}
H_N (X) &= \sum_{j=1} ^N V(x_j) + \lambda \sum_{1\leq i<j\leq N} w (x_i-x_j) \nonumber\\
\E [\mubf_N] &= \int_{\Om ^N } H_N (X) d\mubf_N (X) \nonumber\\
\F [\mubf_N] &= \int_{\Om ^N } H_N (X) d\mubf_N (X) + T \int_{\Om ^N } d \mubf_N (X) \log (\mubf_N (X))
\end{align}
where $\mubf_N \in \PP_s (\Om ^N)$ is a symmetric probability measure in position variables only.

\subsubsection*{Marginals, reduced densities.} Given a $N$-particles mixed state, it is very useful to consider its marginals, or reduced densities, obtained by integrating out some variables:
\begin{equation}\label{eq:defi marginale}
\mubf_N ^{(n)} (x_1,\ldots,x_n) = \int_{\Om ^{N-n}} \mubf(x_1,\ldots,x_n,x'_{n+1},\ldots,x'_N) dx'_{n+1}\ldots dx'_N \in \PP_s (\Om ^n).
\end{equation}
The $n-$th reduced density $\mubf_N ^{(n)}$ is interpreted as the probability density for having one particle at $x_1$, one particle at $x_2$, etc... one particle at $x_n$. In view of the symmetry of $\mubf_N$, the choice of which $N-n$ variables over which one integrates in  Definition~\eqref{eq:defi marginale} is irrelevant.

A first use of these marginals is a rewriting of the energy using only the first two marginals\footnote{More generally, an energy depending only on a $n$-body potential may be rewritten by using the $n$-th marginal only.}:
\begin{align}\label{eq:intro ener marginales}
\E [\mubf_N] &= N \int_{\Om} V(x) d\mubf_N ^{(1)} (x) + \lambda \frac{N(N-1)}{2} \iint_{\Om \times \Om} w(x-y) d\mubf_N ^{(2)}(x,y)\nonumber\\
&= \iint_{\Om\times\Om} \left( \frac{N}{2} V (x) + \frac{N}{2} V (y) + \lambda \frac{N(N-1)}{2} w(x-y) \right)d\mubf_N ^{(2)}(x,y),
\end{align}
where we used the symmetry of the Hamiltonian.

\subsubsection*{Mean-field approximation.} Solving the above minimization minimization problems is a very difficult task when the particle number $N$ gets large. In order to obtain more tractable theories from which useful information can be extracted, one often has to rely on approximations. The simplest and most well-known of these is the mean-field approximation. One can introduce it in several ways, the goal being to obtain a self-consistent one-body problem starting from the $N$-body problem.

Here we follow the ``molecular chaos'' point of view on mean-field theory, originating in the works of Maxwell and Boltzmann. The approximation consists in assuming that all particles are independent and identically distributed~(iid). We thus take an ansatz of the form
\begin{equation}\label{eq:intro MF ansatz}
\mubf_N (x_1,\ldots,x_N) = \rho ^{\otimes N} (x_1,\ldots,x_N) = \prod_{j=1} ^N \rho(x_j) 
\end{equation}
where $\rho \in \PP (\Om)$ is a one-body probability density describing the typical behavior of one of the iid particles under consideration.

The mean-field energy and free-energy functionals are obtained by inserting this ansatz in~\eqref{eq:def ener func class} or~\eqref{eq:def free ener func class}. The mean-field energy functional is thus
\begin{equation}\label{eq:intro MF ener func class}
\MFEf [\rho] = N ^{-1} \E [\rho ^{\otimes N}] = \int_{\Om} V(x)d\rho(x) + \lambda \frac{N-1}{2} \iint_{\Om\times \Om} w(x-y) d\rho(x)d\rho(y). 
\end{equation}
We shall denote $\MFEe$ its infimum amongst probability measures. In a similar manner, the mean-field free energy functional is given by
\begin{align}\label{eq:intro MF free ener class}
\MFf [\rho] &= N ^{-1} \F [\rho ^{\otimes N}] \nonumber\\
&= \int_{\Om} V(x)d\rho(x) + \lambda \frac{N-1}{2} \iint_{\Om\times \Om} w(x-y) d\rho(x)d\rho(y) + T \int_{\Om} \rho \log \rho 
\end{align}
and its infimum shall be denoted $\MFe$. The term ``mean-field'' is motivated by the fact that~\eqref{eq:intro MF ener func class} corresponds to having an interaction between the particles' density $\rho$ and the self-consistant potential 
$$\rho \ast w = \int_{\Om} w(.-y) d\rho(y)$$ 
whose gradient is the so-called mean-field. 

\subsection{Quantum mechanics and Bose-Einstein condensation}\label{sec:forma quant}

After this classical intemezzo we can now introduce the quantum objects that are the main theme of this course. We shall be content with a rapid review of the basic principles of quantum mechanics. Other ``mathematician-friendly'' presentations may be found in~\cite{LieSei-09,Solovej-notes}. Our discussion of the relevant concepts is in places voluntarily simplified.

\subsubsection*{Wave-functions and quantum kinetic energy.} One of the basic postulates of quantum mechanics is the identification of pure states of a system with normalized vectors of a complex Hilbert space $\gH$. For particles living in the configuration space $\R ^d$, the relevant Hilbert space is $L^2 (\R ^d)$, the space of complex square-integrable functions over $\R ^d$. 

Given a particle in the state $\psi \in L ^2 (\R ^d)$, one identifies $|\psi| ^2$ with a probability density: $|\psi (x)| ^2$ gives the probability for the particle to be located at $x$. We thus impose the normalization
$$ \int_{\R ^d} |\psi| ^2 =1. $$
We thus see that, even in the case of a pure state, one cannot specify with certainty the position of the particle. More precisely, \emph{one cannot simultaneously specify its position and its momentum.} This \emph{uncertainty principle} is the direct consequence of another fundamental postulate: $|\hat{\psi}| ^2$ gives the momentum-space probability density of the particle, where $\hat{\psi}$ is the Fourier transform of $\psi$. 

In quantum mechanics, the (non-relativistic) kinetic energy of a particle is thus given by
\begin{equation}\label{eq:ener cin quant}
\int_{\R ^d } \frac{|p| ^2}{2} |\hat{\psi} (p)| ^2 dp = \int_{\R ^d } \frac{1}{2} |\nabla \psi (x)| ^2 dx.  
\end{equation}
That the position and momentum of a particle cannot be simultaneously specified is a consequence of the fact that it is impossible for both $|\psi| ^2$ and $|\hat{\psi}| ^2$ to converge to a Dirac mass. A popular way of quantifying this is Heisenberg's uncertainty principle: for all  ~$x_0\in \R ^d$ 
$$ \left( \int_{\R ^d }|\nabla \psi (x)| ^2 dx \right)\left( \int_{\R ^d } |x-x_0| ^2 |\psi (x)| ^2 dx \right) \geq C.$$
Indeed, the more precisely the particle's position is known, the smaller the second term of the left-hand side (for a certain $x_0$). The first term of the left-hand side must then be very large, which, in view of ~\eqref{eq:ener cin quant} rules out the possibility for the momentum distribution to be concentrated around a single $p_0 \in \R ^d$.

For many applications however (see~\cite{Lieb-76} for a discussion of this point), this inequality is not sufficient. A better way of quantifying the uncertainty principle is given by Sobolev's inequality (here in its 3D version):
$$\int_{\R ^3 } |\nabla \psi | ^2  \geq C \left( \int_{\R ^3} |\psi| ^6  \right) ^{1/3}.$$ 
If the position of the particle is known with precision, $|\psi| ^2$ must approach a Dirac mass, and the right-hand side of the above inequality blows up. So do the integrals~\eqref{eq:ener cin quant}, with the same interpretation as previously.

\subsubsection*{Bosons and Fermions.} For a system of $N$ quantum particles living in $\R ^d$, the appropriate Hilbert space is  $L ^2 (\R ^{dN}) \simeq \bigotimes ^N L ^2 (\R ^d)$. A pure state is thus a certain $\Psi \in L ^2 (\R ^{dN})$ where $|\Psi (x_1,\ldots,x_N)| ^2$ is interpreted as the probability density for having particle $1$ at $x_1$, ..., particle $N$ at $x_N$. As in classical mechanics, the indistinguishability of the particles imposes
\begin{equation}\label{eq:indis quantique moins}
|\Psi (X)| ^2 = |\Psi (X_{\sigma})| ^2 
\end{equation}
for any permutation $\sigma \in \Sigma_N$. This condition is necessary for indistinguishable patricles, but it is in fact not sufficient for quantum particles. To introduce the correct notion, we denote  $U_{\sigma}$ the unitary operator interchanging particles according to $\sigma \in \Sigma_N$:
$$ U_{\sigma} \: u_1 \otimes \ldots \otimes u_N = u_{\sigma (1)}\otimes \ldots \otimes u_{\sigma (N)}$$
for all $u_1,\ldots,u_N \in L ^2 (\R ^d)$, extended by linearity to $L ^2 (\R ^{dN}) \simeq \bigotimes ^N L ^2 (\R  ^d)$ (one may construct a basis with vectors of the form $u_1 \otimes \ldots \otimes u_N$). For $\Psi \in L ^2 (\R ^{dN})$ to describe indistinguishable particles we have to require that  
\begin{equation}\label{eq:indis quantique}
\left\langle \Psi, A \Psi \right\rangle_{L ^2 (\R ^{dN})} = \left\langle U_{\sigma} \Psi, A U_{\sigma} \Psi \right\rangle_{L ^2 (\R ^{dN})} 
\end{equation}
for any bounded operator $A$ acting on $L ^2 (\R ^{dN})$. Details would lead us to far, but suffice it to say that condition~\eqref{eq:indis quantique} corresponds to asking that any measure (corresponding to an observable $A$) on the system must be independent of the particles' labeling. In classical mechanics, all possible measurements correspond to bounded functions on phase-space and thus~\eqref{eq:class sym} alone guarantees the invariance of observations under particle exchanges. In quantum mechanics, observables include all bounded operators on the ambient Hilbert space, and one must thus impose the stronger condition~\eqref{eq:indis quantique}.

An important consequence\footnote{This is an easy but non-trivial exercise.} of the symmetry condition~\eqref{eq:indis quantique} is that a system of indistinguishable quantum particles must satisfy one of the following conditions, stronger than~\eqref{eq:indis quantique moins}: either 
\begin{equation}\label{eq:intro bosons}
\Psi (X) = \Psi(X_\sigma) 
\end{equation}
for all $X\in \R ^{dN}$ and $\sigma \in \Sigma_N$, or 
\begin{equation}\label{eq:intro fermions}
\Psi (X) = \epsilon (\sigma) \Psi(X_\sigma) 
\end{equation}
for all $X\in \R ^{dN}$ and $\sigma \in \Sigma_N$, where $\eps (\sigma)$ is the sign of the permutation $\sigma$. One refers to particles described by a wave-function satisfying~\eqref{eq:intro bosons} (respectively~\eqref{eq:intro fermions}) as bosons (respectively fermions). These two types of fundamental particles have a very different behavior, one speaks of bosonic and fermionic statistics. For example, fermions obey the \emph{Pauli exclusion principle} which stipulates that two fermions cannot simultaneously occupy the same quantum state. One can already see that ~\eqref{eq:intro fermions} imposes 
$$ \Psi (x_1,\ldots,x_i,\ldots,x_j,\ldots,x_N) =  -\Psi (x_1,\ldots,x_j,\ldots,x_i,\ldots,x_N)$$
and thus (formally)
$$ \Psi (x_1,\ldots,x_i,\ldots,x_i,\ldots,x_N) = 0$$ 
which implies that it is impossible for two fermions to occupy the same position~$x_i$. One may consult~\cite{LieSei-09} for a discussion of Lieb-Thirring inequalities, which are one of the most important consequences of Pauli's principle.

Concretely, when studying a quantum system made of only one type of particles, one has to restrict admissible pure states to those being either of bosonic of fermionic type. One thus works 
\begin{itemize}
\item for bosons, in $L^2_s (\R ^{dN}) \simeq \bigotimes_s ^N L ^2 (\R ^d)$, the space of symmetric square-integrable wave-functions, identified with the symmetric tensor product of $N$ copies of $L ^2 (\R ^d)$. 
\item for fermions, in $L^2_{as} (\R ^{dN}) \simeq \bigotimes_{as} ^N L ^2 (\R ^d)$, the space of anti-symmetric square-integrable wave-functions, identified with the anti-symmetric tensor product of $N$ copies of $L ^2 (\R ^d)$.
\end{itemize}
As the name indicates, Bose-Einstein condensation can occur only in a bosonic system, and this course shall thus focus on the first case.

\subsubsection*{Density matrices, mixed states.} We will always identify a pure state  $\Psi \in L ^2 (\R ^{dN})$ with the corresponding density matrix, i.e. the orthogonal projector onto  $\Psi$, denoted $\ketl \Psi \ketr \bral \Psi \brar$. Just as in classical mechanics, mixed states of a system are by definition the statistical superpositions of pure states, that is the convex combinations of orthogonal projectors. Using the spectral theorem, it is clear that the set of mixed states conci\"ides with that of positive self-adjoint operators having trace $1$:
\begin{equation}\label{eq:intro etats mixtes}
\cS (L ^2 (\R ^{dN})) = \left\lbrace \Gamma \in \gS ^1 (L ^2 (\R ^{dN})), \Gamma = \Gamma ^*, \Gamma \geq 0, \tr \Gamma = 1 \right\rbrace
\end{equation}
where $\gS ^1 (\gH)$ is the Schatten-von Neumann class~\cite{ReeSim4,Schatten-60,Simon-79} of trace-class operators on a Hilbert space $\gH$. To obtain the bosonic (fermionic) mixed states one considers respectively
$$ \cS (L _{s/as} ^2 (\R ^{dN})) = \left\lbrace \Gamma \in \gS ^1 (L_{s/as} ^2 (\R ^{dN})), \Gamma = \Gamma ^*, \Gamma \geq 0, \tr \Gamma = 1 \right\rbrace.$$
Note that in the vocabulary of density matrices, bosonic symmetry consists in imposing 
\begin{equation}\label{eq:intro bos sym mat} 
U_{\sigma} \Gamma = \Gamma 
\end{equation}
whereas fermionic symmetry corresponds to 
$$ U_{\sigma} \Gamma = \eps (\sigma) \Gamma.$$
One sometimes considers the weaker symmetry notion
\begin{equation}\label{eq:quant sym gen}
 U_{\sigma} \Gamma U_{\sigma} ^* = \Gamma 
\end{equation}
which is for example satisfied by the (unphysical) superposition of a bosonic and a fermionic state.

\subsubsection*{Energy functionals.} The quantum energy functional corresponding to the classical non-relativistic Hamiltonian~\eqref{eq:intro class hamil 1} is obtained by the substitution
\begin{equation}\label{eq:quant mom}
p \leftrightarrow -i\nabla, 
\end{equation}
consistent with the identification~\eqref{eq:ener cin quant} for the kinetic energy of a quantum particle. The quantized Hamiltonian is a (unbounded) operator on $L ^2 (\R ^{dN})$:
\begin{equation}\label{eq:intro quant hamil}
H_N = \sum_{j=1} ^N \left( - \frac12 \Delta_j + V (x_j) \right) + \lambda \sum_{1\leq i<j \leq N} w(x_i - x_j)  
\end{equation}
where $-\Delta_j = (-i \nabla_j) ^2$ corresponds to the Laplacian in the variable $x_j \in \R ^d$. The corresponding energy for a pure state $\Psi \in L ^2 (\R ^{dN})$ is 
\begin{equation}\label{eq:intro quant ener pur}
\E [\Psi] = \langle \Psi, H_N \Psi \rangle_{L ^2 (\R ^{dN})}.  
\end{equation}
By linearity this generalizes to 
\begin{equation}\label{eq:intro quant ener mixte}
\E [\Gamma] = \tr_{L ^2 (\R ^{dN})} [ H_N \Gamma ].   
\end{equation}
in the case of a mixed state  $\Gamma \in \cS (L ^2 (\R ^{dN}))$. At zero temperature, the equilibrium state of the system is found by minimizing the above energy functional. In view of the linearity of~\eqref{eq:intro quant ener mixte} as a function of $\Gamma$ and by the  spectral theorem, it is clear that one may restrict the minimization to pure states:
\begin{align}\label{eq:def ener quant}
E_{s/as}(N) = \inf\left\{ \E [\Gamma], \Gamma \in \cS (L ^2 _{s/as} (\R ^{dN})) \right\}\nonumber\\
&= \inf\left\{ \E [\Psi], \Psi \in L_{s/as} ^2 (\R ^{dN}), \norm{\Psi}_{L ^2 (\R ^{dN})} = 1 \right\}. 
\end{align}
Here we use again the index $s$ (respectively $as$) to denoted the bosonic (respectively fermionic) energy. In the absence of an index, we mean that the minimization is performed without symmetry constraint. However, because of the symmetry of the  Hamiltonian, one can in this case restrict the minimization to mixed states satisfying~\eqref{eq:quant sym gen}, or to wave-functions satisfying~\eqref{eq:indis quantique moins}. 

In the presence of thermal noise, one must add to the energy a term including the Von-Neumann entropy
\begin{equation}\label{eq:Von Neu entropie}
S[\Gamma] = - \tr_{L ^2 (\R ^{dN})} [\Gamma \log \Gamma] = - \sum_{j} a_j \log a_j 
\end{equation}
where the $a_j$'s are the eigenvalues (real and positive) of $\Gamma$, whose existence is guaranteed by the spectral theorem. Similarly to the classical entropy~\eqref{eq:def class entr}, Von-Neumann's entropy is minimized ($S[\Gamma] = 0$) by pure states, i.e. orthongonal projectors, which of course only have one non-zero eigenvalue, equal to $1$.

The free-energy functional at temperature $T$ is then 
\begin{equation}\label{eq:def free ener func quant}
\F [\Gamma] = \E [\Gamma] - T S [\Gamma] 
\end{equation}
and minimizers are in general mixed states.

\subsubsection*{Alternative forms for the kinetic energy: magnetic fields and relativistic effects.} In the classical case, we have only introduced the simplest possible form of kinetic energy. The reason is that for the minimization problems that shall concern us, the kinetic energy plays no role. Other choices for the relation between $p$ and the kinetic energy\footnote{Other \emph{dispersion relations.}} are nevertheless possible. In quantum mechanics the choice of this relation is crucial even at equilibrium because of the non-trivial minimization in momentum variables. Recalling~\eqref{eq:ener cin quant} and~\eqref{eq:quant mom}, we also see that in the quantum context, different choices will lead to different functional spaces in which to set the problem.

Apart from the non-relativistic quantum energy already introduced, at least two generalizations are physically interesting:

\begin{itemize}
\item When a magnetic field $B:\R ^d \mapsto \R $ interacts with the particles, one replaces 
\begin{equation}\label{eq:quant mom magn}
p \leftrightarrow -i\nabla + A 
\end{equation}
where $A:\R ^d \mapsto \R ^d$ is the \emph{vector potential}, satisfying 
$$ B = \mathrm{curl}\: A.$$
Of course $B$ determines $A$ only up to a gradient. The choice of a particular $A$ is called a \emph{gauge choice}. The kinetic energy operator, taking the Lorentz force into account, becomes in this case 
$$(p+A) ^2 = \left( -i\nabla + A \right) ^2 = -\left( \nabla - i A\right) ^2,$$
called a \emph{magnetic Laplacian}. 

This formalism is also appropriate for particles in a rotating frame: calling $x_3$ the rotation axis one must then take $A = \Om (-x_2,x_1,0)$ with $\Om$ the rotation frequency. This corresponds to taking the Coriolis force into account. In this case one must also replace the potential $V(x)$ by $V(x) - \Om ^2 |x| ^2$ to account for the centrifugal force.

\smallskip

\item When one wishes to take relativistic effects into account, the dispersion relation becomes 
$$ \mbox{ Kinetic energy }  = c\sqrt{p ^2 + m ^2 c^2}  - m c ^2$$
with $m$ the mass and $c$ the speed of light in vacuum. Choosing units so that $c=1$ and recalling ~\eqref{eq:quant mom}, one is lead to consider the kinetic energy operator
\begin{equation}\label{eq:rel kin ener}
\sqrt{p ^2 + m ^2} - m = \sqrt{-\Delta + m ^2} - m 
\end{equation}
which is easily defined in Fourier variables for example. In the non-relativitic limit where $|p|\ll m = mc $ we formally recover the operator  $-\Delta$ to leading order. A caricature of~\eqref{eq:rel kin ener} is sometimes used, corresponding to the ``extreme relativistic'' case $|p|\gg mc$ where one takes
\begin{equation}\label{eq:rel kin ener ext}
\sqrt{p ^2} = \sqrt{-\Delta}
\end{equation}
as kinetic energy operator.
\item One can of course combine the two generalizations to consider relativistic particles in a magnetic field, using the operators
$$\sqrt{(-i\nabla + A) ^2 + m ^2 } - m \mbox{ and } \left|-i\nabla +A\right|$$
based on the relativistic dispersion relation and the correspondance~\eqref{eq:quant mom magn}.
\end{itemize}

\subsubsection*{Reduced density matrices.} It will be useful to define in the quantum context a concept similar to the reduced densities of a classical state. Given $\Gamma \in \cS (L ^2 (\R ^{dN}))$, one defines its $n$-th reduced density matrix by taking a partial trace on the last $N-n$ particles:
\begin{equation}\label{eq:intro def mat red}
\Gamma ^{(n)}= \tr_{n+1 \to N} \Gamma, 
\end{equation}
which precisely means that for any bounded operator $A_n$ acting on $L ^2 (\R ^{dn})$, 
$$ \tr_{L ^2 (\R ^{dn})} [A_n \Gamma ^{(n)}] := \tr_{L ^2 (\R ^{dN})} [A_n \otimes \one^{\otimes N-n} \Gamma] $$
where $\one$ is the identity on $L ^2 (\R ^d)$. The above definition is easily generalized to any Hilbert space, but in the case of $L^2$, the reader is maybe more familiar with the following equivalent definition: We identify  $\Gamma$ with its kernel, i.e. the function $\Gamma (X;Y)$ such that for all $\Psi \in L ^2 (\R ^{dN})$
$$ \Gamma\, \Psi = \int_{\R ^{dN}} \Gamma (X;Y) \Psi (Y) dY.$$ 
One can then also identify $\Gamma ^{(n)}$ with its kernel
\begin{multline*}
 \Gamma ^{(n)} (x_1,\ldots,x_n;y_1,\ldots,y_n) \\= \int_{\R ^{d(N-n)}} \Gamma (x_1,\ldots,x_n,z_{n+1},\ldots,z_{N};y_1,\ldots,y_n,z_{n+1},\ldots,z_{N})dz_{n+1}\ldots dz_N. 
\end{multline*}
In the case where $\Gamma$ has a symmetry, bosonic or fermionic, the choice of the variables over which we take the partial trace is arbitrary. Let us note that even if one starts from a pure state, the reduced density matrices are in general mixed states.

Just as in~\eqref{eq:intro ener marginales} one may rewrite the energy~\eqref{eq:intro quant ener mixte} in the form 
\begin{align}\label{eq:quant ener red mat}
\E [\Gamma] &= N \tr_{L ^2 (\R ^d)} \left[\left( - \half \Delta + V \right) \Gamma ^{(1)} \right] + \lambda \frac{N(N-1)}{2} \tr_{L ^2 (\R ^{2d})} [w (x-y) \Gamma ^{(2)}]\nonumber \\
&= \tr_{L ^2 (\R ^{2d})} \left[ \left(\frac{N}{2} \left( - \half \Delta_x + V (x) - \half \Delta_y + V (y) \right) + \lambda \frac{N(N-1)}{2} w (x-y) \right) \Gamma ^{(2)}\right].  
\end{align}

\subsubsection*{Mean-field approximation.} Even more than so in classical mechanics, actually solving the minimization problem~\eqref{eq:def ener quant} for large $N$ is way too costly, and it is necessary to rely on approximations. This course aims at studying the simplest of those, which consists in imitating~\eqref{eq:intro MF ansatz} by taking an ansatz for iid particles,
\begin{equation}\label{eq:intro quant ansatz}
\Psi (x_1,\ldots,x_N) = u ^{\otimes N} (x_1,\ldots, x_N) = u(x_1)\ldots u(x_N) 
\end{equation}
for a certain $u\in L ^2 (\R ^d).$ Inserting this in the energy functional~\eqref{eq:quant ener red mat} we obtain the Hartree functional
\begin{align}\label{eq:intro def Hartree func}
\EH [u] &= N ^{-1} \E[ u ^{\otimes N}] \nonumber \\
&= \int_{\R ^d } \left( \frac{1}{2}|\nabla u| ^2 + V |u| ^2 \right) + \lambda\frac{N-1}{2} \iint_{\R ^d \times \R ^d} |u(x)| ^2 w(x-y) |u(y)| ^2 
\end{align}
with the corresponding minimization problem
\begin{equation}\label{eq:intro def Hartree ener}
\eH = \inf\left\{ \EH [u], \norm{u}_{L ^2 (\R ^d)} = 1 \right\}. 
\end{equation}
Note that we have transformed a linear problem for the $N$-body wave-function (since the energy functional is quadratic, the variational equation is linear) into a cubic problem for the one-body wave-function $u$ (the energy functional is quartic and hence the variational equation will be cubic).

Ansatz~\eqref{eq:intro quant ansatz} is a symmetric wave-function, appropriate for bosons. It corresponds to looking for the ground state in the form of a Bose-Einstein condensate where all particles are in the state $u$. Because of Pauli's principle, fermions can in fact never be completely uncorrelated in the sense of~\eqref{eq:intro quant ansatz}. The mean-field ansatz for fermions rather consists in taking
$$ \Psi (x_1,\ldots,x_N) = \det \left(u_j (x_k)\right)_{1\leq j,k\leq N}$$
with orthonormal functions $u_1,\ldots,u_N$ (orbitals of the system). This ansatz (Slater determinant) leads to the Hartree-Fock functional that we shall not discuss in these notes (a presentation in the same spirit can be found in~\cite{Rougerie-XEDP})

Note that for classical particles, the mean-field ansatz~\eqref{eq:intro MF ansatz} allows taking a mixed state~$\rho$. To describe a bosonic system, one always takes a pure state $u$ in the ansatz~\eqref{eq:intro quant ansatz}, which calls for the following remarks:
\begin{itemize}
\item If one takes a general $\gamma \in \cS (L ^2 (\R ^d))$, the $N$-body state defined as 
\begin{equation}\label{eq:ansatz MF sans sym}
 \Gamma = \gamma ^{\otimes N} 
\end{equation}
has the bosonic symmetry~\eqref{eq:intro bos sym mat} if and only if $\gamma$ is a pure state (see e.g.~\cite{HudMoo-75}), $\gamma = | u \rangle \langle u |$, in which case $\Gamma = | u ^{\otimes N}\rangle \langle u  ^{\otimes N}|$ and we are back to Ansatz~\eqref{eq:intro quant ansatz}.

\item In the case of the energy functional~\eqref{eq:quant ener red mat}, the minimization problems with and without bosonic symmetry imposed co\"incide (see e.g.~\cite[Chapter 3]{LieSei-09}). One can then minimize with no constraint and find the bosonic energy. This remains true when the kinetic energy is~\eqref{eq:rel kin ener} or~\eqref{eq:rel kin ener ext} but is notoriously wrong in the presence of a magnetic field.

\item For some energy functionals, for example in presence of a magnetic or rotation field, the minimum without bosonic symmetry is strictly lower than the minimum with symmetry, cf~\cite{Seiringer-03}. The ansatz~\eqref{eq:ansatz MF sans sym} is then appropriate to approximate the problem without bosonic symmetry (in view of the Hamiltonian's symmetry, one may always assume the weaker symmetry~\eqref{eq:quant sym gen}) and one then obtains a generalized Hartree functional for mixed one-body states $\gamma \in \cS (L ^2 (\R ^d))$:
\begin{equation}\label{eq:intro Hartree func gen}
\EH [\gamma] = \tr _{L ^2 (\R^d )} \left[\left( -\half \Delta + V \right)\gamma \right] + \lambda \frac{N-1}{2} \tr_{L ^2 (\R ^{2d})} \left[ w(x-y) \gamma ^{\otimes 2}\right].
\end{equation}
\end{itemize}

\vspace{0.5cm}

The next two sections introduce, respectively in the classical and quantum setting, the question that will occupy us in the rest of the course:

\medskip
\begin{center}
\emph{Can one justify, in a certain limit, the validity of the mean-field ans\"atze~\eqref{eq:intro MF ansatz} and~\eqref{eq:intro quant ansatz}  to describe equilibrium states of a system of indistinguishable particles ?} 
\end{center}

\subsection{Mean-field approximation and the classical de Finetti theorem}\label{sec:MF deF class}\mbox{}\\\vspace{-0.4cm}

Is it legitimate to use the simplification~\eqref{eq:intro MF ansatz} to determine the equilibrium states of a classical system ? Experiments show this is the case when the number of particles is large, which we mathematically implement by considering the limit $N\to \infty$ of the problem at hand. 

A simple framework in which one can justify the validity of the mean-field \emph{approximation} is the so-called mean-fied \emph{limit}, where one assumes that all terms in the energy~\eqref{eq:def ener func class} have comparable weights. In view of~\eqref{eq:intro ener marginales}, this is fulfilled if  $\lambda$ scales as $N ^{-1}$, for example 
\begin{equation}\label{eq:intro lambda}
\lambda = \frac{1}{N-1}, 
\end{equation}
in which case one may expect the ground state energy per particle $N ^{-1} E(N)$ to have a well-defined limit. The particular choice~\eqref{eq:intro lambda} helps in simplifying certain expressions, but the following considerations apply as soon as $\lambda$ is of order $N ^{-1}$. We should insist on the simplification we introduced in looking at the  $N\to \infty$ limit under the assumption~\eqref{eq:intro lambda}. This regime is far from covering all the physically relevant cases. It however already leads to an already very non-trivial and instructive problem.

The goal of this section is to dicuss the mean-field approximation for the zero-temperature equilibrium states of a classical system. We are thus looking at
\begin{equation}\label{eq:intro MF lim class}
E(N) = \inf \left\{\int_{\Om ^N} H_N d\mubf_N,\, \mubf_N \in \PP_s (\Om ^N)\right\} 
\end{equation}
with
\begin{equation}\label{eq:intro MF lim hamil}
H_N (X) = \sum_{j=1} ^N V(x_j) + \frac{1}{N-1} \sum_{1\leq i < j \leq N} w(x_i-x_j).
\end{equation}
We will sketch a formal proof of the validity of the mean-field approximation at the level of the ground state energy, i.e. we shall argue that
\begin{equation}\label{eq:intro justif MF class}
\boxed{ \lim_{N\to \infty} \frac{E(N)}{N} = \MFEe = \inf\left\{ \MFEf [\rho], \rho \in \PP (\Om) \right\}}
\end{equation}
where the mean-field functional is obtained as in~\eqref{eq:intro MF ener func class}, taking~\eqref{eq:intro lambda} into account:
\begin{equation}\label{eq:intro MF func lim}
\MFEf[\rho] = \int_{\Om} V d\rho + \frac12 \iint_{\Om \times \Om} w(x-y) d\rho(x)d\rho(y). 
\end{equation}
Estimate~\eqref{eq:intro justif MF class} for the ground state energy is a first step, and the proof scheme in fact yields information on the equilibrium states themselves. To simplify the presentation, we shall postpone the discussion of this aspect, as well as the study of the positive temperature case, to Chapter~\ref{sec:class}.

\subsubsection*{Formally passing to the limit.} Here we wish to make the ``algebraic" structure of the problem apparent. The justification of the manipulations we will perform requires analysis assumptions that we shall discuss in the sequel of the course, but that we ignore for the moment.

We start by using~\eqref{eq:intro ener marginales} and~\eqref{eq:intro lambda} to write 
$$ \frac{E(N)}{N} = \frac{1}{2} \inf\left\{ \iint_{\Om \times \Om} H_2 (x,y) d\mubf_N ^{(2)}(x,y), \, \mubf_N \in \PP_s (\Om ^N) \right\} $$
where $H_2$ is the two-body Hamiltonian defined as in~\eqref{eq:intro MF lim hamil} and $\mubf_N ^{(2)}$ is the second marginal of the symmetric probability $\mubf_N$. Since the energy depends only on the second marginal, one may see the problem we are interested in as a constrained optimization problem for two-body symmetric probability measures: 
$$ \frac{E(N)}{N} = \frac{1}{2} \inf\left\{ \iint_{\Om \times \Om} H_2 (x,y) d\mubf ^{(2)}(x,y), \mubf ^{(2)} \in \PP_N  ^{(2)} \right\} $$
with
$$\PP_N  ^{(2)} = \left\{ \mubf ^{(2)} \in \PP_s (\Om ^2) \,|\, \exists \mubf_N \in \PP_s (\Om ^N), \mubf ^{(2)} = \mubf_N ^{(2)} \right\} $$
the set of two-body probability measures one can obtain as the marginals of a $N$-body state. Assuming that the limit exists (first formal argument) we thus obtain
\begin{equation}\label{eq:MF class form 1}
 \lim_{N\to \infty} \frac{E(N)}{N} = \frac{1}{2} \lim_{N\to \infty}\, \inf\left\{ \iint_{\Om \times \Om} H_2 (x,y) d\mubf ^{(2)}(x,y), \mubf ^{(2)} \in \PP_N  ^{(2)} \right\} 
\end{equation}
and it is very tempting to exchange limit and infimum in this expression (second formal argument) in order to obtain
$$ \lim_{N\to \infty} \frac{E(N)}{N} = \frac{1}{2} \inf \left\{ \iint_{\Om \times \Om} H_2 (x,y) d\mubf ^{(2)}(x,y), \mubf ^{(2)} \in \lim_{N\to \infty} \PP_N  ^{(2)} \right\}. $$
We have here observed that the energy functional appearing in ~\eqref{eq:MF class form 1} is actually independent of $N$. All the $N$-dependence lies in the constrained variational space we consider. This suggests that a natural limit problem consists in minimizing the same functional but on the limit of the variational space, as written above.

To give a meaning to the limit of $ \PP_N  ^{(2)}$, we observe that the sets $\PP_N  ^{(2)}$ form a decreasing sequence
$$ \PP_{N+1} ^{(2)} \subset \PP_N ^{(2)}$$
as is easily shown by noting that if $\mubf ^{(2)}\in \PP_{N+1} ^{(2)}$, then for a certain~$\mubf_{N+1} \in \PP_s (\Om ^{N+1})$
\begin{equation}\label{eq:MF class form obs}
\mubf ^{(2)} = \mubf_{N+1} ^{(2)} = \left(\mubf_{N+1} ^{(N)}\right) ^{(2)} 
\end{equation}
and of course $\mubf_{N+1} ^{(N)}\in \PP_s (\Om ^N)$. One can thus legitimately identify 
$$ \PP_{\infty} ^{(2)} := \lim_{N\to \infty} \PP_N  ^{(2)} = \bigcap_{N\geq 2} \PP_N ^{(2)}$$
and the natural limit problem is then (modulo the justification of the above formal manipulations)
\begin{equation}\label{eq:MF class form 2}
\lim_{N\to \infty} \frac{E(N)}{N}  = E_{\infty} = \frac{1}{2} \inf \left\{ \iint_{\Om \times \Om} H_2 (x,y) d\mubf ^{(2)}(x,y), \mubf ^{(2)} \in  \PP_\infty  ^{(2)} \right\}. 
\end{equation}
This is a variational problem over the set of two-body states that one can obtain as reduced densities of $N$-body states, for all $N$.

\subsubsection*{A structure theorem.} We now explain that actually
$$ E_{\infty}  = \MFEe, $$
as a consequence of a fundamental result on the structure of the space $\PP_\infty  ^{(2)}$.  

Let us take a closer look at this space. Of course it contains the product states of the form $\rho \otimes \rho,\rho \in \PP (\Om)$ since $\rho ^{\otimes 2}$ is the second marginal of $\rho ^{\otimes N}$ for all $N$. By convexity $\PP_\infty  ^{(2)}$ also contains all the convex combinations of product states: 
\begin{equation}\label{eq:MF class form 3}
\left\{ \int_{\rho \in \PP (\Om)} \rho ^{\otimes 2}dP (\rho), P\in \PP (\PP(\Om))\right\} \subset \PP_\infty  ^{(2)}
\end{equation}
with $\PP (\PP(\Om))$ the set of probability measures over $\PP (\Om)$. In view of~\eqref{eq:intro ener marginales} and~\eqref{eq:intro MF func lim} we will have justified the mean-field approximation~\eqref{eq:intro justif MF class} if one can show that the infimum in~\eqref{eq:MF class form 2} is attained for $\mubf ^{(2)} = \rho ^{\otimes 2}$ for a certain $\rho \in \PP (\Om)$.

The structure result that allows us to reach this conclusion is the observation that there is in fact equality in~\eqref{eq:MF class form 3}: 
\begin{equation}\label{eq:MF class form 4}
\left\{ \int_{\rho \in \PP (\Om)} \rho ^{\otimes 2}dP (\rho), P\in \PP (\PP(\Om))\right\} =\PP_\infty  ^{(2)}.
\end{equation}
Indeed, in view of the the linearity of the energy functional as a function of $\mubf ^{(2)},$ one may write 
\begin{align*}
E_{\infty} &= \frac12 \inf \left\{ \int_{\rho \in \PP (\Om)} \iint_{\Om \times \Om} H_2 (x,y) d\rho ^{\otimes 2}(x,y)dP(\rho), P \in \PP (\PP(\Om)) \right\}\nonumber\\
&= \inf \left\{ \int_{\rho \in \PP (\Om)} \MFEf [\rho] dP(\rho), P \in \PP (\PP(\Om)) \right\}\nonumber\\
&= \MFEe
\end{align*}
since it is clear that the infimum over $P\in \PP (\PP(\Om))$ is attained for $P = \delta_{\rhoMF}$, a Dirac mass at $\rhoMF$, a minimizer of $\MFEf$.

We thus see that, if we accept some formal manipulations (that we will justify in Chapter~\ref{sec:class}) to arrive at~\eqref{eq:MF class form 2}, the validity of the mean-field approximation follows by using very little of the properties of the Hamiltonian, but a lot the structure of symmetric $N$-body states, in the form of the equality~\eqref{eq:MF class form 4}.

The latter is a consequence of the Hewitt-Savage, or classical de Finetti, theorem~\cite{DeFinetti-31,DeFinetti-37,HewSav-55}, recalled in Section~\ref{sec:HS} and proved in Section~\ref{sec:DF}. Let us give a few details for the familiarized reader. We have to show the reverse inclusion in~\eqref{eq:MF class form 3}. We thus pick some $\mubf ^{(2)}$ satisfying  
$$\mubf ^{(2)} = \mubf_N ^{(2)} $$
for a certain sequence $\mubf_N \in \PP_s (\Om ^N)$. Recalling~\eqref{eq:MF class form obs} we may assume that 
$$ \mubf_{N+1} ^{(N)} = \mubf_N$$
and we thus have to deal with a sequence (hierarchy) of consistant $N$-body states. There then exists (Kolmogorov's extension theorem) a symmetric probability measure over the sequences of~$\Om$, $\mubf \in \PP_s (\Om ^{\N})$ such that 
$$\mubf_N = \mubf ^{(N)}$$
where the $N$-th marginal is defined as in~\eqref{eq:defi marginale}. The Hewitt-Savage theorem~\cite{HewSav-55} then ensures the existence of a unique probability measure $P\in \PP (\PP (\Om))$ such that 
$$\mubf_N = \int_{\rho \in \PP (\Om)} \rho ^{\otimes N} dP (\rho). $$
We obtain the desired result by taking the second marginal.

\medskip 

Chapter~\ref{sec:class} contains the details of the above proof scheme. We will specify the adequate assumptions to put all these considerations on a rigorours basis. It is however worth to note immediately that this proof (inspired from~\cite{MesSpo-82,Kiessling-89,Kiessling-93,CagLioMarPul-92,KieSpo-99}) used no structure property of the Hamiltonian (e.g. the interactions can be attractive or repulsive or a mixture of both) but only compactness and regularity properties.

\subsection{Bose-Einstein condensation and the quantum de Finetti theorem}\label{sec:MF deF quant}\mbox{}\\\vspace{-0.4cm}

We now sketch a strategy for justifying the mean-field approximation at the level of the ground state energy of a large bosonic system. The method is similar to that for the classical case presented above. We consider a mean-field regime by setting $\lambda = (N-1) ^{-1}$ and thus work with the Hamiltonian
\begin{equation}\label{eq:MF quant hamil}
H_N =\sum_{j=1} ^N -\left(\nabla_j + i A (x_j)\right) ^2+ V(x_j) + \frac{1}{N-1}\sum_{1\leq i<j\leq N} w(x_i-x_j)  
\end{equation}
acting on $L ^2_s (\R^{dN})$. Compared to the previous situation, we have added a vector potential $A$ to start emphasizing the generality of the approach. It can for example correspond to an external magnetic field $B = \rm{curl} \,A$.

The starting point is the bosonic ground state energy defined as previously
\begin{align}\label{eq:MF quant E(N)}
E (N) &= \inf\left\{ \tr_{L ^2 (\R ^{dN})} \left[ H_N \Gamma_N \right], \,\Gamma_N \in \cS (L ^2_s (\R ^{dN}))  \right\} \nonumber\\
&= \inf\left\{ \left\langle \Psi_N, H_N \Psi_N\right\rangle, \, \Psi_N \in L ^2_s (\R ^{dN}), \norm{\Psi_N} _{L ^2 (\R ^{dN})}  = 1  \right\}
\end{align}
where we recall that we can minimize over pure or mixed states indifferently.

Our goal is to show that for large $N$ the bosonic energy per particle can be calculated using Hartree's functional
\begin{equation}\label{eq:MF quant lim}
\boxed{\lim_{N\to \infty}  \frac{E(N)}{N} = \eH }
\end{equation}
where
\begin{align}\label{eq:MF quant Hartree}
\eH &= \inf \left\{  \EH[u], u \in L ^2 (\R ^d), \norm{u}_{L^2 (\R ^d )} = 1 \right\}\nonumber\\
\EH [u] &= \int_{\R ^d} \left(|(\nabla +i A) u| ^2 + V |u| ^2\right) + \frac12 \iint_{\R ^d \times \R ^d} |u(x)| ^2 w(x-y) |u(y)| ^2 dxdy.
\end{align}
Since Hartree's energy is obtained by inserting an anzatz $\Psi_N = u ^{\otimes N}$ in the $N$-body energy functional, the asymptotic result~\eqref{eq:MF quant lim} is already a strong indication in favor of the existence of BEC in the ground state of bosonic system, in the mean-field regime. We will come back later to the consequences of~\eqref{eq:MF quant lim} for minimizers. As in the classical case, the mean-field regime is a very simplified, but already very instructive, model. We will present in Chapter~\ref{sec:NLS} an anlysis of other physically relevant regimes.

\subsubsection*{Formally passing to the limit.} The first step to obtain~\eqref{eq:MF quant lim} is as in the classical case to reduce formally to a simplified limit problem. We start by rewriting the energy using~\eqref{eq:quant ener red mat} and the assumption $\lambda = (N-1) ^{-1}$
\begin{equation}\label{eq:MF quant form 1}
 \frac{E(N)}{N} = \frac12 \inf\left\{ \tr_{L ^2 (\R^{2d})} \left[ H_2 \: \Gamma ^{(2)}\right], \, \Gamma ^{(2)} \in \PP_N ^{(2)}\right\} 
\end{equation}
where
$$ \PP_N ^{(2)} = \left\{ \Gamma ^{(2)} \in \cS (L_s ^2(\R ^{2d})) \,|\, \exists \, \Gamma_N \in \cS (L_s ^2(\R ^{dN})), \Gamma ^{(2)} = \tr_{3\to N}[ \Gamma_N ]\right\}$$
is the set of  ''$N$-representable" two-body density matrices that are the partial trace of a $N$-body state.  

Characterizing the set $\PP_N ^{(2)}$ is a famous problem in quantum mechanics, see e.g.~\cite{ColYuk-00,LieSei-09}, and the formulation ~\eqref{eq:MF quant form 1} is thus not particularly useful at fixed $N$. However, as in the classical case, the representability problem can be given a satisfactory answer in the limit $N\to \infty$, and we shall rely on this fact. We start by noticing that, taking partial traces, we easily obtain 
\begin{equation}\label{eq:MF quant form 1bis}
\PP_{N+1} ^{(2)}\subset \PP_N ^{(2)}
\end{equation}
so that~\eqref{eq:MF quant form 1} can be seen as a variational problem set in a variational set that gets more and more constrained when $N$ increases. 

Assuming again that one can exchange infimum and limit (which is of course purely formal) we obtain
\begin{equation}\label{eq:MF quant form 2}
\lim_{N\to \infty} \frac{E(N)}{N} = E_{\infty} := \frac12 \inf\left\{ \tr_{L ^2 (\R^{2d})} \left[ H_2 \: \Gamma ^{(2)}\right], \Gamma ^{(2)} \in \PP_\infty ^{(2)} \right\} 
\end{equation}
with
\begin{equation}\label{eq:MF quant form 3}
\PP_\infty ^{(2)} = \lim_{N\to \infty} \PP_N ^{(2)} = \bigcap_{N\geq 2} \PP_N ^{(2)} 
\end{equation}
the set of two-body density matrices that are $N$-representable for all $N$. As previously we have used the fact that the energy functional does not depend on $N$ to pass to the limit formally. 

\subsubsection*{A structure theorem.} It so happens that the structure of the set $\PP_\infty ^{(2)}$ is entirely known and implies the equality
\begin{equation}\label{eq:MF quant form 4}
E_{\infty} = \eH,
\end{equation}
which concludes the proof of~\eqref{eq:MF quant lim}, up to the justification of the formal manipulations we have just performed.

The structure property leading to~\eqref{eq:MF quant form 4} is a quantum version of the de Finetti-Hewitt-Savage theorem mentioned in the prevous section. Pick a $\Gamma ^{(2)}\in \PP_\infty ^{(2)}$. We then have a sequence $\Gamma_N \in \cS (L ^2_s (\R ^{dN}))$ such that for all $N$
$$ \Gamma ^{(2)} = \Gamma_N ^{(2)}.$$
Without loss of generality one can assume that this sequence is consistent in the sense that 
$$\Gamma_{N+1} ^{(N)} = \tr_{N+1} [\Gamma_{N+1}] = \Gamma_N. $$
The quantum de Finetti theorem of St\o{}rmer-Hudson-Moody~\cite{Stormer-69,HudMoo-75} then guarantees the existence of a probability measure $P \in \PP (S L ^{2} (\R ^d))$ on the unit sphere of $L^2 (\R ^d)$ such that
$$ \Gamma_N = \int_{u\in SL ^{2} (\R ^d)} |u ^{\otimes N}\rangle \langle u ^{\otimes N} |dP (u)$$
and thus 
$$ \Gamma ^{(2)} = \int_{u\in SL ^{2} (\R ^d)} |u ^{\otimes 2}\rangle \langle u ^{\otimes 2} |dP (u).$$
We can then conclude that
\begin{align*}
 E_\infty &= \inf\left\{ \frac12 \int_{u\in SL ^{2} (\R ^d)} \tr_{L ^2 (\R ^{2d})} [H_2 |u ^{\otimes 2}\rangle \langle u ^{\otimes 2} |]dP(u), P \in \PP (S L ^{2} (\R ^d)) \right\} \\ 
 &= \inf\left\{ \int_{u\in SL ^{2} (\R ^d)} \EH [u]dP(u), P \in \PP (S L ^{2} (\R ^d)) \right\} \\ 
&=\eH
 \end{align*}
where the last equality holds because it is clearly optimal to pick $P= \delta_{\uH}$ with $\uH$ a minimizer of Hartree's functional.

We thus see that the validity of~\eqref{eq:MF quant lim} is (at least formally) a consequence of the structure of the set of bosonic states and only marginally depends on properties of the Hamiltonian~\eqref{eq:MF quant hamil}. The justification (under some assumptions) of the formal manipulations we just performed to arrive at~\eqref{eq:MF quant form 2} (or a variant), as well as the proof the quantum de Finetti theorem (along with generalizations and corollaries) are the main purpose of these notes.

\subsubsection*{Bose-Einstein condensation.} Let us anticipate a bit on the implications of~\eqref{eq:MF quant lim}. We shall see in the sequel that results of the sort
\begin{equation}\label{eq:MF quant BEC 1}
\Gamma_N ^{(n)} \to \int_{u\in S L ^2 (\R^d)} |u ^{\otimes n}\rangle \langle u ^{\otimes n}| dP (u) 
\end{equation}
for all fixed $n\in \N$ when $N\to\infty$, follow very naturally in good cases. Here $\Gamma_N$ is (the density matrix of) a minimizer of the $N$-body energy and $P$ a probability measure concentrated on minimizers of $\EH$. The convergence will take place (still in good cases) in trace-class norm.

When there is a unique (up to a constant phase) minimizer $\uH$ for $\EH$ we thus get 
\begin{equation}\label{eq:MF quant BEC 2}
\Gamma_N ^{(n)} \to |\uH ^{\otimes n}\rangle \langle \uH ^{\otimes n}|  \mbox{ when } N\to \infty,
\end{equation}
which proves BEC at the level of the ground state. Indeed, BEC means \emph{by definition} (see~\cite{LieSeiSolYng-05} and references therein, in particular~\cite{PenOns-56}) the existence of an eigenvalue of order\footnote{Let us recall that all our density matrices are normalized to have trace $1$ in this course.} $1$ in the limit $N\to\infty$ in the spectrum of $\Gamma_N ^{(1)}$. This is clearly implied by~\eqref{eq:MF quant BEC 2}, which is in fact stronger.  

One can certainly wonder whether stronger results than~\eqref{eq:MF quant BEC 2} may be obtained. One could for example think of an approximation in norm like 
$$ \norm{ \Psi_N - \uH ^{\otimes N} }_{L ^2 (\R ^{dN})}\to 0.$$
Let us mention immediately that results of this kind are \emph{wrong} in general. The good condensation notion is formulated using density matrices, as the following two remarks show:
\begin{itemize}
\item Let us think of a $\Psi_N$ of the form ($\otimes_s$ stands for the symmetric tensor product) 
 $$ \Psi_N = \uH ^{\otimes (N-1)} \otimes_s \varphi$$ 
where $\varphi$ is orthogonal to $\uH$. Such a state is ``almost condensed'' since all particles but one are in the state $\uH$. However, in the  usual $L^2 (\R ^ {dN})$ sense,  $\Psi_N$ is orthogonal to $\uH ^{\otimes N}$. The two states thus cannot be close in norm, although their $n$-body density matrices for $n\ll N$ will be close.

\item Following the same line of ideas, it is natural to look for corrections to the $N$-body minimizer under the form
$$ \Psi_N = \varphi_0 \, \uH ^{\otimes N} + \uH ^{\otimes (N-1)} \otimes_s \varphi_1 + \uH ^{\otimes (N-2)} \otimes_s \varphi_2 + \ldots + \varphi_N$$
with $\varphi_0 \in \C$ and $\varphi_k \in L_s ^2 (\R ^{dk})$ for $k=1\ldots N$. It so happens that the above ansatz is correct if the sequence $(\varphi_k)_{k = 0,\ldots,N}$ is chosen to minimize a certain effective Hamiltonian on Fock space. The occurence of non-condensed terms implying the $\varphi_k$ for $k\geq 1$ contributes to leading order to the norm of $\Psi_N$, but not to the reduced density matrices, which rigorously confirms that the correct notion of condensation is necessarily based on reduced density matrices. This remark only scratches the surface of a beautiful subject that we will not discuss here: Bogoliubov's theory. We refer the interested reader to~\cite{CorDerZin-09,LieSol-01,LieSol-04,Solovej-06,Seiringer-11,GreSei-12,LewNamSerSol-13,NamSei-14,DerNap-13} for recent mathematical results in this direction.
\end{itemize}

\subsubsection*{A remark concerning symmetry.} We have been focusing on the bosonic problem because of our main physical motivations. The fermionic problem is (at least at present) not covered by such considerations, but one might be interested in the problem without symmetry constraint mentioned previously.

In the case where $A\equiv 0$ in~\eqref{eq:MF quant hamil}, the problems with and without bosonic symmetry co\"incide, but it is not the case in general. One may nevertheless use the same method as above to study the problem without symmetry constraint. Indeed, because of the symmetry of the Hamiltonian, one may without loss of generality impose the weaker symmetry notion~\eqref{eq:quant sym gen} in the minimization problem. The set of two-body density matrices appearing in the limit is also covered by the St\o{}rmer-Hudson-Moody theorem, and one then has to minimize an energy amongst two-body density matrices of the form
\begin{equation}\label{eq:intro deF gen} 
\Gamma ^{(2)} = \int_{\gamma \in \cS (L ^2 (\R ^d))} \gamma ^{\otimes 2} dP(\gamma) 
\end{equation}
where $P$ is now a probability measure on mixed one-particle states, that is on positive self-adjoint operators over  $L ^ 2 (\R ^d)$ having trace $1$. One deduces in this case the convergence of the ground state energy per particle to the minimum of the generalized Hartree functional (cf~\eqref{eq:intro Hartree func gen}) 
$$
\EH [\gamma] = \tr _{L ^2 (\R^d )} \left[\left( -(\nabla + i A) ^2 + V \right)\gamma \right] + \lambda \frac{N-1}{2} \tr_{L ^2 (\R ^{2d})} \left[ w(x-y) \gamma ^{\otimes 2}\right]$$
and the minimum is in general different from the minimum of Hartree's functional~\eqref{eq:MF quant Hartree} (see e.g.~\cite{Seiringer-03}).

We already recalled that $\gamma ^{\otimes 2}$ can have bosonic symmetry only if $\gamma $ is a pure state $\gamma = |u  \rangle \langle u|$. it is thus clear why the problem without symmetry can in general lead to a strictly lower energy. In fact the minimizer of~\eqref{eq:intro Hartree func gen} when there is no magnetic field is always attained at a pure state, in coherence with the different observations we already made on symmetry issues. We finally note that taking bosonic symmetry into account in the mean-field limit is done completely naturally using the different versions of the quantum de Finetti theorem.

\newpage

\section{\textbf{Equilibrium statistical mechanics}}\label{sec:class}

In this chapter we shall prove the classical de Finetti theorem informally discussed above and present applications to problems in classical mechanics.  

To simplify the exposition and with a view to applications, we will consider particles living in a domain $\Omega \subset \R ^d$ which could be $\R ^d $ itself. We denote $\Omega ^N$ and $\Omega ^{\N}$ the cartesian product of $N$ copies of $\Omega$ and the set of sequences of $\Omega$ respectively. The space of probability measures over a set $\Lambda$ will always be denoted $\PP (\Lambda)$. In places we may make the simplifying assumption that $\Om$ is compact, in which case $\PP (\Omega)$ is compact for the weak convergence of measures.

\subsection{The Hewitt-Savage Theorem}\label{sec:HS}\mbox{}\\\vspace{-0.4cm}

We mention only in passing the first works on what is now known as de Finetti's theorem~\cite{DeFinetti-31,DeFinetti-37,Khintchine-32,Dynkin-53}. 
In this course we shall start from~\cite{HewSav-55} where the classical de Finetti theorem is proved in its most general form. Our point of view on the de Finetti theorem is here resolutely analytic in that we shall only deal with sequences of probability measures. More probabilistic versions of the theorem exist. We refer the reader to~\cite{Aldous-85,Kallenberg-05} and references therein for developments in this direction. 

Informally, the Hewitt-Savage theorem~\cite{HewSav-55} says that every symmetric probability measure over $\Omega ^N$ approaches a convex combination of product probability measures when $N$ gets large. A \emph{symmetric probability measures} $\mubf_N $ has to satisfy
\begin{equation}\label{eq:sym class N}
\mubf_N (A_1 \times \ldots \times A_N) = \mubf_N (A_{\sigma(1)},\ldots,A_{\sigma(N)}) 
\end{equation}
for every borelian domains $A_1,\ldots,A_N \subset \Omega$ and every permutation of $N$ indices $\sigma \in \Sigma_N$. We denote $\PP_s (\Om ^N)$ the set of such probability measures. A \emph{product measure} built on $\rho \in \PP (\Om)$ is of the form
\begin{equation}\label{eq:proba produit}
\rho ^{\otimes N} (A_1,\ldots,A_N) = \rho (A_1) \ldots \rho (A_N)  
\end{equation}
and is of course symmetric. We are thus looking for a result looking like
\begin{equation}\label{eq:class deF formel}
\mubf_N \approx \int_{\rho \in \PP (\Om)} \rho ^{\otimes N} dP_{\mubf_N} (\rho) \mbox{ when } N\to \infty
\end{equation}
where $P_{\mubf_N} \in \PP(\PP (\Om))$ is a probability measure over probability measures. 

\medskip

A first possibility for giving a meaning to~\eqref{eq:class deF formel} consists in taking immediately $N = \infty$, that is, consider a probability with infinitely many variables $\mubf \in \PP (\Om ^{\N})$ instead of $\mubf_N$, which lives on  $\Om ^N$. This is the natural meaning one should give to a ``classical state of a system with infinitely many particles''. We assume a notion of symmetry inherited from~\eqref{eq:sym class N}: 
\begin{equation}\label{eq:sym class inf}
\mubf (A_1,A_2,\ldots) =  \mubf (A_{\sigma(1)},A_{\sigma(2)},\ldots) 
\end{equation}
for every sequence of borelian domains $(A_k)_{k\in \N} \subset \Omega ^{\N}$ and every permutation of infinitely many indices $\sigma \in \Sigma_{\infty}$. We denote $\PP_s (\Om ^{\N})$ the set of probability measures over $\Om ^{\N}$ satisfying~\eqref{eq:sym class inf}. Hewitt and Savage~\cite{HewSav-55} proved the following:

\begin{theorem}[Hewitt-Savage 1955]\label{thm:HS}\mbox{}\\
Let $\mubf \in \PP_s (\Om ^{\N})$ satisfy~\eqref{eq:sym class inf}. Let $\mubf ^{(n)}$ be its $n$-th marginal, 
\begin{equation}\label{eq:HS marginale n}
\mubf ^{(n)} (A_1,\ldots,A_n) = \mubf (A_1,\ldots,A_n,\Omega,\ldots,\Omega,\ldots). 
\end{equation}
There exists a unique probability measure $P_{\mubf} \in \PP (\PP (\Omega))$ such that 
\begin{equation}\label{eq:result HS}
\mubf ^{(n)} =  \int_{\rho \in \PP (\Om)} \rho ^{\otimes n} dP_{\mubf} (\rho). 
\end{equation}
\end{theorem}

In statistical mechanics, this theorem is applied as a weak version of the formal approximation~\eqref{eq:class deF formel} in the following manner. We start from a classical $N$-particle state $\mubf_N \in \PP (\Omega ^N)$ whose marginals 
\begin{equation}\label{eq:class marginale N}
\mubf_N ^{(n)} (A_1,\ldots,A_n) = \mubf (A_1,\ldots,A_n,\Omega ^{N-n})
\end{equation}
converge\footnote{We denote this convergence $\wto_\ast$ to distinguish it from a norm convergence.} as measures in $\PP (\Om ^n)$, up to a (non-relabeled) subsequence:
\begin{equation}\label{eq:weak proba convergence}
\mubf_N ^{(n)} \wto_{\ast} \mubf ^{(n)} \in \PP (\Om ^{(n)}). 
\end{equation}
This means that 
\[
\mubf_N ^{(n)} (A_n) \to \mubf ^{(n)} (A_n) 
\]
for all borelian subset $A_n$ of $\Omega ^n$, and thus  
\[
\int_{\Omega} f_n d\mubf_N ^{(n)} \to  \int_{\Omega} f_n d\mubf ^{(n)}
\]
for all bounded continuous functions decaying at infinity $f_n$ from $\Om ^n$ to $\R$. Modulo a diagonal extraction argument one may always assume that the convergence~\eqref{eq:weak proba convergence} is along the same subsequence for any $n\in \N$. Testing~\eqref{eq:weak proba convergence} with a Borelian $A_n = A_m \times \Om ^{m-n}$ for $m\leq n$ one obtains the consistency relation
\begin{equation}\label{eq:consist class}
\left( \mubf ^{(n)} \right)  ^{(m)} = \mubf ^{(m)},\mbox{ for all } m \leq n
\end{equation}
which implies that $(\mubf ^{(n)})_{n\in \N}$ does describe a system with infinitely many particles. One may then see (Kolmogorov's extension theorem) that there exists $\mubf \in \PP (\Om ^{\N})$ such that $\mubf ^{(n)}$ is its $n$-th marginal (whence the notation). This measure satisfies~\eqref{eq:sym class inf} and we can thus apply Theorem~\ref{thm:HS} to obtain 
\begin{equation}\label{eq:class deF rig}
\mubf_N ^{(n)} \wto_* \int_{\rho \in \PP (\Om)} \rho ^{\otimes n} dP_{\mubf} (\rho) \mbox{ when } N\to \infty 
\end{equation}
where $P_{\mubf} \in \PP (\PP (\Om))$. This is a weak but rigorous version of~\eqref{eq:class deF formel}. In words \textbf{the $n$-th marginal of a classical $N$-particle state may be approximated by a convex combination of product states when $N$ gets large and $n$ is fixed .} Note that the measure $P_{\mubf}$ appearing in~\eqref{eq:class deF rig} does not depend on $n$. 

Of course, one must first be able to use a compactness argument in order to obtain~\eqref{eq:weak proba convergence}. This is possible if $\Om$ is compact (then $\PP (\Om ^n)$ is compact for the convergence in the sense of measures~\eqref{eq:weak proba convergence}). More generally, if the physical problem we are interested in has a confining mechanism, one may show that the marginals of equilibrium states of the system are tight and deduce~\eqref{eq:weak proba convergence}.

\medskip

As for the proof of Theorem~\ref{thm:HS}, there are several possible approaches. We mention first that the uniqueness part is a rather simple consequence of a density argument in $C_b (\PP (\Om))$, the bounded continuous functions over $\PP(\Om)$ due to Pierre-Louis Lions~\cite{Lions-CdF}. 

\begin{proof}[Proof of Theorem~\ref{thm:HS}, Uniqueness]
We easily verify that monomials of the form
\begin{equation}\label{eq:monomes PLL}
C_b (\PP(\Om)) \ni M_{k,\phi} (\rho) := \int \phi (x_1,\ldots,x_k) d\rho ^{\otimes k} (x_1,\ldots,x_k), k\in\N, \phi \in C_b (\Om ^k) 
\end{equation}
generate a dense sub-algebra of $C_b(\PP(\Om))$, the space of bounded continuous functions over $\PP(\Om)$, see Section 1.7.3 in~\cite{Golse-13}. The main point is that for all $\mu \in \PP (\Om)$
\[
M_{k,\phi} (\mu) M_{\ell,\psi} (\mu) = M_{k+\ell,\phi\:\otimes\:\psi} (\mu). 
\]
That this sub-algbra is dense is a consequence of the Stone-Weierstrass theorem.

We thus need only check that if there exists two measures $P_{\mubf}$ and $P'_{\mubf}$ satisfying~\eqref{eq:HS marginale n}, then 
\[
\int_{\rho\in \PP (\Om)} M_{k,\phi} (\rho) d P_{\mubf} (\rho) =  \int_{\rho\in \PP (\Om)} M_{k,\phi} (\rho) d P'_{\mubf} (\rho)
\]
for all $k\in \N$ and $\phi \in C_b (\PP (\Om ^k))$. But this last equation simply means that
\begin{multline*}
\int_{\rho\in \PP (\Om)} \left(\int_{\Om ^k} \phi(x_1,\ldots,x_k) d\rho ^{\otimes k} (x_1,\ldots,x_k) \right) d P_{\mubf} (\rho) 
\\ =  \int_{\rho\in \PP (\Om)}  \left(\int_{\Om ^k} \phi(x_1,\ldots,x_k) d\rho ^{\otimes k} (x_1,\ldots,x_k) \right) d P'_{\mubf} (\rho) 
\end{multline*}
which is obvious since both expressions are equal to 
\[
\int_{\Om ^k}  \phi(x_1,\ldots,x_k) d\mubf ^{(k)} (x_1,\ldots,x_k)  
\]
by assumption.
\end{proof}

For the existence of the measure, which is the most remarkable point, we mention three approaches
\begin{itemize}
\item The original proof of Hewitt-Savage is geometrical: the set of symmetric probability measures over $\Om ^{\N}$ is of course convex. The Choquet-Krein-Milman theorem says that any point of a convex set is a convex combination of the extremal points. It thus suffices to show that the extremal points of $\PP_s (\Om ^{\N})$ co\"incide with product measures, corresponding to sequence  of marginals of the form $(\rho ^{\otimes n})_{n\in \N}$. This approach is not constructive, and the proof that the sequences $(\rho ^{\otimes n})_{n\in \N}$ are the extremal points of $\PP_s (\Om ^{\N})$ is by contradiction. 
\item An entirely constructive approach is due to Diaconis-Freedman~\cite{DiaFre-80}. In this probabilistic argument, Theorem~\ref{thm:HS} is a corollary of an approximation result at finite $N$, giving a quantitative version of~\eqref{eq:class deF rig}.
\item Lions developed a new approach for the needs of mean-field game theory~\cite{Lions-CdF}. This is a dual point of view where one starts from ``weakly dependent continuous functions'' with many variables. A summary of this is in the lecture notes~\cite[Section 1.7.3]{Golse-13}, and a thorough presentation in~\cite[Chapter I]{Mischler-11}. Developments and generalizations may be found in~\cite{HauMis-14}.
\end{itemize}

It so happens that the proof of the Hewitt-Savage theorem following Lions' point of view is for a large part a rediscovery of the method of Diaconis and Freedman. In the sequel we will follow a blend of the two appraoches. See~\cite{Mischler-11} for a more complete discussion.

\subsection{The Diaconis-Freedman theorem}\label{sec:DF}\mbox{}\\\vspace{-0.4cm}

As we just announced, it is in fact possible to give a quantitative version of~\eqref{eq:class deF rig} which implies Theorem~\ref{thm:HS}. Apart from its intrinsic interest this result naturally leads to a constructive proof of the measure existence. The approximation will be quantified in the natural norm for probability measures over a set $S$, the total variation norm:
\begin{equation}\label{eq:TV norm}
\norm{\mu}_{\rm TV} = \int_{S} d|\mu| = \sup_{\phi \in C_b (S)} \left| \int_{\Om} \phi \: d\mu \right|
\end{equation} 
which co\"incides with the $L^1$ for absolutely continuous measures. We shall prove the following result, taken from~\cite{DiaFre-80}:

\begin{theorem}[Diaconis-Freedman]\label{thm:DF}\mbox{}\\
Let $\mubf_N \in \PP_s (\Om ^N)$ be a symmetric probability measure. There exists $P_{\mubf_N} \in \PP(\PP (\Om))$ such that, setting
\begin{equation}\label{eq:Pnu}
\mut_N := \int_{\rho \in \PP (\Om)} \rho ^{\otimes N} dP_{\mubf_N} (\rho)
\end{equation}
we have 
\begin{equation}\label{eq:DiacFreed}
\left\Vert \mubf_N ^{(n)} - \mut_N ^{(n)} \right\Vert _{\rm TV} \leq 2 \frac{n(n-1)}{N}.
\end{equation}
\end{theorem}

\begin{proof}
We slightly abuse notation by writing $\mubf_N(Z) dZ$ instead of $d\mubf_N (Z)$ for integrals in $(z_1,\ldots,z_N) = Z \in \Om ^{N}$. It is anyway already instructive enough to consider the case of an absolutely continuous measure.

By symmetry of $\mubf_N$ we have, for all $X = (x_1,\ldots,x_N)\in \Om ^N$, 
\begin{equation}\label{eq:repre P}
\mubf_N (X) = \int_{\Om ^N} \mubf_N (Z) \sum_{\sigma \in \Sigma_N} (N!) ^{-1} \delta_{X = Z_{\sigma} } dZ
\end{equation}
where $Z_{\sigma}$ is the $N$-tuple $(z_{\sigma (1)},\ldots,z_{\sigma (n)})$. We define
\begin{equation}\label{eq:defi Pnu}
\mut_N (X) = \int_{\Om ^N} \mubf_N (Z) \sum_{\gamma \in \Gamma_N} N ^{-N} \delta_{X = Z_{\gamma}} dZ,
\end{equation}
where $\Gamma_N$ is the set of all \emph{applications}\footnote{Compared to $\Sigma_N$, $\Gamma_N$ thus allows repeated indices.} from $\left\{1,\ldots,N \right\}$ to itself and $X_{\gamma}$ is defined is the same manner as $X_{\sigma}$. The precise meaning of~\eqref{eq:defi Pnu} is 
$$ \int_{\Om ^N} \phi(X) d \mut_N (X) = \sum_{\gamma \in \Gamma_N} N ^{-N} \int_{\Om ^N} \phi (Z_{\gamma}) d\mubf_N (Z)$$
for every regular function $\phi$.

Noting that
\begin{equation}\label{eq:factor urn}
\sum_{\gamma \in \Gamma_N} N ^{-N} \delta_{X = Z_{\gamma}}  = \left( N ^{-1} \sum_{j=1} ^N \delta_{z_j} \right) ^{\otimes N} \left(x_1,\ldots,x_N\right), 
\end{equation}
one can put~\eqref{eq:defi Pnu} under the form~\eqref{eq:Pnu} by taking 
\begin{equation}\label{eq:defi nu}
P_{\mubf_N} (\rho) = \int_{\Om ^N} \delta_{\rho = \bar{\rho}_Z} \mubf_N (Z) dZ, \quad  \bar{\rho} _Z (x) := \sum_{i=1} ^N N ^{-1} \delta_{z_j=x}.
\end{equation}
Note in passing that $P_{\mubf_N}$ charges only empirical measures (of the form $\bar{\rho}_Z$ above). We now estimate the difference between the marginals of $\mubf_N$ and $\mut_N$. Diaconis and Freedman proceed as follows: of course 
\[
\mubf_N ^{(n)} - \mut_N ^{(n)} =  \int_{\Om ^N}  \left( \left(\sum_{\sigma \in \Sigma_N} (N!) ^{-1} \delta_{X = Z_{\sigma}}\right) ^{(n)} -  \left(\sum_{\gamma \in \Gamma_N} N ^{-N} \delta_{X = Z_{\gamma}}\right) ^{(n)}\right) \mubf_N(Z) dZ,
\]
but 
$$\left(\sum_{\sigma \in \Sigma_N} (N!) ^{-1} \delta_{X = Z_{\sigma}}\right) ^{(n)}$$ is the probability law for drowing $n$ balls at random from an urn containing $N$ balls\footnote{Labeled $z_1,\ldots,z_N$.}, \textit{without} replacement, whereas 
$$\left(\sum_{\gamma \in \Gamma_N} N ^{-N} \delta_{X = Z_{\gamma}}\right) ^{(n)}                                                                                                                                                                                                                                                                                                                                                                                                                                                                                                                                                   $$ 
is the probability law for drawing $n$ balls at random from an urn containing $N$, \emph{with} replacement. Intuitively it is clear that when $n$ is small compared to $N$, the fact that we replace the balls or not after each drawing does not significantly influence the result. It is not difficult to obtain quantitative bounds leading to~\eqref{eq:DiacFreed}, see for example~\cite{Freedman-77}.

Another way of obtaining~\eqref{eq:DiacFreed}, which seems to originate in~\cite{Grunbaum-71}, is as follows: in view of~\eqref{eq:defi Pnu} and~\eqref{eq:factor urn}, we have  
$$\mut_N  ^{(n)}(X) = \int_{\Om ^N} \mubf_N (Z) \left( N ^{-1} \sum_{j=1} ^N \delta_{z_j} \right) ^{\otimes n} \left(x_1,\ldots,x_N\right) dZ. $$
We then expand the tensor product and compute the contribution of terms where all indices are different. By symmetry of $\mubf_N ^{(n)}$ we obtain 
\begin{equation}\label{eq:DF astuce}
\mut_N ^{(n)} = \frac{N (N-1) \ldots (N-n+1)}{N ^n} \mubf_N ^{(n)} + \nu_n 
\end{equation}
where $\nu_n$ is a positive measure on $\Om ^n$ (all terms obtained by expanding~\eqref{eq:factor urn} are positive). We thus have 
\begin{equation}\label{eq:preuve DiacFreed} 
\mubf_N  ^{(n)} - \mut_N ^{(n)} = \left(1 - \frac{N (N-1) \ldots (N-n+1)}{N ^n} \right)\mubf_N ^{(n)} - \nu_n 
\end{equation}
and since both terms in the left-hand side are probabilities, we deduce 
\[
\int_{\Om ^n} d\nu_n =  \left(1 - \frac{N (N-1) \ldots (N-n+1)}{N ^n} \right).
\]
Moreover, since the first term of the right-hand side of~\eqref{eq:preuve DiacFreed} is positive and the second one negative we obtain by the triangle inequality
\begin{align*}
\int_{\Om ^n} d\left|\mubf_N  ^{(n)} - \mut_N ^{(n)} \right| &\leq \left(1 - \frac{N (N-1) \ldots (N-n+1)}{N ^n} \right) + \int_{\Om ^n} d\nu_n 
\\&= 2 \left(1 - \frac{N (N-1) \ldots (N-n+1)}{N ^n} \right). 
\end{align*}
It is then easy to see that
\begin{multline*}
\frac{N (N-1) \ldots (N-n+1)}{N ^n} = \prod_{j=1} ^n \frac{N-j+1}{N} = \prod_{j=1} ^n \left(1 - \frac{j-1}{N}\right) 
\\ \geq  \left(1 - \frac{n-1}{N}\right) ^n \geq 1- \frac{n(n-1)}{N},
\end{multline*}
which proves~\eqref{eq:DiacFreed} with $C=2$. A better constant $C=1$ can be obtained, see~\cite{DiaFre-80,Freedman-77}.
\end{proof}

The following remark can be useful in applications:

\begin{remark}[First marginals of the Diaconis-Freedman measure]\label{rem:marg DF}\mbox{}\\
We have
\begin{equation}\label{eq:marg DF 1}
\mut_N ^{(1)} (x) = \mubf_N  ^{(1)}(x) 
\end{equation}
and 
\begin{equation}\label{eq:marg DF 2}
\mut_N ^{(2)} (x_1,x_2) = \frac{N-1}{N} \mubf_N ^{(2)} (x_1,x_2) + \frac1N \mubf_N ^{(1)} (x_1) \delta_{x_1 = x_2}.
\end{equation}
as direct consequences of Definition~\eqref{eq:factor urn}. Indeed, using symmetry, 
\begin{align}\label{eq:marginals DF}
\mut_N ^{(1)} (x) &= N ^{-1} \sum_{j=1} ^N \int_{\Om ^N} \mubf_N (Z) \delta_{z_j=x} dZ = \mubf_N ^{(1)} (x)\nonumber\\
\mut_N ^{(2)} (x_1,x_2) &= N ^{-2} \int_{\Om ^N} \mubf_N (Z) \left( \sum_{j=1} ^N \delta_{z_j = x_1} \right) \left( \sum_{j=1} ^N \delta_{z_j = x_2} \right) dZ \nonumber \\
&= N ^{-2} \sum_{ 1 \leq i \neq j \leq N} \int_{\Om ^N} \mubf_N (Z) \delta_{ z_i = x_1} \delta_{z_j = x_2}  dZ \nonumber\\
&+ N ^{-2} \sum_{i=1} ^N \int_{\Om ^N} \mubf_N (Z) \delta_{z_i = x_1} \delta_{z_i = x_2} dZ \nonumber \\
&= \frac{N-1}{N} \mubf_N ^{(2)} (x_1,x_2) + \frac1N \mubf_N ^{(1)} (x_1) \delta_{x_1 = x_2}.
\end{align}
Higher-order marginals can be obtained by similar but heavier computations. \hfill\qed
\end{remark}

As a corollary of the preceding theorem we obtain a simple proof of the existence part in the Hewitt-Savage theorem:

\begin{proof}[Proof of Theorem~\ref{thm:HS}, Existence]
We start with the case where $\Om$ is compact. We apply Theorem~\ref{thm:DF} to $\mubf ^{(N)}$, obtaining
\begin{equation}\label{eq:preuve HS}
\norm{\mubf ^{(n)} - \int_{\rho \in \PP (\Om)} \rho ^{\otimes n} dP_{N} (\rho)}_{\rm TV} \leq C \frac{n ^2}{N}    
\end{equation}
with $P_N \in \PP (\PP (\Om))$ the measure defined in~\eqref{eq:defi nu}. When $\Om$ is compact, $\PP (\Om)$ and $\PP (\PP (\Om))$ also are. We thus may (up to a subsequence) assume that 
\[
 P_N \to P \in \PP (\PP (\Om))
\]
in the sense of measures and there only remains to pass to the limit $N\to \infty$  at fixed $n$ in~\eqref{eq:preuve HS}.

When $\Om$ is not compact, we follow an idea of Lions~\cite{Lions-CdF} (see also~\cite{Golse-13}). We want to ensure that the measure $P_N$ obtained by applying the Diaconis-Freedman theorem to $\mubf ^{(N)}$ converges. For this it suffices to test against a monomial of the form~\eqref{eq:monomes PLL}:
\begin{align}\label{eq:DF PLL}
\int_{\PP(\PP(\Om))} M_{k,\phi} (\mu) dP_N (\mu) &= \int_{Z \in \Om ^N} M_{k,\phi} \left( \frac{1}{N} \sum_{j=1} ^N \delta_{z_j} \right) d\mubf ^{(N)} (Z)\nonumber\\
&= \int_{Z \in \Om ^N} \int_{X\in \Om ^k} \phi(x_1,\ldots,x_k) \prod_{j=1} ^k \left(\frac{1}{N} \sum_{j=1} ^N \delta_{z_j=x_k} \right) d\mubf ^{(N)} (Z)\nonumber\\
&= \int_{Z\in \Om ^k} \phi(z_1,\ldots,z_k) d\mubf ^{(k)} (Z) + O(N ^{-1})
\end{align}
by a computation similar to that yielding~\eqref{eq:DF astuce}. The limit~\eqref{eq:DF PLL} thus exists for any monomial $M_{k,\phi}$, and by density of monomials for any bounded continuous function over $\PP(\Om)$. We then deduce  
\[
 P_N \wto_* P
\]
and we can conclude as previously.
\end{proof}

Some remarks before going to the applications of Theorems~\ref{thm:HS} and~\ref{thm:DF}:

\begin{remark}[On the Diaconis-Freedman-Lions construction]\label{rem:DiacFreLio}\mbox{}\\
\vspace{-0.5cm}
\begin{enumerate}
\item We first note that the measure defined by~\eqref{eq:defi nu} is that Lions uses in his approach of the Hewitt-Savage theorem. Using empirical measures is the canonical way to construct a measure over $\PP (\PP(\Om))$ given one over $\PP_s(\Om ^N)$. Passing to the limit as in ~\eqref{eq:DF PLL} can replace the explicit estimate~\eqref{eq:DiacFreed} if one is only interested in the proof of Theorem ~\ref{thm:HS}. This construction and the combinatorial trick~\eqref{eq:DF astuce} are also used e.g. in~\cite{GolMouRic-13,MisMou-13,HauMis-14}.
\item Having an explicit estimate of the form~\eqref{eq:DiacFreed} at hand is very satisfying and can prove useful in applications. One might wonder whether the obtained convergence rate is optimal. Perhaps suprisingly it is. One might have expected a useful estimate for $n\ll N$, but it happens that $\sqrt{n}\ll N$ is optimal, see examples in~\cite{DiaFre-80}.
\item The formulae~\eqref{eq:marginals DF} are useful in practice (see~\cite{RouYng-14} for an application). It is quite satisfying that $\mubf_N ^{(1)} = \mut_N ^{(1)}$ and that $\mubf_N ^{(2)}$ can be reconstructed using only $\mut_N ^{(2)}$ and~$\mut_N ^{(1)}$.
\item As a drawback of its generality, the previous construction actually behaves very badly in many cases. Note that~\eqref{eq:defi nu} charges only empirical measures, which all have infinite entropy. This causes problems when employing Theorem~\ref{thm:DF} to the study of a functional with temperature. Moreover, in a situation with strong repulsive interactions, one typically applies the construction to a measure with zero probability of having two particles at the same place,  $\mubf_N ^{(2)}(x,x) \equiv 0$. In this case $\mut_N ^{(2)} (x,x) = N ^{-1} \mubf_N ^{(1)} (x)$ is non-zero and the energy of $\mut_N$ will be infinite if the interaction potential has a singularity at the origin.
\item It would be very interesting to have a construction leading to an estimate of the form~\eqref{eq:DiacFreed}, bypassing the aforementioned inconveniences. For example, is it possible to guarantee that $\mut_N \in L ^1(\Om ^N)$ if $\mubf_N \in L ^1 (\Om ^N)$ ? One might also demand that the construction leave product measures $\rho ^{\otimes N}$ invariant, which is not at all the case for~\eqref{eq:Pnu}.
\item If $\Om$ is replaced by a finite set, say $\Om = \left\lbrace 1,\ldots, d\right\rbrace$, one may obtain an error proportional to $dn/N$ instead of $n ^2/N$ by using the original proof of Diaconis-Freedman. One may thus replace~\eqref{eq:DiacFreed} by 
\begin{equation}\label{eq:DiacFreed 3}
\left\Vert \mubf_N ^{(n)} - \mut_N ^{(n)} \right\Vert _{\rm TV} \leq \frac{C}{N} \min\left( dn,n ^2\right). 
\end{equation}
We will not use this point anywhere in the sequel.
\hfill\qed
\end{enumerate}

\end{remark}

Finally, let us make a separate remark on a possible generalization of the approach above:

\begin{remark}[Weakly dependent symmetric functions of many variables.]\label{rem:LioHauMis}\mbox{}\\
\vspace{-0.5cm}

In applications (see next section), one is lead to apply de Finetti-like theorems in the following manner: given a sequence $(u_N)_{N\in \N}$ of symmetric functions of $N$ variables, we study the quantity
$$ \int_{\Om ^N} u_N (X) d\mubf_N (X)$$
for a symmetric probability measure $\mubf_N$. The results previously discussed imply that if this happens to depend only on a marginal $\mubf_N ^{(n)}$ for fixed $n$, a natual limit object appears when $N\to \infty$. In particular, if 
\begin{equation}\label{eq:simple u_N}
 u_N (X) = {N \choose n} ^{-1} \sum_{1\leq i_1 < \ldots < i_n \leq N} \phi (x_{i_1}, \ldots, x_{i_n}),
\end{equation}
we have
\begin{equation}\label{eq:general DiacFreed}
 \int_{\Om ^N} u_N (X) d\mubf_N (X) = \int_{\Om ^n} \phi (x_1,\ldots,x_n) d\mubf_N ^{(n)} \to \int_{\rho \in \PP (\Om)} \left( \int_{\Om ^n} \phi \: d\rho ^{\otimes n}\right) dP_{\mubf} (\rho) 
\end{equation}
for a certain probability measure $P_{\mubf} \in \PP (\PP (\Om))$. One might ask whether this kind of results is true for a class of functions $u_N$ depending on $N$ in more subtle a manner. The natural assumptions seems to be that $u_N$ \emph{depends weakly on its $N$ variables}, in the sense introduced by Lions~\cite{Lions-CdF} and recalled in~\cite[Section 1.7.3]{Golse-13}. Without entering the details, one may easily see that to such a sequence there corresponds (modulo extraction of a subsequence) a continuous function over probabilities $U\in C(\PP(\Om))$: one may show that, along a subsequence,
\begin{equation}\label{eq:general Lions}
 \int_{\Om ^N} u_N (X) d\mubf_N (X) \to \int_{\rho \in \PP (\Om)} U(\rho) dP_{\mubf} (\rho). 
\end{equation}
Think of the case where $u_N$ depends only on the empirical measure:
\begin{equation}\label{eq:u_N moins simple}
u_N (x_1,\ldots,x_N) = F \left( \frac{1}{N} \sum_{j=1} ^N \delta_{x_j} \right) 
\end{equation}
with a sufficiently regular function $F$. Such a function depends weakly on its $N$ variables in the sense of Lions, but cannot be written in the form~\eqref{eq:simple u_N}. More generally, this kind of considerations can be applied under assumptions of the kind  
$$ |\nabla_{x_j} u_N (x_1,\ldots,x_N)| \leq \frac{C}{N}\quad \forall j= 1\ldots N, \quad \forall (x_1,\ldots,x_N)\in \Om ^N.$$
One may then wonder whether the convergence rate in~\eqref{eq:general Lions} can be quantified. Results in this direction may be found in~\cite{HauMis-14}.\hfill\qed
\end{remark}

\subsection{Mean-field limit for a classical free-energy functional}\label{sec:appli HS}\mbox{}\\\vspace{-0.4cm}

In this section we apply Theorem~\ref{thm:HS} to the study of a free-energy functional at positive temperature, following~\cite{MesSpo-82,CagLioMarPul-92,Kiessling-93,KieSpo-99}. We consider a domain $\Om \subset \R ^d$ and the functional
\begin{equation}\label{eq:aHS free func}
 \F_N [\mubf] = \int_{X \in \Om ^N} H_N (X) \mubf (X) dX + T \int_{\Om ^N} \mubf(X) \log \mubf(X) dX
\end{equation}
defined for probability measures $\mubf \in \PP (\Om ^N)$. Here the temperature $T$ will be fixed in the limit $N\to \infty$ and the Hamiltonian  $H_N$ is in mean-field scaling:
\begin{equation}\label{eq:aHS hamil}
H_N (X) = \sum_{j=1} ^N V(x_j) + \frac{1}{N-1} \sum_{1\leq i < j \leq N} w (x_i-x_j). 
\end{equation}
We denote by $V$ a lower semi-continuous potential. To be in a compact setting we shall assume that either $\Om$ is bounded or  
\begin{equation}\label{eq:aHS confine}
V(x) \to \infty \mbox{ when } |x|\to \infty. 
\end{equation}
The interaction potential $w$ will be bounded below and lower semi-continuous. To be concrete, one may think of $w\in L ^{\infty}$, or a repulsive Coulomb potential:
\begin{align}\label{eq:aHS Coulomb}
w(x) &= \frac{1}{|x| ^{d-2}} \mbox{ if } d = 3\\
w(x) &= - \log |x| \mbox{ if } d=2\\
w(x) &=  -| x| \mbox{ if } d=1,
\end{align}
ubiquitous in applications. We will always take $w$ even,
\[
w(x) = w (-x), 
\]
and if the domain is not bounded we will assume
\begin{align}\label{eq:aHS confine 2}
w(x-y) &+ V(x) + V(y) \to \infty \mbox{ when } |x|\to \infty \mbox{ or } |y|\to \infty \nonumber \\
w(x-y) &+ V(x) + V(y)  \mbox{ is lower semi-continuous.}
\end{align}
We are interested in the limit of the Gibbs measure minimizing~\eqref{eq:aHS free func} in $\PP_s (\Om ^N)$:
\begin{equation}\label{eq:aHS Gibbs}
\mubf_N (X) = \frac{1}{\ZN} \exp\left( - \frac{1}{T} H_N (X) \right) dX
\end{equation}
and to the corresponding free-energy
\begin{equation}\label{eq:aHS free ener}
F_N = \inf_{\mubf \in \PP (\Om_s ^N)} \F_N [\mubf] = \F_N [\mubf_N] = - T \log \ZN. 
\end{equation}
The free-energy functional is rewritten 
\begin{align}\label{eq:aHS free func 2}
\F_N [\mubf] &= N \int_{\Om} V(x) d\mubf ^{(1)} (x) + \frac{N}{2} \iint_{\Om \times \Om} w(x-y) d\mubf ^{(2)} (x,y) + T \int_{\Om ^N} \mubf \log \mubf \nonumber\\
 &= \frac{N}{2} \iint_{\Om \times \Om} \left( w(x-y) + V(x) + V(y) \right) d\mubf ^{(2)} (x,y)  + T \int_{\Om ^N} \mubf \log \mubf. 
\end{align}
using the marginals
\begin{equation}\label{eq:aHS marginales}
\mubf ^{(n)} (x_1,\ldots,x_n) = \int_{x_{n+1},\ldots,x_N\in \Om} d \mubf (x_1,\ldots,x_N).
\end{equation}

Inserting an ansatz of the form
\begin{equation}\label{eq:aHS ansat}
\mubf = \rho ^{\otimes N}, \rho \in \PP (\Om) 
\end{equation}
in~\eqref{eq:aHS free func} we obtain the mean-field functional
\begin{multline}\label{eq:aHS MFf}
\MFf [\rho] := N ^{-1} \F_N [\rho ^{\otimes N}] \\
= \int_{\Om} V(x) d\rho(x) + \frac{1}{2} \iint_{\Om \times \Om} w(x-y) d\rho (x)d\rho(y) + T \int_{\Om} \rho \log \rho 
\end{multline}
with minimum $\MFe$ and minimizer (not necessarily unique) $\rhoMF$ amongst probability measures. Our goal is to justify the mean-field approximation by proving the following theorem:

\begin{theorem}[\textbf{Mean-field limit at fixed temperature}]\label{thm:aHS}\mbox{}\\
We have
\begin{equation}\label{eq:aHS result ener}
\frac{F_N}{N} \to \MFe \mbox{ when } N \to \infty.
\end{equation}
Moreover, up to a subsequence, we have for every $n\in \N$
\begin{equation}\label{eq:aHS result marginal}
\mubf_N ^{(n)} \wto_\ast \int_{\rho \in \MFmin} \rho  ^{\otimes n} dP (\rho).  
\end{equation}
as measures where $P$ is a probability measurer over $\MFmin$, the set of minimizers of $\MFf$.
\end{theorem}

In particular, if $\MFf$ has a unique minimizer we obtain, for the whole sequence, 
\[
\mubf_N ^{(n)} \wto_\ast \left(\rhoMF\right) ^{\otimes n}.
\]

\begin{proof}
We follow~\cite{MesSpo-82,CagLioMarPul-92,Kiessling-89,Kiessling-93}. An upper bound to the free energy is eeasily obtained using test functions of the form $\rho ^{\otimes N}$ and we deduce
\begin{equation}\label{eq:aHS borne sup}
\frac{F_N}{N}  \leq \MFe. 
\end{equation}
The corresponding lower bound requires more work. We start by extracting a subsequence along which
\begin{equation}\label{eq:aHS preuve 1}
\mubf_N ^{(n)} \wto_\ast \mubf ^{(n)} 
\end{equation}
for all $n\in \N$, with $\mubf \in \PP_s (\Om ^{\N})$. This is done as explained in Section~\ref{sec:HS}, using either the fact that $\Om$ is compact or Assumption~\eqref{eq:aHS confine 2}, which implies that the sequence is tight. Indeed, using simple energy upper and lower bounds one easily sees that for all $\eps > 0$ there exists $R_\eps$ such that 
$$ \mubf_N ^{(2)} \left(B (0,R_\eps)\times B (0,R_\eps)\right) > 1 - \eps$$
for all $N$. Tightness of all the other marginals follows. 
 
By lower semi-continuity we immediately have
\begin{multline}\label{eq:aHS lim ener}
\liminf_{N\to \infty} \frac{1}{2} \iint_{\Om \times \Om} \left( w(x-y) + V(x) + V(y) \right) d\mubf_N ^{(2)} (x,y) 
\\ \geq \frac{1}{2} \iint_{\Om \times \Om} \left( w(x-y) + V(x) + V(y) \right) d\mubf ^{(2)} (x,y).     
\end{multline}
For the entropy term we use the sub-additivity property (this is a consequence of Jensen's inequality, see~\cite{RobRue-67} or the previously cited references) 
\[
\int_{\Om ^N} \mubf_N \log \mubf_N \geq \left\lfloor \frac{N}{n} \right\rfloor \int_{\Om ^n} \mubf_N ^{(n)} \log \mubf_N ^{(n)} +  \int_{\Om ^{N-n\left\lfloor \frac{N}{n} \right\rfloor}} \mubf_N ^{\left(N-n\left\lfloor \frac{N}{n} \right\rfloor\right)} \log \mubf_N ^{\left(N-n\left\lfloor \frac{N}{n} \right\rfloor\right)} 
\]
where $\lfloor \: . \: \rfloor$ denotes the integer part. Jensen's inequality implies that for probability measures $\mu$ and $\nu$, the relative entropy of $\mu$ with respect to $\nu$ is positive:
\[ 
\int \mu \log \frac{\mu}{\nu} = \int \nu \frac{\mu}{\nu} \log \frac{\mu}{\nu} \geq \left(\int \nu \frac{\mu}{\nu} \right) \log \left( \int \nu \frac{\mu}{\nu} \right) = 0.
\]
We deduce that for all $\nu_0 \in \PP(\Om)$
\begin{align*}
\int_{\Om ^{N-n\left\lfloor \frac{N}{n} \right\rfloor}} \mubf_N ^{\left(N-n\left\lfloor \frac{N}{n} \right\rfloor\right)} \log \mubf_N ^{\left(N-n\left\lfloor \frac{N}{n} \right\rfloor\right)} &= \int_{\Om ^{N-n\left\lfloor \frac{N}{n} \right\rfloor}} \mubf_N ^{\left(N-n\left\lfloor \frac{N}{n} \right\rfloor\right)} \log \left(\frac{\mubf_N ^{\left(N-n\left\lfloor \frac{N}{n} \right\rfloor\right)}}{\nu_0 ^{\otimes \left(N-n\left\lfloor \frac{N}{n} \right\rfloor\right)}}\right) \\
&+ \int_{\Om ^{N-n\left\lfloor \frac{N}{n} \right\rfloor}} \mubf_N ^{\left(N-n\left\lfloor \frac{N}{n} \right\rfloor\right)} \log \nu_0 ^{\otimes \left(N-n\left\lfloor \frac{N}{n} \right\rfloor\right)}\\
&\geq \left(N-n\left\lfloor \frac{N}{n} \right\rfloor\right) \int_{\Om } \mubf_N ^{(1)} \log \nu_0. 
\end{align*}
Choosing $\nu_0\in\PP(\Om)$ of the form $\nu_0 = c_0 \exp(-c_1 V)$ it is not difficult to see that the last integral is bounded below independently of $N$ and we thus obtain that for all $n\in \N$
\begin{equation}\label{eq:aHS lim entr}
\liminf_{N\to \infty} \frac{1}{N} \int_{\Om ^N} \mubf_N \log \mubf_N \geq \frac{1}{n} \int_{\Om ^n} \mubf ^{(n)} \log \mubf ^{(n)}
\end{equation}
by lower semi-continuity of (minus) the entropy.

Gathering~\eqref{eq:aHS lim ener} and~\eqref{eq:aHS lim entr} we obtain a lower bound in terms of a functional of $\mubf$:
\begin{multline}\label{eq:aHS somme}
\liminf_{N\to \infty} \frac{1}{N} \F_N [\mubf_N] \geq \F [\mubf] :=  \frac{1}{2} \iint_{\Om \times \Om} \left( w(x-y) + V(x) + V(y) \right) d\mubf ^{(2)} (x,y) 
\\ + T \limsup_{n\to \infty} \frac{1}{n} \int_{\Om ^n} \mubf ^{(n)} \log \mubf ^{(n)}.
\end{multline}
The second term in the right-hand side is called (minus) the mean entropy of $\mubf\in \PP (\Om ^{\N})$. We will next apply the Hewitt-Savage theorem to $\mubf$. The first term in $\F$ is obviously affine as a function of $\mubf$, which is perfect to apply~\eqref{eq:result HS}. One might however worry at the sight of the second term which rather looks convex. In fact a simple argument of~\cite{RobRue-67} shows that this mean entropy is affine. It is a part of the statement that the $\limsup$ is in fact a limit:

\begin{lemma}[\textbf{Properties of the mean entropy}]\label{lem:mean entropy}\mbox{}\\
The functional
\begin{equation}\label{eq:mean entropy}
\sup_{n\to \infty} \frac{1}{n} \int_{\Om ^n} \mubf ^{(n)} \log \mubf ^{(n)} = \lim_{n\to \infty}  \frac{1}{n} \int_{\Om ^n} \mubf ^{(n)} \log \mubf ^{(n)}
\end{equation}
is affine over $\PP (\Om ^{\N})$.
\end{lemma}

\begin{proof}
We start by proving that the $\sup$ equals the limit. Denote the former by $S$ and pick $j$ such that 
$$ \frac{1}{j} \int_{\Om ^j} \mubf ^{(j)} \log \mubf ^{(j)} = \sup_{n\to \infty} \frac{1}{n} \int_{\Om ^n} \mubf ^{(n)} \log \mubf ^{(n)} - \eps_j  = S- \eps_j.$$
Next let $k \geq j$ and write (Euclidean division)
$$k = mj + r \mbox{ with }  0\leq r <j.$$
Then, by the subbaditivity property already mentioned 
\begin{align*}
\frac{1}{k} \int_{\Om ^k} \mubf ^{(k)} \log \mubf ^{(k)} &\geq \frac{m}{k} \int_{\Om ^j} \mubf ^{(j)} \log \mubf ^{(j)} + \frac{1}{k} \int_{\Om ^r} \mubf ^{(r)} \log \mubf ^{(r)}
\\&\geq \frac{mj}{k} (S-\eps_j) - \frac{C_j}{k} 
\end{align*}
where $C_j$ only depends on $j$. We then pass to the limit $k\to \infty$ at fixed $j$. Noting that then $m\to k/j$ we obtain
$$ \liminf_{k\to \infty} \frac{1}{k} \int_{\Om ^k} \mubf ^{(k)} \log \mubf ^{(k)} \geq S-\eps_j.$$
Passing finally to the limit $j\to \infty$ along an appropriate subsequence we can make $\eps_j \to 0$, which proves the equality in~\eqref{eq:mean entropy}.  

To prove that the functional is affine, let $\mubf_1,\mubf_2 \in \PP (\Om ^{\N})$ be given. We use the convexity of  $x\mapsto x\log x$ and the monotonicity of $x\mapsto \log x$ to obtain 
\begin{align*}
\frac{1}{2}\int_{\Om ^n} \mubf_1 ^{(n)} \log \mubf_1 ^{(n)} &+ \frac{1}{2}\int_{\Om ^n} \mubf_2 ^{(n)} \log \mubf_2 ^{(n)}\geq
\int_{\Om ^n} \left( \frac{1}{2}\mubf_1 ^{(n)} + \frac{1}{2} \mubf_2 ^{(n)} \right) \log \left( \frac{1}{2}\mubf_1 ^{(n)} + \frac{1}{2} \mubf_2 ^{(n)} \right) 
\\&\geq  \frac{1}{2}\int_{\Om ^n} \mubf_1 ^{(n)} \log \mubf_1 ^{(n)} + \frac{1}{2}\int_{\Om ^n} \mubf_2 ^{(n)} \log \mubf_2 ^{(n)} 
\\&- \frac{\log (2)}{2} \left( \int_{\Om ^n} \mubf_1 ^{(n)} + \int_{\Om ^n} \mubf_2 ^{(n)}\right) 
\\&= \frac{1}{2}\int_{\Om ^n} \mubf_1 ^{(n)} \log \mubf_1 ^{(n)} + \frac{1}{2}\int_{\Om ^n} \mubf_2 ^{(n)} \log \mubf_2 ^{(n)} - \log (2).
\end{align*}
Dividing by $n$ and passing to the limit we deduce that the functional is indeed affine.
\end{proof}

Thus $\F[\mu]$ is affine. There remains to use~\eqref{eq:result HS}, which gives a probability $P_{\mubf} \in \PP (\PP (\Om ))$ such that
\begin{align*}
\liminf_{N\to \infty} \frac{1}{N} \F_N [\mubf_N] &\geq \int_{\rho \in \PP (\Om)} \F [\rho ^{\otimes \infty}] dP_{\mubf}(\rho)
\\& = \int_{\rho \in \PP (\Om)} \MFf [\rho ] dP_{\mubf}(\rho) \geq \MFe.
\end{align*}
Here we denote $\rho ^{\otimes \infty}$ the probability over $\PP (\Om ^{\N})$ which has $\rho ^{\otimes n}$ for $n$-th marginal for all $n$ and we have used the fact that $P_{\mubf}$ has integral $1$. This concludes the proof of~\eqref{eq:aHS result ener} and~\eqref{eq:aHS result marginal} follows easily. It is indeed clear in view of the previous inequalities that $P_{\mubf}$ must be concentrated on the set of minimizers of the mean-field free-energy functional.
\end{proof}

\subsection{Quantitative estimates in the mean-field/small temperature limit}\label{sec:appli DF}\mbox{}\\ \vspace{-0.4cm}

Here we give an example where the Diaconis-Freedman construction is useful to supplement the use of the Hewitt-Savage theorem. As mentioned in Remark~\ref{rem:DiacFreLio}, this construction behaves rather badly with respect to the entropy, but there is a fair number of interesting problems where it makes sense to consider a small temperature in the limit $N\to \infty$, in which case entropy plays a small role.

A central example is that of log-gases. It is well-known that the distribution of eigenvalues of certain random matrices ensembles is given by the Gibbs measure of a classical gas with logarithmic interactions. Moreover, it so happens that the relevant limit for large matrices is a mean-field regime with temperature of the order of $N^{-1}$. Consider the following Hamiltonian (assumptions on $V$ are as previously, taking $w = -\log |\:.\:|$):
\begin{equation}\label{eq:aDF Hamil}
H_N (X) = \sum_{j=1} ^N V(x_j) - \frac{1}{N-1} \sum_{1\leq i < j \leq N} \log |x_i-x_j|  
\end{equation}
where $X = (x_1,\ldots,x_N) \in \R ^{dN}$. The associated Gibbs measure 
\begin{equation}\label{eq:aDF Gibbs}
\mubf_N (X) = \frac{1}{\ZN} \exp \left(-\beta N H_N (X) \right) dX 
\end{equation}
corresponds (modulo a $\beta$-dependent change of scale) to the distribution of of the eigenvalues of a random matrix in the following cases:
\begin{itemize}
\item $d=1,\beta=1,2,4$ and $V(x) = \frac{|x| ^2}{2}$. We respectively obtain the gaussian real symmetric, complex hermitian and quarternionic self-dual matrices.
\item $d=2,\beta = 2$ and $V(x) = \frac{|x| ^2}{2}$. We then obtain the gaussian matrices without symmetry conditon, the so-called Ginibre ensemble~\cite{Ginibre-65}.
\end{itemize}
In these notes we give no further precisions on the random matrix aspect, and will simply take the previous facts as a sufficient motivation to study the $N\to \infty$ limit of the measures~\eqref{eq:aDF Gibbs} with $\beta$ fixed. This corresponds (compare with~\eqref{eq:aHS Gibbs}) to taking $T=\beta ^{-1} N ^{-1}$, i.e. a very small temperature. For an introduction to random matrices and log-gases we refer to~\cite{AndGuiZei-10,Forrester-10,Mehta-04}. For precise studies of the measures~\eqref{eq:aDF Gibbs} following different methods than what we shall do here, one may consult e.g.~\cite{BenGui-97,BenZei-98,BouPasShc-95,BouErdYau-12,BouErdYau-14,BouErdYau-12,ChaGozZit-13,RouSer-13,SanSer-12,SanSer-13}.

In the case $d=2$, measures of the form~\eqref{eq:aDF Gibbs} also have a natural application to the study of certain quantum wave-functions appearing in the study of the fractional quantum Hall effect (see~\cite{RouSerYng-13b,RouSerYng-13,RouYng-14,RouYng-15} and references therein). Here too it is sensible to consider $\beta$ as being fixed.

The singularity at the origin of the logarithm poses difficulties in the proof, as indicated in Remark~\ref{rem:DiacFreLio}, but one may bypass them easily, contrarily to those linked to the entropy. The following method is not limited to log-gases and can be re-employed in various contexts.

\medskip

Again,~\eqref{eq:aDF Gibbs} minimizes a free-energy functional
\begin{equation}\label{eq:aDF free func}
 \F_N [\mubf] = \int_{X \in \Om ^N} H_N (X) \mubf (X) dX + \frac{1}{\beta N} \int_{\Om ^N} \mubf(X) \log \mubf(X) dX 
\end{equation}
whose minimum we denote $F_N$. The natural limit object is this time an energy functional with no entropy term:
\begin{equation}\label{eq:aDF MFf}
\MFfo [\rho] := \int_{\R ^d} V d\rho -\frac{1}{2}\iint_{\R ^d \times \R ^d } \log|x-y| d\rho (x) d\rho (y),
\end{equation}
obtained by inserting the ansatz $\rho ^{\otimes N}$ in~\eqref{eq:aDF free func} and neglecting the entropy term, which is manifestly of lower order for  fixed $\beta$. We denote $\MFeo$ and $\rhoMF$ respectively the minimum energy and the minimizer (unique in this case by strict convexity of the functional). It is well-known (see the previous references as well as~\cite{KieSpo-99}) that 
\begin{equation}\label{eq:aDF ener lim}
N ^{-1} F_N = -\frac{1}{\beta} \log \ZN \to \MFeo \mbox{ when } N\to \infty
\end{equation}
and
\begin{equation}\label{eq:aDF Gibbs lim}
\mubf_N ^{(n)} \wto_* \left(\rhoMF \right) ^{\otimes n}.
\end{equation}
We shall prove~\eqref{eq:aDF Gibbs lim} and give a quantitative version of~\eqref{eq:aDF ener lim}, taking inspiration from~\cite{RouYng-14}:

\begin{theorem}[\textbf{Free-energy estimate for a log-gas}]\label{thm:aDF}\mbox{}\\
For all $\beta \in \R$, we have 
\begin{equation}\label{eq:aDF result}
\MFeo - C N ^{-1} \left(\beta ^{-1} + \log N +1 \right) \leq  - \frac{1}{\beta} \log \ZN \leq \MFeo + C \beta ^{-1}N ^{-1}. 
\end{equation}
\end{theorem}

Fine estimates for the partition function $\ZN$ of the log-gase seem to have become available only recently~\cite{LebSer-15,PetSer-14,RouSer-13,SanSer-12,SanSer-13}. Amongst other things, the previous references indicate that the correction to $\MFeo$ is exactly of order $N ^{-1}\log N$. 

\begin{proof}
We shall not elaborate on the upper bound, which is easily obtained by taking the usual product ansatz and estimating the entropy. For the lower bound we shall use Theorem~\ref{thm:DF}. We first need a crude bound on the entropy term: by positivity of the relative entropy (via Jensen's inequality)
\[
\int_{\R ^{dN}} \mubf \log \frac{\mubf}{\nubf} \geq 0 \mbox{ for all } \mubf,\nubf \in \PP (\Om ^N) 
\]
one may write, using the probability measure
$$\nubf_N = \left(c_0 \exp\left( -V(|x|)\right)\right) ^{\otimes N},$$
the lower bound
\begin{equation}\label{eq:aDF rel entr}
\int_{\R ^{dN}} \mubf_N \log \mubf_N \geq  \int_{\R ^{dN}} \mubf_N \log \nubf_N = -N \int_{\R ^d} V d \mubf_N ^{(1)} - N \log c_0.
\end{equation}
To obtain a lower bound to the energy we first need to regularize the interaction potential: Let $\alpha >0$ be a small parameter to be optimized over later and 
\begin{equation}\label{eq:log alpha 2}
- \logal |z| = \begin{cases}
             \displaystyle -\log\alpha+\half \left(1-\frac{|z|^2}{\alpha^2}\right) &\mbox{ if } |z|\leq \alpha \\
             - \log |z| &\mbox{ if } |z|\geq \alpha.
            \end{cases}
\end{equation}
Clearly $-\log_\alpha |z|\leq -\log |z|$ is regular at the origin. Moreover, we have 
\begin{equation}\label{eq:log alpha 3}
- \frac d{d\alpha}\logal |z| = \begin{cases}
             \displaystyle -\frac 1\alpha+\frac {|z|^2}{\alpha^3} &\mbox{ si } |z|\leq \alpha\\
             0 &\mbox{ si } |z|\geq \alpha.
            \end{cases}
\end{equation}
Using the lower bound $-\log_\alpha |z|\leq -\log |z|$ to obtain
\[
\int_{X \in \Om ^N} H_N (X) \mubf (X) dX \geq  N \int_{\R ^d} V d\mubf_N ^{(1)} - \frac{N}{2} \iint_{\R ^d \times \R ^d} \logal |x-y| d\mubf_N ^{(2)} (x,y) 
\]
we are now in a position where we can apply Theorem~\ref{thm:DF}. More precisely we use the explicit formulae~\eqref{eq:marginals DF}:   
\begin{multline}\label{eq:aDF use DF}
 \int_{X \in \Om ^N} H_N (X) \mubf_N (X) dX \geq 
 \\ N \int_{\R ^d} V d\mut_N ^{(1)} - \frac{N ^2}{2(N-1)} \iint_{\R ^d \times \R ^d} \logal |x-y| d\mubf_N ^{(2)} (x,y) + C \logal(0). 
\end{multline}
Combining~\eqref{eq:aDF rel entr},~\eqref{eq:Pnu} and recalling that the temperature $T$ is equal to $\left(\beta N \right) ^{-1}$ we obtain
\begin{multline}\label{eq:aDF final}
N ^{-1} \F_N [\mubf_N] \geq \int_{\rho \in \PP(\R^d)} \MFfal [\rho] dP_{\mubf_N} (\rho) 
\\ +C N ^{-1} \left( \logal(0) - \beta ^{-1}\right) \geq \MFeal - C N ^{-1} \log (\alpha) - C (\beta N) ^{-1}  
\end{multline}
where $\MFeal$ is the minimum (amongst probability measures) of the modified functional
\[
\MFfal [\rho] := \int_{\R ^d} V (1-\beta ^{-1} N ^{-1}) d\rho -\frac{N}{2(N-1)}\iint_{\R ^d \times \R ^d } \logal |x-y| d\rho (x) d\rho (y).
\]
Exploiting the associated variational equation,~\eqref{eq:log alpha 3} and the Feynman-Hellmann principle, it is not difficult to see that for small enough $\alpha$
$$ \left| \MFeal - \MFeo \right| \leq C \alpha ^d + C N ^{-1}$$
and we thus conclude
\[
N ^{-1} \F_N [\mubf_N] \geq \MFeo - C N ^{-1} \log (\alpha) - C \beta N ^{-1} - C N ^{-1} - C \alpha ^d - C, 
\]
which gives the desired lower bound after optimizing over $\alpha $ (take $\alpha = N ^{-1/d}$).
\end{proof}

\begin{remark}[Possible extensions]\mbox{}\label{rem:log gas}\\\vspace{-0.4cm}
\begin{enumerate}
 \item One can also prove quantitative versions of~\eqref{eq:aDF Gibbs lim} following essentially the above method of proof. We shall not elaborate on this point for which we refer to~\cite{RouYng-14}.
 \item Another case we could deal with along the same lines is that of unitary, orthogonal and symplectic gaussian random matrices ensembles introduced by Dyson~\cite{Dyson-62a,Dyson-62b,Dyson-62c}. In this case $\R ^d$ is replaced by the unit circle, $\beta = 1,2,4$, $V\equiv 0$ in~\eqref{eq:aDF Hamil}. In this case, dividing the interaction by $N-1$ is irrelevant.
%
\hfill \qed
\end{enumerate}
\end{remark}

\newpage

\section{\textbf{The quantum de Finetti theorem and Hartree's theory}}\label{sec:quant}

Now we enter the heart of the matter, i.e. mean-field limits for large bosonic systems. We present first the derivation of the ground state of Hartree's theory for confined particles, e.g. living in a bounded domain. In such a case, the result is a rather straightforward consequence of the quantum de Finetti theorem proved by St\o{}rmer and Hudson-Moody~\cite{Stormer-69,HudMoo-75}. The latter describes all the \emph{strong} limits (in the sense of the trace-class norm) of the reduced density matrices of a large bosonic system. 

We will then move on to the more complex case of non-confined systems. In this chapter we will assume that the interaction potential has no bound state (it could be purely repulsive for example), the general case being dealt with later (Chapter~\ref{sec:Hartree}). In the absence of bound states it is sufficient to have at our disposal a de Finetti theorem describing all the \emph{weak} limits (in the sense of the weak-$\ast$ topology on $\gS ^1$) of the reduced density matrices.

The weak de Finetti theorem (introduced in~\cite{LewNamRou-13}) implies the strong de Finetti theorem and in fact the two results can be deduced from an even more general theorem appearing in~\cite{Stormer-69,HudMoo-75}. In these notes I chose not to follow this approach, but rather that of~\cite{AmmNie-08,LewNamRou-13} which is more constructive. This will be discussed in details in Section~\ref{sec:rel deF}, which announces the plan of the next chapters.

\subsection{Setting the stage}\label{sec:quant stage}\mbox{}\\\vspace{-0.4cm}

To simplify the discusstion, we will focus on the case of non relativistic quantum particles, in the absence of a magnetic field. The Hamiltonian will thus have the general form 
\begin{equation}\label{eq:quant hamil}
H_N = \sum_{j=1} ^N T_j + \frac{1}{N-1}\sum_{1\leq i < j \leq N} w(x_i-x_j), 
\end{equation}
acting on the Hilbert space $\gH_s ^N = \bigotimes_s ^N \gH$, i.e. the symmetric tensor product of $N$ copies of $\gH$ where $\gH$ denotes the space $L ^2(\Om)$ for $\Om \subset \R ^d$. The operator $T$ is a Schr\"odinger operator 
\begin{equation}\label{eq:op Schro}
T = - \Delta + V 
\end{equation}
with $V:\Om \mapsto \R$ and $T_j$ acts on the $j$-th variable: 
\[
T_j \psi_1 \otimes \ldots \otimes \psi_N = \psi_1 \otimes \ldots \otimes T_j \psi_j \otimes \ldots \otimes \psi_N. 
\]
We assume that $T$ is self-adjoint and bounded below, and that the interaction potential $w:\R \mapsto \R$ is bounded relatively to $T$ (as operators): for some $ 0 \leq \beta_-, \beta_+ < 1$
\begin{equation}\label{eq:T controls w}
 -\beta_-(T_1+ T_2)-C\leq w(x_1-x_2) \leq \beta_+ (T_1+T_2) + C.
\end{equation}
We also take $w$ symmetric
\[
w(-x) = w(x), 
\]
and decaying at infinity
\begin{equation}\label{eq:decrease w}
w\in L^p (\Om) + L ^{\infty} (\Om), \max(1,d/2) < p < \infty \to 0, w(x) \to 0 \mbox{ when }  |x|\to\infty.
\end{equation}

We will always make an abuse of notation by writing $w$ for the multiplication operator by $w(x_1-x_2)$ on $L ^2 (\Om)$.

\begin{remark}[Checking assumptions on $w$]\label{rem:asum w}\mbox{}\\
One may check that assumptions~\eqref{eq:decrease w} imply operator bounds of the form~\eqref{eq:T controls w}, using standard functional inequalities. We quickly sketch the method here for completness. Suppose that $w \in L ^{p} (\R ^d)$, such that the conjugate H\"older exponent $q$ is  smaller than half the critical exponent for the Sobolev injection $H^1 (\R ^d) \hookrightarrow L ^q (\R ^d)$:
\begin{equation}\label{eq:exponents}
 q = \frac{p}{p-1} <  \frac{p^*}{2}, \quad p ^* = \frac{2d}{d-2}. 
\end{equation}
We decompose $|w|$ as 
$$ |w| = |w| \1_{|w| \leq R} + |w| \1_{|w| \geq R}.$$
Then, pick $f\in C ^{\infty}_c (\R ^{2d})$ and write
\begin{align*}
\left| \left\langle f , w f \right\rangle \right| &= \left| \iint_{\R ^d \times \R^d}  w(x_1-x_2) |f(x_1,x_2)| ^2 d{x_1} d{x_2} \right| \\
&\leq \iint_{\R ^d \times \R^d}  \left| w(x_1-x_2) \one_{|w|\geq R} \right| |f(x_1,x_2)| ^2 d{x_1} d{x_2} + R \iint_{\R ^d \times \R^d} |f(x_1,x_2)| ^2 d{x_1} d{x_2} .
\end{align*}
Next, using H\"older's and Sobolev's inequalities:
\begin{multline*}
\iint_{\R ^d \times \R^d}  \left| w(x_1-x_2) \one_{|w|\geq R} \right| |f(x_1,x_2)| ^2 d{x_1} d{x_2} \\
\leq \int_{\R ^d}  \left(\int_{\R ^d} |w (x_1-x_2)| ^p \one_{|w|\geq R} d{x_2} \right) ^{1/p} \left(\int_{\R ^d} |f (x_1,x_2)| ^{2q} d{x_2} \right) ^{1/q} d{x_1}
\\\leq \norm{w\one_{|w|\geq R}}_{L ^p} \int_{\R ^d} \norm{f(.,x_2)}_{H ^1} ^2 d{x_2}
\end{multline*}
with $p ^{-1} + q ^{-1} = 1$. All in all we thus have 
\begin{equation*}
\left| \left\langle f , w f \right\rangle \right| \leq \norm{w\one_{|w|\geq R}}_{L ^p} \iint_{\R ^d} |\nabla_{x_1} f(x_1,x_2)| ^2   dx_1 d{x_2} + R  \iint_{\R ^d \times \R^d} |f(x_1,x_2)| ^2 d{x_1} d{x_2}.
\end{equation*}
Since $w \in L ^p$, $\norm{w\one_{|w|\geq R}}_{L ^p} \to 0$ when $R\to \infty$ and we thus have
$$ \left| \left\langle f , w f \right\rangle \right| \leq h(R) \left| \left\langle f , -\Delta_{x_1} f \right\rangle \right| + (R + h (R) ) \left\langle f ,  f \right\rangle,$$
where $h(R)$ can be made arbitrarily small by taking $R$ large. In particular we can take $h(R)<1$, and this leads to an operator inequality of the desired form~\eqref{eq:T controls w}.

\hfill\qed

\end{remark}

The above assumptions ensure by well-known methods that $H_N$ is self-adjoint on the domain of the Laplacian, and bounded below, see~\cite{ReeSim2}. 
Our goal is to describe the ground state of~\eqref{eq:quant hamil}, i.e. a state achieving\footnote{The ground state might not exist, in which case we think of a sequence of states asymptotically achieving the infimum.} 
\begin{equation}\label{eq:quant ground state}
E(N) = \inf \sigma_{\gH ^N} H_N  = \inf_{\Psi \in \gH ^N, \norm{\Psi} = 1} \left\langle \Psi, H_N \Psi \right\rangle_{\gH ^N}.
\end{equation}
In the mean-field regime that concerns us here, we expect that any ground state satifies 
\begin{equation}\label{eq:quant factor form}
\Psi_N \approx u ^{\otimes N} \mbox{ when } N\to \infty
\end{equation}
in a sense to be made precise later, which naturally leads us to the Hartree functional
\begin{align}\label{eq:Hartree func}
\EH [u] &=  N ^{-1} \left\langle u ^{\otimes N}, H_N u ^{\otimes N} \right\rangle_{\gH ^N} = \langle u,T u \rangle_{\gH} +  \frac{1}{2} \langle u\otimes u,w \:u \otimes u \rangle_{\gH_s ^2} \nonumber\\
&= \int_{\Om} |\nabla u | ^2 + V |u| ^2 + \frac{1}{2} \iint_{\Om\times \Om} |u(x)| ^2 w(x-y) |u(y)| ^2 dxdy. 
\end{align}
We shall denote $\eH$ and $\uH$ the minimum and a minimizer for $\EH$ respectively. By the variational principle we of course have the upper bound 
\begin{equation}\label{eq:up bound Hartree}
\frac{E(N)}{N} \leq \eH 
\end{equation}
and we aim at proving a matching lower bound to obtain 
\begin{equation}\label{eq:Hartree lim formel}
 \frac{E(N)}{N} \to \eH \mbox{ when } N\to \infty.
\end{equation}

\begin{remark}[Generalizations]\label{rem:energie cin}\mbox{}\\
All the main ideas can be introduced in the preceding framework, we refer to~\cite{LewNamRou-13} for a discussion of generalizations. One can even think of the case where both $V$ and $w$ are smooth compactly supported functions if one wishes to understand the method in the simplest possible case.

A very interesting generalization consists in substituting the Laplacian in~\eqref{eq:op Schro} by a relativistic kinetic energy operator and/or including a magnetic field as described in Section~\ref{sec:forma quant}. One must then adapt the assumptions ~\eqref{eq:T controls w} and~\eqref{eq:decrease w} but the message stays the same: the approach applies as long as the many-body Hamiltonian and the limit functional are both well-defined.

Another possible generalization is the inclusion of interactions involving more than two particles at a time, to obtain functionals with higher-order non-linearities in the limit. It is of course necessary to consider a Hamiltonian in mean-field scaling by adding terms of the form e.g.
$$ \lambda_N\sum_{1\leq i < j,k \leq N} w (x_i-x_j,x_i-x_k)$$
with $\lambda_N \propto N^{-2}$ when $N\to \infty$. It is also possible to take into account a more general form than just the multiplication by a potential, under assumptions of the same type as~\eqref{eq:T controls w}.\hfill \qed
\end{remark}

\subsection{Confined systems and the strong de Finetti theorem}\label{sec:Hartree deF fort}\mbox{}\\\vspace{-0.4cm}

By ``confined system'' we mean that we are dealing with a compact setting. We will make one the two following assumptions: either
\begin{equation}\label{eq:hartree confine 1}
\Om \subset \R ^d  \mbox{ is a bounded set } 
\end{equation}
or
\begin{equation}\label{eq:hartree confine 2}
\Om = \R ^d  \mbox{ and } V(x)\to \infty \mbox{ when } |x|\to \infty
\end{equation}
with $V$ the potential appearing in~\eqref{eq:op Schro}. We will also assumte that 
\[
 V \in L ^p_{\rm loc} (\Om), \max(1,d/2) < p \leq \infty. 
\]
In both cases it is well-known~\cite{ReeSim4} that 
\begin{equation}\label{eq:resol comp}
T = -\Delta + V \mbox{ has compact resolvent}, 
\end{equation}
which allows to easily obtain strong convergence of the reduced density matrices of a ground state of~\eqref{eq:quant hamil}. One may then really think of the limit object as a quantum state with infinitely many particles. Taking inspiration from the classical setting discussed previously, the natural definition is the following:

\begin{definition}[\textbf{Bosonic state with infinitely many particles}]\label{def:etat infini}\mbox{}\\
Let $\gH$ be a complex separable Hilbert space and, for all $n\in \N$, let $\gH_s ^n$ be the corresponding bosonic $n$-particles space. We call a \emph{bosonic state} with infinitely many particles a sequence $(\gamma ^{(n)})_{n\in \N}$ of trace-class operators satisfying
\begin{itemize}
\item $\gamma ^{(n)}$ is a bosonic $n$-particles state: $\gamma ^{(n)}\in \gS ^1 (\gH_s ^n)$ is self-adjoint, positive and   
\begin{equation}\label{eq:defi inf trace}
\tr_{\gH_s ^n}[\gamma ^{(n)}] = 1. 
\end{equation}
\item the sequence $(\gamma ^{(n)})_{n\in \N}$ is consistent:
\begin{equation}\label{eq:defi inf consistance}
\tr_{n+1} [\gamma ^{(n+1)}] = \gamma ^{(n)} 
\end{equation}
where $\tr_{n+1}$ is the partial trace with respect to the last variable in $\gH ^{n+1}$.
\end{itemize}
\hfill \qed
\end{definition}

The key property is the consistency~\eqref{eq:defi inf consistance}, which ensures that the sequence under consideration does describe a physical state. Note that $\gamma ^{(0)}$ is just a real number and that consistency implies that $\tr_{\gH ^n}[\gamma ^{(n)}] = 1$ for all $n$ as soon as $\gamma ^{(0)}=1$. 

A particular case of symmetric state is a product state:

\begin{definition}[\textbf{Product state with infinitely many particles}]\mbox{}\label{def:etat produit}\\
A \emph{product state} with infinitely many particles is a sequence of trace-class operators $\gamma ^{(n)}\in \gS ^1 (\gH_s ^n)$ with
\begin{equation}\label{eq:etat produit}
\gamma ^{(n)} = \gamma ^{\otimes n},  
\end{equation}
for all $n\geq 0$ where $\gamma$ is a one-particle state. A bosonic product state is necessarily of the form 
\begin{equation}\label{eq:etat produit bosonique}
\gamma ^{(n)} = \left(|u \rangle \langle u | \right) ^{\otimes n}  = |u ^{\otimes n}\rangle \langle u ^{\otimes n} | 
\end{equation}
with $u\in S\gH$, the unit sphere of the one-body Hilbert space.\hfill\qed
\end{definition}

That bosonic product states are all of the form~\eqref{eq:etat produit bosonique} comes from the observation that if $\gamma \in \gS ^1 (\gH)$ is not pure (i.e. is not a projector), then $\gamma ^{\otimes 2}$ cannot have bosonic symmetry~\cite{HudMoo-75}.

The strong de Finetti theorem is the appropriate tool to describe these objects and specify the link between the two previous definitions. In the following form, it is due to Hudson and Moody~\cite{HudMoo-75}:

\begin{theorem}[\bf Strong quantum de Finetti]\label{thm:DeFinetti fort}\mbox{}\\
Let $\gH$ be a separable Hilbert space and $(\gamma^{(n)})_{n\in \N}$ a bosonic state with infinitely many particles on $\gH$. There exists a unique Borel probability measure $\mu \in \PP (S\gH)$ on the sphere $S\gH = \left\{ u \in\gH, \norm{u} = 1\right\}$ of $\gH$, invariant under the action\footnote{Multiplication by a constant phase $e ^{i\theta}, \theta \in \R$.} of $S ^1$, such that
\begin{equation}
\gamma^{(n)}=\int_{S\gH}|u^{\otimes n}\rangle\langle u^{\otimes n}| \, d\mu(u)
\label{eq:melange}
\end{equation}
for all $n\geq0$.
\end{theorem}

In other words, \textbf{every bosonic state with infinitely many particles is a convex combination of bosonic product states}. To deduce the validity of the mean-field approximation at the level of the ground state of~\eqref{eq:quant hamil}, it then suffices to show that the limit problem is set in the space of states with infinitely many particles, which is relatively easy in a compact setting. We are going to prove the following result:

\begin{theorem}[\textbf{Derivation of Hartree's theory for confined bosons}]\label{thm:confined}\mbox{}\\
Under the preceding assumptions, in particular~\eqref{eq:hartree confine 1} or~\eqref{eq:hartree confine 2}
$$\lim_{N\to\ii}\frac{E(N)}{N}=\eH.$$
Let $\Psi_N$ be a ground state for $H_N$, achieving the infimum~\eqref{eq:quant ground state} and 
$$ \gamma_N ^{(n)} := \tr_{n+1 \to N} \left[ |\Psi_N\rangle\langle \Psi_N|\right]$$
be its $n$-th reduced density matrix. There exists a unique probability measure $\mu$ on $\cM_{\rm H}$, the set of minimizers of $\EH$ (modulo a phase), such that, along a subsequence and for all $n\in \N$ 
\begin{equation}\label{eq:def fort result}
\lim_{N\to\ii} \gamma^{(n)}_{N}=\int_{\cM_{\rm H}} d\mu(u)\;|u^{\otimes n}\rangle\langle u^{\otimes n}|
\end{equation}
strongly in the  $\gS ^1 (\gH ^n)$ norm. In particular, if $\EH$ has a unique minimizer (modulo a constant phase), then for the whole sequence 
\begin{equation}
\lim_{N\to\ii} \gamma^{(n)}_N=|\uH^{\otimes n}\rangle\langle \uH^{\otimes n}|.
\label{eq:BEC-confined} 
\end{equation}
\end{theorem}

The ideas of the proof are essentially contained in~\cite{FanSpoVer-80,PetRagVer-89,RagWer-89}, applied to a somewhat different context however. We follow some clarifications given in~\cite[Section~3]{LewNamRou-13}. 

\begin{proof}[Proof]
We have to prove the lower bound corresponding to~\eqref{eq:up bound Hartree}. We start by writing
\begin{align}\label{eq:afort ener matrices}
\frac{E(N)}{N} &=  \frac1N \left\langle \Psi_N, H_N \Psi_N \right\rangle_{\gH ^N} = \tr_{\gH} [T \gamma_N ^{(1)}] + \frac12 \tr_{\gH_s ^2} [w\gamma_N ^{(2)}]
\nonumber 
\\ &= \frac12 \tr_{\gH_s ^2} \left[ \left(T_1 + T_2 + w \right)\gamma_N ^{(2)}\right]
\end{align}
and we now have to describe the limits of the reduced density matrices $\gamma_N ^{(1)}$ and $\gamma_N ^{(2)}$. Since the sequences $(\gamma ^{(n)}_N)_{N\in \N}$ are by definition bounded in $\gS ^1$, by a diagonal extraction argument we may assume that for all $n\in \N$
\[
 \gamma_N ^{(n)}\wto_\ast \gamma ^{(n)} \in \gS^1 (\gH_s ^n)
\]
weakly-$\ast$ in $\gS ^1 (\gH ^n)$. That is, for every compact operator $K_n$ over $\gH ^n$ we have 
$$ \tr _{\gH ^n} \left[\gamma_N ^{(n)} K_n\right] \to \tr _{\gH ^n} \left[\gamma ^{(n)} K_n\right].$$ 
We are going to show that the limit is actually strong. For this it is sufficient (see~\cite{dellAntonio-67,Robinson-70} or~\cite[Addendum H]{Simon-79}) to show that
\begin{equation}\label{eq:afort masse}
\tr _{\gH_s ^n} \left[\gamma ^{(n)} \right] = \tr _{\gH_s ^n} \left[\gamma_N ^{(n)} \right] = 1,
\end{equation}
i.e. no mass is lost in the limit. We start by noting that, using Assumption~\eqref{eq:T controls w},
\begin{align*}
\eH &\geq  \frac12 \tr_{\gH_s ^2} \left[ \left(T_1 + T_2 + w \right)\gamma_N ^{(2)}\right] 
\\&\geq (1-\beta_-) \frac12 \tr\left[(T_1 + T_2)\gamma_N ^{(2)}\right] - C \tr [\gamma_N ^{(2)}]
\\& = (1-\beta_-) \tr\left[T \gamma_N ^{(1)}\right] - C. 
\end{align*}
Thus $\tr_{\gH} [T\gamma_N ^{(1)}]$ is uniformly bounded and, up to a possible further extraction, we have 
$$ \left(T  + C_0 \right) ^{1/2} \gamma_N ^{(1)} \left(T  + C_0 \right) ^{1/2} \wto_\ast \left(T  + C_0 \right) ^{1/2} \gamma ^{(1)} \left(T  + C_0 \right) ^{1/2}$$ 
for a certain constant $C_0$. Consequently 
\begin{multline*}
1= \tr _{\gH } [\gamma_N ^{(1)} ] = \tr _{\gH } \left[ \left(T  + C_0 \right) ^{-1} \left(T  + C_0 \right) ^{1/2} \gamma_N ^{(1)}\left(T  + C_0 \right) ^{1/2}  \right] \\ \to \tr _{\gH} \left[ \left(T  + C_0 \right) ^{-1} \left(T  + C_0 \right) ^{1/2} \gamma ^{(1)}\left(T  + C_0 \right) ^{1/2}  \right] = \tr _{\gH } [\gamma ^{(1)} ]
\end{multline*}
since $\left(T  + C_0 \right) ^{-1}$ is by Assumption~\eqref{eq:resol comp} a compact operator. One obtains~\eqref{eq:afort masse} similarly by noting that 
\[
\tr_{\gH} [T\gamma_N ^{(1)}] = \frac{1}{n}\tr_{\gH ^n} \left[ \sum_{j=1} ^n T_j \gamma_N ^{(n)} \right] 
\]
is uniformly bounded in $N$ and that $\sum_{j=1} ^n T_j$ also has compact resolvent which allows for a similar argument.

We thus have, for all $n\in \N$
\[
 \gamma_N ^{(n)}\to \gamma ^{(n)}
\]
strongly in trace-class norm, and in particular, for all bounded operator $B_n$ on $\gH ^n$
$$ \tr _{\gH ^n} [\gamma_N ^{(n)} B_n] \to \tr _{\gH ^n} [\gamma ^{(n)} B_n].$$ 
Testing this convergence with $B_{n+1} = B_n \otimes \one$ we deduce
$$ \tr_{n+1} [\gamma ^{(n+1)} ] = \gamma ^{(n)}$$
and thus the sequence $(\gamma ^{(n)})_{n\in \N}$ describes a bosonic state with infinitely many particles in the sense of Definition~\ref{def:etat infini}. We apply Theorem~\ref{thm:DeFinetti fort}, which yields a measure $\mu \in \PP (S\gH)$. In view of Assumption~\eqref{eq:T controls w}, the operator $T_1+T_2 + w$ is bounded below on $\gH ^2$, say by $2 C_T$. Since $\tr_{\gH ^2} \gamma ^{(2)} = 1$ we may write
\begin{align*}
\liminf_{N\to \infty}\frac12 \tr_{\gH ^2} \left[ \left(T_1 + T_2 + w \right)\gamma_N ^{(2)}\right] &=  \liminf_{N\to \infty} \frac12 \tr_{\gH ^2} \left[ \left(T_1 + T_2 + w -2C_T \right)\gamma_N ^{(2)}\right] + C_T\\
& \geq \frac12 \tr_{\gH ^2} \left[ \left(T_1 + T_2 + w -2C_T \right)\gamma ^{(2)}\right] + C_T 
\\&= \frac12 \tr_{\gH ^2} \left[ \left(T_1 + T_2 + w  \right)\gamma ^{(2)}\right]
\end{align*}
using Fatou's lemma for positive operators. Using the linearity of the energy as a function of $\gamma ^{(2)}$ and~\eqref{eq:melange} 
\[
\liminf_{N\to \infty} \frac{E(N)}{N} \geq \int_{u\in S\gH} \frac12 \tr_{\gH ^2} \left[ \left(T_1 + T_2 + w  \right)|u ^{\otimes 2}\rangle \langle u ^{\otimes 2} |\right] d\mu(u) = \int_{u\in S\gH} \EH [u] d\mu(u) \geq \eH,
\]
which is the desired lower bound. The other statements of the theorem follow by noting that equality must hold in all the previous inequalities and thus that $\mu$ may charge only minimizers of Hartree's functional.
\end{proof}

It is clear from the preceding proof that the structure of bosonic states with infinitely many particles plays the key role. The Hamiltonian itself could be chosen in a very abstract way provided it includes a confining mechanism allowing to obtain strong limits. Several examples are mentioned in~\cite[Section~3]{LewNamRou-13}.

\subsection{Systems with no bound states and the weak de Finetti theorem}\label{sec:Hartree deF faible}\mbox{}\\\vspace{-0.4cm}

In the previous section we have used in a strong way the assumption that the system was confined in the sense of~\eqref{eq:hartree confine 1}-\eqref{eq:hartree confine 2}. These assumptions are sufficient to understand many physical cases, but it is highly desirable to be able to relax them. One is then lead to study cases where the convergence of reduced density matrices is not better than weak-$\ast$, and to describe as exhaustively as possible the possible scenarii, in the spirit of the concentration-compactness principle. A first step, before asking the question of how compactness may be lost, consists in describing the weak limits themselves. It turns out that we still have a very satisfying description. In fact, one could not have hoped for better than the following theorem, proven in~\cite{LewNamRou-13}:

\begin{theorem}[\bf Weak quantum de Finetti]\label{thm:DeFinetti faible}\mbox{}\\
Let $\gH$ be a complex separable Hilbert space and $(\Gamma_N)_{N\in \N}$ a sequence of bosonic states with $\Gamma_N \in \gS ^1 (\gH_s ^N)$. We assume that for all $n\in\N$ 
\begin{equation}\label{eq:def faible convergence}
\Gamma_N ^{(n)} \wto_\ast \gamma ^{(n)}
\end{equation}
in $\gS ^1 (\gH_s ^n)$. Then there exists a unique probability measure $\mu \in \PP (B\gH)$ on the unit ball $B\gH = \left\{ u \in\gH, \norm{u} \leq 1\right\}$ of $\gH$, invariant under the action of $S^1$, such that
\begin{equation}
\gamma^{(n)}=\int_{B\gH}|u^{\otimes n}\rangle\langle u^{\otimes n}| \, d\mu(u)
\label{eq:melange faible}
\end{equation}
for all $n\geq0$.
\end{theorem}

\begin{remark}[On the weak quantum de Finetti theorem]\label{rem:weak deF}\mbox{}\\\vspace{-0.4cm}
\begin{enumerate}
\item Assumption~\eqref{eq:def faible convergence} can always be satisfied in practice. Modulo a diagonal extraction, one may always assume that convergence holds along a subsequence. This theorem thus exactly describes all the possible weak limits for a sequence of bosonic $N$-body states when $N\to \infty$.
\item The fact that the measure lives over the unit ball in~\eqref{eq:melange faible} is not surprising since there may be a loss of mass in cases covered by the theorem. In particular it is possible that $\gamma ^{(n)}=0$ for all $n$ and then $\mu = \delta_0$, the Dirac mass at the origin.  
\item The term weak refers to ``weak convergence'' and does not indicate that this result is less general than the strong de Finetti theorem. It is in fact more general. To see this, think of the case where no mass is lost, $\tr_{\gH ^n} [\gamma ^{(n)}] = 1$. The measure $\mu$ must then be supported on the sphere and convergence must hold in trace-class norm. It is by the way sufficient to assume that $\tr_{\gH ^n} [\gamma ^{(n)}] = 1$ for a certain $n\in \N$, and the convergence is strong for all $n$ since the measure $\mu$ does not depend on $n$.
\item Ammari and Nier have slightly more general results~\cite{Ammari-hdr,AmmNie-08,AmmNie-09,AmmNie-11}. In particular, it is not necessary to start from a state with a fixed particle number. One can consider a state on Fock space provided suitable bounds on its particle number (seen as a random variable in this framework) are available.
\item Uniqueness of the measure follows from a simple argument. Here we will mostly be interested in the existence, which is sufficient for static problems. For time-dependent problems on the contrary, uniqueness is crucial~\cite{AmmNie-08,AmmNie-09,AmmNie-11,CheHaiPavSei-13}.
\end{enumerate}

\hfill\qed
\end{remark}

Coming back to the derivation of Hartree's theory, let us recall that the energy upper bound~\eqref{eq:up bound Hartree} is always true. We are thus only looking for lower bounds. A case where knowning the weak-$\ast$ limit of reduced density matrices is then that of a weakly-$\ast$ lower semi-continuous functional. This remark may seem of marginal interest, but this case does cover a number of physically relevant systems, those with no bound states. 

We are going to prove the validity of Hartree's approximation in this case, using Theorem~\ref{thm:DeFinetti faible}. We shall work in $\R ^d$ and assume that the potential $V$ in~\eqref{eq:op Schro} is non-trapping in every direction:
\begin{equation}\label{eq:sans confinement}
V \in L ^p (\R ^d) + L ^{\infty} (\R ^d), \max{1,d/2}\leq p <\infty,  V(x) \to 0 \mbox{ when } |x|\to \infty.
\end{equation}
The assumption about the absence of bound states concerns the interaction potential $w$. It is materialized par the inequality\footnote{Or an appropriate variant when one considers a different kinetic energy, cf~Remark~\ref{rem:energie cin}.}
\begin{equation}\label{eq:no bound states}
- \Delta + \frac{w}{2} \geq 0 
\end{equation}
as operators. This means that $w$ is not attractive enough for particles to form bound states such as molecules etc ... Indeed, because of Assumption~\eqref{eq:decrease w}, $- \Delta + \frac{w}{2}$ can have only negative eigenvalues, its essential spectrum starting at $0$. In view of~\eqref{eq:no bound states}, it can in fact not have any eigenvalue at all, and thus no eigenfunctions which are by definition the bound states of the potential. A particular example is that of a purely repulsive potential $w\geq 0$.

\medskip

Under Assumption~\eqref{eq:no bound states}, particles that might escape to infinity no longer see the one-body potential $V$ and then necessarily carry a positive energy that can be neglected for a lower bound to the total energy. Particles staying confined by the potential $V$ are described by the weak-$\ast$ limits of density matrices. We can apply the weak de Finetti theorem to the latter, and this leads to the next theorem, for whose statement we need the notation
\begin{equation}\label{eq:Hartree perte masse}
\eH (\lambda) := \inf_{\norm{u} ^2 = \lambda} \EH [u],\quad 0\leq \lambda \leq 1
\end{equation}
for the Hartree energy in the case of a loss of mass. Under assumption~\eqref{eq:no bound states} it is not difficult to show that for all $ 0 \leq \lambda \leq 1$
\begin{equation}\label{eq:Hartree lambda}
\eH (\lambda) \geq \eH (1) = \eH 
\end{equation}
by constructing trial states made of two well-separated pieces of mass.

\begin{theorem}[\textbf{Derivation of Hartree's theory in the absence of bound states}]\label{thm:wlsc}\mbox{}\\
Under the previous assumptions, in particular~\eqref{eq:sans confinement} and~\eqref{eq:no bound states} we have
$$\lim_{N\to\ii}\frac{E(N)}{N}=\eH.$$
Let $\Psi_N$ be a sequence of approximate ground states, satisfying 
$$ \langle \Psi_N, H_N \Psi_N \rangle = E(N) + o (N)$$ 
in the limit $N\to \infty$ and 
$$ \gamma_N ^{(n)} := \tr_{n+1 \to N} \left[ |\Psi_N\rangle\langle \Psi_N|\right]$$
be the corresponding $n$-th reduced density matrices. There exists a probability measure $\mu$ supported on 
$$\cM_{\rm H} =\left\{ u\in B\gH, \; \EH[u] = \eH \left(\norm{u} ^2 \right) = \eH (1)\right\}$$ 
such that, along a subsequence and for all $n\in \N$ 
\begin{equation}\label{eq:def faible result}
\gamma^{(n)}_{N} \wto_\ast \int_{\cM_{\rm H}} d\mu(u)\;|u^{\otimes n}\rangle\langle u^{\otimes n}|
\end{equation}
in $\gS ^1 (\gH ^n).$ In particular, if $\eH$ has a unique minimizer $\uH$, with $\norm{\uH} = 1$, then for the whole sequence,
\begin{equation}
\lim_{N\to\ii} \gamma^{(n)}_N=|\uH^{\otimes n}\rangle\langle \uH^{\otimes n}|
\label{eq:BEC-confined faible} 
\end{equation}
strongly in trace-class norm.
\end{theorem}

\begin{remark}[Case where mass is lost]\label{rem:perte masse}\mbox{}\\
Note that it is possible for the convergence in the above statement to be only weak-$\ast$, which covers a certain physical reality. If the one-body potential is not attractive enough to retain all particles, we will typically have a scenario where
\[
\begin{cases}
\eH (\lambda) = \eH (1) \mbox{ for } \lambda_c \leq \lambda \leq 1 \\
\eH(\lambda) < \eH (1) \mbox{ for } 0 \leq \lambda < \lambda_c 
\end{cases}
\]
where $\lambda_c$ is a critical mass that can be bound by the potential $V$. In this case, $\eH(\lambda)$ will not be achieved if $\lambda_c < \lambda \leq 1$ and one will have a minimizer for Hartree's energy only for a mass $0 \leq \lambda \leq \lambda_c$. If for example the minimizer  $\uH$ at mass $\lambda_c$ is unique modulo a constant phase, Theorem~\ref{thm:wlsc} shows that 
\[
\gamma^{(n)}_{N} \wto_\ast  |\uH^{\otimes n}\rangle\langle \uH^{\otimes n}|
\]
and one should note that the limit has a mass $\lambda_c ^n <1$. This scenario actually happens in the case of a ``bosonic atom'', see Section 4.2 in~\cite{LewNamRou-13}.\hfill\qed
\end{remark}

Theorem~\ref{thm:wlsc} was proved in~\cite{LewNamRou-13}. In order to be able to apply Theorem~\ref{thm:DeFinetti faible} we start with the following observation:

\begin{lemma}[\textbf{Lower semi-continuity of an energy with no bound state}]\label{lem:wlsc}\mbox{}\\
Under the previous assumptions, let $\gamma^{(1)}_N,\gamma^{(2)}_N\geq0$ be two sequences satisfying 
\begin{equation*}
\tr_{\gH^2}\gamma^{(2)}_N=1, \quad \gamma_N^{(1)} =\Tr_{2}\gamma^{(2)}_N
\end{equation*}
as well as $\gamma^{(k)}_N\wto_*\gamma^{(k)}$ weakly-$\ast$ in $\gS^1(\gH^k)$ for $k=1,2$. Then
\begin{equation}
\liminf_{N\to \infty} \left( \Tr_\gH[T \gamma_N^{(1)}]+\frac{1}{2}\Tr_{\gH^2}[w \gamma_N^{(2)}] \right) \ge \Tr_\gH[T \gamma^{(1)}]+\frac{1}{2}\Tr_{\gH^2}[w \gamma^{(2)}].
\label{eq:energy wlsc} 
\end{equation}
\end{lemma}

\begin{proof}
We shall need two trunctation functions $0\leq \chi_R,\eta_R\leq 1$ satisfying
\[
 \chi_R ^2 + \eta_R ^2 \equiv 1, \supp (\chi_R)\subset B(0,2R), \supp (\eta_R) \subset B(0,R) ^c , |\nabla \chi_R| + |\nabla \eta_R| \leq C R ^{-1}. 
\]
It is easy to show the IMS formula: 
\begin{multline}\label{eq:IMS}
-\Delta  = \chi_R (-\Delta) \chi_R +  \eta_R (-\Delta) \eta_R - |\nabla \chi_R | ^2 -   |\nabla \eta_R | ^2 
\\ \geq \chi_R (-\Delta) \chi_R +  \eta_R (-\Delta) \eta_R - C R ^{-2}
\end{multline}
as an operator. For the one-body part of the energy we then easily have 
$$  \tr_{\gH} [T \gamma_N ^{(1)}] \geq \tr_{\gH} [T \chi_R\gamma_N ^{(1)} \chi_R] + \tr_{\gH} [-\Delta \eta_R \gamma_N ^{(1)} \eta_R] + r_1(N,R)$$
with
$$ r_1 (N,R) \geq - C  R ^{-2} + \tr_{\gH} [ \eta_R ^2 V \gamma_N ^{(1)} ]$$
and thus
$$\liminf_{R\to \infty} \liminf_{N\to \infty} r_1(N,R) = 0$$
because in view of~\eqref{eq:sans confinement} we have 
\begin{equation}\label{eq:infinity one body}
  \tr_{\gH} [\eta_R ^2 V \gamma_N ^{(1)}] = \int_{\R ^d} \eta_R ^2 (x) V(x) \rho_N ^{(1)} (x) dx \to 0 \mbox{ when } R\to\infty 
\end{equation}
uniformly in $N$. Here we have denoted $\rho_N ^{(1)} $ the one-body density of $\gamma_N ^{(1)}$, formally defined by
$$\rho_N ^{(1)} (x) = \gamma_N ^{(1)}(x,x),$$
where we identify $\gamma_N ^{(1)}$ and its kernel. 
The details of the proof of~\eqref{eq:infinity one body} are left to the reader. One should use the a priori bound $\tr [ T \gamma_N ^{(1)}] \leq C$ to obtain a bound on $\rho_N ^{(1)}$ in some appropriate $L ^q$ space via Sobolev embeddings, then use H\"older's inequality and the fact that, for the conjugate exponent $p$, 
$$ \norm{\eta_R ^2 V}_{L^p (\R ^d)} \to 0$$ 
by assumption.

To deal with the interaction term we introduce
\[
 \rho_N ^{(2)} (x,y):= \gamma_N ^{(2)} (x,y;x,y)
\]
the two-body density of $\gamma_N ^{(2)}$ (identified with its kernel $\gamma_N ^{(2)} (x',y';x,y)$) and write 
\begin{align*}
\tr_{\gH ^2} [w\gamma_N ^{(2)}] &= \iint_{\R ^d\times\R^d } w(x-y) \rho_N ^{(2)} (x,y) dxdy\\
&= \iint_{\R ^d\times\R^d } w(x-y) \chi_R ^2 (x) \rho_N ^{(2)} (x,y) \chi_R ^2 (y)  dxdy 
\\&+ \iint_{\R ^d\times\R^d } w(x-y) \eta_R ^2 (x) \rho_N ^{(2)} (x,y) \eta_R ^2 (y) dxdy 
\\&+ 2 \iint_{\R ^d\times\R^d } w(x-y) \chi_R ^2 (x)  \rho_N ^{(2)}(x,y) \eta_R ^2 (y) dxdy 
\\&= \iint_{\R ^d\times\R^d } w(x-y) \chi_R ^2 (x)\rho_N ^{(2)} (x,y)  \chi_R ^2 (y) dxdy 
\\&+ \iint_{\R ^d\times\R^d } w(x-y) \eta_R ^2 (x)  \rho_N ^{(2)} (x,y) \eta_R ^2 (y) dxdy 
\\&+ r_2(N,R)
\\&= \tr_{\gH ^2} [w\:\chi_R \otimes \chi_R \gamma_N ^{(2)} \chi_R \otimes \chi_R] + \tr_{\gH ^2} [w \:\eta_R \otimes \eta_R \gamma_N ^{(2)} \eta_R \otimes \eta_R]
\\&+ r_2(N,R)
\end{align*}
where
$$\liminf_{R\to \infty} \liminf_{N\to \infty} r_2(N,R) = 0$$
by a standard concentration-compactness argument that we now sketch.  One has to show that 
$$ \lim_{R\to \infty} \lim_{N\to \infty} \iint_{\R ^d\times\R^d } |w(x-y)| \chi_R ^2 (x)  \rho_N ^{(2)}(x,y) \eta_R ^2 (y) dxdy. $$
Since $\chi_R ^2 (x) |w(x-y)| \eta_{4R} ^2 (y) \leq \1_{|x-y|\geq R} |w(x-y)|$, the term 
$$ \iint_{\R ^d\times\R^d } |w(x-y)| \chi_R ^2 (x)  \rho_N ^{(2)}(x,y) \eta_{4R} ^2 (y) dxdy$$
can be dealt with similarly as~\eqref{eq:infinity one body}. We leave again the details to the reader. There thus remains to control 
$$ \iint_{\R ^d\times\R^d } |w(x-y)| \chi_R ^2 (x)  \rho_N ^{(2)}(x,y) \left( \eta_R ^2 (y) - \eta_{4R}  ^2(y)\right)dxdy.$$
for which we claim that
\begin{multline*}
\lim_{R\to \infty}\lim_{k\to \infty} \left(  \iint |w(x-y)| \chi_R^2(x)[\eta_R^2 (y) -\eta_{4R}^2(y) ] \rho^{(2)}_{N_k}(x,y)\,dx\,dy \right)\\
\leq\lim_{R\to \infty}\lim_{k\to \infty} \left(  \iint |w(x-y)| \1(R\leq|y|\leq 8R)\rho^{(2)}_{N_k}(x,y)\,dx\,dy \right)=0 
\end{multline*}
for an appropriate subsequence $(N_k)$. Let us introduce the concentration functions
$$ Q_{N}(R):= \iint |w(x-y)| \1(|y|\geq R) \rho^{(2)}_N(x,y)\,dx\,dy. $$
For every $N$, the function $R\mapsto Q_N(R)$ is decreasing on $[0,\infty)$. Moreover, 
$$0\le Q_N (R) \le \Tr_{\gH^2}[|w| \gamma_N^{(2)}] \le C_0$$
using our assumptions on $w$ and the a priori bound $\tr \left[ T \gamma_N ^{(1)}\right] \leq C$ again. Therefore, by Helly's selection principle, there exists a subsequence $N_k$ and a decreasing function $Q:[0,\infty) \to [0,C_0]$ such that $Q_{N_k}(R)\to Q(R)$ for all $R\in [0,\infty)$. Since $\lim_{R\to \infty} Q(R)$ exists, we conclude that 
$$\lim_{R\to \infty} \lim_{k\to \infty} (Q_{N_k}(R) - Q_{N_k}(8R))= \lim_{R\to \infty} (Q(R)-  Q(8R) )=0.$$
This is the desired convergence.

At this stage we thus have
\begin{multline}\label{eq:proof wlsc}
\liminf_{N\to \infty} \Tr_\gH[T \gamma_N^{(1)}]+\frac{1}{2}\Tr_{\gH^2}[w \gamma_N^{(2)}] \geq 
\\\liminf_{R\to \infty} \liminf_{N\to \infty} \Tr_\gH[T \chi_R \gamma_N^{(1)} \chi_R] + \frac{1}{2}\Tr_{\gH^2}[w \:\chi_R \otimes \chi_R \gamma_N^{(2)}\chi_R \otimes \chi_R] 
\\+ \liminf_{R\to \infty} \liminf_{N\to \infty} \Tr_\gH[-\Delta \eta_R \gamma_N^{(1)} \eta_R]+\frac{1}{2} \Tr_{\gH^2}[w \:\eta_R \otimes \eta_R \gamma_N^{(2)}\eta_R \otimes \eta_R]
\end{multline}
The terms on the second line give the right-hand side of~\eqref{eq:energy wlsc}. Indeed, recalling that $T= - \Delta + V$, using Fatou's lemma for $-\Delta \geq 0$ and the fact that 
$$\chi_R \gamma_N ^{(1)}\chi_R \to \chi_R \gamma ^{(1)}\chi_R$$
in norm since $\chi_R$ has compact support, we have
$$ \liminf_{N\to \infty} \Tr_\gH[T \chi_R \gamma_N^{(1)} \chi_R] \geq \Tr_\gH[T \chi_R \gamma^{(1)} \chi_R]. $$
It then suffices to recall that $\chi_R \to 1$ pointwise to conclude. The interaction term is dealt with in a similar way, using the strong convergence 
$$\chi_R \otimes \chi_R \gamma_N ^{(2)}\chi_R \otimes \chi_R  \to \chi_R \otimes \chi_R  \gamma ^{(2)} \chi_R \otimes \chi_R $$
and then the pointwise convergence $\chi_R \to 1$.

The terms on the third line of~\eqref{eq:proof wlsc} form a positive contribution that we may drop from the lower bound. To see this, note that since $0\leq \eta_R \leq 1$ we have $\eta_R \otimes \one \geq \eta_R \otimes \eta_R$ and (with $\tr_2$ the partial trace with respect to the second variable)
\[
 \eta_R \gamma_N ^{(1)} \eta_R \geq \tr_2 [\eta_R \otimes \eta_R \gamma_N ^{(2)} \eta_R \otimes \eta_R],
\]
which gives, by symmetry of $\gamma_N ^{(2)}$, 
\begin{multline*}
\Tr_\gH[-\Delta \eta_R \gamma_N^{(1)} \eta_R]+\frac{1}{2} \Tr_{\gH^2}[w \:\eta_R \otimes \eta_R \gamma_N^{(2)}\eta_R \otimes \eta_R] \\ 
\geq \frac{1}{2} \Tr_{\gH^2}\left[\left( (-\Delta) \otimes \one + \one \otimes (-\Delta) + w \right)\:\eta_R \otimes \eta_R \gamma_N^{(2)}\eta_R \otimes \eta_R\right]
\end{multline*}
and it is not difficult\footnote{Just decouple the center of mass from the relative motion of the two particles, i.e. make the change of variables $(x_1,x_2) \mapsto (x_1 + x_2, x_1-x_2)$.} to see that Assumption~\eqref{eq:no bound states} implies
\[
(-\Delta) \otimes \one + \one \otimes (-\Delta) + w \geq 0
\]
which ensures the positivity of the third line of~\eqref{eq:proof wlsc} and concludes the proof.
\end{proof}

We now conclude the 

\begin{proof}[Proof of Theorem~\ref{thm:wlsc}.]
Starting from a sequence $\Gamma_N = |\Psi_N\rangle \langle\Psi_N| $ of $N$-body states we extract subsequences as in the proof of Theorem~\ref{thm:confined} to obtain 
\[
 \Gamma_N ^{(n)}\wto_\ast \gamma ^{(n)}.
\]
Using Lemma~\ref{lem:wlsc}, we obtain 
\begin{align*}
\liminf_{N\to \infty} N ^{-1} \tr_{\gH ^N} [H_N \Gamma_N] &= \liminf_{N\to \infty} \left(\tr_{\gH} [T \Gamma_N ^{(1)}] + \frac12 \tr_{\gH ^2} [w\Gamma_N ^{(2)}] \right)
\\&\geq \tr_{\gH} [T \gamma ^{(1)}] + \frac12 \tr_{\gH ^2} [w\gamma ^{(2)}]
\end{align*}
and there remains to apply Theorem~\ref{thm:DeFinetti faible} to the sequence $(\gamma ^{(n)})_{n\in \N}$ to obtain
\[
 \liminf_{N\to \infty} N ^{-1} \tr_{\gH ^N} [H_N \Gamma_N] \geq \int_{B\gH} \EH[u] d\mu(u) \geq \eH
\]
using~\eqref{eq:Hartree lambda} and the fact that $\int_{B\gH} d\mu(u) = 1$. Once again, the other conclusions of the theorem easily follow by inspecting the cases of equality in the previous estimates.
\end{proof}

For later use, we note that during the proof of Lemma~\ref{lem:wlsc} we have obtained the intermediary result~\eqref{eq:proof wlsc} without using the assumption that $w$ has no bound state. We write this as a lemma that we shall use again in Chapter~\ref{sec:Hartree}.

\begin{lemma}[\textbf{Energy localization}]\label{lem:loc ener}\mbox{}\\
Under assumptions~\eqref{eq:decrease w} and~\eqref{eq:sans confinement}, let $\Psi_N$ be a sequence of almost minimizers for $E(N)$: 
$$ \langle \Psi_N, H_N \Psi_N \rangle = E(N) + o (N)$$ 
and $\gamma_N ^{(k)}$ the associated reduced density matrices. We have  
\begin{multline}\label{eq:loc ener}
\liminf_{N\to \infty} \frac{E(N)}{N} = \liminf_{N\to \infty} \Tr_\gH[T \gamma_N^{(1)}]+\frac{1}{2}\Tr_{\gH^2}[w \gamma_N^{(2)}] \geq 
\\\liminf_{R\to \infty} \liminf_{N\to \infty} \Tr_\gH[T \chi_R \gamma_N^{(1)} \chi_R] + \frac{1}{2}\Tr_{\gH^2}[w \:\chi_R ^{\otimes 2} \gamma_N^{(2)}\chi_R ^{\otimes 2}] 
\\+ \liminf_{R\to \infty} \liminf_{N\to \infty} \Tr_\gH[-\Delta \eta_R \gamma_N^{(1)} \eta_R]+\frac{1}{2} \Tr_{\gH^2}[w \:\eta_R ^{\otimes 2} \gamma_N^{(2)}\eta_R ^{\otimes 2}]
\end{multline}
where $0\leq \chi_R \leq 1$ is $C^1$, with support in $B(0,R)$, and $\eta_R = \sqrt{1-\chi_R ^2}$.
\end{lemma}

\subsection{Links between various structure theorems for bosonic states}\label{sec:rel deF}\mbox{}\\\vspace{-0.4cm}

We have just introduced two structure theorems for many-particles bosonic systems. These indicate that, morally, if $\Gamma_N$ is a $N$-body bosonic state on a separable Hilbert space~$\gH$, there exists a probability measure $\mu \in \PP (\gH)$ on the one-body Hilbert space such that 
\begin{equation}\label{eq:def quant formel}
\Gamma_N ^{(n)} \approx \int_{u\in \gH} |u ^{\otimes n} \rangle \langle u ^{\otimes n}| d\mu (u) 
\end{equation}
when $N$ is large and $n$ is fixed. Chapters~\ref{sec:deF finite dim} and~\ref{sec:locFock} of these notes are for a large devoted to the proofs of these theorems ``\`a la de Finetti''. As mentioned in Remark~\ref{rem:weak deF}, Theorem~\ref{thm:DeFinetti faible} is actually more general than Theorem~\ref{thm:DeFinetti fort}, and we will thus prove the former. 

In order to better appreciate the importance of the weak theorem in infinite dimensional spaces, we emphasize that the key property allowing to prove the strong theorem is the consistency~\eqref{eq:defi inf consistance}. When starting from the reduced density matrices $\Gamma_N ^{(n)}$ of a $N$-body state $\Gamma_N$ and extracting weakly-$\ast$ convergent subsequences to define a hierarchy $\left(\gamma ^{(n)}\right)_{n\in \N}$,
$$ \Gamma_N ^{(n)}\wto_\ast \gamma ^{(n)},$$
one only has
$$ \tr_{n+1}\left[\gamma ^{(n+1)}\right] \leq \lim_{N\to \infty} \tr_{n+1} \left[\gamma_N ^{(n+1)}\right] = \lim_{N\to \infty} \gamma_N ^{(n)} = \gamma ^{(n)}.$$
because \emph{the trace is not continuous }\footnote{This is a fancy way of characterizing infinite dimensional spaces.} \emph{for the weak-$\ast$ topology}, only lower semi-continuous. Obviously, the relation 
\begin{equation}\label{eq:sous consist}
\tr_{n+1}\left[\gamma ^{(n+1)}\right] \leq  \gamma ^{(n)}
\end{equation}
is not sufficient to prove a de Finetti theorem. A simple counter-example is given by the sequence of reduced density matrices of a one-body state $v\in S\gH$: 
$$\gamma ^{(0)} = 1,\: \gamma ^{(1)} = |v\rangle \langle v|,\: \gamma ^{(n)} = 0 \mbox{ for }n\geq 2.$$ 

\medskip

In these notes we have chosen to prove the weak de Finetti theorem as constructively as possible. Before we announce the plan of the proof, we shall say a few words of the much more abstract approach of the historical references~\cite{Stormer-69,HudMoo-75}. These papers actually contain a version of the theorem that is even stronger than Theorem~\ref{thm:DeFinetti faible}. This result applies to ``abstract'' states that are not necessarily normal or bosonic. These may be defined as follows:

\begin{definition}[\textbf{Abstract state with infinitely many particles}]\label{def:etat infini abs}\mbox{}\\
Let $\gH$ be a complex separable Hilbert space and for all $n\in \N$, $\gH ^n = \bigotimes ^n \gH$ be the corresponding $n$-body space. We call an \emph{abstract state with infinitely many particles} a sequence $(\omega ^{(n)})_{n\in \N}$ where
\begin{itemize}
\item $\omega ^{(n)}$ is an abstract $n$-body state: $\omega ^{(n)}\in (\gB (\gH ^n)) ^*$, the dual of the space of bounded operators on $\gH ^n$, $\omega ^{(n)}\geq 0$ and 
\begin{equation}\label{eq:masse abstraite}
\omega ^{(n)} (\one_{\gH ^n}) = 1. 
\end{equation}
\item $\omega ^{(n)}$ is symmetric in the sense that
\begin{equation}\label{eq:sym abstraite}
\omega ^{(n)} \left( B_1 \otimes \ldots \otimes B_n \right) = \omega ^{(n)} \left( B_{\sigma(1) }\otimes \ldots \otimes B_{\sigma(n)} \right) 
\end{equation}
for all $B_1,\ldots,B_n \in \gB (\gH)$ and all permutation $\sigma \in \Sigma_n$. 
\item the sequence $(\omega ^{(n)})_{n\in \N}$ is consistent:
\begin{equation}\label{eq:consistance abstraite}
\omega ^{(n+1)} \left(B_1 \otimes \ldots \otimes B_n \otimes \one_{\gH}\right)=  \omega ^{(n)} \left(B_1 \otimes \ldots \otimes B_n \right)
\end{equation}
for all $B_1,\ldots,B_n \in \gB (\gH)$.
\end{itemize}
\hfill \qed
\end{definition}

An abstract state is in general not normal, i.e. it does not fit in the following definition:

\begin{definition}[\textbf{Normal state, locally normal state}]\label{def:etat infini normal}\mbox{}\\
Let $\gH$ be a complex separable Hilbert space and $(\omega ^{(n)})_{n\in \N}$  be an abstract state with infinitely many particles. We say that   $(\omega ^{(n)})_{n\in \N}$ is locally normal if $\omega ^{(n)}$ is normal for all $n\in \N$, i.e. there exists $\gamma ^{(n)}\in \gS ^1 (\gH ^n)$ a trace-class operator such that
\begin{equation}\label{eq:etat normal}
\omega ^{(n)} (B_n) = \tr_{\gH ^n} [\gamma ^{(n)} B_n] 
\end{equation}
for all $B_n \in \gB (\gH ^n)$.
\hfill \qed
\end{definition}

Identifying trace-class operators with the associated normal states we readily see that Definition~\ref{def:etat infini} is a particular case of abstract state with infinitely many particles. Note that (by the spectral theorem), the set of convex combinations of pure states (orthogonal projectors) co\"incides with the trace class. A non-normal abstract state is thus not a mixed state, i.e. not a statistical superposition of pure states. The physical interpretation of a non-normal state is thus not obvious, at least in the type of settings that these notes are concerned with.

The consistency notion~\eqref{eq:consistance abstraite} is the natural generalization of~\eqref{eq:defi inf consistance} but it is important to note that the symmetry~\eqref{eq:sym abstraite} is  weaker than bosonic symmetry. It in fact corresponds to classicaly indistinguishable particles (the modulus of the wave-function is symmetric but not the wave-function itself). One may for example note that if $\omega  ^{(n)}$ is normal in the sense of~\eqref{eq:etat normal}, then $\gamma ^{(n)}\in \gS ^1 (\gH ^n)$ satisfies
\[
U_{\sigma}  \gamma ^{(n)} U_{\sigma} ^* = \gamma ^{(n)}
\]
where $U_{\sigma}$ is the unitary operator permuting the $n$ particles according to $\sigma \in \Sigma_n$. Bosonic symmetry corresponds to the stronger constraint 
\[
U_{\sigma}  \gamma ^{(n)} = \gamma ^{(n)} U_{\sigma} ^* = \gamma ^{(n)},
\]
cf Section~\ref{sec:forma quant}. 

We also have a notion of product state that generalizes Definition~\ref{def:etat produit}.

\begin{definition}[\textbf{Abstract product state with infinitely many particles}]\mbox{}\label{def:etat produit abs}\\
We call \emph{abstract product state} an abstract state with infinitely many particles such that
\begin{equation}\label{eq:etat produit abstrait}
\omega ^{(n)} = \omega ^{\otimes n} 
\end{equation}
for all $n\in \N$, where $\omega \in (\gB(\gH))^*$ is an abstract one-body state (in particular $\omega \geq 0$ and $\omega (\one_{\gH}) = 1$).\hfill\qed
\end{definition}

The most general form of the quantum de Finetti theorem\footnote{At least the most general form the author is aware of.} says that \textbf{every abstract state with infinitely many particles is a convex combination of product states}:

\begin{theorem}[\bf Abstract quantum de Finetti]\label{thm:DeFinetti abs}\mbox{}\\
Let $\gH$ be a complex separable Hilbert space and $(\omega^{(n)})_{n\in \N}$ an abstract bosonic state with infinitely many particles built on $\gH$. There exists a unique probability measure $\mu \in \PP \left( (\gB (\gH)) ^* \right)$ on the dual of the space of bounded perators on  $\gH$ such that
\begin{equation}\label{eq:sup mes deF abs}
\mu \left( \left\{ \omega \in (\gB (\gH)) ^*, \; \omega \geq 0, \; \omega (\one_{\gH}) = 1 \right\} \right) = 1 
\end{equation}
and
\begin{equation}
\omega^{(n)}=\int \omega ^{\otimes n} \, d\mu(\om)
\label{eq:melange abs}
\end{equation}
for all $n\geq0$.
\end{theorem}

\begin{remark}[On the abstract quantum de Finetti theorem]\label{rem:abs deF}\mbox{}\\\vspace{-0.4cm}
\begin{enumerate}
\item This result was first proven by St\o{}rmer~\cite{Stormer-69}. Hudson and Moody~\cite{HudMoo-75} then gave a simpler proof by adapting the Hewitt-Savage proof of the classical de Finetti theorem~\ref{thm:HS}: they prove that product states are the extremal points of the convex set of abstract states with infinitely many particles. The existence of the measure is then a consequence of the Choquet-Krein-Milman theorem. This approach does require the notion of abstract states and does not give a direct proof of Theorem~\ref{thm:DeFinetti fort}.
\item Hudson and Moody~\cite{HudMoo-75} deduce the strong de Finetti theorem from the abstract theorem. An adaptation of their method (see~\cite[Appendix A]{LewNamRou-13}) shows that the weak de Finetti theorem is also a consequence of the abstract theorem.
\item This theorem has been used to derive Hartree-type theories for abstract states without bosonic symmetry in~\cite{FanSpoVer-80,RagWer-89,PetRagVer-89}. To recover the usual Hartree theory one must be able to show that the limit state is (locally) normal. In finite dimensional spaces, $\gB (\gH)$ of course co\"incides with the space of compact operators, which implies that any abstract state is normal. This difficulty thus does not occur in this setting.  
\item There also exist quantum generalizations of the probabilistic versions of the classical de Finetti theorem (that is, those dealing with sequences of random variables~\cite{Aldous-85,Kallenberg-05}), see e.g.~\cite{KosSpei-09}. These are formulated in the context of free probability and require additional symmetry assumptions besides permutation symmetry.
\end{enumerate}
\hfill\qed
\end{remark}

At this stage we thus have the scheme (``deF'' stands for de Finetti)
\begin{equation}\label{eq:schema deF 1}
\boxed{\mbox{abstract deF } \Rightarrow \mbox{ weak deF } \Rightarrow \mbox{ strong deF } },
\end{equation}
but the proof of the weak de Finetti theorem we are going to present follows a different route, used in~\cite{LewNamRou-13,LewNamRou-13b}. It starts from the finite dimensional theorem\footnote{In finite dimension there is no need to distinguish between the weak and the strong version.}:
\begin{equation}\label{eq:schema deF 2}
\boxed{\mbox{finite-dimensional deF } \Rightarrow \mbox{ weak deF } \Rightarrow \mbox{ strong deF } }.
\end{equation}
This approach leads to a somewhat longer proof than the scheme~\eqref{eq:schema deF 1} starting from the Hudson-Moody proof of Theorem~\ref{thm:DeFinetti abs}. This detour is motivated by five main practical and aesthetic reasons:
\begin{enumerate}
\item The proof following~\eqref{eq:schema deF 2} is simpler from a conceptual point of view: it requires neither the notion of abstract states nor the use of the Choquet-Krein-Milman theorem.
\item Thanks to recent progress, due mainly to the quantum information community~\cite{ChrKonMitRen-07,Chiribella-11,Harrow-13,KonRen-05,FanVan-06,LewNamRou-13b}, we have a completely constructive proof the finite dimensional quantum de Finetti theorem at our disposal. One first proves explicit estimates thanks for a construction for finite $N$, in the spirit of the Diaconis-Freedman approach to the classical case. Then one passes to the limit as in the proof of the Hewitt-Savage theorem we have presented in Section~\ref{sec:DF}.
\item The first implication in Scheme~\eqref{eq:schema deF 2} is also essentially constructive, thanks to Fock-space localization techniques used e.g. in~\cite{Ammari-04,DerGer-99,Lewin-11}. These tools are inherited from the so-called ``geometric'' methods~\cite{Enss-77,Enss-78,Sigal-82,Simon-77} that adapt to the $N$-body problem localization ideas natural in the one-body setting. These allow (amongst other things)  a fine description of the lack of compactness due to loss of mass at infinity, in the spirit of the concentration-compactness principle~\cite{Lions-84,Lions-84b}.
\item In particular, the proof of the first implication in~\eqref{eq:schema deF 2} yields a few corollaries which will allow us to prove the validity of Hartree's theory in the general case. When the assumptions made in Section~\ref{sec:Hartree deF faible} do not hold, the weak de Finetti theorem and its proof according to~\eqref{eq:schema deF 1} are not sufficient to conclude: Particles escaping to infinity may form negative-energy bound states. The localization methods we are going to discuss will allow us to analyze this phenomenon. 
\item In Chapter~\ref{sec:NLS} we will deal with a case where the interaction potential depends on $N$ to derive non-linear Schr\"odinger theories in the limit. This amounts to taking a limit where $w$ converges to a Dirac mass \emph{simultaneously} to the $N\to \infty$ limit. In this case, compactness arguments will not be sufficient and the explicit estimates we shall obtain along the proof of the finite dimensional de Finetti theorem will come in handy.
\end{enumerate}

An alternative point of view on the proof strategy~\eqref{eq:schema deF 2} is given by the Ammari-Nier approach~\cite{Ammari-hdr,AmmNie-08,AmmNie-09,AmmNie-11}, based on semi-classical analysis methods. The relation between the two approaches will be discussed below. 

As noted above, the second implication in~\eqref{eq:schema deF 2} is relatively easy. The following chapters deal with the first two steps of the strategy. They contain the proof of the weak de Finetti theorem and several corollaries and intermediary results. The finite dimensional setting  (where the distinction between the strong and the weak theorems is irrelevant) is discussed in Chapter~\ref{sec:deF finite dim}. The localization methods allowing to prove the first implication in~\eqref{eq:schema deF 2} are the subject of Chapter~\ref{sec:locFock}.

\newpage

\section{\textbf{The quantum de Finetti theorem in finite dimensonal spaces}}\label{sec:deF finite dim}

This chapter deals with the starting point of the proof strategy~\eqref{eq:schema deF 2}, that is a proof of the strong de Finetti theorem in the case of a \emph{finite dimensional} complex Hilbert space~$\gH$,
\[
 \dim \gH = d.
\]
In this case, strong and weak-$\ast$ convergences in $\gS ^1 (\gH ^n)$ are the same and thus there is no need to distinguish between the strong and the weak de Finetti theorem. The main advantage of working in finite dimensions is the possibility to obtain explicit estimates, in the spirit of the Diaconis-Freedman theorem (with a completely different method, though). We are going to prove the following result, which gives bounds in trace-class norm:

\begin{theorem}[\textbf{Quantitative quantum de Finetti}]\label{thm:DeFinetti quant} \mbox{}\\
Let $\Gamma_N \in \gS ^1 (\gH_s ^N)$ be a bosonic state over $\gH_s^N$ and $\gamma_N ^{(n)}$ its reduced density matrices. There exists a probability measure $\mu_N \in \PP (S\gH)$ such that, denoting 
\begin{equation}\label{eq:finite deF etat}
\Gammat _N = \int_{u\in S \gH} |u ^{\otimes N}\rangle \langle u ^{\otimes N}| d\mu_N(u) 
\end{equation}
the associated state and $\gammat_N ^{(n)}$ its reduced density matrices, we have
\begin{equation} \label{eq:error finite dim deF}
\Tr_{\gH^n} \Big| \gamma_N ^{(n)} - \widetilde \gamma _N^{(n)}  \Big| \leq \frac{2 n(d+2n)}{N}
\end{equation}
for all $n=1 \ldots N$.
\end{theorem}

\begin{remark}[On the finite dimensional quantum de Finetti theorem]\label{rem:quant deF}\mbox{}\\\vspace{-0.4cm}
\begin{enumerate}
\item This result is due to Christandl, K\"onig, Mitchison and Renner~\cite{ChrKonMitRen-07}, important earlier work being found in ~\cite{KonRen-05} and~\cite{FanVan-06}. One may find developments along theses lines in~\cite{Chiribella-11,Harrow-13,LewNamRou-13b}. The quantum information community also considered several variants, see for example~\cite{CavFucSch-02,ChrKonMitRen-07,ChrTon-09,RenCir-09,Renner-07,BraHar-12}.
\item One can add a step in the strategy~\eqref{eq:schema deF 2}:
\begin{equation}\label{eq:schema deF 3}
\boxed{\mbox{quantitative deF } \Rightarrow \mbox{finite dimensional deF } \Rightarrow \mbox{ weak deF } \Rightarrow \mbox{ strong deF } }.
\end{equation}
Indeed, in finite dimension one may identify the sphere $S\gH$ with a usual, compact, Euclidean sphere (with dimension $2d-1$, i.e. the unit sphere in $\R ^{2d}$). The space of probability measures on $S\gH$ is then compact for the usual weak topology and one may extract from $\mu_N$ a converging subsequence to prove Theorem~\ref{thm:DeFinetti fort} in the case where $\dim \gH <\infty$, exactly as we did to deduce Theorem~\ref{thm:HS} from Theorem~\ref{thm:DF} in Section~\ref{sec:DF}.
\item The bound~\eqref{eq:error finite dim deF} is not optimal. One may in fact obtain the estimate
\begin{equation} \label{eq:error finite dim deF 2}
\Tr_{\gH^n} \Big| \gamma_N ^{(n)} - \widetilde \gamma _N^{(n)}  \Big| \leq \frac{2 nd}{N},
\end{equation}
with the same construction, see~\cite{Chiribella-11,ChrKonMitRen-07,LewNamRou-13b}. The proof we shall present only gives~\eqref{eq:error finite dim deF} but seems more instructive to me. For the applications we have in mind, $n$ will always be fixed anyway (equal to $2$ most of the time), and in this case~\eqref{eq:error finite dim deF} and~\eqref{eq:error finite dim deF 2} give the same order of magnitude in terms of $N$ and $d$.
\item The bound~\eqref{eq:error finite dim deF 2} in the quantum case is the equivalent of the estimate in $dn/N$ mentioned in Remark~\ref{rem:DiacFreLio} for the classical case. One may ask if this order of magnitude is optimal. It clearly is with the construction we are going to use, but it would be very interesting to know if one can do better with another construction. In particular, can one find a bound independent from $d$, reminiscent of~\eqref{eq:DiacFreed} in the Diaconis-Freedman theorem ? 
\end{enumerate}

\hfill\qed
\end{remark}

The construction of $\Gammat_N$ is taken from~\cite{ChrKonMitRen-07}. It is particularly simple but it does use in a strong manner the fact that the underlying Hilbert space has finite dimension. The approach we shall follow for the proof of Theorem~\ref{thm:DeFinetti quant} is originally due to Chiribella~\cite{Chiribella-11}. We are going to prove an explicit formula giving the density of $\Gammat_N$ as a function of those of  $\Gamma_N$, in the spirit of Remark~\ref{rem:marg DF}. This formula implies~\eqref{eq:error finite dim deF} in the same manner as~\eqref{eq:DF astuce} implies~\eqref{eq:DiacFreed}.

In Section~\ref{sec:CKMR constr} we present the construction, state Chiribella's explicit formula and obtain Theorem~\ref{thm:DeFinetti quant} as a corollary. Before giving a proof of Chiribella's result, it is useful to discuss some informal motivation and some heuristics on the Christandl-K\"onig-Mitchison-Renner (CKMR) construction, which happens to be connected to well-known ideas of semi-classical analysis. This is the purpose of Section~\ref{sec:CKMR heur}. Finally, we prove Chiribella's formula in Section~\ref{sec:CKMR proof}, following the approach of~\cite{LewNamRou-13b}. It has been independently found by Lieb and Solovej~\cite{LieSol-13} (with a different motivation), and was inspired by the works of Ammari and Nier~\cite{AmmNie-08,AmmNie-09,AmmNie-11}. Related considerations appeared also in~\cite[Chapter~3]{Knowles-thesis}. Other proofs are available in the literature, cf~\cite{Chiribella-11} and~\cite{Harrow-13}. 

\subsection{The CKMR construction and Chiribella's formula}\label{sec:CKMR constr}\mbox{}\\\vspace{-0.4cm}

We first note that the Diaconis-Freedman construction introduced before is purely classical since it is based on the notion of empirical measure, which has no quantum counterpart. A different approach is thus clearly necessary for the proof of Theorem~\ref{thm:DeFinetti quant}.

In a finite dimensional space, one may identify the unit sphere $S\gH = \left\{ u \in \gH,\norm{u}=1 \right\}$ with a Euclidean sphere. One may thus equip it with a uniform measure (Haar measure of the rotation group, simply the Lebesgue measure on the Euclidean sphere), that we shall denote $du$, taking the convention that 
$$\int_{S\gH} du = 1. $$
We then have a nice resolution of the identity as a simple consequence of the invariance of $du$ under rotations (see Section~\ref{sec:CKMR proof} below for a proof). We state this as a lemma:

\begin{lemma}[\textbf{Schur's formula}]\label{lem:Schur}\mbox{}\\
Let $\gH$ be a complex finite dimensional Hilbert space and $\gH ^N$ the corresponding bosonic $N$-body space. Then
\begin{equation}\label{eq:Schur}
 \dim \gH_s^N \int_{S\gH} | u^{\otimes N} \rangle  \langle u^{\otimes N} | \, du = \1_{\gH^N}.  
\end{equation}
\end{lemma}

The idea of Christandl-K\"onig-Mitchison-Renner is to simply define
\begin{align}\label{eq:def CKMR}
d\mu_N (u) &:=  \dim \gH_s^N \tr_{\gH_s ^N} \left[\Gamma_N | u^{\otimes N} \rangle  \langle u^{\otimes N} |\right] du \nonumber\\
&= \dim \gH_s^N \tr_{\gH_s ^N}  \left\langle u^{\otimes N}, \Gamma_N u^{\otimes N} \right\rangle  du, 
\end{align}
i.e. to take
\begin{equation}\label{eq:def CKMR 2}
\Gammat_N =  \dim \gH_s^N \int_{S\gH} | u^{\otimes N} \rangle  \langle u^{\otimes N} | \left\langle u^{\otimes N}, \Gamma_N u^{\otimes N} \right\rangle du.
\end{equation}
Chriibella's observation\footnote{In the quantum information vocabulary this is formulated as a relation between ``optimal cloning'' and ``optimal measure and prepare channels''.} is the following: 

\begin{theorem}[\textbf{Chiribella's formula}] \label{thm:CKMR-identity} \mbox{}\\
With the previous definitions, it holds 
\begin{equation}\label{eq:CKMR exact}
\gammat _N^{(n)} = {{N+n+d-1}\choose n}^{-1}\sum_{\ell=0}^{n} {N \choose \ell}  \gamma_N^{(\ell)} \otimes _s \1_{\gH^{n-\ell}}
\end{equation}
with the convention
$$ \gamma_N^{(\ell)} \otimes _s \1_{\gH^{n-\ell}}= \frac{1}{\ell!\,(n-\ell)!}\sum_{\sigma\in S_n} (\gamma_N^\ell)_{\sigma(1),...,\sigma(\ell)} \otimes (\1_{\gH^{n-\ell}})_{{\sigma(\ell+1)},...,{\sigma(n)}}$$
where $(\gamma_N^\ell)_{\sigma(1),...,\sigma(\ell)}$ acts on the  $\sigma(1)\ldots,\sigma(\ell)$ variables.
\end{theorem} 

From this result we deduce a simple proof of the quantitative de Finetti theorem:

\begin{proof}[Proof of Theorem~\ref{thm:DeFinetti quant}]
We proceed as in~\eqref{eq:preuve DiacFreed}. Only the first term in the sum~\eqref{eq:CKMR exact} is really relevant:  
\begin{equation}
 \gammat _N^{(n)} - \gamma_N ^{(n)}  = ( C(d,n,N) - 1) \gamma_N ^{(n)} + B = -A + B \label{eq:estim CKMR}
\end{equation}
where 
\[
C(d,n,N) = \frac{(N+d-1)!}{(N+n+d-1)!} \frac{N!}{(N-n)!} < 1,  
\]
and $A,B$ are positive operators. We have 
$$\tr_{\gH ^n} [-A+B] = \tr \left[\gammat _N^{(n)} - \gamma_N ^{(n)} \right]= 0,$$
and thus, by the triangle inequality,
\[
\Tr \Big|\gammat _N^{(n)} - \gamma_N ^{(n)} \Big| \leq \tr A + \tr B = 2 \Tr A = 2 (1- C(d,n,N)).
\]
Next, the elementary inequality
\begin{align*}
C(d,n,N) &=  \prod_{j=0} ^{n-1} \frac{N-j}{N+j+d}\ge  \left( 1 - \frac{2n + d -2}{N + d + n - 1}\right) ^n \geq 1 -n \frac{2n + d -2}{N + d + n - 1}
\end{align*}
gives
\begin{equation} \label{eq:error-NLR}
\Tr \Big| \gamma_N ^{(n)} - \gammat _N^{(n)} \Big| \le \frac{2 n(d+2n)}{N},
\end{equation}
which is the desired result.
\end{proof}

Everything now relies on the proof of Theorem~\ref{thm:CKMR-identity}, which is the subject of Section~\ref{sec:CKMR proof}. Before we give it, some heuristics regarding the relevance of the construction~\eqref{eq:def CKMR 2} shall be discussed.

\subsection{Heuristics and motivation}\label{sec:CKMR heur}\mbox{}\\\vspace{-0.4cm}

Schur's formula~\eqref{eq:Schur} expresses the fact that the family $\left( u ^{\otimes N} \right)_{u\in S\gH}$ forms an over-complete basis of $\gH_s ^N$. Such a basis labeled by a continuous parameter is reminiscent of a coherent state decomposition~\cite{KlaSka-85,ZhaFenGil-90}. This basis in fact turns our to be ``less and less over-complete'' when $N$ gets large. Indeed, we clearly have
\begin{equation}\label{eq:less over complete}
 \langle u ^{\otimes N}, v ^{\otimes N} \rangle_{\gH ^N}  = \langle u , v \rangle_{\gH} ^N \to 0 \mbox{ when } N \to \infty 
\end{equation}
as soon as $u$ and $v$ are not exactly colinear. The basis $\left( u ^{\otimes N} \right)_{u\in S\gH}$ thus becomes ``almost orthonormal'' when $N$ tends to infinity.

In the vocabulary of semi-classical analysis~\cite{Lieb-73b,Simon-80,Berezin-72}, trying to write
\[
 \Gamma_N = \int_{u\in S \gH} d\mu_N (u) |u ^{\otimes N} \rangle \langle u ^{\otimes N}| 
\]
amounts to looking for an \emph{upper symbol} $\mu_N$ representing $\Gamma_N$. In fact, it so happens~\cite{Simon-80} that one may always find such a symbol, only $\mu_N$ is in general not a positive measure. The problem we face is to find a way to approximate the upper symbol (for which no explicit expression as a function of the state itself exists, by the way) with a positive measure. 

Note that \emph{if} the coherent state basis were orthogonal, then the upper symbol would be positive, simply by positivity of the state. Since we noted that the family $\left( u ^{\otimes N} \right)_{u\in S\gH}$ becomes ``almost orthonormal'' when $N$ gets large, it is very natural to expect that the upper symbol may be approximate by a positive measure in this limit.

On the other hand, the measure introduced in~\eqref{eq:def CKMR} is exactly what one calls the \emph{lower symbol} of the state $\Gamma_N$. One of the reasons why lower and upper symbols were introduced is that these two a priori different objects have a tendency to co\"incide in semi-classical limits. But the $N\to \infty$ limit we are concerned with may indeed be seen as a semi-classical limit and it is thus very natural to take the lower symbol as an approximation of the upper symbol in this limit.

One can motivate this choice in a slightly more precise way. Assume we have a sequence of $N$-body states defined starting from an upper symbol independent of $N$,
\[
 \Gamma_N = \dim (\gH _s ^N) \int_{u\in S \gH}  \mu ^{\rm sup} (u) |u ^{\otimes N} \rangle \langle u ^{\otimes N}|du,  
\]
and let us compute the corresponding lower symbols:
\begin{align*}
 \mu ^{\rm inf}_N (v) = \langle v ^{\otimes N}, \Gamma_N v ^{\otimes N} \rangle &= \dim (\gH _s ^N) \int_{u\in S \gH}  \mu ^{\rm sup} (u) \left|\langle u ^{\otimes N}, v ^{\otimes N}\rangle\right| ^2 du
 \\&=\dim (\gH _s ^N) \int_{u\in S \gH}  \mu ^{\rm sup} (u) \left|\langle u , v \rangle\right| ^{2N} du.
\end{align*}
In view of the observation~\eqref{eq:less over complete} and the necessary invariance of $\mu ^{\rm sup}$ under the action of $S ^1$ it is clear that we have 
$$ \mu ^{\rm inf}_N (v) \to \mu ^{\rm sup} (v) \mbox{ when } N\to\infty.$$ 
In other words, the lower symbol is, for large $N$, an approximation of the upper symbol, that has the advantage of being positive. Without consituting a rigorous proof of Theorem~\ref{thm:DeFinetti quant}, this point of view shows that the CKMR constuction is extremely natural.

\subsection{Chiribella's formula and anti-Wick quantization}\label{sec:CKMR proof}\mbox{}\\\vspace{-0.4cm}

The proof of Theorem~\ref{thm:CKMR-identity} we are going to present uses the second quantization formalism. We start with a very useful lemma. In the vocabulary alluded to in the previous section it says that a state is entirely characterized by its lower symbol, a well-known fact~\cite{Simon-80,KlaSka-85}.

\begin{lemma}[\textbf{The lower symbol determines the state}]\label{lem:uk-g-uk=0}\mbox{} \\
If an operator $\gamma^{(k)}$ on $\gH_s^k$ satisfies  
\bq \label{eq:uk-g-uk=0}
 \langle u^{\otimes k}, \gamma^{(k)} u^{\otimes k} \rangle =0\qquad \text{for~all}~u\in \gH,
 \eq
then $\gamma^{(k)} \equiv 0$.
\end{lemma}

\begin{proof}
We use the symmetric tensor product
$$\Psi_k\otimes_s\Psi_{k-\ell} (x_1,...,x_{k})=\frac{1}{\sqrt{\ell!(k-\ell)!k!}}\sum_{\sigma\in {\Sigma}_{k}}\Psi_\ell(x_{\sigma(1)},...,x_{\sigma(\ell)})\Psi_{k-\ell}(x_{\sigma(\ell+1)},...,x_{\sigma(k)})
$$
for two vectors $\Psi_\ell \in \gH^{\ell}$ and $\Psi_{k-\ell}\in\gH^{k-\ell}$ and assume that~\eqref{eq:uk-g-uk=0} holds for all $u\in \gH$.

Picking two unit vectors $u,v$, replacing $u$ by $u+tv$ in~\eqref{eq:uk-g-uk=0} and taking the derivative with respect to $t$ (which must be zero), we obtain
$$ \Re \left\langle u ^{\otimes k -1} \otimes_s v, \gamma ^{(k)} u ^{\otimes k} \right\rangle = 0.$$ 
Next, replacing $u$ by $u+itv$ and taking the derivative with respect to $t$ we get
$$ \Im \left\langle u ^{\otimes k -1} \otimes_s v, \gamma ^{(k)} u ^{\otimes k} \right\rangle = 0.$$ 
Hence for all $u$ and $v$
\begin{equation}\label{eq:step low symb}
\left\langle u ^{\otimes k -1} \otimes_s v, \gamma ^{(k)} u ^{\otimes k} \right\rangle = 0. 
\end{equation}
Doing the same manipulations but taking now second derivatives with respect to $t$ we also have 
$$ \left\langle u ^{\otimes k-1} \otimes_s v , \gamma ^{(k)} u ^{\otimes k-1}\otimes_s v \right\rangle + 2  \left\langle u ^{\otimes k-2} \otimes_s v ^{\otimes 2} , \gamma ^{(k)} u ^{\otimes k}\right\rangle = 0.$$
But, on the other hand, replacing $u$ by $u+tv$ and $u+itv$ in~\eqref{eq:step low symb} and taking first derivatives we get 
$$ \left\langle u ^{\otimes k-1} \otimes_s v , \gamma ^{(k)} u ^{\otimes k-1} \otimes_s v \right\rangle + \left\langle u ^{\otimes k-2} \otimes_s v ^{\otimes 2} , \gamma ^{(k)} u ^{\otimes k}\right\rangle = 0.$$
Combining the last two equations we infer 
$$ \left\langle u ^{\otimes k-1} \otimes_s v , \gamma ^{(k)} u ^{\otimes k-1}\otimes_s v \right\rangle = 0$$
for all unit vectors $u$ and $v$.
Taking $v$ of the form $v=v_1\pm \widetilde v_1$ then $v = v_1 \pm i \widetilde v_1$ and iterating the argument, we conclude 
   \bqq
\langle v_1 \otimes_s v_2 \otimes_s \ldots \otimes_s v_k, \gamma^{(k)} \widetilde v_1 \otimes_s \widetilde v_2 \otimes_s \ldots \otimes_s \widetilde v_k \rangle =0
 \eqq
for all $v_j, \widetilde v_j \in \gH$. Vectors of the form  $v_1 \otimes_s v_2 \otimes_s \ldots \otimes_s v_k$ form a basis of $\gH_s ^k$, thus the proof is complete.
\end{proof}
 
For self-containedness we use the previous lemma to give a short proof of Schur's formula~\eqref{eq:Schur}:

\begin{proof}[Proof of Lemma~\ref{lem:Schur}]
In view of Lemma~\ref{lem:uk-g-uk=0}, it suffices to show that 
$$ \dim \gH_s^N \int_{S\gH} \left|\langle u , v \rangle \right| ^{2N} \, du = 1$$
for all $v\in S\gH$. Pick $v$ and $\tilde{v}$ in $S\gH$ and $U$ a unitary mapping such that $U v = \tilde{v}$. Then, by invariance of $du$ under rotations, 
$$ \dim \gH_s^N \int_{S\gH} \left|\langle u , v \rangle \right| ^{2N} \, du = \dim \gH_s^N \int_{S\gH} \left|\langle U u , U v \rangle \right| ^{2N} \, du = \dim \gH_s^N \int_{S\gH} \left|\langle u , \tilde{v} \rangle \right| ^{2N} \, du$$
and thus the quantity
$$ \dim \gH_s^N \int_{S\gH} \left|\langle u , v \rangle \right| ^{2N} \, du$$ 
does not depend on $v$. By Lemma~\ref{lem:uk-g-uk=0} this implies that 
$$\dim \gH_s^N \int_{S\gH} | u^{\otimes N} \rangle  \langle u^{\otimes N} | \, du = c \1_{\gH^N}$$
for some constant $c$. Taking the trace of both sides of this equation shows that $c=1$ and the proof is complete.
\end{proof}

In the sequel we use standard bosonic creation and annihilation operators. For all $f\in \gH$, we define the creation operator $a^*(f): \gH_s^{k-1} \to \gH_s^{k}$ by
$$
{a^*}(f)\left( {\sum\limits_{\sigma  \in {\Sigma_{k-1}}} {{f_{\sigma (1)}} } \otimes ... \otimes {f_{\sigma (k-1)}}} \right) = (k) ^{-1/2} \sum\limits_{\sigma  \in {\Sigma_{k}}} {{f_{\sigma (1)}} }  \otimes ... \otimes {f_{\sigma (k)}}
$$
where in the right-hand side we set $f_k = f$. The annihilation operator $a(f): \gH^{k+1} \to \gH^{k}$ is the formal adjoing of  $a^*(f)$ (whence the notation), defined by
$$ a(f) \left( {\sum\limits_{\sigma  \in {\Sigma_{k+1}}} {{f_{\sigma (1)}} } \otimes ... \otimes {f_{\sigma (k+1)}}} \right) = (k+1) ^{1/2} \sum\limits_{\sigma  \in {\Sigma_{k+1}}} \left\langle f,f_{\sigma(1)} \right\rangle {{f_{\sigma (2)}} }  \otimes ... \otimes {f_{\sigma (k)}}$$
for all $f,f_1,...,f_{k}$ in $\gH$. These operators satisfy the {\it canonical commutation relations} (CCR)
\begin{equation}\label{eq:CCR}
[a(f),a(g)]=0,\quad[a^*(f),a^*(g)]=0,\quad [a(f),a^*(g)]= \langle f,g \rangle_{\gH}. 
\end{equation}
One of the uses of these objects is that the reduced density matrices $\gamma_N ^{(n)}$ of a bosonic state $\Gamma_N$ are characterized by the relations\footnote{We recall our convention that $\tr [\gamma_N ^{(n)}] = 1$.} 
\begin{equation} \label{eq:Wick}
 \langle v^{\otimes n}, \gamma_N^{(n)} v^{\otimes n} \rangle =\frac{(N-n)!}{N!} \tr_{\gH_s ^N} \left[ a^*(v)^n a(v)^n \Gamma_N \right].
 \end{equation}
Lemma~\ref{lem:uk-g-uk=0} guarantees that this determines $\gamma_N ^{(n)}$ completely. The definition above is called a Wick quantization: Creation and annihilation operators appear in the normal order, all creators on the left and all annihilators on the right.

The key observation in the proof of Theorem~\ref{thm:CKMR-identity} is that the density matrices of the state~\eqref{eq:def CKMR 2} can be alternatively defined from $\Gamma_N$ via an anti-Wick quantization where creation and annihilation operators appear in anti-normal order: All annihilators on the left and all creators on the right.

\begin{lemma}[\textbf{The CKMR construction and anti-Wick quantization}]\label{lem:A Wick}\mbox{}\\
Let $\Gammat_N$ be defined by~\eqref{eq:def CKMR 2} and $\gammat_N ^{(n)}$ be its reduced density matrices. We have 
\begin{equation}\label{eq:A Wick}
\langle v^{\otimes n}, \gammat _N^{(n)} v^{\otimes n} \rangle = \frac{(N+d-1)!}{(N+n+d-1)!} \tr_{\gH_s ^N} \left[ a(v)^n a ^*(v)^n \Gamma_N \right]
\end{equation}
for all $v\in \gH$.
\end{lemma}

\begin{proof}
It suffices to consider the case of a pure state $\Gamma_N = |\Psi_N \rangle\langle \Psi_N |$ and write
\begin{align*}
\langle v^{\otimes k}, \widetilde \gamma _N^{(k)} v^{\otimes k} \rangle &= \dim \gH_s^N \int_{S\gH} du |\langle u^{\otimes N}, \Psi_N \rangle|^2 | \langle  u^{\otimes k},v^{\otimes k} \rangle|^2 \\
&= \dim \gH_s^N \int_{S\gH} du |\langle u^{\otimes (N +k)}, v^{\otimes k} \otimes \Psi_N \rangle|^2 \\
&= \frac{N!}{(N+k)!} \dim \gH_s^N \int_{S\gH} du |\langle u^{\otimes (N +k)}, a^*(v)^{k} \Psi_N \rangle|^2 \\
&=  \frac{N!}{(N+k)!}\frac{\dim \gH_s^N}{\dim \gH_s^{N+k}} \langle a(v)^{k} \Psi_N,a (v)^k \Psi_N \rangle \\
&= \frac{(N+d-1)!}{(N+k+d-1)!} \langle \Psi_N, a(v)^k a^*(v)^{k} \Psi_N \rangle
\end{align*}
using Schur's lemma~\eqref{eq:Schur} in $\gH_s ^{N+k}$ in the third line, the fact that $a^*(v)$ is the adjoint of $a (v)$ in the fourth line and recalling that
\begin{equation}\label{eq:dim boson}
\dim \gH_s ^N = { N+d-1 \choose d-1}, 
\end{equation}
the number of ways of choosing $N$ elements from $d$, allowing repetitions, without taking the order into account. This is the number of orthogonal vectors of the form 
$$ u_{i_1} \otimes_s \ldots \otimes_s u_{i_N}, \quad (i_1,\ldots, i_N) \in \left\{ 1,\ldots, d\right\}$$ 
one may form starting from an orthogonal basis $(u_1, \ldots,u_d)$ of $\gH$, and these form an orthogonal basis of $\gH ^N$.
\end{proof}

The way forward is now clear: we have to compare polynomials in $a ^*(v)$ and $a(v)$ written in normal and anti-normal order. This standard operation leads to the final lemma of the proof:

\begin{lemma}[\textbf{Normal and anti-normal order}]\label{lem:Wick A Wick}\mbox{}\\
Let $v\in S\gH$. We have  
\begin{equation}\label{eq:Wick A Wick}
a(v) ^n a ^*(v) ^n = \sum_{k=0} ^n \binom{n}{k} \frac{n!}{k!} a ^*(v) ^k a (v) ^k \mbox{ for all } n\in \N. 
\end{equation}
\end{lemma}

\begin{proof}
The computation is made easier by recalling the expression for the $n$-th Laguerre polynomial
\[
L_n (x) = \sum_{k=0} ^n \binom{n}{k} \frac{(-1) ^k}{k!} x ^k.
\]
These polynomials satisfy the recurrence relation 
\[
(n+1) L_{n+1} (x) = (2n +1) L_n (x) - x L_n (x) - n L_{n-1} (x)  
\]
and one may see that~\eqref{eq:Wick A Wick} may be rewritten
\begin{equation}\label{eq:Wick A Wick proof}
a(v) ^n a ^*(v) ^n = \sum_{k=0} ^n c_{n,k} \, a^*(v) ^k a (v) ^k 
\end{equation}
where the $c_{n,k}$ are the coefficients of the polynomial
$$\tilde{L}_n (x) := n! \, L_n (-x).$$
It thus suffices to show that, for any $n\geq 1$, 
\[
a(v) ^{n+1} a ^*(v) ^{n+1} = a^* (v)  a(v) ^n a ^*(v) ^n  a(v) + (2n+1) a(v) ^n a ^*(v) ^n - n ^2  a(v) ^{n-1} a ^*(v) ^{n-1}.
\]
Note the order of creation and annihilation operators in the first term of the right-hand side: knowing a normal-ordered representation of $a(v) ^n a ^*(v) ^n $ and $a(v) ^{n-1} a ^*(v) ^{n-1}$ we deduce a normal-ordered representation of the left-hand side.

A repeated application of the CCR~\eqref{eq:CCR} gives the relations 
\begin{align}\label{eq:CCR n}
a(v) a ^* (v) ^n &= a ^* (v) ^n a(v) + n  a ^* (v) ^{n-1} \nonumber \\
a(v) ^n a ^* (v) &= a ^* (v)  a(v) ^n + n  a (v) ^{n-1}
\end{align}
Then 
\begin{align*}
a^* (v) a(v) ^n a ^*(v) ^n  a(v) &= a(v) ^n a ^* (v) ^{n+1} a (v) - n a (v) ^{n-1} a ^* (v) ^n a(v)\\
& = a(v) ^{n+1} a ^*(v) ^{n+1} - (n + 1) a(v) ^n a ^*(v) ^n \\
&- n a(v) ^n a ^*(v) ^n +  n ^2  a(v) ^{n-1} a ^*(v) ^{n-1},
\end{align*}
and the proof is complete.
\end{proof}

The final formula~\eqref{eq:CKMR exact} is deduced by combining Lemmas~\ref{lem:uk-g-uk=0},~\ref{lem:A Wick} and~\ref{lem:Wick A Wick} with~\eqref{eq:Wick}, simply noting that for $j\leq n$
\begin{align*}
\tr \left[ a ^* (v) ^j a(v) ^j \Gamma_N \right] &= \frac{N!}{(N-j)!} \langle v ^{\otimes j}, \Gamma_N ^{(j)}  v ^{\otimes j}\rangle \\
&= \frac{N!}{(N-j)!} \langle v ^{\otimes n}, \Gamma_N ^{(j)}  \otimes \1 ^{\otimes n-j}  v ^{\otimes n}\rangle \\
&= \frac{N!}{(N-j)!} \langle v ^{\otimes n}, \Gamma_N ^{(j)}  \otimes_s \1 ^{\otimes n-j} v ^{\otimes n}\rangle.
\end{align*}
The first two equalities are just definitions and the third one comes from the bosonic symmetry of $v ^{\otimes n}$

\newpage

\section{\textbf{Fock-space localization and applications}}\label{sec:locFock}

We now turn to the first implication in the proof strategy~\eqref{eq:schema deF 2}. We shall need to convert weak-$\ast$ convergence of reduced density matrices into strong convergence, in order to apply Theorem~\ref{thm:DeFinetti fort}. The idea is to localize the state $\Gamma_N$ one starts from using either compactly supported functions or finite rank orthogonal projectors. One may then work in a compact setting with $\gS^1$-strong convergence, apply Theorem~\ref{thm:DeFinetti fort} and then pass to the limit in the localization as a last step. More precisely, we will use localization in finite dimensional spaces, in order to show that the general theorem can be deduced from the finite-dimensional constructive proof discussed in the previous chapter.

The subtlelty here is that the appropriate localization notion for a $N$-body state (e.g. a wave-function $\Psi_N \in L ^2 (\R ^{dN})$) is more complicated than that one is used to for one-body wave-functions $\psi \in L ^2 (\R ^d)$. One in fact has to work directly on the reduced density matrices and localize them in such a way that the localized matrices correspond to a quantum state. The localization procedure may lead to particle losses and thus the localizaed state will in general not be a $N$-body state but a superposition of $k$-body states, $0\leq k \leq N$, i.e. a state on Fock space. 

The localization procedure we shall use is described in Section~\ref{sec:loc rig}. We shall first give some heurtistic considerations in Section~\ref{sec:loc heur}, in order to make precise what has been said above, namely that the correct localization procedure in $L ^2 (\R ^{dN})$ must differ from the usual localization in $L ^2 (\R ^d)$. Section~\ref{sec:proof deF faible} contains the proof of the weak quantum de  Finetti theorem and a useful auxiliary result which is a consequence of the proof using localization. 

\subsection{Weak convergence and localization for a two-body state}\label{sec:loc heur}\mbox{}\\\vspace{-0.4cm}

The following considerations are taken from~\cite{Lewin-11}. Let us take a particularly simple sequence of bosonic two-body states 
\begin{equation}\label{eq:deux corps}
\Psi_n := \psi_n \otimes_s \phi_n = \frac{1}{\sqrt{2}} \left( \psi_n \otimes \phi_n + \phi_n \otimes \psi_n \right) \in L_s ^2 (\R ^{2d})
\end{equation}
with $\psi_n$ and $\phi_n$ being normalized in $L ^2 (\R ^d)$. This corresponds to having one patricle in the state $\psi_n$ and one particle in the state $\phi_n$. We will assume
\[
 \bral \phi_n,\psi_n \ketr_{L ^2 (\R ^d)} = 0,
\]
which ensures that $\norm{\Psi_n} = 1$. Extracting a subsequence if need be we have
\[
 \Psi_n \wto \Psi \mbox{ weakly in }  L ^2 (\R ^{2d})
\]
and the convergence is strong if and only if $\norm{\Psi} = \norm{\Psi_n}= 1$. In the case where some mass is lost in the limit, i.e. $\norm{\Psi}<1$, the convergence is only weak. 

We will always work in a locally compact setting and thus the only possible source for the loss of mass is that it disappears at infinity~\cite{Lions-84,Lions-84b,Lions-85a,Lions-85b}. A possibility is that both particles $\phi_n$ and $\psi_n$ are lost at infinity
\[
\psi_n \wto 0, \: \phi_n \wto 0 \mbox{ in } L^2 (\R ^{2d}) \mbox{ in } n\to \infty, 
\]
in which case $\Psi_n \wto 0$ in $L ^2 (\R ^{2d})$. In $L ^2 (\R ^{2d})$ this is the scenario that is closest to the usual loss of mass in $L ^2 (\R ^d)$, but there are other possibilities.
 
A typical case is that where only one of the two particles is lost at infinity, which we can materialize by 
\[
 \psi_n \wto 0 \mbox{ weakly in } L ^2 (\R ^d), \: \phi_n \to \phi \mbox{ strongly in } L ^2 (\R ^{2d}).
\]
For the loss of mass of $\psi_n$ one may typically think of the example
\begin{equation}\label{eq:fuite masse}
 \psi_n = \psi \left( . + x_n \right) 
\end{equation}
with $|x_n| \to \infty$ when $n\to \infty$ and $\psi$ say smooth with compact support. We have in this case
\[
\Psi_n \wto 0 \mbox{ in } L ^2 (\R ^{2d}) 
\]
but for obvious physical reasons we would prefer to have a weak convergence notion ensuring
\begin{equation}\label{eq:geo conv}
\Psi_n \wto_{ g}  \frac{1}{\sqrt{2}} \phi.  
\end{equation}
In particular, since only the parrticle in the state $\psi_n$ is lost at infinity it is natural that the limit state be one with only the particle described by $\phi$ left. We denote this convergence $\wto_{ g}$ because this is precisely the geometric convergence discussed by Mathieu Lewin in~\cite{Lewin-11}. The difficulty is of course that the two sides of~\eqref{eq:geo conv} live in different spaces.

To introduce the correct convergence notion, one has to look at the density matrices of~$\Psi_n$:
\begin{align}\label{eq:matr deux corps}
\gamma_{\Psi_n} ^{(2)} &= \ketl \phi_n \otimes_s \psi_n \ketr \bral \phi_n \otimes_s \psi _n \brar \wto_\ast 0 \mbox{ in } \gS ^1\left (L ^2 (\R ^{2d}) \right) \nonumber\\
\gamma_{\Psi_n} ^{(1)} &= \half \ketl \phi_n \ketr \bral \phi_n  \brar + \half \ketl \psi_n \ketr \bral \psi_n  \brar \wto_\ast \frac{1}{2} \ketl \phi \ketr \bral \phi  \brar \mbox{ in } \gS ^1\left (L ^2 (\R ^{d}) \right).
\end{align}
One then sees that the pair $\left( \gamma_{\Psi_n} ^{(2)}, \gamma_{\Psi_n} ^{(1)}\right)$ converges to the pair $\left( 0, \half \ketl \phi \ketr \bral \phi\brar\right)$ that corresponds to the density matrices of the one-body state $\sqrt{2} ^{-1} \phi \in L ^2 (\R ^{d})$. More precisely, the geometric convergence notion is formulated in the Fock space (here bosonic with two particles)
\begin{equation}\label{eq:Fock deux corps}
\cF_s ^{\leq 2} (L ^2 (\R ^{d})) :=  \C \oplus L ^2 (\R ^d) \oplus L_s ^2 (\R^{2d}) 
\end{equation}
and we have in the sense of geometric convergence on $\gS ^1 \left( \cF_s ^{\leq 2} \right)$
\[
 0 \oplus 0 \oplus \ketl \Psi_n \ketr \bral \Psi_n \brar \wto_g  \half \oplus \half \ketl \phi \ketr \bral \phi  \brar \oplus 0,
\]
which means that all the reduced density matrices of the left-hand side converge to those of the right-hand side. We note that the limit does have trace $1$ in $\gS ^1 \left( \cF_s ^{\leq 2} \right)$, there is thus no loss of mass in $\cF_s ^{\leq 2}$. More precisely, in $\cF_s ^{\leq N}$ \textbf{the loss of mass for a pure $N$-particles state is materialized by the convergence to a mixed state with less particles.} 

\medskip

Just as the appropriate notion of weak convergence for $N$-body problems is different from the usual weak convergence in $L^2 (\R^{dN})$ (a fortiori when taking the limit $N\to \infty$), the appropriate procedure to localize a state and turn weak convergence into strong convergence must be thought anew. Given a self-adjoint positive localization operator $A$, say $A= P$ a finite rank projector or $A = \chi$ the multiplication by a compactly supported function $\chi$, one usually localizes a wave-function $\psi \in L ^2 (\R ^d)$ by defining 
\[
\psi_A = A \psi 
\]
which amounts to associate
\[
 \ketl \psi \ketr\bral \psi \brar \leftrightarrow \ketl A \psi \ketr\bral A \psi \brar.
\]
Emulating this procedure for the two-body state~\eqref{eq:deux corps} one might imagine to consider a localized state defined by its two-body density matrix
\[
\gamma_{n,A}  ^{(2)}= \ketl A\otimes A \Psi_n \ketr \bral A\otimes A \Psi_n \brar.
\]
It is then clear that
\[
 \gamma_{n,A} ^{(2)} \wto_\ast 0 \mbox{ in } \gS^1 (\gH_s ^2),
\]
which was to be expected, but it is more disturbing that we also have for the corresponding one-body density matrix 
\[
 \gamma_{n,A} ^{(1)} \wto_\ast 0 \mbox{ in } \gS ^1 (\gH)
\]
whereas, in view of~\eqref{eq:matr deux corps} one would rather like to have 
\[
 \gamma_{n,A} ^{(1)} \to  \frac{1}{2} \ketl A \phi \ketr \bral A \phi  \brar \mbox{ strongly in } \gS ^1 (\gH).
\]
The solution to this dilemma is to define a localized state by asking that its reduced density matrices be $A\otimes A \gamma_{\Psi_n} ^{(2)} A\otimes A$ and $A\gamma_{\Psi_n} ^{(1)}A$. The corresponding state is then uniquely determined, and it so happens that it is a state on Fock space, as we explain in the next section.

\subsection{Fock-space localization}\label{sec:loc rig}\mbox{}\\\vspace{-0.4cm}

After the preceding heuristic considerations, we now introduce the notion of localization in the bosonic Fock space\footnote{The procedure is the same for fermionic particles.}
\begin{align}\label{eq:Fock space}
\cF_s (\gH) &=  \C \oplus \gH \oplus \ldots \oplus \gH_s ^n \oplus \ldots\nonumber\\
\cF_s (L ^2 (\R ^d) ) &= \C \oplus L ^2 (\R ^d) \oplus \ldots \oplus L_s ^2 (\R ^{dn}) \oplus \ldots.
\end{align}
In this course we will always start from $N$-body states, in which case it is sufficient to work in the truncated Fock space 
\begin{align}\label{eq:Fock space tronc}
\cF_s  ^{\leq N}(\gH) &=  \C \oplus \gH \oplus \ldots \oplus \gH_s ^n \oplus \ldots \oplus \gH_s ^N\nonumber\\
\cF_s ^{\leq N} (L ^2 (\R ^d) ) &= \C \oplus L ^2 (\R ^d) \oplus \ldots \oplus L_s ^2 (\R ^{dn}) \oplus \ldots \oplus L_s ^2 (\R ^{dN}).
\end{align}

\begin{definition}[\textbf{Bosonic states on the Fock space}]\label{def:Fock state}\mbox{}\\
A bosonic state on the Fock space is a positive self-adjoint operator with trace $1$ on $\cF_s$. We denote $\cS (\cF_s(\gH))$ the set of bosonic states 
\begin{equation}\label{eq:Fock state}
\cS (\cF_s(\gH)) = \left\{ \Gamma \in \gS ^1 (\cF_s(\gH)), \Gamma = \Gamma ^* , \Gamma\geq 0, \tr_{\cF_s(\gH)} [\Gamma] = 1 \right\}.
\end{equation}
We say that a state is diagonal (stricto sensu, block-diagonal) if it can be written in the form
\begin{equation}\label{eq:Fock state diag}
\Gamma = G_{0}\oplus G_{1}\oplus \ldots \oplus G_n \oplus \ldots 
\end{equation}
with $G_n \in \gS ^1 (\gH_s ^n)$. A state $\Gamma_N$ on the truncated Fock space $\cF_s ^{\leq N}(\gH)$, respectively a diagonal state on the truncated Fock space, are defined in the same manner. For a diagonal state on $\cF_s ^{\leq N}(\gH)$, of the form 
\[
\Gamma = G_{0,N}\oplus G_{1,N}\oplus \ldots \oplus G_{N,N},  
\]
its $n$-th reduced density matrix $\Gamma_N ^{(n)}$ is the operator on $\gH_s ^n$ defined by 
\begin{equation}\label{eq:Fock mat red}
\Gamma_N ^{(n)} = {N\choose n}^{-1}\sum_{k=n}^N{k\choose n}\tr_{n+1\to k}G_{N,k}.
\end{equation}
\hfill\qed
\end{definition}

The acquainted reader will notice two things:
\begin{itemize}
\item We introduce only those concepts that will be crucial to the sequel of the course. One may of course define the density matrices of general states, but we will not need this hereafter. For a diagonal state, the reduced density matrices~\eqref{eq:Fock mat red} characterize the state completely. For a non-diagonal state, one must also specify its ``off-diagonal'' density matrices $\Gamma ^{(p,q)}:\gH_s ^p \mapsto \gH_s ^q$ for $p\neq q$.
\item The normalization we chose in~\eqref{eq:Fock mat red} is not standard. It is chosen such that, in the spirit of the rest of the course, the $n$-th reduced matrix of a $N$- particles state (i.e. one with $G_{0,N} = \ldots = G_{N-1,N} = 0$ in~\eqref{eq:Fock mat red}) be of trace $1$. The standard convention would rather be to fix the trace at ${N \choose n}$, which is less convenient to apply Theorem~\ref{thm:DeFinetti fort}. In other words the normalization takes into account the fact that throughout the course we work we a prefered particle number $N$. 
\end{itemize}

We may now introduce the concept of localization of a state. We shall limit ourselves to $N$-body states and self-adjoint localization operators, which is sufficient for our needs in the sequel. The following lemma/definition is taken from~\cite{Lewin-11}. Other versions may be found e.g. in~\cite{Ammari-04,DerGer-99,HaiLewSol_thermo-09}.

\begin{lemma}[\textbf{Localization of a $N$-body state}]\label{lem:Fock loc}\mbox{}\\
Let $\Gamma_N \in \cS (\gH_s ^N)$ be a bosonic $N$-body state and $A$ a self-adjoint operator on $\gH$ with $0\leq A^2 \leq 1$. There exists a unique diagonal state $\Gamma_N ^A \in \cS (\cF_s ^{\leq N} (\gH))$ such that  
\begin{equation}\label{eq:Fock loc mat}
\left(\Gamma_N ^A\right) ^{(n)} = A ^{\otimes n} \Gamma_N ^{(n)} A ^{\otimes n}
\end{equation}
for all $0 \leq n \leq N$. Moreover, writing $\Gamma_N ^A$ in the form  
\[
\Gamma_N ^A = G_{0,N} ^A \oplus G_{1,N} ^A \oplus \ldots \oplus G_{N,N} ^A,  
\]
we have the fundamental relation
\begin{equation}\label{eq:Fock funda rel}
\tr_{\gH_s ^n} \left[G_{N,n} ^A \right] =\tr_{\gH_s ^{N-n}} \left[G_{N,N-n} ^{\sqrt{1-A ^2}} \right].
\end{equation}
\end{lemma}

\begin{remark}[Fock-space localization]\label{rem:Fock com loc}\mbox{}\\\vspace{-0,4cm}
\begin{enumerate}
\item The uniqueness part of the lemma shows that one has to work on Fock space. The localized state is in fact unique in $\cS (\cF_s (\gH))$, but to see this we would need slightly more general definitions, cf~\cite{Lewin-11}.
\item The relation~\eqref{eq:Fock funda rel} is one of the cornerstones of the method. Loosely speaking it expresses the fact that, in the state  $\Gamma_N$, \textbf{the probability of having $n$ particles $A$-localized is equal to the probability of having $N-n$ particles $\sqrt{1-A^2}$ localized}. Think of the case of a very simple localization operator, $A= \one _{B(0,R)}$, the (multiplication by the) indicative function of the ball of radius $R$. We are then simply saying that the probability of having exactly $n$ particles in the ball equals the probability of having exactly $N-n$ particles outside of the ball. Indeed, in probabilistic terms, these two possibilities correspond to the same event.\hfill\qed
\end{enumerate}

\begin{proof}[Proof of Lemma~\ref{lem:Fock loc}.]
Uniqueness, at least amongst diagonal states on the truncated Fock space, is a simple consequence of the fact that the reduced density matrices uniquely characterize the state. Details can be found in~\cite{Lewin-11}.

For the existence, one can use the usual identification (in the sense of unitary equivalence)
$$\cF_s (A  \gH \oplus \sqrt{1-A ^2}\gH) \simeq \cF_s (A\gH) \otimes \cF_s (\sqrt{1-A ^2}\gH)$$
and define the localized state by taking a partial trace with respect to the second Hilbert space in the tensor product of the right-hand side.  We shall follow a more explicit but equivalent route. To simplify notation we shall restrict to the case where $A=P$ is an orthogonal projector and thus $\sqrt{1-A ^2} = \Pp$. 

First note that for any $n$-body observable $O_n$
\begin{align*}
\tr \left[ O_n P ^{\otimes n} \Gamma_N ^{(n)} P ^{\otimes n} \right] &= \tr \left[ P ^{\otimes n} O_n P ^{\otimes n} \Gamma_N ^{(n)}  \right] \\
&= \tr \left[(P ^{\otimes n} O_n P ^{\otimes n}) \otimes \1 ^{\otimes N-n} \Gamma_N  \right]\\
&= \tr \left[ O_n \otimes \1 ^{\otimes N-n} ( P ^{\otimes n} \otimes \1 ^{\otimes N-n} \Gamma_N  P ^{\otimes n} \otimes \1 ^{\otimes N-n}) \right]
\end{align*}
so that 
$$ P ^{\otimes n} \Gamma_N ^{(n)} P ^{\otimes n} = \tr_{n+1 \to N} \left[ P ^{\otimes n} \otimes \1 ^{\otimes N-n} \Gamma_N  P ^{\otimes n} \otimes \1 ^{\otimes N-n}\right].$$

Thus, by cyclicity of the trace, we obtain
\begin{align*}
P ^{\otimes n} \Gamma_N ^{(n)} P ^{\otimes n} &= \tr_{n+1\to N} \left[(P ^{\otimes n} \otimes \one ^{\otimes (N-n)} ) P ^{\otimes n} \otimes \one ^{\otimes (N-n)} \Gamma_N  P ^{\otimes n} \otimes \one ^{\otimes (N-n)} \right]\nonumber\\
&= \sum_{k=0} ^{N-n} {N-n \choose k}  \tr_{n+1\to N} \left[P ^{\otimes n + k} \otimes \Pp ^{\otimes (N-n-k)} \Gamma_N  P ^{\otimes n} \otimes \one ^{\otimes (N-n)} \right]\nonumber
\\&= \sum_{k=0} ^{N-n} {N-n \choose k}  \tr_{n+1\to N} \left[P ^{\otimes n + k} \otimes \Pp ^{\otimes (N-n-k)} \Gamma_N  P ^{\otimes n+k} \otimes \Pp ^{\otimes (N-n-k)} \right]\nonumber
\\&= \sum_{k=n} ^{N} {N-n \choose k-n} \tr_{n+1\to N} \left[P ^{\otimes k } \otimes \Pp ^{\otimes (N-k)} \Gamma_N  P ^{\otimes k} \otimes \Pp ^{\otimes (N-k)} \right].\nonumber
\end{align*}
Here we inserted $\one = P + \Pp$ in the first occurence of $\1 ^{\otimes N-n}$ in the first line and expanded. Since the state $P ^{\otimes n} \otimes \one ^{\otimes (N-n)} \Gamma_N  P ^{\otimes n} \otimes \one ^{\otimes (N-n)}$ we act on is symmetric under exchange of the last $N-n$ variables,  we can reorganize the terms in the expansion to obtain the sum of the second line, involving binomial coefficients. It then suffices to note that
\begin{equation*}\label{eq:fact calcul}
{N-n \choose k-n} = { N \choose k} { N \choose n} ^{-1} { k \choose n}
\end{equation*}
to deduce
\begin{equation*}
P ^{\otimes n} \Gamma_N ^{(n)} P ^{\otimes n}  = \sum_{k=n} ^{N} { N \choose n} ^{-1} { k \choose n}  \tr_{n+1\to N} \left[G_{N,k} ^P \right]= \left(G_N  ^P \right)^{(n)}
\end{equation*}
with (cf Definition~\eqref{eq:Fock mat red})
\begin{equation}\label{eq:Fock calcul loc}
G_{N,k} ^P =  { N \choose k} \tr_{k+1\to N} \left[P ^{\otimes k } \otimes \Pp ^{\otimes (N-k)} \Gamma_N  P ^{\otimes k} \otimes \Pp ^{\otimes (N-k)} \right]
\end{equation}
and
\[
G_N  ^P = G_{N,0} ^P \oplus \ldots \oplus G_{N,N} ^P.
\]
This is indeed an operator on $\cF_s ^{\leq N} (P\gH)$: in~\eqref{eq:Fock calcul loc}, one can interchange the first $k$ particles and they live on $P\gH$, whereas the last $N-k$ particles are traced out. 

There remains to show that $G_N ^P$ is indeed a state, i.e. that its trace is $1$. To see this, we write
\begin{align*}
1 &= \tr_{\gH ^N} [\Gamma_N] = \tr_{\gH ^N} \left[\left(P+\Pp \right) ^{\otimes N} \Gamma_N \left(P+\Pp \right) ^{\otimes N}\right] \\
&= \sum_{k=0} ^N { N \choose k} \tr_{\gH ^N} \left[P ^{\otimes k } \otimes \Pp ^{\otimes (N-k)} \Gamma_N  P ^{\otimes k} \otimes \Pp ^{\otimes (N-k)}\right]
\\&= \sum_{k=0} ^N  \tr_{\gH ^k} \left[G_{N,k} ^P\right] = \tr_{\cF(\gH)} [G_{N} ^P].
\end{align*}
The relation~\eqref{eq:Fock funda rel} is an immediate consequence of~\eqref{eq:Fock calcul loc} and the symmetry of $\Gamma_N$.
\end{proof}
\end{remark}

\subsection{Proof of the weak quantum de Finetti theorem and corollaries}\label{sec:proof deF faible}\mbox{}\\\vspace{-0.4cm}

We are now going to use the localization procedure just described to prove the first implication of the strategy~\eqref{eq:schema deF 2}. The idea is to use a finite-rank projector $P$, and then combine~\eqref{eq:Fock mat red} with~\eqref{eq:Fock loc mat}  to write (with $n\in \N$ fixed)
\begin{align}\label{eq:wdeF formal}
P ^{\otimes n} \gamma_N ^{(n)} P ^{\otimes n} &= \sum_{k=n}^N {N\choose n}^{-1}{k\choose n}\tr_{n+1\to k} G_{N,k} ^P \nonumber
\\&\approx \sum_{k=n}^N \left(\frac{k}{N}\right) ^n \tr_{n+1\to k} G_{N,k} ^P.
\end{align}
Here we inserted the simple estimate (see the computation in~\cite{LewNamRou-13}, Equation (2.13))
\begin{equation}\label{eq:wdeF calcul}
{N\choose n}^{-1}{k\choose n} =  \left(\frac{k}{N}\right) ^n + O(N ^{-1}).
\end{equation}
We then argue as follows: The terms where $k$ is small contribute very little to the sum~\eqref{eq:wdeF formal} because of the factor $\left(\frac{k}{N}\right) ^n$. For the terms where $k$ is large, we note that, up to normalization, $G_{N,k} ^P$ is a $k$-particles bosonic state over $P\gH$. One may thus apply the de Finetti theorem discussed in Chapter~\ref{sec:deF finite dim} to it, without worrying about compactness issues since $P\gH$ has finite dimension. Since $k$ is large in these terms, and $n$ is fixed, we obtain (formally)
$$\tr_{n+1\to k} G_{N,k} ^P \approx \tr_{\gH ^k} [G_{N,k} ^P] \int_{u \in SP\gH} d\nu_k (u) |u ^{\otimes n}\rangle \langle u ^{\otimes n}| $$
for a certain measure $\nu_k$, and thus 
\[
P ^{\otimes n} \gamma_N ^{(n)} P ^{\otimes n} \approx \sum_{k \simeq N}^N \tr_{\gH ^k} [G_{N,k} ^P] \left(\frac{k}{N}\right) ^n \int_{u \in SP\gH} d\nu_k (u) |u ^{\otimes n}\rangle \langle u ^{\otimes n}|.
\]
In the limit $N\to \infty$, the discrete sum should become an integral in $\lambda = k/N $. Using the fact that $G_{N} ^P$ is a state to deal with normalization  (recall that~$\sum_{k=1} ^N \tr [G_{N,k} ^P] = 1$), it is natural to hope that we can obtain
\[
P ^{\otimes n} \gamma_N ^{(n)} P ^{\otimes n} \approx \int_{0} ^1  \lambda ^n \int_{u \in SP\gH} d\nu (\lambda,u)  (u) |u ^{\otimes n}\rangle \langle u ^{\otimes n}|,
\]
which may be rewritten in the form~\eqref{eq:melange faible} by defining (in ``spherical'' coordinates on the unit ball $BP\gH$)
\[
 d\mu (u) = d\mu \left( \norm{u},\frac{u}{\norm{u}} \right) := d \norm{u} ^2 \times d\nu_{\norm{u} ^2} \left(\frac{u}{\norm{u}}\right).
\]
There remains, as a last step, to apply this procedure for a sequence of projectors $P_{\ell} \to \one$ and to check some compatibility relations to conclude. The final measure might not be a probability measure, which one can compensate by adding a delta function at the origin without changing any of the previous formulae.

Note that this proof, which combines the methods of Chapter~\ref{sec:deF finite dim} and Section~\ref{sec:loc rig}, gives a recipe to construct the de Finetti measure ``by hand'' (up to the fact that we have to pass to the limit at some point). This is very useful in practice, see Chapters~\ref{sec:Hartree} and~\ref{sec:NLS}. The spirit of the proof using localization is reminiscent of some aspects of the method of Ammari and Nier~\cite{Ammari-hdr,AmmNie-08} who define cylindrical projections of the measure first.

Let us now present the details of the proof, following~\cite[Section 2]{LewNamRou-13}.

\begin{proof}[Proof of Theorem~\ref{thm:DeFinetti faible}.]
\textbf{Existence.} We carry on with the above notation. The following formula defines $M_{P,N}^{(n)}$ as a measure on $[0,1]$, with values in the positive hermitian matrices of size $\dim\,\left(\otimes_s^n (P\gH)\right)$:
$$\dM_{P,N}^{(n)}(\lambda):=\sum_{k=n}^N\;\delta_{k/N}(\lambda)\;\tr_{n+1\to k}G^P_{N,k}.$$
We then have, using~\eqref{eq:wdeF calcul} 
\begin{equation}\label{eq:wdeF preuve 1}
\tr\left|P^{\otimes n}\gamma^{(n)}_{\Gamma_N}P^{\otimes n}-\int_0^1 \lambda^n\,\dM_{P,N}^{(n)}(\lambda)\right| \leq \frac{C}{N}\sum_{k=n}^N\tr G^P_{N,k} \to 0 \mbox{ when } N\to \infty.
\end{equation}
Since $P$ is a finite rank projector and $\gamma_N ^{(n)}$ converges weakly-$\ast$ by assumption, 
\begin{equation}\label{eq:wdeF preuve 2}
P^{\otimes n}\gamma^{(n)}_{N}P^{\otimes n}\to P^{\otimes n}\gamma^{(n)}P^{\otimes n} 
\end{equation}
strongly in trace-class norm. On the other hand 
$$\tr_{\gH^n}\left[\int_0^1\dM_{P,N}^{(n)}(\lambda)\right]=\sum_{k=n+1}^N\;\tr_{\gH^k}G^P_{N,k}\leq\sum_{k=0}^N\;\tr_{\gH^k}G^P_{N,k}=1,$$
so $M_{P,N}^{(n)}$ is a sequence of measures with bounded total variation over a compact finite dimensional space (positive hermitian matrices of size $\dim P \gH$ having a trace less than $1$). One may thus extract from it a subsequence converging weakly as measures to $M_P^{(n)}$. Combining with~\eqref{eq:wdeF preuve 1} and~\eqref{eq:wdeF preuve 2} we have 
\begin{equation}\label{eq:wdeF preuve 3}
P^{\otimes n}\gamma^{(n)}P^{\otimes n}=\int_0^1 \lambda^n\,\dM_P^{(n)}(\lambda).
\end{equation}

We now have to show that the sequence of measures $\left(M_P^{(n)}\right)_{n\in \N}$ that we just obtained is consistent in the sense that, for all $n\geq0$ and for all continuous functions $f$ over $[0,1]$ vanishing at $0$,
\begin{equation}
\int_{0} ^1 f(\lambda) \tr_{n+1} dM_P^{(n+1)} (\lambda)= \int_{0} ^1 f(\lambda) d M_P^{(n)} (\lambda).
\label{eq:A_P_consistent 2}
\end{equation}
Here $\tr_{n+1}$ denotes partial trace with respect to the last variable. We have 
\begin{align*}
\tr_{n+1} d M_{P,N}^{(n+1)}(\lambda)&=\sum_{k=n+1}^N\;\delta_{k/N}(\lambda)\;\tr_{n \to k} G_{N,k} \\
&=d M_{P,N}^{(n)}(\lambda)- \delta_{n/N}(\lambda) G^P_{N,n}
\end{align*}
and thus
\begin{align*}
\int_{0} ^1 f(\lambda) \tr_{\gH^n}\left|\tr_{n+1}d M_{P,N}^{(n+1)}(\lambda)-d M_{P,N}^{(n)}(\lambda)\right|&\leq  \int_{0} ^1 f(\lambda) \delta_{n/N}(\lambda) \tr_{\gH ^n} G^P_{N,n} 
\\&\leq  f\left(\frac{n}{N}\right)
\end{align*}
since $\tr_{\gH ^n} G^P_{N,n}\leq 1$. Passing to the limit we obtain~\eqref{eq:A_P_consistent 2} for all function $f$ such that $f(0) = 0$. 

We now apply Theorem~\ref{thm:DeFinetti fort} in finite dimension. Stricto sensu, this result applies only at fixed $\lambda$, but approaching  $dM_P ^{(n)}$ by step functions and then passing to the limit, we obtain a  measure $\nu_P$ on $[0,1]\times S\gH\cap (P\gH)$ such that
$$\int_0 ^1 f(\lambda) \dM_P^{(n)}(\lambda)=\int_{S\gH}\int_0 ^1 f(\lambda) d\nu_P(\lambda,u)|u^{\otimes n}\rangle\langle u^{\otimes n}|$$
for all continuous functions $f$ vanishing at $0$. We thus obtain 
\begin{align*}
P^{\otimes n}\gamma^{(n)}P^{\otimes n}&=\int_{0}^1\int_{S\gH}d\nu_{P}(\lambda,u)\,\lambda^n |u^{\otimes n}\rangle\langle u^{\otimes n}|\\
&=\int_{0}^1\int_{S\gH}d\nu_{P}(\lambda,u)\,|(\sqrt\lambda u)^{\otimes n}\rangle\langle (\sqrt\lambda u)^{\otimes n}|\\
&= \int_{B \gH} d\mu_P (u) |u^{\otimes n}\rangle\langle u^{\otimes n}|, 
\end{align*}
defining the measure $\mu_P$ in spherical coordinates. We are free to add a Dirac mass at the origin to turn $\mu_P$ into a probability measure.

The argument can be applied to a sequence of finite rank projectors converging to the identity. We then have a sequence of probability measures  $\mu_k$ on $B\gH$ such that
$$P_k^{\otimes n}\gamma^{(n)}P_k^{\otimes n}=\int_{B\gH}d\mu_k(u)\,|u^{\otimes n}\rangle\langle u^{\otimes n}|.$$
Taking an increasing sequence of projectors (i.e. $P_k \gH \subset P_{k+1} \gH$) it is clear that $\mu_k$ co\"incides with $\mu_\ell$ on $P_{\ell} \gH$ for $\ell\leq k$ (cf the argument for measure uniqueness below). Since all these measures have their supports in a bounded set, there exists (see for example~\cite[Lemma 1]{Skorokhod-74}) a unique probability measure\footnote{To construct it, note that the $\sigma$-closure of the union for $k\geq 0$ of the borelians of $P_k\gH$ co\"incides with the borelians of $\gH$.} $\mu$ on $B\gH$ which co\"incides with $\mu_k$ on $P_k\gH$ in the sense that:  
$$\int_{B\gH}d\mu_k(u)\,|u^{\otimes n}\rangle\langle u^{\otimes n}|=\int_{B\gH}d\mu(u)\,|(P_ku)^{\otimes n}\rangle\langle (P_ku)^{\otimes n}|.$$
We thus conclude
$$P_k^{\otimes n}\gamma^{(n)}P_k^{\otimes n}=P_k^{\otimes n}\left(\int_{B\gH}d\mu(u)\,|u^{\otimes n}\rangle\langle u^{\otimes n}|\right)P_k^{\otimes n}$$
and there only remains to take the limit $k\to \infty$ to deduce the existence of a measure satisfying~\eqref{eq:melange faible}. 

\medskip

\noindent \textbf{Uniqueness.} Let us now prove that the measure just constructed must be unique. Let $\mu$ and $\mu'$ satisfy
\begin{equation}
\int_{B\gH}|u^{\otimes k}\rangle\langle u^{\otimes k}|d\mu(u) -\int_{B\gH}|u^{\otimes k}\rangle\langle u^{\otimes k}|d\mu'(u) = 0
\label{eq:uniqueness}
\end{equation}
for all $k\geq1$. Let $V = \mathrm{vect} (e_1,\ldots,e_d)$ be a finite-dimensional subspace of $\gH$, $P_i=|e_i\rangle\langle e_i|$ the associated projectors and $\mu_V,\mu'_V$ the cylindrical projections of $\mu$ and $\mu'$ over~$V$. Applying $P_{i_1}\otimes\cdots\otimes P_{i_k}$ on the left  and $P_{j_1}\otimes\cdots\otimes P_{j_k}$ on the right to~\eqref{eq:uniqueness}, we obtain 
$$\int_{BV}u_{i_1}\cdots u_{i_k}\overline{u_{j_1}}\cdots \overline{u_{j_k}}\,d(\mu_V-\mu_V')(u)=0$$
for all mutli-indices $i_1,...,i_k$ and $j_1,...,j_k$ (we denote $u=\sum_{j=1}^du_je_j$). On the other hand, by $S^1$ invariance of both measures it is clear that if $k\neq \ell$,
$$\int_{BV}u_{i_1}\cdots u_{i_k}\overline{u_{j_1}}\cdots \overline{u_{j_\ell}}\,d(\mu_V-\mu_V')(u)=0.$$
Since polynomials in $u_i$ and $\overline{u_j}$ are dense in $C^0(BV,\C)$ (continuous functions on the unit ball of $V$), we deduce that the cylindrical projections of $\mu$ and $\mu'$ over $V$ co\"incide. That being true for all finite dimensional subspace $V$, we conclude that the two measures must co\"incide everywhere.
\end{proof}

We now give a particular case, adapted to our needs in the next chapter, of a very useful corollary of the above method of proof (see~\cite[Theorem 2.6]{LewNamRou-13} for the general statement):

\begin{corollary}[\textbf{Localization and de Finetti measure}]\label{cor:wdeF loc}\mbox{}\\
Let $(\Gamma_N)_{N\in \N}$ be a sequence of $N$-body states over $\gH = L ^2 (\R ^d)$ satisfying the assumptions of Theorem~\ref{thm:DeFinetti faible} and $\mu$ be the associated de Finetti measure. Assume that 
\begin{equation}\label{eq:wdeF cor assum}
\tr \left[ -\Delta \gamma_N^{(1)} \right] = \tr \left[ |\nabla| \gamma_N^{(1)} |\nabla| \right] \leq C 
\end{equation}
for some constant independent of $N$. Let $\chi$ be a localization function with compact support with $0\leq \chi \leq 1$ and $G_N^{\chi}$ be the localized state defined in Lemma~\ref{lem:Fock loc}. Then
\begin{equation}
\lim_{N\to\ii}\sum_{k=0}^Nf\left(\frac{k}{N}\right)\tr_{\gH^k} G_{N,k}^\chi = \int_{B\gH} d\mu(u)\;f(\|\chi u\|^2)
\label{eq:wdeF cor}
\end{equation}
for all continuous functions $f$ on $[0,1]$.
\end{corollary}
 
\begin{proof}
Since polynomials are dense in the continuous functions on $[0,1]$ it is sufficient to consider the case $f(\lambda)=\lambda^n$ with $n=0,1,...$. We then use~\eqref{eq:Fock mat red},~\eqref{eq:Fock loc mat} and~\eqref{eq:wdeF calcul} again to write
\[
\left|\sum_{k=0}^N \left(\frac{k}{N}\right) ^n \tr_{\gH^k} G_{N,k}^\chi  - \tr_{\gH ^n} \left[\chi ^{\otimes n} \gamma_N ^{(n)} \chi ^{\otimes n} \right]\right| \to 0 \mbox{ when } N\to \infty.
\]
Assumption~\eqref{eq:wdeF cor assum} ensures 
\[
 \tr \left[ \left( \sum_{j=1}^{n} -\Delta_j \right) \gamma_N^{(n)} \right] \leq C n
\]
and we may thus assume that
\[
\left( \sum_{j=1}^{n} |\nabla|_j \right) \gamma_N^{(n)} \left( \sum_{j=1}^{n} |\nabla|_j \right) \wto_{\ast} \left( \sum_{j=1}^{n} |\nabla|_j \right) \gamma ^{(n)} \left( \sum_{j=1}^{n} |\nabla|_j \right)\mbox{ when } N\to \infty. 
\]
The multiplication operator by $\chi$ is relatively compact with respect to the Laplacian, thus
$$D_{n} ^\chi:= \chi ^{\otimes n} \left( \sum_{j=1}^{n} |\nabla|_j \right) ^{-1}$$
is a compact operator. We then have 
\begin{align*}
 \tr_{\gH ^n} \left[\chi ^{\otimes n} \gamma_N ^{(n)} \chi ^{\otimes n} \right] & = \tr_{\gH ^n} \left[ D_{n} ^\chi \left( \sum_{j=1}^{n} |\nabla|_j \right) \gamma_N ^{(n)} \left( \sum_{j=1}^{n} |\nabla|_j \right)  D_{n} ^\chi \right] \\
&\to \tr_{\gH ^n} \left[D_{n} ^\chi  \left( \sum_{j=1}^{n} |\nabla|_j \right) \gamma ^{(n)}  \left( \sum_{j=1}^{n} |\nabla|_j \right)  D_{n} ^\chi \right]\\ 
 &= \tr_{\gH ^n} \left[\chi ^{\otimes n} \gamma ^{(n)} \chi ^{\otimes n} \right].
\end{align*}
We conclude by using~\eqref{eq:melange faible} that
\begin{align*}
\sum_{k=0}^N \left(\frac{k}{N}\right) ^n \tr_{\gH^k} G_{N,k}^\chi &\to  \int_{B\gH} d\mu(u) \tr_{\gH ^n} \left[ |(\chi u)  ^{\otimes n}\rangle \langle (\chi u)  ^{\otimes n}|\right]\\
&= \int_{B\gH} d\mu(u) \norm{\chi u} ^{2n} = \int_{B\gH} d\mu(u) f(\norm{\chi u} ^{2}).
\end{align*}
\end{proof}

\begin{remark}[Weak de Finetti measure and loss of mass.]\mbox{}\\
\vspace{-0.4cm}
\begin{enumerate}
\item Assumption~\eqref{eq:wdeF cor assum} is used to ensure strong compactness of density matrices, see the proof. It is of course very natural for the applications we have in mind (uniformly bounded kinetic energy). The convergence in~\eqref{eq:wdeF cor} means that the mass of the de Finetti measure $\mu$ on the sphere $\{\|u\|^2=\lambda\}$ corresponds to the probability that a fraction $\lambda$ of the particles does not escape to infinity.
\item We will use this corollary to obtain information on particles escaping to infinity in the following manner. We define the function
$$ \eta = \sqrt{1-\chi ^2}$$ 
which localizes ``close to infinity''. One of course cannot apply the result directly to the localized state $G_N^{\eta}$. Instead one may use Relation~\eqref{eq:Fock funda rel} to obtain
\begin{align}
\lim_{N\to\ii}\sum_{k=0}^N f\left(\frac{k}{N}\right)\tr_{\gH^k}G_{N,k}^{\eta}&=\lim_{N\to\ii}\sum_{k=0}^N f\left(1-\frac{k}{N}\right)\tr_{\gH^k}G_{N,k}^{\chi}\nn\\
&=\int_{B\gH}d\mu(u)\;f\left(1- \|\chi u \|^2\right), \label{eq:wdeF cor 2}
\end{align}
which gives some control on the loss of mass at infinity, encoded in the de Finetti measure. This is a bit surprising since the latter by definition describes particles which stay trapped. Typically we will use~\eqref{eq:wdeF cor 2} with $f (\lambda) \simeq \eH (\lambda)$, the Hartree energy at mass $\lambda$, see the next chapter.  
\end{enumerate}
\hfill\qed
\end{remark}

\newpage

\section{\textbf{Derivation of Hartree's theory: general case}}\label{sec:Hartree}

We now turn to the general case of the derivation of Hartree's functional as a limit of the $N$-body problem in the mean-field regime. In Chapter~\ref{sec:quant} we saw that, under simplifying physical assumptions, the result is a rather direct consequence of the weak and strong de Finetti theorems. The general case requires a more thorough analysis and we shall use fully the localization tools introduced in Chapter~\ref{sec:locFock}.

The setting is now that of particles in a non-trapping external potential $V$, interacting via a potential that may have bound states. Let us recap the notation: The Hamiltonian of the full system is 
\begin{equation}\label{eq:bis quant hamil}
H_N ^V = \sum_{j=1} ^N T_j + \frac{1}{N-1}\sum_{1\leq i < j \leq N} w(x_i-x_j), 
\end{equation}
acting on the Hilbert space $\gH_s ^N = \bigotimes_s ^N \gH$, with $\gH = L ^2 (\R ^d)$. The one-body Hamiltonian is
\begin{equation}\label{eq:bis op Schro}
T = - \Delta + V, 
\end{equation}
that we assume to be self-adjoint and bounded below. We have emphasized the dependence on the potential $V$ in the notation~\eqref{eq:bis quant hamil} because we will be lead to consider the system where particles are lost at infinity described by the Hamiltonian $H_N ^0$ where one sets  $V\equiv 0$. One can generalize in the same directions as mentioned in Remark~\ref{rem:energie cin}, but we shall for simplicity stick to the above model case.

The interaction potential $w:\R \mapsto \R$ will be asumed bounded relatively to $T$: for some $0 \leq \beta_-, \beta_+ \leq 1$ and $C>0$
\begin{equation}\label{eq:bis T controls w}
 -\beta_-(T_1+ T_2)-C\leq w(x_1-x_2) \leq \beta_+ (T_1+T_2) + C.
\end{equation}
We also assume symmetry
\[
w(-x) = w(x), 
\]
and some decay at infinity
\begin{equation}\label{eq:bis decrease w}
w\in L^p (\R ^d) + L ^{\infty} (\R ^d), \max(1,d/2) < p < \infty \to 0, w(x) \to 0 \mbox{ when }  |x|\to\infty.
\end{equation}
Again, it is rather vain to consider partially trapping one-body potentials, and we thus assume that $V$ is non-trapping in all directions:
\begin{equation}\label{eq:bis sans confinement}
V \in L ^p (\R ^d) + L ^{\infty} (\R ^d), \max{1,d/2}\leq p <\infty,  V(x) \to 0 \mbox{ when } |x|\to \infty.
\end{equation}
This ensures~\cite{ReeSim2} that $H_N$ is self-adjoint and bounded below. The ground state energy of~\eqref{eq:quant hamil} is always given by
\begin{equation}\label{eq:bis quant ground state}
E ^V (N) = \inf \sigma_{\gH ^N} H_N ^V  = \inf_{\Psi \in \gH ^N, \norm{\Psi} = 1} \left\langle \Psi, H_N ^V \Psi \right\rangle_{\gH ^N}.
\end{equation}
Finally let us recall what the limit objects are. The Hartree functional with potential $V$ is given by 
\begin{equation}\label{eq:bis hartree func}
\EH  ^V [u] := \int_{\R ^d } |\nabla u| ^2 + V |u| ^2 + \frac12 \iint_{\R ^d \times \R ^d } |u(x)| ^2 w(x-y) |u(y)| ^2 dx dy
\end{equation}
and we shall use the notation $\EH ^0$ for the translation-invariant functional where $V\equiv 0$. The Hartree energy at mass $\lambda$ is given by 
\begin{equation}\label{eq:bis Hartree perte masse}
\eH  ^V(\lambda) := \inf_{\norm{u} ^2 = \lambda} \EH ^V [u],\quad 0\leq \lambda \leq 1.
\end{equation}
Under the previous assumptions we will always have the binding inequality
\begin{equation}\label{eq:hartree liaison}
\eH ^V (1) \leq \eH ^V (\lambda) + \eH ^0 (1-\lambda) 
\end{equation}
which is easily proved by evaluating the energy of a sequence of functions with a mass $\lambda$ in the well of the potential $V$ and a mass $1-\lambda$ escaping to infinity. We will prove the following theorem, extracted from~\cite{LewNamRou-13} (particular case of Theorem 1.1 therein).

\begin{theorem}[\textbf{Derivation of Hartree's theory, general case}]\label{thm:hartree general}\mbox{}\\
Under the preceding assumptions, we have:

\medskip

\noindent$(i)$ \underline{Convergence of energy.}
\begin{equation}
\boxed{\lim_{N\to\ii}\frac{E^V(N)}{N}=\eH^V(1).}
\label{eq:general ener}
\end{equation}

\medskip

\noindent $(ii)$ \underline{Convergence of states.} Let $\Psi_N$ be a sequence of $L ^2 (\R ^{dN})$-normalized quasi-minimizers for $H_N ^V$:  
\begin{equation}\label{eq:quasi min}
\pscal{\Psi_N,H_N^V\Psi_N}= E^V(N)+o(N) \mbox{ when } N \to \infty, 
\end{equation}
and $\gamma^{(k)}_N$ be the corresponding reduced density matrices. There exists $\mu \in \PP (B\gH)$ a probability measure on the unit ball of  $\gH$ with $\mu (\cM ^V) = 1$, where
\begin{equation}
\cM^V =\Big\{u\in B\gH\ :\ \cEH^V [u]=e^V_H(\norm{u}^2)=e^V_H(1)-e^0_H(1-\norm{u}^2)\Big\},
\label{eq:def_M_V} 
\end{equation}
such that, along a subsequence, 
\begin{equation}
\boxed{\gamma^{(k)}_{N_j}\wto_*\int_{\cM^V}|u^{\otimes k}\rangle\langle u^{\otimes k}|\,d\mu(u)}
\label{eq:general state}
\end{equation}
weakly-$\ast$ in $\gS^1(\gH^k)$, for all $k\geq1$.

\medskip

\noindent $(iii)$ If in addition the strict binding inequality 
\begin{equation}
\eH^V(1)<\eH^V(\lambda)+\eH^0(1-\lambda)
\label{eq:hartree liaison stricte}
\end{equation}
holds for all $0\leq\lambda<1$, the measure $\mu$ has its support in the sphere $S\gH$ and the limit~\eqref{eq:general state} holds in trace-class norm. In particular, if $\eH^V(1)$ has a unique (modulo a constant phase) minimizer $u_H$, then for the whole sequence
\begin{equation}
\boxed{\gamma^{(k)}_{N}\to|u_H^{\otimes k}\rangle\langle u_H^{\otimes k}| \mbox{  strongly in } \gS ^1(\gH ^k)} 
\label{eq:general BEC}
\end{equation} 
for all fixed $k\geq1$.
\end{theorem}

\begin{remark}[Characterization of the limit set]\label{rem:limit set}\mbox{}\\
It is a classical fact that  
$$\cM^V  = \left\lbrace u\in B\gH \: | \: \exists (u_n) \mbox{ minimizing sequence for } e ^V_H (1) \mbox{ such that } u_n \wto u \mbox{ weakly in } L ^2 (\R ^d) \right\rbrace.$$
This follows from the usual concentration compactness method for non-linear one-body variational problems. One can thus interpret~\eqref{eq:general state} as saying that the weak limits for the $N$-body problem are fully characterized in terms of the weak limits for the one-body problem. \hfill\qed
\end{remark}

The proof proceeds in two steps. In Section~\ref{sec:trans inv} we first consider the completely translation-invariant case where $V\equiv 0$, which will describe particles escaping from the potential well $V$ in the general case. We show that the energy $\eH ^0 (1)$ is the limit of $ N ^{-1} E ^0 (N)$. In this case one cannot hope for much more since there always exist minimizing sequences whose density matrices converge to $0$ in trace-class norm. In addition to the tools already introduced we shall rely on ideas of L\'evy-Leblond, Lieb, Thirring and Yau~\cite{LevyLeblond-69,LieThi-84,LieYau-87} to recover a bit of compactness. 

We then use fully the localization methods  of Chapter~\ref{sec:locFock} to treat the general case in Section~\ref{sec:general general}. In the spirit of the concentration-compactness principle we will localize minimizing sequences inside and outside of a ball. The inside-localized part is described by the weak-$\ast$ limit of reduced density matrices and we can thus use the weak de Finetti theorem. The outside-localized part no longer sees the potential $V$ and we may thus apply to it the results of Section~\ref{sec:trans inv}.

\subsection{The translation-invariant problem}\label{sec:trans inv}\mbox{}\\\vspace{-0.4cm}

Here we deal with the case where $V \equiv 0$. It is then possible to construct quasi-minimizing sequences for $E ^0 (N)$ with $\gamma_N ^{(1)} (.-y_N) \wto_{\ast} 0$ for any sequence of translations $x_N$. One may thus have \emph{vanishing} in Lions' terminology~\cite{Lions-84,Lions-84b} and without any specific trick one cannot hope for more than the convergence of the energy. Indeed, it is possible to construct a state where the relative motion of the particles is bound by the interaction potential $w$ but where the center of mass vanishes. This implies vanishing for the whole sequence. We shall thus be content with proving the convergence of the energy:

\begin{theorem}[\textbf{Energy of translation-invariant systems}]\label{thm:trans inv}\mbox{}\\
Under the preceding assumptions, we have 
\begin{equation}\label{eq:triv ener}
\boxed{ \lim_{N\to \infty} \frac{E^0(N)}{N} =e^0_{\rm H}(1).}
\end{equation}
\end{theorem}

The first step is to use part of the interaction potential to create an attractive one-body potential (this roughly amounts to taking out the center of mass degree of freedom). This way we will define an auxiliary problem whose energy is close to the original one, but for which particles always stay trapped. This is done in the following lemma, inspired from~\cite{LieThi-84,LieYau-87}:

\begin{lemma}[\textbf{Auxiliary problem with binding}]\label{lem:triv aux}\mbox{}\\
We split $w = w ^+ - w ^-$ into positive and negative parts and define, for some $\ep >0$
$$ \wep (x) = w (x) + \ep w ^- (x).$$ 
Consider the auxiliary Hamiltonian
\begin{equation}\label{eq:aux hamil}
H_N ^{\eps}= \sum_{i=1}^{N}\Big(K_i -\eps w_-(x_i) \Big) + \frac{1}{N-1} \sum_{1\leq i<j \leq N} w_\eps (x_i-x_j) 
\end{equation}
and $\Eep (N)$ the associated ground state energy. Then
\begin{equation}\label{eq:aux low bound}
\aep:= \lim_{N\to \infty} \frac{\Eep(N)}{N} \leq \lim_{N\to \infty} \frac{E ^0 (N)}{N} =:a.
\end{equation}
\end{lemma}

\begin{proof}
Using symmetry, we can for all $\Psi \in L_s ^2 (\R ^{dN})$ write 
\[
\left\langle \Psi, \left( \sum_{i=1}^{N} -\Delta_i \right) \Psi \right\rangle = \frac{N}{N-1} \left\langle \Psi, \left( \sum_{i=1}^{N-1} -\Delta_i \right) \Psi \right\rangle.
\]
Similarly
\[
\frac{2}{N(N-1)} \left\langle \Psi, \sum_{1\leq i<j \leq N}  w (x_i-x_j)\Psi \right\rangle = \tr_{\gH ^2} \left[w \gammaP ^{(2)}\right] =  \tr_{\gH ^2} \left[\wep \gammaP ^{(2)}\right] - \eps \tr_{\gH ^2} \left[w^- \gammaP ^{(2)}\right]
\]
with $\gammaP = |\Psi \rangle \langle \Psi |$. Then 
\[
 \tr_{\gH ^2} \left[\wep \gammaP ^{(2)}\right] = \frac{2}{(N-1)(N-2)} \tr \left[\sum_{1\leq i<j \leq N-1} w_\eps(x_i-x_j) \gammaP \right]
\]
and
\[
 \eps \tr_{\gH ^2} \left[w^- \gammaP ^{(2)}\right]  = \frac{\eps}{N-1} \tr \left[\sum_{i=1} ^{N-1} w ^- (x_i-x_N) \gammaP \right].
\]
All this implies
\begin{multline}
\frac{N-1}{N}\langle \Psi , H_N^0  \Psi \rangle\\ =\left\langle \Psi, \left( \sum_{i=1}^{N-1}\Big(-\Delta_i -\eps w ^- (x_i-x_N) \Big) + \frac{1}{N-2} \sum_{1\leq i<j \leq N-1} w_\eps(x_i-x_j) \right) \Psi \right\rangle.
\label{eq:trick-Nam}
\end{multline}
In the preceding equation the Hamiltonian between parenthesis depends on $x_N$ via the one-body potential $\eps w ^- (x_i-x_N)$  but since the other terms are translation-invariant the bottom of the spectrum is in fact independent of $x_N$. We thus have 
\begin{equation}\label{eq:aux low bound pre}
\frac{N-1}{N}\langle \Psi , H_N^0  \Psi \rangle \geq E^\eps(N-1) \langle \Psi,\Psi\rangle  
\end{equation}
for all $\Psi\in L ^2_s (\R ^{dN})$, which implies
$$\frac{E^0(N)}{N}\ge \frac{E ^\eps(N-1)}{N-1}.$$
The sequences $E^0(N)/N$ and $E^\eps(N)/N$ are increasing since they are the infimum of variational problems set on decreasing sets (cf~\eqref{eq:MF quant form 1}). Using trial states spread over larger and larger domains method, one may on the other hand show that 
$$ \frac{E^0(N)}{N} \leq 0,\quad \frac{E^\eps(N)}{N} \leq 0$$
and thus that the limits $a$ and $\aep$ exist. Then~\eqref{eq:aux low bound pre} clearly implies~\eqref{eq:aux low bound}.
\end{proof}

We next derive a lower bound to $\Eep(N)$ for $\epsilon$ small enough. It is much easier to work on this energy because the corresponding sequences of reduced density matrices are strongly compact in $\gS ^1$. Indeed, compactness or its absence (physically, binding or its absence) resuts from a comparison between the attractive and repulsive parts of the one- and two-body potentials. Here the one-body potential $\ep w ^-$ and the two-body potential $\wep$ are well equilibrated because they have been built precisely for this purpose, starting from the original two-body potential $w$. This trick will allow us to conclude the 

\begin{proof}[Proof of Theorem~\ref{thm:trans inv}.]
As usual, only the lower bound is non-trivial. Since $\eH ^0 (1) \leq 0$ one may assume that $a<0$ since otherwise $a=e^0_{\rm H}(1)=0$ and there is nothing to prove. We are going to prove the lower bound
\begin{equation}\label{eq:aux prob bound}
 \aep \geq e_{\rm H} ^{\eps}(1) := \inf_{\norm{u}^2=1} \left\{ \langle u,(-\Delta-\eps w_-) u\rangle + \frac{1}{2} \iint |u(x)|^2 w_\eps (x-y) |u(y)|^2 dxdy \right\},
\end{equation}
and it is then easy to show that $e_{\rm H} ^{\eps}(1) \to \eH  ^{0} (1)$ when $\ep \to$ to obtain~\eqref{eq:triv ener} by combining with~\eqref{eq:aux low bound}.  

Let $\Psi_N$ be a sequence of wave-functions such that
$$\langle \Psi_N, H_N ^{\eps} \Psi_N \rangle = \Eep(N)+o(N).$$ 
and $\gamma_N^{(k)}$ the corresponding reduced density matrices. Then
$$
a_\eps = \lim_{N\to \infty}\frac{\langle \Psi_N, H_N ^{\eps} \Psi_N \rangle}{N}=  \lim_{N\to \infty} \left( \Tr_\gH\left[(-\Delta-\eps w_-)\gamma_N^{(1)}\right] + \frac{1}{2} \Tr_{\gH^2}\left[w_\eps \gamma_N^{(2)}\right] \right).
$$ 
After the usual diagonal extraction we can assume that  
$$\gamma^{(k)}_N\wto_\ast\gamma^{(k)}$$
and we are going to show that the convergence is actually strong. We shall afterwards use the strong de Finetti theorem to obtain Hartree's energy as a lower bound.

We pick a smooth partition of unity $\chi_R^2 +\eta_R^2=1$ and use Lemma~\ref{lem:loc ener} to get 
\begin{multline}
\label{eq:split-energy-translation}
a_\eps \ge  \liminf_{R\to \infty}\liminf_{N\to \infty} \left\{ \Tr_\gH\left[(-\Delta-\eps w_-) \chi_R \gamma_N^{(1)}\chi_R\right] + \frac{1}{2} \Tr_{\gH^2}\left[w_\eps \chi_R^{\otimes 2}\gamma_N^{(2)}\chi_R^{\otimes 2}\right] \right. \\
\left. + \Tr_\gH\left[ -\Delta \eta_R \gamma_N^{(1)}\eta_R\right] + \frac{1}{2} \Tr_{\gH^2}\left[w_\eps \eta_R^{\otimes 2}\gamma_N^{(2)}\eta_R^{\otimes 2}\right] \right \}. 
\end{multline}
We define the $\chi_R$-- and $\eta_R$--localized states $G_N^\chi$ and $G_N^\eta$ by applying Lemma~\ref{lem:Fock loc} to $\Psi_N$. We will use those to estimate separately the two terms in the right side of~\eqref{eq:split-energy-translation}.

\medskip

\paragraph*{\bf The $\chi_R$-localized term.} Using~\eqref{eq:Fock mat red} and~\eqref{eq:Fock loc mat} we have 
\begin{multline}
\label{eq:localization-chi-eps}
\Tr_\gH\left[(-\Delta-\eps w_-) \chi_R \gamma_N^{(1)}\chi_R\right] + \frac{1}{2} \Tr_{\gH^2}\left[w_\eps \chi_R^{\otimes 2}\gamma_N^{(2)}\chi_R^{\otimes 2}\right]\\
= \frac{1}{N} \sum_{k=1}^N\Tr_{\gH^k}\left[ \left( \sum_{i=1}^k (-\Delta-\eps w_-)_i + \frac{1}{N-1} \sum_{i<j}^k w_\eps (x_i-x_j)\right) G^\chi_{N,k} \right]
\end{multline}
We apply the inequality 
\begin{equation}\label{eq:A-tB}
A+tB = (1-t)A+ t(A+B) \ge (1-t) \inf \sigma(A) + t \inf\sigma  (A+B) 
\end{equation}
with
\begin{equation*}
A=\sum_{\ell =1}^k (-\Delta-\eps w_-)_\ell,\quad  A+B =H_{\eps,k},\quad t=(k-1)/(N-1). 
\end{equation*}
We have 
$$ \lim_{\eps \to 0} \inf \sigma(-\Delta-\eps w_-) = \inf \sigma(-\Delta) = 0$$
and since we assume that $a<0$, for $\eps$ small enough
$$\inf \sigma(-\Delta-\eps w_-) > a \ge a_\eps \ge k^{-1}\inf \sigma(H_{\eps,k}).$$
Thus
$$\inf\sigma(A)\ge \inf \sigma(A+B)$$
and we may then write 
$$\Tr_\gH \left[(-\Delta-\eps w_-) \chi_R \gamma_N^{(1)}\chi_R\right] + \frac{1}{2} \Tr_{\gH^2}\left[w_\eps \chi_R^{\otimes 2}\gamma_N^{(2)}\chi_R^{\otimes 2}\right]  \ge \sum_{k=1}^N  \frac{k\Tr G^{\chi}_{N,k}}{N} \; \frac{\Eep (k)}{k}.$$
But since 
$$\sum_{k=0}^N \frac{k\Tr G^\chi_{N,k}}{N} = \Tr \left[\chi_R^2 \gamma_N^{(1)}\right]\underset{N\to\ii}{\longrightarrow}\Tr \left[\chi_R^2 \gamma^{(1)}\right]\quad \text{and}\quad \lim_{k\to \infty} \frac{\Eep (k)}{k} =a_\eps,$$
and $k\mapsto \frac{\Eep (k)}{k}$ is increasing we conclude
\begin{equation}
\label{eq:energy-chiR-translation}
\liminf_{N\to \infty} \left( \Tr_\gH\left[(-\Delta-\eps w_-) \chi_R \gamma_N^{(1)}\chi_R\right] + \frac{1}{2} \Tr_{\gH^2}\left[w_\eps \chi_R^{\otimes 2}\gamma_N^{(2)}\chi_R^{\otimes 2}\right] \right)
\geq a_\eps \Tr \left[\chi_R^2 \gamma^{(1)}\right]
\end{equation}
by monotone convergence. 
 
\medskip

\paragraph*{\bf The $\eta_R$-localized term.} Using the results of Section~\ref{sec:loc rig} as above we have
\begin{multline}
\label{eq:localization-translation-eta}
\Tr_\gH\left[ T \eta_R \gamma_N^{(1)}\eta_R\right] + \frac{1}{2} \Tr_{\gH^2}\left[w_\eps \eta_R^{\otimes 2}\gamma_N^{(2)}\eta_R^{\otimes 2}\right]\\
= \frac{1}{N} \sum_{k=1}^N \Tr_{\gH^k}\left[ \left( \sum_{i=1}^k -\Delta_i + \frac{1}{N-1} \sum_{i<j}^k w_\eps (x_i-x_j)\right) {G}^{\eta}_{N,k} \right].
\end{multline}
Here we use
$$-\Delta \ge 0,\quad w_\eps = w+ 2\eps w ^- \ge (1-2\eps)w \mbox{ and } E^0(k)\le ak<0$$ 
to obtain
\bqq
\sum_{i=1}^k -\Delta_i + \frac{1}{N-1} \sum_{i<j}^k w_\eps (x_i-x_j) &\ge& \frac{(1-2\eps)(k-1)}{N-1} H^0_k \hfill\\
&\ge &  \frac{(1-2\eps)(k-1)}{N-1} E^0 (k) \ge E^0(k) - 2\eps a k
\eqq
for all $k\ge 1$. Combining with~\eqref{eq:localization-translation-eta} we obtain
\begin{equation*}
\Tr_\gH\left[ -\Delta \eta_R \gamma_N^{(1)}\eta_R\right] + \frac{1}{2} \Tr_{\gH^2}\left[w_\eps \eta_R^{\otimes 2}\gamma_N^{(2)}\eta_R^{\otimes 2}\right] \ge \sum_{k=1}^N  \frac{k\Tr G^{\eta}_{N,k}}{N} \cdot \left(\frac{E^0(k)}{k} - 2\eps a \right).
\end{equation*}
We finally deduce that
\begin{multline}
\label{eq:energy-etaR-translation}
\liminf_{N\to \infty} \left( \Tr_\gH\left[ T \eta_R \gamma_N^{(1)}\eta_R\right] + \frac{1}{2} \Tr_{\gH^2}\left[w_\eps \eta_R^{\otimes 2}\gamma_N^{(2)}\eta_R^{\otimes 2}\right]\right)\geq (1-2 \eps) a \big(1-\Tr[\chi_R^2 \gamma^{(1)}]\big) 
\end{multline}
upon using
$$\sum_{k=0}^N \frac{k}{N} \Tr G^\eta_{N,k} = \Tr \left[\eta_R^2 \gamma_N^{(1)} \right]\underset{N\to\ii}{\longrightarrow}1-\Tr \left[\chi_R^2 \gamma^{(1)}\right]\quad \text{and}\quad \lim_{k\to \infty} \frac{E^0(k)}{k}  =a$$
as well as the fact that $k\mapsto \frac{E^0(k)}{k}$ is increasing.

\medskip

\paragraph*{\bf Conclusion.}

Inserting~\eqref{eq:energy-chiR-translation} and~\eqref{eq:energy-etaR-translation} in~\eqref{eq:split-energy-translation} we find
\begin{align*}
a_\eps &\ge  \liminf_{R\to \infty}\left( a_\eps \Tr \left[\chi_R^2 \gamma^{(1)}  \right] + (1-2\eps) a \big(1-\Tr \left[\chi_R^2\gamma^{(1)}\right]\big) \right)\\
&=  a_\eps \Tr [\gamma^{(1)}] + (1-2\eps) a \big(1-\Tr [\gamma^{(1)}]\big)
\end{align*}
Since we assumed $a_\eps \le a <0$, we obtain
\bq \label{eq:concentration-case}
\tr[\gamma^{(1)}]=1.
\eq
There is thus no loss of mass for the auxiliary problem defined in Lemma~\ref{lem:triv aux}. This suffices to conclude that the reduced density matrices $\gamma_N ^{(k)}$ converge strongly. Indeed, one may apply the weak de Finetti theorem to the sequence of weak limits $\gamma ^{(k)}$ to obtain a measure $\mu$ living on the unit ball of $\gH$. But~\eqref{eq:concentration-case} combined with~\eqref{eq:melange faible} implies that the measure must in fact live on the unit sphere. We thus have, for all $k \geq 0$, 
$$ \tr[\gamma^{(k)}] = 1,$$
which implies that the convergence of $\gamma_N ^{(k)}$ to $\gamma ^{(k)}$ is actually strong in trace-class norm. We can then go back to~\eqref{eq:split-energy-translation} to obtain 
\begin{multline*}
\liminf_{N\to\ii}\left( \Tr_\gH\left[(-\Delta-\eps w_-)\gamma_N^{(1)}\right] + \frac{1}{2} \Tr_{\gH^2}\left[w_\eps \gamma_N^{(2)}\right] \right)\\
\geq \Tr_\gH\left[(-\Delta-\eps w_-)\gamma^{(1)}\right] + \frac{1}{2} \Tr_{\gH^2}\left[w_\eps \gamma^{(2)}\right].
\end{multline*}
We then apply the strong de Finetti theorem to the limits of the reduced density matrices to conclude that the right side is necessarily larger than $e_{\rm H} ^{\eps}(1)$. This gives~\eqref{eq:aux prob bound} and concludes the proof.
\end{proof}

\subsection{Concluding the proof in the general case}\label{sec:general general}\mbox{}\\\vspace{-0.4cm}

We have almost all the ingredients of the proof of Theorem~\ref{thm:hartree general} at our disposal. We only need a little bit more information on the translation-invariant problem, as we now explain.

For $k\geq 2$ consider the energy 
\begin{equation}\label{eq:b_j}
b_k(\lambda):= \frac{1}{k}\inf \sigma_{\gH^k} \left( \sum_{i=1}^k -\Delta_i + \frac{\lambda}{k-1} \sum_{i<j}^k w (x_i-x_j) \right),
\end{equation}
i.e. an energy for $k$ particles with an interaction strength proportional to $\frac{\lambda}{k-1}$. In the last section we have already shown that the limit $k\to \infty$ of such an energy is given by $\lambda \eH ^0 (\lambda)$ when $\lambda$ is a fixed parameter. Using Fock-space localization methods, the energy of particles lost at infinity in the minimization of the energy for $N\geq k$ particles will be naturally described as a superposition of energies for $k$-particles systems with an interaction of strength $1/(N-1)$ inherited from the original problem. In other words, we will have to evaluate a superposition of the energies $b_k (\lambda)$ with  
$$\lambda = \frac{k-1}{N-1},$$
a bit as in~\eqref{eq:localization-chi-eps} and~\eqref{eq:localization-translation-eta}. Since this $\lambda$ depends on $k$ it will be useful to know that $b_k(\lambda)$ is equicontinuous as a function of $\lambda$:

\begin{lemma}[\textbf{Equi-continuity of the energy as a function of the interaction}]\label{lem:equi conti}\mbox{}\\
We set the convention
\begin{equation*}
b_0 (\lambda) = b_1 (\lambda) = 0.
\end{equation*}
Then, for all $\lambda \in [0,1]$ 
\begin{equation}\label{eq:lim b lambda}
\lim_{k\to\ii} \lambda\, b_k(\lambda) =\eH^0(\lambda).
\end{equation}
Moreover, for all $0\leq \lambda \leq \lambda ' \leq 1$ 
\begin{equation}\label{eq:equi conti b lambda}
0\leq b_k (\lambda) -b_k (\lambda') \leq C |\lambda - \lambda'|
\end{equation}
where $C$ does not depend on $k$.
\end{lemma}

\begin{proof}
We start by vindicating our claim that~\eqref{eq:lim b lambda} is a direct consequence of the analysis of the previous section. For $\lambda = 0$ there is nothing to prove. For $\lambda > 0$ we use Theorem~\ref{thm:trans inv} to obtain  
\begin{align*}
\lim_{k\to\ii} \lambda\, b_k(\lambda)&=\lambda \inf_{\norm{u}^2=1}\left(\pscal{u,Ku}+\frac{\lambda}{2}\iint w(x-y)|u(x)|^2|u(y)|^2\right)\\
&=\inf_{\norm{u}^2=\lambda}\left(\pscal{u,Ku}+\frac{1}{2}\iint w(x-y)|u(x)|^2|u(y)|^2\right)=\eH^0(\lambda).
\end{align*} 
The fact that
$$b_k(\lambda)\ge b_k(\lambda')~~\text{for all}~0\le \lambda<\lambda'\le 1$$
is a consequence of~\eqref{eq:A-tB}. Indeed, one can see that either $b_k (\lambda) = 0$ for all $\lambda$ or the same kind of argument as in~\eqref{eq:A-tB} applies.

For the equicontinuity~\eqref{eq:equi conti b lambda} we fix some $ 0 <\alpha \leq 1$ and remark that, with 
$$\delta := (\lambda'-\lambda)(\alpha^{-1}-\lambda)^{-1},$$
we have
\begin{align*}
\frac{1}{k}\left( \sum_{i=1}^k K_i + \frac{\lambda'}{k-1} \sum_{i<j}^k w_{ij} \right) &= \frac{1-\delta}{k} \left( \sum_{i=1}^k K_i + \frac{\lambda}{k-1} \sum_{i<j}^k w_{ij} \right) \\
&+ \frac{\delta}{k\alpha} \left( \alpha \sum_{i=1}^k K_i + \frac{1}{k-1}\sum_{i<j}^k w_{ij} \right)\\
&\ge (1-\delta) b_k(\lambda) - \frac{C\delta}{\alpha}
\end{align*}
using the fact the spectrum of the operator appearing in the second line is bounded below. We deduce
$$ 0\le b_k(\lambda)- b_k(\lambda') \le \delta (b_k(\lambda)+C\alpha^{-1})\le C|\lambda'-\lambda|$$
since $b_k(\lambda)$ is uniformly bounded and $|\delta| \leq C |\lambda- \lambda '|$.
\end{proof}

We may know conclude the 

\begin{proof}[Proof of Theorem~\ref{thm:hartree general}.] Let $\Psi_N$ be a sequence of  $N$-body wave-functions such that 
$$\langle \Psi_N, H_N^V \Psi_N \rangle = E^V(N)+o(N),$$
with $\gamma_N^{(k)}$ the associated reduced density matrices. After a diagonal extraction we have 
$$\gamma_N^{(k)}\wto \gamma^{(k)}$$ 
weakly-$\ast$. Theorem~\ref{thm:DeFinetti faible} ensures the existence of a probability measure $\mu$ such that
\begin{equation}\label{eq:gen melange}
 \gamma ^{(k)} = \int_{u\in B\gH} d\mu (u) |u ^{\otimes k} \rangle \langle u ^{\otimes k} |.
\end{equation}
We pick a smooth partition of unity $\chi_R^2 +\eta_R^2=1$ as previously and define the localized state $G_N^\chi$ and $G_N^\eta$ constructed from $|\Psi_N \rangle \langle \Psi_N |$. Using Lemma~\ref{lem:loc ener} again we obtain
\begin{align} \label{eq:split-energy-general}
\lim_{N\to \infty}\frac{E^V(N)}{N} &= \lim_{N\to \infty} \left( \Tr_\gH\left[T\gamma_N^{(1)}\right] + \frac{1}{2} \Tr_{\gH^2}\left[w \gamma_N^{(2)}\right] \right) \nn\\
&\ge  \liminf_{R\to \infty}\liminf_{N\to \infty} \left\{ \Tr_\gH\left[T \chi_R \gamma_N^{(1)}\chi_R\right] + \frac{1}{2} \Tr_{\gH^2}\left[w \chi_R^{\otimes 2}\gamma_N^{(2)}\chi_R^{\otimes 2}\right] \right. \nn\\
&\qquad\left. +\Tr_\gH\left[ -\Delta \eta_R \gamma_N^{(1)}\eta_R\right] + \frac{1}{2} \Tr_{\gH^2}\left[w\eta_R^{\otimes 2}\gamma_N^{(2)}\eta_R^{\otimes 2}\right] \right\}.
\end{align}
First we use the strong local compactness and~\eqref{eq:gen melange} for the $\chi_R$-localized term:
\begin{multline}
\liminf_{N\to \infty} \left\{ \Tr_\gH\left[ T \chi_R \gamma_N^{(1)}\chi_R\right] + \frac{1}{2} \Tr_{\gH^2}\left[w \chi_R^{\otimes 2}\gamma_N^{(2)}\chi_R^{\otimes 2}\right]\right\}\\
\geq  \Tr_\gH\left[T \chi_R \gamma^{(1)}\chi_R\right] + \frac{1}{2} \Tr_{\gH^2}\left[w \chi_R^{\otimes 2}\gamma^{(2)}\chi_R^{\otimes 2}\right]= \int_{B\gH} \E^V_{\rm H}[\chi_R u] d\mu(u).
\label{eq:localization-chi-general}
\end{multline}
Our main task is to control the second term in the right side of~\eqref{eq:split-energy-general}. We claim that
\begin{equation}
\label{eq:localization-eta-general}
\liminf_{N\to \infty} \left( \Tr\left[T \eta_R \gamma_N^{(1)}\eta_R\right] + \frac{1}{2} \Tr_{\gH^2}\left[w \eta_R^{\otimes 2}\gamma_N^{(2)}\eta_R^{\otimes 2}\right]\right) \ge \int_{B\gH} e^{0}_{\rm H}(1-\norm{\chi_R u}^2) d\mu(u). 
\end{equation}
Indeed, using the $\eta_R$-localized state $G_N^\eta$  we have  
\begin{align*}
\Tr_\gH\left[ -\Delta \eta_R \gamma_N^{(1)}\eta_R\right] &+ \frac{1}{2} \Tr_{\gH^2}\left[w \eta_R^{\otimes 2}\gamma_N^{(2)}\eta_R^{\otimes 2}\right]
=\frac{1}{N} \sum_{k=1}^N \Tr_{\gH^k}\left[ \left( \sum_{i=1}^k -\Delta_i + \frac{1}{N-1} \sum_{i<j}^k w_{ij}\right) {G}^{\eta}_{N,k} \right] \\
&\ge\sum_{k=1}^N \frac{1}{N} \Tr {G}^{\eta}_{N,k} \inf \sigma_{\gH^k} \left( \sum_{i=1}^k -\Delta_i + \frac{1}{N-1} \sum_{i<j}^k w_{ij}\right) \\
&\geq \sum_{k=0}^N \tr G_{N,k}^\eta \frac{k}{N} b_k \left( \frac{k-1}{N-1}\right)
\end{align*}
where $b_k$ is the function defined in~\eqref{eq:b_j}. On the other hand
\begin{equation} \label{eq:localization-eta-bk}
\lim_{N\to\ii}\sum_{k=0}^N \tr G_{N,k}^\eta \left( \frac{k}{N} b_k \left( \frac{k-1}{N-1}\right) - e_{\rm H}^0 \left(\frac{k}{N}\right) \right) = 0,
\end{equation}
because, using the equicontinuity of $\{b_k\}_{k=1}^\infty$ and the convergence $\lim_{k\to \infty} \lambda b_k(\lambda)=e_{\rm H}^{0}(\lambda)$ given by Lemma~\ref{lem:equi conti}, we get 
$$ \lim_{N \to \infty} \sup_{k=1,2,...,N} \left|\frac{k}{N} b_k\left(\frac{k-1}{N-1}\right)-e_{\rm H}^{0}\left(\frac{k}{N}\right)\right| =0.$$
We just have to combine this with 
$$\sum_{k=0}^N \tr G_{N,k}^\eta = 1$$
to obtain~\eqref{eq:localization-eta-bk}. At this stage we thus have 
\[
\liminf_{N\to \infty} \Tr_\gH\left[ -\Delta \eta_R \gamma_N^{(1)}\eta_R\right] + \frac{1}{2} \Tr_{\gH^2}\left[w \eta_R^{\otimes 2}\gamma_N^{(2)}\eta_R^{\otimes 2}\right] \geq \lim_{N\to\ii}\sum_{k=0}^N \tr G_{N,k}^\eta e_{\rm H}^0 \left(\frac{k}{N}\right) .
\]
We now use the fundamental relation~\eqref{eq:Fock funda rel} and Corollary~\ref{cor:wdeF loc} as indicated in Section~\ref{sec:proof deF faible} to deduce 
\begin{multline*}
\lim_{N\to\ii}\sum_{k=0}^N \tr G_{N,k}^\eta\; e_{\rm H}^0 \left(\frac{k}{N}\right) =\lim_{N\to\ii}\sum_{k=0}^N \tr G_{N,N-k}^\chi \;e_{\rm H}^0 \left(\frac{k}{N}\right)\\
=\lim_{N\to\ii}\sum_{k=0}^N \tr G_{N,k}^\chi \;e_{\rm H}^0 \left(1-\frac{k}{N}\right)=\int_{B\gH}\;e_{\rm H}^0 (1-\|\chi_R u\|^2)d\mu(u),
\end{multline*}
which concludes the proof of~\eqref{eq:localization-eta-general}.

There remains to insert~\eqref{eq:localization-chi-general} and~\eqref{eq:localization-eta-general} in~\eqref{eq:split-energy-general} and use Fatou's lemma. This gives
\begin{align}\label{eq:conclu preuve general}
e_{\rm H} (1) \geq \lim_{N\to \infty}\frac{E^V(N)}{N} &\ge \liminf_{R\to \infty} \left( \int_{B\gH} \left[\E^V_{\rm H}[\chi_R u]+ e^0_{\rm H}(1-\norm{\chi_R u}^2) \right] d\mu(u)\right) \nn\\
&\ge \int_{B\gH} \liminf_{R\to \infty} \left[\E^V_{\rm H}[\chi_R u]+ e^0_{\rm H}(1-\norm{\chi_R u}^2)\right] d\mu(u) \nn\\
&= \int_{B\gH} \left[\E^V_{\rm H}[u]+ e^0_{\rm H}(1-\norm{u}^2) \right] d\mu(u) \nn\\
&\geq \int_{B\gH} \left[\eH^V(\norm{u}^2)+ e^0_{\rm H}(1-\norm{u}^2) \right] d\mu(u)\ge e_{\rm H}(1), 
\end{align}
using the continuity of $\lambda \mapsto e_{\rm H}^0(\lambda)$ and the binding inequality~\eqref{eq:hartree liaison}. This concludes the proof of~\eqref{eq:general ener}, and the other results of the theorem follow by inspecting the cases of equality in~\eqref{eq:conclu preuve general}.
\end{proof}

\newpage

\section{\textbf{Derivation of Gross-Pitaevskii functionals}}\label{sec:NLS}

We now turn to the derivation of non-linear Schr\"odinger (NLS) functionals with local non-linearities: 
$$\ENLS [\psi] := \int_{\R ^d } \left| \nabla \psi \right| ^2 + V  |\psi | ^2 + \frac{a}{2} |\psi| ^4.$$
One may obtain such objects starting from a Hartree functional such as ~\eqref{eq:Hartree func} and taking an interaction potential converging  (in the sense of distributions) to a a Dirac mass
$$ w_L (x) = L ^{-d} w \left( \frac{x}{L}\right) \wto \left(\int_{\R ^d} w \right) \delta_0 \mbox{ when } L\to 0.$$ 
Since we already obtained~\eqref{eq:Hartree func} as the limit of a $N$-body problem, one might be tempted to see the derivation of the NLS functional as a simple passage to the limit in a one-body problem. The issue with such an approach is the total lack of control on the relationship between the physical parameters $N$ and $L$. One can see this as a problem of non-commuting limits: it is not clear at all that one can interchange the order of the limits  $N\to \infty$ and\footnote{The $L\to 0$ limit of the $N$-body problem is by the way very hard to define properly.} $L\to 0$. 

From a physical point of view, the good question is ``What relation should $N$ and $L$ satisfy if one is to obtain the NLS energy by taking \emph{simultaneously} the limits $N\to \infty$ and $L\to 0$ of the $N$-body problem ?'' The process leading to the NLS theory is thus more subtle than that leading to the Hartree theory. We shall first elaborate a little bit more on this point in the following subsection.

\subsection{Preliminary remarks}\label{sec:GP rem}\mbox{}\\\vspace{-0.4cm}

Until now we have worked with only two physical parameters: the particle number $N$ and the interactions strength $\lambda$. In order to have a well-defined limit problem we have been lead to considering the mean-field limit where $\lambda \propto N ^{-1}$. In this case, the range of the interaction potential is fixed and each particle interacts with all the others. The interaction strength felt by a typical particle is thus of order $\lambda N \propto 1$, comparable with its self-energy (kinetic + potential). We saw that this equilibration of forces acting on each particle, combined with structure results \`a la de Finetti, naturally leads to the conclusion that particles approximately behave as if they were independent. It follows that Hartree-type descriptions are valid in the limit $N\to \infty$.

There are other ways to justify such models: the equilibration of forces which allows the limit problem to emerge may result from a more subtle mechanism. For example, in the ultra-cold alkali gases where BEC has been observed, it has more to do with the \emph{diluteness} of the system than with the weakness of the interaction strength. For a theoretical description of this situation, we can introduce in our model a length scale $L$ characterizing the range of the interactions. Taking the total system size as a reference, a dilute system is materialized by taking the limit $L\to 0$. The interaction strength for a typical particle is then of order $ \lambda N L ^{d}$ (interaction strength $\times$ number of particles in a ball of radius $L$ around a given particle). This parameter is the one we should fix when $N\to \infty$ to obtain a limit problem. Different regimes are then possible, depending on the ratio between $\lambda$ and $L$.  

One may discuss the different possibilities starting from the $N$-body Hamiltonian\footnote{Once again, it is possible to add fractional Laplacians and/or magnetic fields, cf Remark~\ref{rem:energie cin}. For simplicity we do not pursue this here.} 
\begin{equation}\label{eq:GP start hamil}
H_N = \sum_{j=1} ^N  - \Delta_j  + V(x_j)  + \frac{1}{N-1} \sum_{1\leq i<j \leq N}N ^{d\beta} w( N ^{\beta} (x_i-x_j))
\end{equation}
which corresponds to choosing 
$$L=N ^{-\beta},\quad  \lambda \propto N ^{d\beta - 1}.$$ 
The fixed parameter $0\leq \beta$ is used to set the ratio between $L$ and $\lambda$. We consider the reference interaction potential $w$ as fixed, and we shall denote 
\begin{equation}\label{eq:w_N}
 w_N (x) := N ^{d\beta} w( N ^{\beta} x). 
\end{equation}
For $\beta >0$, $w_N$ converges in the sense of distributions to a Dirac mass
\begin{equation}\label{eq:GP conv delta}
 w_N \to \left(\int_{\R ^d} w \right) \delta_0,  
\end{equation}
materializing the short range of the interactions/diluteness of the system. Reasoning formally, one may want to directly replace $w_N$ by $\left(\int_{\R ^d} w \right) \delta_0$ in~\eqref{eq:GP start hamil}. In this case we are back to a mean-field limit, with a Dirac mass as interaction potential. This is of course purely formal (except in 1D where the Sobolev injection $H ^1 \hookrightarrow C^0$ allows to properly define the contact interaction). Accepting that this manipulation has a meaning and that one may approximate the ground state of~\eqref{eq:GP start hamil} under the form 
\begin{equation}\label{eq:GP ansatz}
\Psi_N = \psi ^{\otimes N} 
\end{equation}
we obtain the Hartree functional
\begin{equation}\label{eq:GP intro Hartree}
\EH [\psi] := \int_{\R ^d } \left| \nabla \psi \right| ^2 + V  |\psi | ^2 + \frac{1}{2} |\psi| ^2 (w \ast |\psi| ^2 ).
\end{equation}
Replacing the interaction potential by a Dirac mass at the origin leads to the Gross-Pitaevskii functional
\begin{equation}\label{eq:GP intro nls}
\ENLS [\psi] := \int_{\R ^d } \left| \nabla \psi \right| ^2 + V  |\psi | ^2 + \frac{a}{2} |\psi| ^4.
\end{equation}
In view of~\eqref{eq:GP conv delta}, the logical choice seems to imagine that when $\beta >0$, we obtain, starting from~\eqref{eq:GP start hamil}, the above functional with
\begin{equation}\label{eq:GP defi a}
a = \int_{\R ^d} w. 
\end{equation}
In fact one may prove the following results (described in the case $d=3$; $d\leq 2$ leads to subtleties and $d\geq 4$ does not have much physical meaning):
\begin{itemize}
\item If $\beta = 0$ we have the previously studied mean-field (MF) regime. The range of the interaction potential is fixed and its strength decreases proportionally to $N ^{-1}$. The limit problem is then~\eqref{eq:GP intro Hartree}, as previously shown. 
\item If $ 0 < \beta < 1$, the limit problem is~\eqref{eq:GP intro nls} with the parameter choice~\eqref{eq:GP defi a}, as can be expected. We will call this case the non-linear Schr\"odinger (NLS) limit. This case does not seem to have been considered in the literature prior to~\cite{LewNamRou-14} but, at least when $w\geq 0$ and $V$ is confining, one may adapt the analysis of the more difficult case $\beta = 1$. 
\item If $\beta = 1$, the limit functional is now~\eqref{eq:GP intro nls} with 
$$a = 4\pi \times \mbox{ scattering length of } w$$
(see~\cite[Appendix C]{LieSeiSolYng-05} for a definition). In this case, the ground state of~\eqref{eq:GP start hamil} includes a correction to the ansatz~\eqref{eq:GP ansatz}, in the form of short-range correlations. In fact, one should expect to have 
\begin{equation}\label{eq:GP ansatz 2}
\Psi_N (x_1,\ldots,x_N) \approx \prod_{j=1} ^N \psi (x_j) \prod_{1\leq i<j\leq N} f \left( N ^{\beta} (x_i - x_j)\right)  
\end{equation}
where $f$ is linked to the two-body problem defined by $w$ (zero-energy scattering solution). It so happens that when $\beta = 1$, the correction has a leading order effect on the energy, that of replacing $\int_{\R ^d}w$ by the corresponding scattering length, as noted first in~\cite{Dyson-57}. We will call this case the Gross-Pitaevskii (GP) limit. It is studied in a long and remarkable series of papers by Lieb, Seiringer and Yngvason  (see for example~\cite{LieYng-98,LieYng-01,LieSeiYng-01,LieSeiYng-00,LieSeiYng-05,LieSei-06,LieSeiSolYng-05}).   
\end{itemize}
As already mentioned, the corresponding evolution problems have also been thoroughly studied. Here too one has to distinguish the MF limit~\cite{BarGolMau-00,AmmNie-08,FroKnoSch-09,RodSch-09,Pickl-11}, the NLS limit~\cite{ErdSchYau-07,Pickl-10} and the GP limit~\cite{ErdSchYau-09,Pickl-10b,BenOliSch-12}. 

There is a fundamental physical difference between the MF and GP regimes: In both cases the effective interaction parameter $\lambda N L ^3$ is of order $1$, but one goes (still in 3D) from a case with numerous weak collistions when $\lambda = N ^{-1}$, $L=1$ to a case with rare but strong collisions when $\lambda = N ^2$, $L = N ^{-1}$. The different NLS scaling sort of interpolate between these two extremes, the transition from ``frequent weak collistions'' to ``rare strong collisions'' happens at $\lambda = 1, L = N ^{-1/3}$, i.e. $\beta = 1/3$.

\medskip

The difficulty for obtaining~\eqref{eq:GP intro nls} by taking the limit $N\to \infty$ is that there are in fact two distinct limits $N\to \infty$ and $w_N \to a \delta_0$ to control at the same time. A simple compactness argument will not suffice in this case and one has to work with quantitative estimates. The goal of this chapter is to explain how one may proceed starting from the quantitative de Finetti theorem stated in Chapter~\ref{sec:deF finite dim}. Since this theorem is only valid in finite dimension, we will need a natural way to project the problem onto finite dimensional spaces. We thus deal with the case of trapped bosons, assuming that for some constants~$c,C~>~0$   
\begin{equation}\label{eq:GP asum V}
c |x| ^s - C\leq V(x).
\end{equation}
In this case, the one-body Hamiltonian $-\Delta + V$ has a discrete spectrum on which we have a good control thanks to Lieb-Thirring-like inequalities. 

The results we are going to obtain are valid for $0 < \beta < \beta_0$ where $\beta_0 = \beta_0 (d,s)$ depends only on the dimension of the configuration space and the potential $V$. We shall give explicit estimates of $\beta_0$, but the method, introduced in~\cite{LewNamRou-14}, is for the moment limited to relatively small $\beta$. In particular, we will always obtain~\eqref{eq:GP defi a} as interaction parameter. The main advantage compared to the Lieb-Seiringer-Yngvason method~\cite{LieYng-98,LieSeiYng-00,LieSeiYng-05,LieSei-06} is that we can in some cases avoid the assumption $w\geq 0$, always made in these papers (see also~\cite{LieSeiSolYng-05}). In particular we present the first derivation of attractive NLS functionals\footnote{One often uses the non-linear optics vocabulary to distinguish the repulsive and attractive cases : repulsive = defocusing, attractive = focusing.} in 1D and 2D.

Recently, a blend of the following arguments, the tools of~\cite{LieYng-98,LieSeiYng-00,LieSeiYng-05,LieSei-06} and new a priori estimates for the many-body ground state has allowed to extend the analysis to the GP regime~\cite{NamRouSei-15}. One can also adapt the method to deal with particles with statistics slightly deviating from the bosonic one (``almost bosonic anyons''), cf~\cite{LunRou-15}.

\subsection{Statements and discussion}\label{sec:GP statement}\mbox{}\\\vspace{-0.4cm}

We take comfortable assumptions on $w$: 
\begin{equation}\label{eq:GP asum w}
w\in L^\infty(\R^d,\R) \mbox{ and } w(x) = w(-x) .
\end{equation}
Without loss of generality we assume 
$$ \sup_{\R ^d} |w| = 1$$
to simplify some expressions. We also assume that 
\begin{equation}\label{eq:replace delta}
x \mapsto (1+|x|) w (x) \in L ^1 (\R ^d), 
\end{equation}
which simplifies the replacement of $w_N$ by a Dirac mass. As usual we use the same notation $w_N$ for the interaction potential~\eqref{eq:w_N} and the multiplication operator by $w_N (x-y)$ acting on $L ^2 (\R ^{2d})$.

\medskip

For $\beta = 0$ we have shown previously that
\begin{equation}\label{eq:GP Hartree lim}
\lim_{N\to \infty} \frac{E(N)}{N} = \eH. 
\end{equation}
We now deal with the case $0 < \beta < \beta_0 (d,s) < 1$ where we obtain the ground state energy of~\eqref{eq:GP intro nls}:
\begin{equation}\label{eq:GP eNLS}
\eNLS := \inf_{\norm{\psi}_{L ^2 (\R ^d)} = 1} \ENLS [\psi]
\end{equation}
with $a$ defined in~\eqref{eq:GP defi a}. Because of the local non-linearity, the NLS theory is more delicate than Hartree's theory. We shall need some structure assumptions on the interactions potential (see~\cite{LewNamRou-14} for a more thorough discussion):  
\begin{itemize}
\item When $d=3$, it is well-known that a ground state for~\eqref{eq:GP intro nls} exists if and only if $a \geq 0$. This is because the cubic non-linearity is super-critical\footnote{One may for example consult~\cite{KilVis-08} for a classification of non-linearities in the NLS equation.} in this case. Moreover, it is easy to see that $N ^{-1} E(N) \to -\infty$ if $w$ is negative at the origin. The optimal assumption happens to be a classical stability notion for the interaction potential:
\begin{equation}\label{eq:GP hyp 3}
\iint_{\R ^d \times \R ^d} \rho (x) w(x-y) \rho(y) dx dy \geq 0, \mbox{ for all } \rho \in L^1 (\R ^d), \: \rho \geq 0. 
\end{equation}
This is satisfied as soon as $w= w_1 + w_2$, $w_1 \geq 0$ and $\hat{w_2}\geq 0$ with $\hat{w_2}$ the Fourier transform of $w_2$. This assumption clearly implies $\int_{\R ^d} w \geq 0$, and one may easily see by changing scales that if it is violated for a certain $\rho \geq 0$, then $E(N)/N \to -\infty$.
\item When $d=2$, the cubic non-linearity is critical. A minimizer for~\eqref{eq:GP intro nls} exists if and only if $a > - a ^*$ with 
\begin{equation}\label{eq:GP a star}
a ^* = \norm{Q}_{L ^2} ^2  
\end{equation}
where $Q\in H ^1 (\R^2)$ is the unique~\cite{Kwong-89} (modulo translations) solution to 
\begin{equation}\label{eq:GP defi Q}
-\Delta Q + Q - Q ^3 = 0.  
\end{equation}
The critical interaction parameter $a ^*$ is the best possible constant in the interpolation inequality
\begin{equation}\label{eq:GP interpolation}
\int_{\R ^2} |u| ^4 \leq C \left(\int_{\R^2} |\nabla u | ^2\right) \left(\int_{\R ^2} |u | ^2\right). 
\end{equation}
We refer to~\cite{GuoSei-13,Maeda-10} for the existence of a ground state and to~\cite{Weinstein-83} for the inequality~\eqref{eq:GP interpolation}. A pedagogical discussion of this kind of subjects may be found in~\cite{Frank-14}.

In view of the above conditions, it is clear that in 2D we have to assume $\int w \geq - a ^*$, but this is in fact not sufficient: as in 3D, if the interaction potential is sufficiently negative at the origin, one may see that $N ^{-1} E(N) \to -\infty$. The appropriate assumption is now
\begin{equation}\label{eq:GP hyp 2}
\norm{u}_{L^2}^2\norm{\nabla u}_{L^2}^2+\frac12\iint_{\R^2\times \R^2}|u(x)|^2|u(y)|^2 w(x-y)\,dx\,dy > 0
\end{equation}
for all $u\in H^1(\R^2)$. Replacing $u$ by $\lambda u(\lambda x)$ and taking the limit $\lambda\to 0$ we obtain
$$\norm{u}_{L^2}^2\norm{\nabla u}_{L^2}^2+\frac12\left(\int_{\R^2}w\right)\int_{\R^2}|u(x)|^4\,dx\geq0,\qquad\forall u\in H^1(\R^2),$$
which implies that 
$$\int_{\R^2}w(x)\,dx\geq-a^*.$$
A scaling argument shows that if the strict inequality in~\eqref{eq:GP hyp 2} is reversed for a certain $u$, then $E(N)/N \to - \infty$. The case where equality may occur in~\eqref{eq:GP hyp 2} is left aside in these notes. It requires a more thorough analysis, see~\cite{GuoSei-13} where this is provided at the level of the NLS functional.
\item When $d=1$, the cubic non-linearity is sub-critical and there is always a minimizer for the functional~\eqref{eq:GP intro nls}. In this case we need no further assumptions. 
\end{itemize}

We may now state the 

\begin{theorem}[\textbf{Derivation of NLS ground states}]\label{thm:deriv nls}\mbox{}\\
Assume that either $d=1$, or $d=2$ and~\eqref{eq:GP hyp 2} holds, or $d= 3$ and~\eqref{eq:GP hyp 3} holds. Further suppose that 
\begin{equation}\label{eq:GP beta 0}
0 < \beta \leq \beta_0 (d,s) :=\frac{s}{2ds + s d^2 + 2d ^2} < 1.
\end{equation}
where $s$ is the exponent appearing in~\eqref{eq:GP asum V}. We then have  
\begin{enumerate}
\item \underline{Convergence of the energy:}
\begin{equation}\label{eq:GP ener convergence}
\frac{E(N)}{N} \to \eNLS. 
\end{equation}
\item \underline{Convergence of states:} Let $\Psi_N$ be a ground state of~\eqref{eq:GP start hamil} and 
$$ \gamma_N ^{(n)}:= \tr_{n+1\to N} \left[ |\Psi_N\rangle \langle \Psi_N | \right] $$
its reduced density matrices. Along a subsequence we have, for all $n\in \N$,  
\begin{equation}\label{eq:GP state convergence}
\lim_{N\to \infty}\gamma_N ^{(n)} = \int_{u\in \MNLS} d\mu (u) |u ^{\otimes n} \rangle \langle u ^{\otimes n}| 
\end{equation}
strongly in  $\gS ^1 (L ^2 (\R ^{dn}))$. Here $\mu$ is a probability measure supported on
\begin{equation}\label{eq:GP nls set}
\MNLS = \left\{ u \in L ^2 (\R ^d), \norm{u}_{L ^2} = 1, \ENLS [u] = \eNLS \right\}. 
\end{equation}
In particular, when~\eqref{eq:GP intro nls} has a unique minimizer $\uNLS$ (up to a constant phase), we have for the whole sequence
\begin{equation}\label{eq:GP BEC}
\lim_{N\to \infty} \gamma_N ^{(n)} =|\uNLS ^{\otimes n} \rangle \langle \uNLS ^{\otimes n}|.
\end{equation}
\end{enumerate}
\end{theorem}

Uniqueness of $\uNLS$ is ensured if $a \geq 0$ or $|a|$ is small. If these conditions are not satisfied, one can show that there are several minimizers in certain trapping potentials having degenerate minima~\cite{AshFroGraSchTro-02,GuoSei-13}. 

\begin{remark}[On the derivation of NLS functionals.]\label{rem:GP NLS}\mbox{}\\
\vspace{-0.4cm}
\begin{enumerate}
\item The assumption $\beta < \beta_0 (d,s)$ is dictated by the method of proof but it is certainly not optimal,  see~\cite{LieSei-06,LieSeiYng-00,LieYng-98,LieYng-01,NamRouSei-15}. One may relax a bit the condition on $\beta$, at the price of heavier computations, something we prefer to avoid in these notes, see~\cite{LewNamRou-14}. In 1D, one may obtain the result for any $\beta >0$. 
\item It is in fact likely that one could replace $\beta_0(d,s)$ by $\beta(d,\infty)$ for any $s>0 $, using localization methods from Chapter~\ref{sec:locFock}. Indeed, ground states always essentially live in a bounded region. We refrain from pursuing this in order not to add another layer of localization (in space) to the proof: we will already rely a lot on finite dimensional localization. In any case, $\beta < \beta(d,\infty)$ is not optimal either.     
\item Let us inspect in more details the conditions on $\beta_0 (d,s)$ we obtain. For the case of a quadratic trapping potential $V(x) = |x| ^2$ for example, we can afford $\beta < 1/24$ in 3D, $\beta < 1/12$ in 2D and $\beta < 1/4$ in 1D. The method adapts with no difficulty to the case of particles in a bounded domain which corresponds to setting formally $s=\infty$. We then otain $\beta_0(d,s) = 1/15$ in 3D, $1/8$ in 2D and $1/3$ in 1D. Improving these thresholds in the case of (partially) attractive potentials remains an open problem.
\item When $\beta$ is smaller than the critical $\beta_0 (d,s)$ by a given amount, the method gives quantitative estimates for the convergence~\eqref{eq:GP ener convergence}, see below. We refer to~\cite[Remark 4.2]{LewNamRou-14} for a discussion of the cases where a convergence rate for the minimizer can be deduced, based on tools from~\cite{CarFraLie-14,Frank-14} and assumptions on the behavior of the NLS functional.\hfill\qed 
\end{enumerate}
\end{remark}

The proof of this result occupies the rest of the chapter. We proceed in two steps. The bulk of the analysis consists in obtaining a quantitative estimate of the discrepancy between the $N$-body energy per particle $N ^{-1} E(N)$ and the Hartree energy 
\begin{equation}\label{eq:GP hartree e}
\eH := \inf_{\norm{u}_{L ^2 (\R ^d)} = 1} \EH [u] 
\end{equation}
given by minimizing the functional
\begin{equation}\label{eq:GP hartree f}
\EH [u] := \int_{\R ^d} \left(\left|\nabla u \right| ^2 + V |u| ^2 \right) dx + \frac{1}{2} \iint_{\R ^d \times \R ^d} |u(x)| ^2 w_N (x-y) |u(y)| ^2 dx dy.  
\end{equation}
The latter objects still depend on $N$ when $\beta >0$, whence the necessity to avoid compactness arguments and obtain precise estimates. Once the link between $N ^{-1} E(N)$ and~\eqref{eq:GP hartree e} is established, there remains to estimate the difference $|\eNLS - \eH|$, which is a much easier task. Most of the restrictive assumptions we have made on $w$ are used only in this second step. The estimates on the difference $|\eH-N ^{-1} E(N)|$ are valid without assuming~\eqref{eq:replace delta} and~\eqref{eq:GP hyp 3} or~\eqref{eq:GP hyp 2}. They thus give some information on the divergence of $N ^{-1} E(N)$ in the case where $\eH$ does not converge to~$\eNLS$:

\begin{theorem}[\textbf{Quantitative derivation of Hartree's theory}]\label{thm:error Hartree}\mbox{}\\
Assume~\eqref{eq:GP asum V} and~\eqref{eq:GP asum w}. Let 
\begin{equation}\label{eq:GP rate intro}
t :=  \frac{1+2d\beta}{2 + d/2 + d/s}.
\end{equation}
If
\begin{equation}\label{eq:GP asum t}
t > 2d \beta ,
\end{equation}
then, for all $d\geq 1$ there exists a constant $C_{d}>0$ such that
\begin{equation}\label{eq:GP energy estimate}
\eH \geq \frac{E(N)}{N} \geq \eH - C_{d} N ^{-t + 2 d\beta}.  
\end{equation} 
\end{theorem}

\begin{remark}[Explicit estimates in the mean-field limit]\label{rem:GP Hartree quant}\mbox{}\\
\vspace{-0.4cm}
\begin{enumerate}
\item Condition~\eqref{eq:GP asum t} is satisfied if $0 \leq \beta < \beta_0 (d,s)$. For the proof of Theorem~\ref{thm:deriv nls} we are only interested in cases where $|\eH|$ is bounded independently of $N$, and~\eqref{eq:GP energy estimate} then gives non-trivial information only if~\eqref{eq:GP asum t} holds.
\item The result is valid in the mean-field case where $\beta = 0$ and thus $\eH$ does not depend on $N$. We then obtain explicit estimates improving on Theorem~\ref{thm:confined}. These present a novelty in the case where Hartee's functional has several minimizers, or a unique degenerate minimizer. In other cases\footnote{The simplest example ensuring uniqueness and non-degeneracy is that where $\hat{w}> 0$ with  $\hat{w}$ the Fourier transform of $w$.}, better estimates are known, with an error of order $N ^{-1}$ given by Bogoliubov's theory~\cite{LewNamSerSol-13,Seiringer-11,GreSei-12,DerNap-13}. See~\cite{NamSei-14} for extensions of Bogoliubov's theory to cases of mutiple and/or degenerate minimizers.
\end{enumerate}
 \hfill\qed
\end{remark}

The proof of Theorem~\ref{thm:error Hartree} occupies Section~\ref{sec:GP Hartree bounds}. We then complete the proof of Theorem~\ref{thm:deriv nls} in Section~\ref{sec:GP Hartee to NLS}.

\subsection{Quantitative estimates for Hartree's theory}\label{sec:GP Hartree bounds}\mbox{}\\\vspace{-0.4cm}

The main idea of the proof is to apply Theorem~\ref{thm:DeFinetti quant} on a low-energy eigenspace of the one-body operator
$$T = -\Delta + V $$
acting on $\gH = L ^2 (\R ^d)$. Assumption~\eqref{eq:GP asum V} ensures that the resolvent of this operator is compact and thus that its specturm is made of a sequence of eigenvalues tending to infinity. We denote $P_-$ and $P_+$ the spectral projectors corresponding to energies respectively below and above a given truncation $\Lambda$:
\begin{equation}\label{eq:GP projectors}
P_- = \1_{(-\infty, \, \Lambda)} \left( T \right), \quad P_+ = \1_{\gH} - P_-=P_- ^{\perp}. 
\end{equation}
We also denote
\begin{equation}\label{eq:NT}
N_\Lambda := \dim (P_- \gH ) = \mbox{ number of eigenvalues of } T \mbox{ smaller than } \Lambda.
\end{equation}
Since the precision of the quantitative de Finetti theorem depends on the dimension of the space on which it applies, it is clearly necessary to have at our disposal a convenient control of~$N_\Lambda$. The tools to achieve this are well-known under the name of Lieb-Thirring inequalities, or rather Cwikel-Lieb-Rosenblum in this case. We shall use the following lemma:

\begin{lemma}[\textbf{Number of bound states of a Schr\"odinger operator}] \label{lem:GP nb bound states}\mbox{}\\
Let $V$ satisfy~\eqref{eq:GP asum V}. For all $d\geq 1$, there exists a constant $C_{d}>0$ such that, for all $\Lambda$ large enough
\begin{equation}\label{eq:GP bound NT}
N_\Lambda \leq C_{d} \Lambda ^{d/s + d/2}. 
\end{equation}
\end{lemma}

\begin{proof}
When $d\geq 3$, this is an application of~\cite[Theorem 4.1]{LieSei-09}. For $d\leq 2$, the result follows easily by applying~\cite[Theorem 2.1]{ComSchSei-78} or~\cite[Theorem 15.8]{Simon-05}, see~\cite{LewNamRou-14} for details. The familiarized reader can convince herself that the right side of~\eqref{eq:GP bound NT} is proportional to the expected number of energy levels in the semi-classical approximation. We refer to ~\cite[Chapter 4]{LieSei-09} for a more thorough discussion of this kind of inequalities.
\end{proof}

In the sequel we shall argue as follows:
\begin{enumerate}
\item The eigenvectors of $T$ form an orthogonal basis of $L ^2 (\R ^d)$ on which the $N$ particles should be distributed. The methods of Chapter~\ref{sec:locFock} provide the right tools to analyze the repartition of the particles between $P_- \gH$ and $P_+ \gH$.
\item If the truncation $\Lambda$ is chosen large enough, particles living on excited energy levels will have a much larger energy per unit mass than the Hartree energy we are aiming at. There can thus only be a small number of particles living on excited energy levels.
\item Particles living on $P_- \gH$ form a state on $\cF_s  ^{\leq N}(P_- \gH)$ (truncated bosonic Fock space built on $P_- \gH$). Since $P_- \gH$ has finite dimension, one may use Theorem~\ref{thm:DeFinetti quant} to describe these particles. These will give the Hartree energy, up to an error depending on $\Lambda$ and the expected number of  $P_-$-localized particles. More precisely, in view of~\eqref{eq:error finite dim deF}, we should expect an error of the form
\begin{equation}\label{eq:GP error heurist}
 \frac{(\Lambda + N ^{d\beta})\times N_\Lambda}{N_-}, 
\end{equation}
i.e. dimension of the localized space $\times$ operator norm of the projected Hamiltonian $/$ number of localized particles.
\item We next have to optimize over $\Lambda$, keeping the following heuristic in mind: if $\Lambda$ is large, there will be many $P_-$-localized particles, which favors the denominator of~\eqref{eq:GP error heurist}. On the other hand, taking $\Lambda$ small favors the numerator of~\eqref{eq:GP error heurist}. Picking an optimal value to balance these two effects leads to the error terms of Theorem~\ref{thm:error Hartree}.
\end{enumerate}

\begin{proof}[Proof of Theorem~\ref{thm:error Hartree}.] The upper bound in~\eqref{eq:GP energy estimate} is as usual proven by taking a trial state of the form $u ^{\otimes N}$. Only the lower bound is non-trivial. We proceed in several steps.

\medskip

\noindent\textbf{Step 1, truncated Hamiltonian.} We first have to convince ourselves that it is legitimate to think only in terms of $P_+$ and $P_-$-localized particles, as we did above. This is the oject of the following lemma, which bounds from below the two-body Hamiltonian 
\begin{equation}\label{eq:GP two body hamil}
H_2 = T\otimes \one + \one \otimes T + w_N 
\end{equation}
in terms of its $P_-$-localization, $P_- \otimes P_- \, H_2 \, P_-\otimes P_-$ and a crude bound on the energy of the $P_+$-localized particles.

\begin{lemma}[\textbf{Truncated Hamiltonian}]\label{lem:GP localize-energy}\mbox{}\\
Assume that $\Lambda \ge C N^{d\beta}$ for a large enough constant $C>0$. Then   
\begin{align} \label{eq:GP H2-localized-error}
H_2 \ge  P_-^{\otimes 2} H_2 P_-^{\otimes 2} + \frac{\Lambda}{4} P_+^{\otimes 2} - \frac{2 N^{2d\beta}}{\Lambda}
\end{align}
\end{lemma}

\begin{proof}
We denote 
\begin{equation}\label{eq:GP two body hamil non int}
H_2 ^0 =  T \otimes \1 + \1 \otimes T 
\end{equation}
the two-body Hamiltonian with no interactions. We may then write
\begin{align}\label{eq:GP loc 1 body}
H_2 ^0 = \left( P_- + P_+ \right) ^{\otimes 2} H_2 ^0 \left( P_- + P_+ \right) ^{\otimes 2} &= \sum_{i,j,k,\ell \in \{ -,+\}} P_i \otimes P_j \, H_2^{0} \, P_k \otimes P_\ell \nonumber \\
&= \sum_{i,j \in \{ -,+\}} P_i \otimes P_j \, H_2^{0} \, P_i \otimes P_j. 
\end{align}
Indeed,
$$P_i \otimes P_j \, H_2 ^0 \, P_k \otimes P_\ell =0$$
if $i\ne k$ or $j \neq \ell$, since $T$ commutes with $P_{\pm}$ and $P_- P_+ =0$. We then note that 
$$P_+ T P_+\ge \Lambda P_+ \mbox{ and } P_- T P_- \ge -C P_-,$$ 
which gives
\begin{align*}
P_+ \otimes P_+ \, H_2 ^0\, P_+ \otimes P_+ &\ge 2 \Lambda \, P_+ \otimes P_+ \\
P_+ \otimes P_- \, H_2 ^0\, P_+ \otimes P_- &\ge (\Lambda - C) P_+ \otimes P_-\\
P_- \otimes P_+ \,H_2 ^0\, P_- \otimes P_+ &\ge (\Lambda - C) P_+ \otimes P_-.
\end{align*}
We conclude
\begin{equation}\label{eq:GP loc 1 body lower bound}
H_2 ^0  \ge P_- ^{\otimes 2} \, H_2 ^0 \, P_- ^{\otimes 2} + (\Lambda -C) \Pi
\end{equation}
where 
$$ \Pi =  P_+ \otimes P_+ + P_- \otimes P_+ + P_+ \otimes P_-  = \1_{\gH ^2} - P_- ^{\otimes 2}.$$

We turn to the interactions:
\begin{equation}\label{eq:GP split interaction}
w_N = \left( P_-^{\otimes 2} + \Pi \right)  w_N \left( P_-^{\otimes 2} + \Pi \right). 
\end{equation}
We have to bound from below the difference
$$w_N- P_- ^{\otimes 2}  w_N P_- ^{\otimes 2}$$
by controling the off-diagonal terms 
$$ \Pi w_N P_- ^{\otimes 2} + P_- ^{\otimes 2} w_N \Pi$$
in~\eqref{eq:GP split interaction}. To this end we write 
$$w_N=w_N^+ - w_N^- \mbox{ with } w_N^{\pm} \ge 0$$
and we use the well-known fact that the diagonal elements of a self-adjoint operator control\footnote{Cf for a positive hermitian matrix $(m_{i,j})_{1\leq i, j \leq n}$ the inequality $2 |m_{i,j}| \leq m_{i,i} + m_{j,j}$.} the off-diagonal elements.

Since $w_N ^{\pm}$ are positive as multiplication operators on $L ^2 (\R ^{2d})$ we have for all $b>0$ 
$$
\left(b^{1/2} P_- ^{\otimes 2} \pm b^{-1/2} \Pi \right) \, w_N^{\pm} \, \left(b^{1/2} P_- ^{\otimes 2} \pm b^{-1/2} \Pi  \right) \ge 0.
$$
Combining these inequalities appropriately we obtain 
$$
\Pi w_N P_- ^{\otimes 2} + P_- ^{\otimes 2} w_N \Pi \geq - b P_- ^{\otimes 2} |w_N| P_- ^{\otimes 2} - b ^{-1} \Pi |w_N| \Pi
$$
for all $b>0$. We then recall that, as an operator,
$$|w_N| \le \|w_N\|_{L^\infty} \le N^{d\beta}$$
and we choose $b=2 N^{d\beta}/\Lambda$ to obtain
$$
\Pi w_N P_- ^{\otimes 2} + P_- ^{\otimes 2} w_N \Pi \geq - \frac{2 N^{2d\beta}}{\Lambda} P_- ^{\otimes 2} - \frac{\Lambda}{2} \Pi $$
Inserting this bound in~\eqref{eq:GP split interaction} we get 
\begin{equation} \label{eq:GP loc interaction lower bound}
w_N \ge  P_- ^{\otimes 2}  w_N P_- ^{\otimes 2}  - \frac{2 N^{2d\beta}}{\Lambda} P_- ^{\otimes 2} -  \left( \frac{\Lambda}{2} + N ^{d\beta} \right) \Pi
\end{equation}
by simply bounding from below
$$ \Pi w_N \Pi \geq - \Pi |w_N| \Pi \geq - N ^{d\beta} \Pi.$$
Combining with~\eqref{eq:GP loc 1 body lower bound} and~\eqref{eq:GP loc interaction lower bound} we obtain for all $\Lambda\ge 1$ the lower bound
$$
H_2 \ge  P_- ^{\otimes 2}  H_2  P_- ^{\otimes 2} - \frac{2 N^{2d\beta}}{\Lambda} P_- ^{\otimes 2} + \left( \frac{\Lambda}{2}- N ^{d\beta} - C\right) \Pi
$$
Since we assume  $\Lambda \ge C N^{d\beta}$ for a large constant $C>0$, we may use $P_- \otimes P_- \le \1_{\gH ^2}$ and $P_- \otimes P_+,\, P_+ \otimes P_-\ge 0$ to deduce
$$
H_2 \ge  P_- ^{\otimes 2}  H_2  P_- ^{\otimes 2} - \frac{2 N^{2d\beta}}{\Lambda} + \frac{\Lambda}{4} P_+ ^{\otimes 2},
$$
which concludes the proof.
\end{proof}

\medskip

\noindent\textbf{Step 2, estimating the localized energy.} 
Let $\Psi_N$ be a minimizer for the $N$-body energy, $\Gamma_N = |\Psi_N \rangle \langle \Psi_N |$ and 
$$ \gamma_N ^{(n)} = \tr_{n+1 \to N} [\Gamma_N] $$
the corresponding reduced density matrices. We can now think only in terms of the $P_-$ and $P_+$-localized states defined as in  Lemma~\ref{lem:Fock loc} by the relations
\begin{equation}\label{eq:GP loc pm}
\left(G_N ^{\pm}\right) ^{(n)} = P_{\pm} ^{\otimes n} \gamma_N ^{(n)} P_\pm ^{\otimes n}. 
\end{equation}
We recall that $G_N ^{\pm}$ are states on the truncated Fock space, i.e.
\begin{equation}\label{eq:GP nomalization-localized-state}
\sum_{k=0} ^N \tr_{\gH ^k }[G_{N,k} ^{\pm}] = 1. 
\end{equation}
We now compare the Hartree energy $\eH$ and the localized energy of $\Gamma_N$ defined by the truncated Hamiltonian of Lemma~\ref{lem:GP localize-energy}:

\begin{lemma}[\textbf{Lower bound for the localized energy}]\label{lem:GP main terms}\mbox{}\\
If $\Lambda\geq C N ^{d\beta}$ for a large enough constant $C>0$ we have,
\begin{equation}\label{eq:low bound main}
  \frac{1}{2} \tr \left[ P_- ^{\otimes 2} H_2 P_- ^{\otimes 2} \gamma_{N}^{(2)} \right] + \frac{\Lambda}{4} \tr
  \left[P_+ ^{\otimes 2} \gamma_N^{(2)} \right] \ge \eH  - C \frac{\Lambda N_\Lambda}{N}- C\frac{\Lambda}{N ^2} .
\end{equation}
\end{lemma}

This lemma is proved by combining Theorem~\ref{thm:DeFinetti quant} and the methods of Chapter~\ref{sec:locFock}. We define an approximate de Finetti measure starting from $G_N ^-$. The idea is related to that we used for the proof of the weak de Finetti theorem in Section~\ref{sec:proof deF faible}:
\begin{equation}\label{eq:GP def-mu-N-localized}
d\mu_N(u) = \sum_{k=2}^N {N \choose 2}^{-1} {k \choose 2} \dim \left( (P_-\gH)_s^k \right) \pscal{u^{\otimes k},G_{N,k}^- u^{\otimes k}} du
\end{equation}
where $du$ is the uniform measure on the finite-dimensional sphere $SP_-\gH$. The choice of the weights in the above sum comes from the fact that we want to approximate the localized two-body density matrix $P_- ^{\otimes 2}  \gamma_N ^{(2)} P_- ^{\otimes 2}$, which is the purpose of 

\begin{lemma} [\textbf{Quantitative de Finetti for a localized state.}] \label{lem:GP deF-localized-state}\mbox{}\\
For all  $\Lambda>0$,  we have 
$$
\Tr_{\gH^2} \left| P_-^{\otimes 2} \gamma_{N}^{(2)} P_-^{\otimes 2} - \int_{SP_-\gH} |u^{\otimes 2}\rangle \langle u^{\otimes 2}| d\mu_N(u)\right| \le \frac{8 N_\Lambda}{N}.
$$
\end{lemma}

\begin{proof} Up to normalization, $G_{N,k}^-$ is a state on $(P_-\gH)_s^{\otimes k}$. Applying Theorem~\ref{thm:DeFinetti quant} with the explicit construction~\eqref{eq:def CKMR} we thus have 
\begin{align*}
\tr_{\gH ^2}\left| \Tr_{3\to k}\left[G_{N,k}  ^-\right]  - \int_{SP_- \gH} |u ^{\otimes 2}\rangle \langle u ^{\otimes 2} | d\mu_{N,k}(u) \right| \leq 8 \frac{N_\Lambda}{k} \tr_{\gH ^k} \left[ G_{N,k} ^-\right]
\end{align*}
with
$$d\mu_{N,k}(u)= \dim (P_-\gH)_s^k \pscal{u^{\otimes k},G_{N,k}^- u^{\otimes k}} du.$$
In view of~\eqref{eq:GP loc pm} and~\eqref{eq:GP def-mu-N-localized} we deduce
\begin{align*}
\Tr_{\gH^2} \left| P_-^{\otimes 2} \gamma_{N}^{(2)} P_-^{\otimes 2} - \int_{SP_-\gH} |u^{\otimes 2}\rangle \langle u^{\otimes 2}| d\mu_N(u)\right|  \leq 
\sum_{k=2}^N {N \choose 2}^{-1} {k \choose 2} 
\frac{8N_\Lambda}{k} \tr_{\gH ^k} \left[ G_{N,k} ^-\right].
\end{align*}
There remains to use the normalization~\eqref{eq:GP nomalization-localized-state} and
$$
{N \choose 2}^{-1} {k \choose 2}  = \frac{k(k-1)}{N(N-1)} \le \frac{k}{N}
$$
to conclude the proof.
\end{proof}

We proceed to the 

\begin{proof}[Proof of Lemma~\ref{lem:GP main terms}.] 
We start with the $P_-$-localized term. By cyclicity of the trace we have
\begin{align*}
\tr \left[ P_- ^{\otimes 2} H_2 P_-^{\otimes 2} \gamma_{N}^{(2)} \right] &= \tr \left[ P_-^{\otimes 2} H_2 P_-^{\otimes 2} \big( P_-^{\otimes 2} \gamma_{N}^{(2)} P_-^{\otimes 2} \big) \right].
\end{align*}
We then apply Lemma~\ref{lem:GP deF-localized-state}, which gives 
$$
\Tr_{\gH^2} \left| P_-^{\otimes 2} \gamma_{N}^{(2)} P_-^{\otimes 2} - \int_{SP_-\gH} |u^{\otimes 2}\rangle \langle u^{\otimes 2}| d\mu_N(u)\right| \le \frac{8 N_\Lambda}{N}.
$$
On the other hand we of course have
\begin{equation}\label{eq:GP bound H-}
\norm{P_-^{\otimes 2} H_2 P_-^{\otimes 2} } \leq 2\Lambda + \|w_N\|_{L^\infty} \le 3\Lambda
\end{equation}
in operator norm, and thus 
\begin{align} \label{eq:GP P-term-0}
\frac{1}{2}\tr \left[ P_-  ^{\otimes 2} H_2 P_-  ^{\otimes 2} \gamma_{N}^{(2)} \right] &\ge \frac{1}{2}\int_{SP_- \gH} \tr_{\gH ^2} \left[ H_2 |u ^{\otimes 2}\rangle \langle u ^{\otimes 2}| \right]  d\mu_N -\frac{C\Lambda N_\Lambda}{N} \nn\\
& = \int_{SP_- \gH} \EH [u] d\mu_N -\frac{C\Lambda N_\Lambda}{N}. 
\end{align}
By the variational principle $\EH [u] \geq \eH$, we deduce
\begin{align} \label{eq:GP P-term}
\frac{1}{2}\tr \left[ P_-  ^{\otimes 2} H_2 P_-  ^{\otimes 2} \gamma_{N}^{(2)} \right] \ge 
 \eH \sum_{k=2} ^N \binom{N}{2} ^{-1} \binom{k}{2} \tr_{\gH ^k} \left( G_{N,k} ^- \right) -\frac{C\Lambda N_\Lambda}{N}
\end{align}
where the computation of $\int d\mu_N$ is straightforward using Schur's formula~\eqref{eq:Schur}. 

For the $P_+$-localized term we use~\eqref{eq:GP loc pm},~\eqref{eq:Fock mat red}  and~\eqref{eq:Fock funda rel} to obtain
\begin{align} \label{eq:GP P+term}
\frac{\Lambda}{4}\tr[ P_+ ^{\otimes 2} \gamma_N^{(2)} P_+ ^{\otimes 2}] &= \frac{\Lambda}{4} \sum_{k=2} ^N \binom{N}{2} ^{-1} \binom{k}{2} \tr \left[ G_{N,k}^+ \right] \nn\\
&= \frac{\Lambda}{4} \sum_{k=0} ^{N-2} \binom{N}{2} ^{-1} \binom{N-k}{2} \tr \left[ G_{N,k}^- \right]. 
\end{align}
Gathering~\eqref{eq:GP P-term},~\eqref{eq:GP P+term} and recalling that
$$
\binom{N}{2} ^{-1} \binom{k}{2} = \frac{k^2}{N^2} + O(N^{-1}), \quad \binom{N}{2} ^{-1} \binom{N-k}{2} = \frac{(N-k)^2}{N^2} + O(N^{-1}),
$$
we find
\begin{align} \label{eq:GP P-P+-together}
&\quad \quad \frac{1}{2}\tr \left[ P_-  ^{\otimes 2} H_2 P_-  ^{\otimes 2} \gamma_{N}^{(2)} \right] + \frac{\Lambda}{4} \tr[P_+  ^{\otimes 2} \gamma_N^{(2)}] \nn\\
&\ge \sum_{k=0}^N \tr \left[ G_{N,k}^- \right] \left( \frac{k^2}{N^2}  \eH + \frac{(N-k)^2}{N^2}\frac{\Lambda}{4}\right) - \frac{C(|\eH|+\Lambda)}{N ^2}- \frac{C\Lambda N_\Lambda}{N}. 
\end{align}
The first error term comes from the fact that the sums in~\eqref{eq:GP P-term} and~\eqref{eq:GP P+term} do not run exactly from $0$ to $N$, and we have used the normalization of the localized states~\eqref{eq:GP nomalization-localized-state} to control the missing terms. 

It is easy to see that for all $p,q$, $0 \leq \lambda \leq 1$
$$ p \lambda ^2 + q(1-\lambda) ^2 \geq   p - \frac{p ^2}{q}.$$
We then take $p=\eH$, $q=\Lambda/4$, $\lambda=k/N$ and use~\eqref{eq:GP nomalization-localized-state} again to deduce
\begin{align*} 
 &  \tr \left[ P_-  ^{\otimes 2} H_2 P_-  ^{\otimes 2} \gamma_{N}^{(2)} \right] + \frac{\Lambda}{2} \tr(P_+  ^{\otimes 2} \gamma_N^{(2)}) \\
& \quad\quad \quad\quad\quad\quad\quad\quad\ge \eH - \frac{\eH^2}{\Lambda} - \frac{C(|\eH|+\Lambda)}{N ^2}- \frac{C\Lambda N_\Lambda}{N} .
\end{align*}
from~\eqref{eq:GP P-P+-together}. There remains to insert the simple estimate
\begin{equation}\label{eq:GP naive bound}
|\eH| \le C+ \|w_N\|_{L^\infty} \le C + N^{d\beta} \leq C + \Lambda
\end{equation}
to obtain the sought-after result.
\end{proof}

\medskip 

\noindent\textbf{Step 3, final optimization.}. There only remains to optimize over $\Lambda$. Indeed, we recall that by definition
$$ \frac{E(N)}{N} = \frac{1}{2} \tr_{\gH ^2}[ H_2 \gamma_N ^{(2)}] $$ 
with the two-body Hamiltonian~\eqref{eq:GP two body hamil}. Combining Lemmas~\ref{lem:GP localize-energy} and~\ref{lem:GP main terms} we have the lower bound
\begin{equation*}
\frac{E(N)}{N} \geq \eH - \frac{C N^{2d\beta}}{\Lambda} -  C \frac{\Lambda N_\Lambda}{N}- C\frac{\Lambda}{N ^2}
\end{equation*}
for all $\Lambda \ge CN^{d\beta}$ with $C>0$ large enough. Using Lemma~\ref{lem:GP nb bound states}, this reduces to 
\begin{equation*}
\frac{E(N)}{N} \geq \eH - \frac{C N^{2d\beta}}{\Lambda} - C_{d} \frac{\Lambda^{1+ d/s + d/2}}{N}.
\end{equation*}
Optimizing over $\Lambda$ we get
\begin{equation}\label{eq:choice e}
\Lambda= N ^{t}
\end{equation}
where 
$$ t = 2s \frac{1+2d\beta}{4s + ds + 2d}$$
and the condition $t>2d\beta$ in~\eqref{eq:GP asum t} ensures that $\Lambda \gg N^{2d\beta}$ for large $N$. We thus conclude
\[
\eH \ge \frac{E(N)}{N} \geq \eH - C_{d} N ^{-t + 2d \beta}, 
\]
which is the desired result.
\end{proof}

\begin{remark}[Note for later use.]\label{rem:GP by product}\mbox{}\\
Following the steps of the proof more precisely we obtain information on the asymptotic behavior of minimizers. More specifically, going back to~\eqref{eq:GP P-term-0}, using Lemma~\ref{lem:GP localize-energy} and dropping the positive $P_+$ localized terms we have
$$ \eH \geq \frac{E(N)}{N} \ge \frac{1}{2}\Tr [P_-^{\otimes 2} H_2 P_-^{\otimes 2} \gamma_N^{(2)}] + o(1) \ge \int_{S\gH } \EH [u] d\mu_N (u) + o(1)$$
and thus
\begin{equation} \label{eq:GP mu-N-localized-cv-energy}
o(1) \geq \int_{S\gH } \left( \EH [u] - \eH \right) d\mu_N (u)
\end{equation}
when $N\to\infty$. We do not specify here (cf Item (3) in Remark~\ref{rem:GP NLS}) the exact order of magnitude of the $o(1)$ obtained by optimizing in Step 3 above. Estimate~\eqref{eq:GP mu-N-localized-cv-energy} morally says that $\mu_N$ must be concentrated close to the minimizers of  $\EH$. This will be of use in the proof of Theorem~\ref{thm:deriv nls}. 
\hfill\qed
\end{remark}

\subsection{From Hartree to NLS}\label{sec:GP Hartee to NLS} \mbox{}\\\vspace{-0.4cm}

There remains to deduce Theorem~\ref{thm:deriv nls} as a corollary of the above analysis. We start with the following lemma:

\begin{lemma}[\textbf{Stability of one-body functionals}]\label{lem:GP Hartree NLS}\mbox{}\\
Consider the functionals~\eqref{eq:GP hartree f} and~\eqref{eq:GP intro nls}. Under the assumptions of Theorem~\ref{thm:deriv nls}, there exists a minimizer for $\ENLS$. Moreover, for all normalized function $u\in L^2(\R^d)$ we have 
\bq \label{eq:GP hartree NLS 0}
\norm{|u|}_{H^1}^2 \le C (\EH [u]+C)
\eq
and
\bq \label{eq:GP hartree NLS 1}
\left| \EH [u] - \ENLS[u] \right| \le C N^{-\beta} \left( 1 + \int_{\R ^d} |\nabla u| ^2 \right) ^2. 
\eq
Consequently
\begin{equation}\label{eq:GP hartree NLS 2}
|\eH - \eNLS | \le C N ^{-\beta}.
\end{equation}
\end{lemma}

\begin{proof}
The stability assumptions we have made guarantee that minimizing sequences for $\ENLS$ are bounded in $H^1$. Assumption~\eqref{eq:GP conv delta} allows to easily estimate the difference between the Hartree and NLS functionals for $H^1$ functions. Details are omitted and may be found in~\cite{LewNamRou-14}. 
\end{proof}

We are now equiped to complete the derivation of NLS functional.

\begin{proof}[Proof of Theorem~\ref{thm:deriv nls}]
Combining~\eqref{eq:GP hartree NLS 2} with~\eqref{eq:GP energy estimate} concludes the proof of~\eqref{eq:GP ener convergence}. There thus remains to prove convergence of states, the second item of the theorem. We proceed in four steps:

\noindent\textbf{Step 1, strong compactness of reduced density matrices.} We extract a diagonal subsequence along which
\begin{equation}\label{eq:weak CV}
\gamma_N ^{(n)} \wto_* \gamma ^{(n)} 
\end{equation}
when $N\to \infty$, for all $n\in\N$. On the other hand we have 
\begin{equation}\label{eq:unif bound kinetic}
\tr \Big[T \gamma_N^{(1)}\Big]=\Tr \Big[ \left(- \Delta + V \right) \gamma_N^{(1)} \Big] \leq C, 
\end{equation}
independently of $N$. To see this, pick some $\alpha >0$, define
\begin{equation*}\label{eq:start hamil eta}
H_{N,\alpha} = \sum_{j=1} ^N  \left( - \Delta_j + V(x_j) \right) + \frac{1+\alpha}{N-1} \sum_{1\leq i<j \leq N}N ^{d\beta} w( N ^{\beta} (x_i-x_j))
\end{equation*}
and apply Theorem~\ref{thm:error Hartree} to this Hamiltonian. We find in particular $H_{N,\alpha}\geq -CN$ and deduce
$$\eNLS+o(1)\geq \frac{\pscal{\Psi_N,H_N\Psi_N}}N\geq -C(1+\alpha)^{-1}+\frac{\alpha}{1+\alpha}\tr \big[T\gamma_N^{(1)}\big],$$
which gives~\eqref{eq:unif bound kinetic}. Since $T=- \Delta + V$ has compact resolvent,~\eqref{eq:weak CV} and~\eqref{eq:unif bound kinetic} imply that, up to a subsequence $\gamma_{N}^{(1)}$ strongly converges in trace-class norm. As noted previously, Theorem~\ref{thm:DeFinetti fort} then implies that also $\gamma_N^{(n)}$ strongly converges, for all $n\geq1$.

\medskip

\noindent\textbf{Step 2, defining the limit measure.} We simplify notation by calling $r_N$ the best bound on  $\left|E(N)/N - \eNLS\right|$ obtained previously. Let $d\mu_N$ be defined as in Lemma~\ref{lem:GP deF-localized-state}. It satisfies 
$$\mu_N(SP_-\gH)=\tr \left[P_-^{\otimes 2}\gamma_N^{(2)}P_-^{\otimes 2}\right].$$
We have
$$
\Tr \left| P_-^{\otimes 2} \gamma_N^{(2)} P_-^{\otimes 2} - \int_{SP_-\gH} |u^{\otimes 2}\rangle \langle u^{\otimes 2}| d\mu_N (u) \right| \leq \frac{8N_\Lambda}{N}\leq C\frac{\Lambda^{1+d/s+d/2}}{N}\to0.
$$
One may on the other hand deduce from the energy estimates of Section~\ref{sec:GP Hartree bounds} a control on the number of excited particles:
\begin{equation}\label{eq:bound excited particles}
1-\mu_N(SP_-\gH)=\Tr\left[(1-P_-^{\otimes 2}) \gamma_N^{(2)}\right] \leq \frac{r_N}{\Lambda}. 
\end{equation}
By the triangle and Cauchy-Schwarz inequalities we deduce
\bq \label{eq:mu-N-localized-cv-gamma2}
\Tr \left|\gamma_N^{(2)} - \int_{SP_-\gH} |u^{\otimes 2}\rangle \langle u^{\otimes 2}| d\mu_N (u) \right| \leq C\frac{\Lambda^{1+d/s+d/2}}{N}+C\sqrt{\frac{r_N}{\Lambda}}.
\eq
We now denote $P_K$ the spectral projector of  $T$ on energies below a truncation $K$, defined as in~\eqref{eq:GP projectors}. Since $\gamma_N ^{(2)} \to \gamma ^{(2)}$ and $P_K \to \1$  
$$ 
\lim_{K \to \infty} \lim_{N\to \infty} \mu_N (S P_K\gH) = 1.
$$
This condition allows us to use Prokhorov's Theorem and~\cite[Lemma~1]{Skorokhod-74} to ensure that, after a possible further extracion, $\mu_N$ converges to a measure $\mu$ on the ball $B\gH$. Passing to the limit, we find
$$\gamma^{(2)}=\int_{B\gH} |u^{\otimes 2}\rangle \langle u^{\otimes 2}| d\mu(u)$$
and it follows that $\mu$ has its support in the sphere $S\gH$ since $\tr [\gamma ^{(2)}] = 1$ by strong convergence of the subsequence.

\medskip

\noindent\textbf{Step 3, the limit measure only charges NLS minimizers.} Using~\eqref{eq:GP mu-N-localized-cv-energy} and  
$$\mu_N(SP_-\gH)=1+O\left(\frac{r_N}{\Lambda}\right),$$
we deduce that
$$\int_{SP_-\gH } \big(\EH  [u] -\eH \big)d\mu_N (u)\leq o(1)$$
in the limit $N\to \infty$. By the estimates of Lemma~\ref{lem:GP Hartree NLS}, it follows that, for a large enough constant $B>0$ independent of $N$,
$$\frac{B^2}{C}\int_{\|\nabla u\|_{L ^2}\geq B} d\mu_N (u)\leq \int_{\|\nabla u \|_{L ^2}\geq B} \big(\EH [u] -\eH \big)d\mu_N (u)\leq o(1),$$
and
$$\int_{\|\nabla u\|_{L^2}\leq B} \big(\ENLS [u] -\eNLS\big)d\mu_N\leq C(1+B^4) N^{-\beta} +\int_{\|\nabla u\|_{L^2}\leq B} \big(\EH [u] -\eH\big)d\mu_N (u)\leq o(1).$$
Passing to the limit $N\to\ii$, we now see that $\mu$ is supported in $\MNLS$.

At this stage, using~\eqref{eq:mu-N-localized-cv-gamma2} and the convergence of $\mu_N$ we have, strongly in trace-class norm,  
$$ \gamma_N^{(2)} \to \int_{\MNLS} |u^{\otimes 2}\rangle \langle u^{\otimes 2}| d\mu(u),$$
where $\mu$ is a probability measure supported in $\MNLS$. Taking a partial trace we obtain
$$ \gamma_N^{(1)} \to \int_{\MNLS} |u\rangle \langle u| d\mu(u)$$
and there only remains to obtain the convergence of reduced density matrices of order $n >2$. 

\medskip

\noindent\textbf{Step 4, higher order density matrices.} We want to obtain
$$ \gamma_N^{(n)} \to \int_{\MNLS} |u^{\otimes n}\rangle \langle u^{\otimes n}| d\mu(u),$$
in trace-class norm when $N\to \infty$. In view of the definition of $\mu$, it suffices to show that
\begin{equation}\label{eq:claim higher DM}
\Tr \left|\gamma_N^{(n)} - \int_{SP_-\gH} |u^{\otimes n}\rangle \langle u^{\otimes n}| d\mu_N (u) \right| \to 0
\end{equation}
where $\mu_N$ is the measure defined by applying Lemma~\ref{lem:GP deF-localized-state} to $\gamma_N ^{(2)}$. To this end we start by approximating $\gamma_N ^{(n)}$ using a new measure, a priori different from $\mu_N = \mu_N^2$ :
\bq \label{eq:def-mu-N-localized bis}
d\mu_N ^n (u) = \sum_{k=n}^N  { N \choose n} ^{-1} {k\choose n} \dim (P_-\gH)_s^k \pscal{u^{\otimes k},G_{N,k}^- u^{\otimes k}} du.
\eq
Proceeding as in the proof of  Lemma~\ref{lem:GP deF-localized-state} we obtain 
\begin{equation}\label{eq:deF higher DM}
\Tr_{\gH^n} \left| P_-^{\otimes n} \gamma_{N}^{(n)} P_-^{\otimes n} - \int_{SP_-\gH} |u^{\otimes n}\rangle \langle u^{\otimes n}| d\mu_N ^n (u)\right| \le C\frac{ n N_\Lambda}{N}.
\end{equation}
An estimate similar to~\eqref{eq:bound excited particles} next shows that 
$$\Tr_{\gH^n} \left| \gamma_{N}^{(n)} - \int_{SP_-\gH} |u^{\otimes n}\rangle \langle u^{\otimes n}| d\mu_N ^n (u)\right| \to 0.$$
Using again the bound
$$ { N \choose n} ^{-1} {k\choose n} = \left(\frac{k}{N}\right) ^n + O (N ^{-1})$$
as well as the triangle inequality and Schur's formula~\eqref{eq:Schur} we deduce from~\eqref{eq:deF higher DM} that 
\begin{align}\label{eq:higher DM final}
\Tr \left|\gamma_N^{(n)} - \int_{SP_-\gH} |u^{\otimes n}\rangle \langle u^{\otimes n}| d\mu_N (u) \right| &\leq  \sum_{k=0} ^N \left( \left(\frac{k}{N}\right) ^2 - \left(\frac{k}{N}\right) ^n \right) \tr_{\gH ^k} \left[ G_{N,k} ^- \right]  \nonumber
\\&+ \sum_{k=0} ^{n-1} \left(\frac{k}{N}\right) ^n \tr_{\gH ^k} \left[ G_{N,k} ^- \right] \nonumber
\\&+ \sum_{k=0} ^{2} \left(\frac{k}{N}\right) ^2 \tr_{\gH ^k} \left[ G_{N,k} ^- \right] + o(1).
\end{align}
Finally, combining the various bounds we have obtained
$$ \sum_{k=2} ^N  \left(\frac{k}{N}\right) ^2 \tr_{\gH ^k} \left[ G_{N,k} ^- \right] \to 1.$$
But, because of~\eqref{eq:GP nomalization-localized-state} it follows that
$$\sum_{k=0} ^N  \left(\frac{k}{N}\right) ^2 \tr_{\gH ^k} \left[ G_{N,k} ^- \right] \to 1.$$
Recalling the normalization 
$$ \sum_{k=0} ^N \tr_{\gH ^k} \left[ G_{N,k} ^- \right] = 1,$$
we may thus apply Jensen's inequality to obtain
$$1\geq \sum_{k=0} ^N  \left(\frac{k}{N}\right) ^n \tr_{\gH ^k} \left[ G_{N,k} ^- \right]\geq \left( \sum_{k=0} ^N  \left(\frac{k}{N}\right) ^2 \tr_{\gH ^k} \left[ G_{N,k} ^- \right]\right) ^{n/2} \to 1.$$
There only remains to insert this and~\eqref{eq:GP nomalization-localized-state} in~\eqref{eq:higher DM final} to deduce~\eqref{eq:claim higher DM} and thus conclude the proof of the theorem. 
\end{proof}

\newpage

\appendix

\section{\textbf{A quantum use of the classical theorem}}\label{sec:class quant}

This appendix is devoted to an alternative proof of a weakened version of Theorem~\ref{thm:confined}. The method, introduced in~\cite{Kiessling-12} is less general than those described previously. This cannot be helped since it consists in an application of the Hewitt-Savage (classical de Finetti) theorem to a quantum problem. Here we follow an unpublished note of Mathieu Lewin and Nicolas Rougerie~\cite{LewRou-unpu12}.

In some cases (absence of magnetic fields essentially), the wave-function $\Psi_N$ minimizing a $N$-body energy can be chosen strictly positive. The ground state of the quantum problem may then be entirely analyzed in terms of the $N$-body density $\rho_{\Psi_N} = |\Psi_N| ^2$, which is a purely classical object (a symmetric probability measure) whose limit can be described using Theorem~\ref{thm:HS}. This approach works only under assumptions on the one-body Hamiltonian that are much more restrictive than those discussed in Remark~\ref{rem:energie cin}. 

\medskip

\noindent\emph{\underline{Note added, June 2020}. I thank Samir Salem for pointing out mistakes in the 2015 version of this appendix. The proof of Proposition~\ref{lem:app A kin lim} presented therein is withdrawn. I have replaced it by a plausibility argument and references to the literature for alternative, correct, proofs.}

\subsection{Classical formulation of the quantum problem}\label{sec:hyp class quant}\mbox{}\\\vspace{-0.4cm}

We here consider a quantum $N$-body Hamiltonian acting on $L ^2 (\R ^{dN})$
\begin{equation}\label{eq:app A hamil N}
H_N = \sum_{j=1} ^N \left(T_j + V (x_j)\right) + \frac{1}{N-1} \sum_{1\leq i < j \leq N} w(x_i-x_j)
\end{equation}
where the trapping potential $V$ and the interaction potential $w$ are chosen as in Section~\ref{sec:Hartree deF fort}. In particular we assume that $V$ is confining. We shall need rather strong assumptions on the kinetic energy operator $T$. The approach we shall discuss in this appendix is based on the following notion:

\begin{definition}[\textbf{Kinetic energy with positive kernel}]\label{def:hyp class quant}\mbox{}\\
We say that $T$ has a positive kernel if there exists $T(x,y):\R ^d \times \R ^d \to \R ^+$ such that  
\begin{equation}\label{eq:hyp class quant}
\bral \psi, T \psi \ketr = \iint_{\R ^d\times \R ^d} T(x,y) \left| \psi (x) - \psi (y) \right| ^2 dxdy 
\end{equation}
for all functions $\psi \in L ^2 (\R^d)$.\hfill\qed
\end{definition}

It is well-known that the pseudo-relativitic kinetic energy is of this form. Indeed 
\begin{equation}\label{eq:app A cin}
\bral \psi, \sqrt{-\Delta} \, \psi \ketr =  \frac{\Gamma(\frac{d+1}{2})}{2\pi ^{(n+1)/2}} \iint_{\R ^d \times \R ^d} \frac{\left| \psi (x) - \psi (y) \right| ^2}{|x-y| ^{d+1}} dxdy,
\end{equation}
see~\cite[Theorem 7.12]{LieLos-01}. More generally one may consider $T = |p| ^{2s}$, $0< s < 1$:
\begin{equation}\label{eq:app A pseudo rel} \bral \psi, |p| ^{2s} \psi \ketr = \bral \psi, (-\Delta) ^{s} \psi \ketr =  C_{d,s} \iint_{\R ^d \times \R ^d} \frac{\left| \psi (x) - \psi (y) \right| ^2}{|x-y| ^{d+2s}} dxdy
\end{equation}
recalling the correspondance~\eqref{eq:quant mom}. 

The non-relavistic kinetic energy does not fit in this framework, but one may nevertheless apply the considerations of this appendix to it because
\begin{equation}\label{eq:app A cin bis}
\bral \psi, -\Delta \, \psi \ketr = \int_{\R ^d} |\nabla \psi| ^2 = C_d \: \underset{s\uparrow 1}{\lim} \:(1-s) \iint_{\R ^d \times \R ^d} \frac{\left| \psi (x) - \psi (y) \right| ^2}{|x-y| ^{d+2s}} dxdy 
\end{equation}
with
$$ C_d = \left( \int_{S^{d-1}} \cos \theta \, d\sigma \right) ^{-1}.$$
Here $S^{d-1}$ is the euclidean sphere equiped with its Lebesgue measure $d\sigma$ and $\theta$ represents the angle with respect to the vertical axis. One may thus see $-\Delta$ as a limiting case of Definition~\ref{def:hyp class quant}. Formula~\eqref{eq:app A cin bis} is proved in ~\cite[Corollary 2]{BouBreMir-01} and \cite{MasNag-78}, see also~\cite{BouBreMir-02,MazSha-02}.

The cases with magnetic field $T= \left( p + A \right) ^2$ and $T= \left| p + A \right|$ are \emph{not} covered by this framework. One may deal with them with the methods of the main body of the course, as already mentioned, but not with those of this appendix. 

\medskip

An important consequence of the above choice of kinetic energy is that, by the triangle inequality
$$ \bral \psi, T \psi \ketr \geq \bral |\psi|, T |\psi| \ketr $$
and thus the total $N$-body energy
$$ \E_N [\Psi_N] = \bral \Psi_N, H_N \Psi_N \ketr$$
satisfies
$$ \E_N [\Psi_N] \geq \E_N \left[| \Psi_N |\right].$$
The ground state energy can thus be calculated using only positive test functions
\begin{equation}\label{eq:app A formul class pre}
E(N) = \inf \left\{ \E_N [\Psi_N], \Psi_N \in L_s ^2 (\R ^{dN})\right\} = \inf \left\{ \E_N [\Psi_N], \Psi_N \in L_s ^2 (\R ^{dN}), \Psi_N \geq 0 \right\}.
\end{equation}
This remark actually allows to prove a fact mentioned previously: the bosonic ground state is identical to the absolute ground state in the case of a kinetic energy of the form~\eqref{eq:hyp class quant} or~\eqref{eq:app A cin bis}, see~\cite[Chapter 3]{LieSei-09}. 

We recall the definition of Hartree's functional: 
$$ \EH [u] = \bral u,T u \ketr + \int_{\R ^d } V |u| ^2 + \frac{1}{2} \iint_{\R ^d \times \R ^d } |u(x)| ^2 w(x-y) |u(y)| ^2 dx dy,$$
its infimum being denoted $\eH$. In the sequel of this appendix we prove the following statement, which is a weakened version of Theorem~\ref{thm:confined}:

\begin{theorem}[\textbf{Derivation of Hartree's theory, alternative statement}]\label{thm:class quant confined}\mbox{}\\
Let $V$ and $w$ satisfy the assumptions of Section~\ref{sec:Hartree deF fort}, in particular~\eqref{eq:hartree confine 2}. We moreover assume that $T$ has a positive kernel in the sense of Definition~\ref{def:hyp class quant}, or is a limit case of this definition, as in~\eqref{eq:app A cin bis}. We then have
$$\lim_{N\to\ii}\frac{E(N)}{N}=\eH.$$
Let $\Psi_N \geq 0$ be a ground state of $H_N$, achieving the infimum in~\eqref{eq:app A formul class pre} and 
$$ \rho_N ^{(n)} (x_1,\ldots,x_n) := \int_{\R ^{d(N-n)}} \left|\Psi_N\left( x_1,\ldots,x_N \right)\right| ^2 dx_{n+1}\ldots dx_N $$
be its $n$-body reduced density. There exists a probability measure $\mu$ on $\cM_{\rm H}$, the set of minimizers of $\EH$ (modulo a phase), such that, along a subsequence
\begin{equation}\label{eq:app A def fort result}
\lim_{N\to\ii} \rho^{(n)}_{N}= \int_{\cM_{\rm H}} \left|u ^{\otimes n}\right| ^2 d\mu(u) \mbox{ for all } n\in \N,
\end{equation}
strongly in $L ^1 \left(\R ^{dn}\right).$ In particular, if $\eH$ has a unique minimizer (modulo a constant phase), then for the whole sequence
\begin{equation}
\lim_{N\to\ii} \rho^{(n)}_{N} = \left|\uH^{\otimes n}\right| ^2 \mbox{ for all } n\in \N.
\label{eq:app A BEC-confined} 
\end{equation}
\end{theorem}

\begin{remark}[Uniqueness for Hartree's theory]\label{rem:app A unic Hartree}\mbox{}\\
In the case of a kinetic energy with positive kernel, uniqueness of the minimizer of $\EH$ is immediate if $w\geq 0$. Indeed, the kinetic energy  $\bral \psi,T \psi \ketr$ is then a strictly convex functional of $\rho = \sqrt{|\psi| ^2}$, see~\cite[Chapter 7]{LieLos-01}. If $w$ is positive the functional $\EH$ is thus itself strictly convex as a function of $\rho$. \hfill\qed 
\end{remark}

The particular case where $T= -\Delta$ has been dealt with by Kiessling in~\cite{Kiessling-12}, and we shall follow his method in the general case. It consists in treating the problem as a purely classical one, which explains that we only obtain the convergence of reduced densities~\eqref{eq:app A def fort result} instead of the convergence of the full reduced density matrices~\eqref{eq:def fort result}. We pursue in the direction of~\eqref{eq:app A formul class pre} by writing
\begin{equation}\label{eq:app A formul class}
E(N) =  \inf \left\{ \E_N \left[\sqrt{\mubf_N}\right], \mubf_N \in \PP_s  (\R ^{dN})\right\}
\end{equation}
where $\mubf_N$ plays the role of $|\Psi_N| ^2$ and we used the fact that we may assume $\Psi_N \geq 0$. The object we have to study is a  symmetric probability measure of $N$ variables, and our strategy shall be similar to that used to prove Theorem~\eqref{thm:aHS}:  
\begin{itemize}
\item Since the system is confined, one may easily pass to the limit and obtain a problem in terms of a classical state with infinitely many particles $\mubf \in \PP_s (\R ^{d\N})$. We then use Theorem~\ref{thm:HS} to describe the limit $\mubf ^{(n)}$ of $\mubf_N ^{(n)}$, for all $n$, using a unique probability measure $P_{\mubf} \in \PP (\PP (\R ^d))$.  
\item The subtle point is to prove that the limit energy is indeed an affine function of $\mubf$. This uses in an essential way the fact that the kinetic energy has a positive kernel (or is a limit case of such energies), as well as the Hewitt-Savage theorem. 
\end{itemize}

These two steps are contained in the two following sections. We then briefly conclude the proof of Theorem~\ref{thm:class quant confined} in a third section.

\subsection{Passing to the limit}\label{sec:lim class quant}\mbox{}\\\vspace{-0.4cm}

The limit problem we aim at deriving is described by the functional (compare with~\eqref{eq:aHS somme})
\begin{multline}\label{eq:app A prob lim}
\E [\mubf]:= \limsup_{n\to \infty} \frac{1}{n} I \left(\mubf ^{(n)} \right) 
\\ + \int_{\R ^d} V(x) d\mubf ^{(1)}(x) + \frac{1}{2}\iint_{\R ^d\times \R ^d } w(x-y) d\mubf ^{(2)}(x,y), 
\end{multline}
where $\mubf \in \PP_s(\R ^{d\N})$ and we set 
\begin{equation}\label{eq:app A kin class}
 I \left(\mubf ^{(n)} \right) := \bral \sqrt{\mubf_n}, \left(\sum_{j=1} ^n T_j \right) \sqrt{\mubf_n} \ketr 
\end{equation}
for all probability measure $\mubf_n \in \PP (\R ^{dn})$. In fact we shall prove in Lemma~\eqref{lem:app A kin lim} below that the $\limsup$ in~\eqref{eq:app A prob lim} is actually a limit.

\begin{lemma}[\textbf{Passing to the limit}]\label{lem:app A lim}\mbox{}\\
Let $\mubf_N \in \PP_s(\R ^{dN})$ achieve the infimum in~\eqref{eq:app A formul class}. Along a subsequence we have
$$ \mubf_N ^{(n)} \wto_* \mubf ^{(n)} \in \PP_s (\R ^{dn})$$
for all $n\in \N$, in the sense of measures. The sequence $\left( \mubf ^{(n)}\right)_{n\in \N}$ defines a probability measure $\mubf\in\PP_s (\R^{d\N})$ and we have
\begin{equation}\label{eq:app A lim inf}
\liminf_{N\to \infty} \frac{E(N)}{N} \geq \E[\mubf]. 
\end{equation}
\end{lemma}

\begin{proof}
Extracting convergent subsequences is done as in Section~\ref{sec:appli HS}. The existence of the measure $\mubf \in \PP_s (\R ^{d\N})$ follows, using Kolmogorov's theorem.

Passing to the liminf in the terms
$$ \frac{1}{N} \sum_{j=1} ^N \int_{\R^{dN}} V(x_j) \: d\mubf_N (x_1,\ldots,x_N) = \int_{\R ^d} V(x) \: d\mubf_N ^{(1)} (x)$$ 
and
$$ \frac{1}{N} \frac{1}{N-1} \sum_{1\leq i<j\leq N} ^N \int_{\R^{dN}} w(x_i-x_j) \:d \mubf_N (x_1,\ldots,x_N) = \frac{1}{2}\iint_{\R ^d\times \R ^d } w(x-y) \: d\mubf_N ^{(2)}(x,y)$$
uses the same ideas as in Chapters~\ref{sec:class} and~\ref{sec:quant}. We shall not elaborate on this point.  

The new point is to deal with the kinetic energy in order to obtain
\begin{equation}\label{eq:app A claim 1}
 \liminf_{N\to \infty} \frac{1}{N} T \left(\sqrt{\mubf_N }\right) \geq \limsup_{n\to \infty} \frac{1}{n} I \left(\mubf ^{(n)} \right). 
\end{equation}
To this end we denote $\Psi_N = \sqrt{\mubf_N}$ a minimizer in~\eqref{eq:app A formul class pre} and we then have 
\begin{align*}
\frac{1}{N} T \left(\sqrt{\mubf_N }\right) &= \frac{1}{N} \tr \left[ \sum_{j=1} ^N T_j \ketl \Psi_N \ketr \bral \Psi_N \brar \right]\\
&=\frac{1}{n} \tr \left[ \sum_{j=1} ^n T_j \gamma_N ^{(n)} \right]
\end{align*}
where $\gamma_N ^{(n)}$ is the $n$-th reduced density matrix of $\Psi_N$. It is a trace-class operator, that we decompose under the form
$$ \gamma_N ^{(n)} = \sum_{k=1} ^{+\infty} \lambda_n ^k  | u_n ^k \rangle \langle u_n ^k |$$
with $u_n ^k$ normalised in $L_s ^2(\R ^{dk})$ and $\sum_{k=1} ^{\infty} \lambda_k ^n = 1$. Inserting this decomposition in the previous equation and recalling~\eqref{eq:app A kin class} we obtain by linearity of the trace
\begin{align*}
\frac{1}{N} T \left(\sqrt{\mubf_N }\right) &= \frac{1}{n}  \sum_{k=1} ^{+\infty} \lambda_n ^k \: T\left( \sqrt{|u_n ^k| ^2} \right) \\
&\geq \frac{1}{n}  T\left( \sqrt{ \sum_{k=1} ^{+\infty} \lambda_n ^k |u_n ^k| ^2} \right) \\
&= \frac{1}{n}  T\left( \sqrt{ \rho_{\gamma_N ^{(n)}}} \right) = \frac{1}{n} I \left(\mubf_N ^{(n)} \right)
\end{align*}
where the inequality in the second line uses the convexity of the kinetic energy as a functions of the density $\rho$, already recalled in Remark~\ref{rem:app A unic Hartree}, cf~\cite[Chapter 7]{LieLos-01}. In the last equality
$$ \rho_{\gamma_N ^{(n)}} = \sum_{k=1} ^\infty \lambda_n ^k |u_n ^k| ^2$$
is the density\footnote{Formally, the diagonal part of the kernel.} of $\gamma_N ^{(n)}$ and it is easy to see that 
$$ \rho_{\gamma_N ^{(n)}} = \rho_{N} ^{(n)} = \int_{\R^{d(N-n)}} |\Psi_N(x_1,\ldots,x_N)| ^2 dx_{n+1}\ldots dx_N,$$
which yields
$$ \frac{1}{N} T \left(\sqrt{\mubf_N }\right) \geq \frac{1}{n} I \left(\mubf_N ^{(n)} \right).$$
To obtain~\eqref{eq:app A claim 1}, we first pass to the liminf in $N$, then to the limsup in $n$.
\end{proof}

The notion of kinetic energy with positive kernel is already crucial at this level. It provides the convexity property that we just used. It will play an even greater role in the next section.

\subsection{The limit problem}\label{sec:linear class quant}

We now have to show that the functional~\eqref{eq:app A prob lim} is affine on $\PP_s (\R ^{d\N})$. The last two terms obviously are, which is not suprising since they are classical in nature. It thus suffices to show that the first term, which encodes the quantum nature of the problem, is also linear in the density:    

\begin{proposition}[\textbf{Linearity of the limit kinetic energy}]\label{lem:app A kin lim}\mbox{}\\
The functional
$$ I \left(\mubf \right) := \limsup_{n\to \infty} \frac{1}{n} I \left(\mubf ^{(n)} \right) = \lim_{n\to \infty} \frac{1}{n} I \left(\mubf ^{(n)} \right)$$
is affine on $\PP_s (\R ^{d\N})$.
\end{proposition}

Kiessling~\cite{Kiessling-12} gave a proof of this result in the case of the non-relativistic kinetic energy ($s=1$). He notes that in this case
$$\frac{1}{n} I \left(\mubf ^{(n)} \right)= \frac{1}{n} T \left(\sqrt{\mubf ^{(n)}} \right) = \frac{1}{n} \sum_{j=1} ^n \int_{\R^{dn}} \left|\nabla_j \sqrt{\mubf ^{(n)}}\right| ^2 = \frac{1}{4n} \sum_{j=1} ^n \int_{\R^{dn}} \left|\nabla_j \log \mubf ^{(n)}\right| ^2 \mubf ^{(n)}$$
and that the last expression is identical to the Fisher information of the probability measure $\mubf ^{(n)}$. The quantity we study can thus be interpreted as a ``mean Fisher information'' of the measure $\mubf \in \PP_s(\R ^{d\N})$, in analogy with the mean entropy introduced in~\eqref{eq:aHS somme}.

This quantity has an interesting connection to the classical entropy of a probability measure. Letting $\mubf ^{(n)}$ evolve following the heat flow, one may show that at each time along the flow, the Fisher information is the derivative of the entropy. Since the heat flow is linear and the mean entropy is affine (cf Lemma~\ref{lem:mean entropy}), Kiessling deduces that the mean Fisher information is affine. Another point of view on this question is given in~\cite[Section~5]{HauMis-14}. The generalizations to fractional Fisher informations (derivatives of the entropy along fractional heat flows) are in~\cite{Salem-19,Salem-19b}, which give the $0<s<1$ cases of~\ref{lem:app A kin lim}. 

Yet another proof (for all cases $0<s \leq 1$) is in~\cite{Rougerie-19b}, but it is actually based on the quantum de Finetti theorem. It thus does not qualify for the purpose of giving a classical mechanics proof of Theorem~\ref{thm:class quant confined}.

Here we shall only give a plausibility argument for Proposition~\ref{lem:app A kin lim add}. We do not think it can be turned into a full proof, except possibly by combining with tools from~\cite{HauMis-14}. We prove the following weak version of linearity:

\begin{lemma}[\textbf{Linearity for nice convex combinations}]\label{lem:app A kin lim add}\mbox{}\\
Let $\rho_1,\ldots,\rho_J \in \PP(\R^d)$ be pairwise distinct positive measures with a common compact support $\Omega$. Assume there exists two positive constants $c_-,c_+$ such that  
\begin{equation}\label{eq:app A regul}
 c_- \leq \rho_j \leq c_+
\end{equation}
almost everywhere on $\Omega$, for all $j=1,\ldots,J$. Let $\lambda_1,\ldots,\lambda_J$ be positive numbers adding to $1$. Denote 
$$ \mubf_n:= \sum_{j=1} ^J \lambda_j \rho_j ^{\otimes n}.$$
Then 
\begin{equation}\label{eq:app A key}
\left| T \left( \sqrt{\mubf_n} \right) - \sum_{j=1} ^J \lambda_j T \left( \sqrt{\rho_j^{\otimes n} }\right) \right| \leq C (c_-,c_+,J) \left( \max_{1\leq i \neq j \leq J} \int_{\Omega} \sqrt{\rho_i} \sqrt{\rho_j} \right) ^{n-1} \max_{1\leq j \leq J} T\left( \sqrt{\rho_j} \right)  
\end{equation}
where the constant $C (c_-,c_+,J)$ depends only on $c_-,c_+$ and $J$ and $T$ is understood with all intergrations restricted to $\Omega$.
\end{lemma}

The crucial point here is that, if the $\rho_j$ are distinct probability measures
\begin{equation}\label{eq:app A trick}
 \int_{\Omega} \sqrt{\rho_i} \sqrt{\rho_j} < \frac{1}{2} \left( \int_{\Omega} \rho_i + \int_{\Omega} \rho_j\right) < 1 
\end{equation}
by Cauchy-Schwarz, so that this quantity raised to a large power $n$ yields an exponentially small term. Passing to the limit $n\to \infty$ shows that for $\mubf$ of the form
$$ \mubf = \int \rho ^{\otimes \infty} dP(\rho), \rho \in \PP (\PP (\R^d))$$
with $P$ an atomic measure on nice functions as in Lemma~\ref{lem:app A kin lim add} we indeed have 
$$ I(\mubf) = \int I(\rho) dP(\rho).$$

\begin{proof}
By convexity of $T$, the quantity inside the absolute value is actually non-positive. We thus aim at an upper bound on   
\begin{equation}\label{eq:rev1}
\sum_{j=1} ^J \lambda_j T \left( \sqrt{\rho_j^{\otimes n}} \right) -  T \left( \sqrt{\mubf_n}\right). 
\end{equation}
We consider first the case $s<1$ in~\eqref{eq:app A pseudo rel}, and extend to $s=1$ in the end. We denote 
$$ U_j = \sqrt{\rho_j ^{\otimes n}}$$
and 
$$ X = (x,x_2,\ldots,x_n), \quad Y= (y,x_2,\ldots,x_n), \quad \hat{X} = (x_2,\ldots,x_n).$$
to write~\eqref{eq:rev1} as 
\begin{align*} 
D &:= \int_{\R^{dn+1}} T(x,y) \sum_{1\leq j \leq J} \lambda_j \left( U_j (X)  - U_j (Y) \right)^2 dx \, dy\,  d\hat{X} \\
&- \int_{\R^{dn+1}} T(x,y) \left( \left(\sum_{j=1}^{J} \lambda_j U_j ^2 (X) \right)^{1/2} - \left(\sum_{j=1}^{J} \lambda_j U_j ^2 (Y) \right)^{1/2} \right)^2  dx \, dy \, d\hat{X} \\
&= 2 \int_{\R^{dn+1}} T (x,y) I (X,Y)
\end{align*}
where
\begin{align*}
 I(X,Y)&:= \frac{1}{2} \sum_{1\leq j \leq J} \lambda_j \left( U_j (X)  - U_j (Y) \right)^2 -  \frac{1}{2} \left( \left(\sum_{j=1}^J \lambda_j U_j ^2 (X) \right)^{1/2} -  \left(\sum_{j=1}^J \lambda_j U_j ^2 (Y) \right)^{1/2} \right)^2 \\
 &= \sqrt{ \sum_{1\leq i,j\leq J} \lambda_j \lambda_j U_i ^2 (X) U_j ^2 (Y) } - \sum_{j=1} ^J \lambda_j U_j(X) U_j (Y)  \\
 &= \sum_{j=1} ^J \lambda_j U_j(X) U_j (Y) \left( \sqrt{1 + \frac{\sum_{1\leq i,j \leq J} \lambda_i \lambda_j \left( U_i ^2 (X) U_j ^2 (Y) - U_i (X) U_j (X) U_i (Y) U_j (Y) \right) }{\left(\sum_{j=1} ^J \lambda_j U_j(X) U_j (Y)\right)^2}} -1 \right)\\
 &\leq \frac{1}{2} \frac{\sum_{1\leq i<j \leq J} \lambda_i \lambda_j \left( U_i ^2 (X) U_j ^2 (Y) + U_j ^2 (X) U_i ^2 (Y) - 2 U_i (X) U_j (X) U_i (Y) U_j (Y) \right) }{\sum_{j=1} ^J \lambda_j U_j(X) U_j (Y)}. 
\end{align*}
We consider terms of the double sum one at a time. Denoting 
$$u_j = \sqrt{\rho_j}$$
we write
\begin{align*}
&\frac{ U_i ^2 (X) U_j ^2 (Y) + U_j ^2 (X) U_i ^2 (Y) - 2 U_i (X) U_j (X) U_i (Y) U_j (Y) }{\sum_{j=1} ^J \lambda_j U_j(X) U_j (Y)}  \\
&\leq \prod_{k=2} ^n u_i ^2 (x_k) u_j ^2 (x_k) \frac{\left|u_i (x)u_j(y) - u_i(y)u_j(x) \right|^2}{ \lambda_i u_i(x)u_i(y) \prod_{k=2} ^n u_i ^2 (x_k) + \lambda_j u_j(x)u_j(y) \prod_{k=2} ^n u_i ^2 (x_k)} \\
&\leq \frac{1}{c_- \min(\lambda_i,\lambda_j)} \prod_{k=2} ^n u_i ^2 (x_k) u_j ^2 (x_k) \frac{\left|u_i (x)u_j(y) - u_i(y)u_j(x) \right|^2}{ \prod_{k=2} ^n u_i ^2 (x_k) + \prod_{k=2} ^n u_i ^2 (x_k)}\\
&\leq \frac{2}{c_- \min(\lambda_i,\lambda_j)} \prod_{k=2} ^n u_i  (x_k) u_j  (x_k) \left|u_i (x)u_j(y) - u_i(y)u_j(x) \right|^2\\
&\leq \frac{2 c_+ }{c_- \min(\lambda_i,\lambda_j)} \prod_{k=2} ^n u_i  (x_k) u_j  (x_k) \left( \left|u_i (x) - u_i(y)\right|^2 + \left|u_j (x) - u_j(y)\right|^2 \right).
\end{align*}
After multiplying by $\lambda_i \lambda_j$, performing the sum in $i,j$ (recall that the $\lambda_j$'s add up to one), multiplying by $T(x,y)$ and integrating in $x,y$ and $x_2,\ldots x_n$ (the integrals separate) this yields the claimed result, with 
$$ C (c_-,c_+,J) \propto J \frac{c_+}{c_-}.$$
The proof is complete in the cases covered by Definition~\ref{def:hyp class quant}. 

To extend it to~\eqref{eq:app A cin bis} we first apply the preceding arguments to the kinetic energy $T_s$ defined by the positive kernel
$$ T_s (x,y) = C_d |x-y| ^{-(d+2s)}$$
with $0<s<1$ fixed. We then multiply the so-obtained inequality by $(1-s)$ and take the limit $s\to 1$. Using~\eqref{eq:app A cin bis} this gives the result when $T$ is the non-relativistic kinetic energy.
\end{proof}

\subsection{Conclusion}\label{sec:concl class quant}\mbox{}\\\vspace{-0.4cm}

The upper bound follows by the usual trial state argument. Combining Lemma~\ref{lem:app A lim} and Proposition~\ref{lem:app A kin lim} as well as the representation of $\mubf$ given by Theorem~\ref{thm:HS} we deduce  
\begin{align*}
\liminf_{N\to \infty} \frac{E(N)}{N} &= \E \left[  \int_{\PP (\R ^d)} \rho ^{\otimes \infty} dP (\rho) \right] \\
&=  \int_{\PP (\R ^d)}  \E \left[  \rho ^{\otimes \infty}  \right] dP (\rho) \\
&= \int_{\PP (\R ^d)}  \EH \left[  \sqrt{\rho}  \right] \d\mu (\rho) \geq \int_{\PP (\R ^d)}  \eH  \; dP (\rho) = \eH,
\end{align*}
where $P$ is the de Finetti measure. This gives the energy convergence. The convergence of reduced densities follows as usual by noting that equality must hold in all the previous inequalities. 

\newpage

\section{\textbf{Finite-dimensional bosons at large temperature}}\label{sec:large T}

Until now we have only considered mean-field quantum systems at zero temperature, and obtained in the limit $N\to \infty$ de Finetti measures concentrated on the minimizers of the limit energy functional. It is possible, taking a large temperature limit at the same time as the mean-field limit, to obtain a Gibbs measure. In this appendix we explain this for the case of bosons with a finite dimensional state-space, following~\cite{Gottlieb-05,LewRou-unpu13} and~\cite[Chapter 3]{Knowles-thesis}. 

In infinite dimension, important problems arise, in particular for the definition of the limit problem. The non-linear Gibbs measures one obtains play an important role in quantum field theory~\cite{Derezinski-13,Simon-74,Summers-12,Gallavotti-85,GliJaf-87} and in the construction of rough solutions to the NLS equation, see for example~\cite{LebRosSpe-88,Bourgain-94,Bourgain-96,Tzvetkov-08,BurTzv-08,BurTzv-08b,BurThoTzv-10,ThoTzv-10,Suzzoni-11}. We refer to the paper~\cite{LewNamRou-14b} for results on the ``mean-field/large temperature'' limit in infinite dimension and a more thorough discussion of these subjects.

\subsection{Setting and results}\label{sec:app B cadre}\mbox{}\\\vspace{-0.4cm}

In this appendix, the one-body state space will be a complex Hilbert space $\gH$ with finite dimension
$$\dim \gH = d.$$
We consider the mean-field type Hamiltonian
\begin{equation}\label{eq:app B hamil}
H_N = \sum_{j=1} ^N h_j + \frac{1}{N-1} \sum_{1\leq i<j \leq N} w_{ij} 
\end{equation}
where $h$ is a self-adjoint operator on $\gH$ and $w$ a self-adjoint operator on $\gH \otimes \gH$, symmetric in the sense that
$$ w (u\otimes v) = w (v\otimes u), \quad \forall u,v\in \gH.$$
The energy functional is as usual defined by 
$$ \E_N [\Psi_N] = \bral \Psi_N, H_N \Psi_N \ketr$$ 
for $\Psi_N \in \bigotimes_s ^N \gH $ and extended to mixed states $\Gamma_N$ over $\gH ^N = \bigotimes_s ^N \gH$ by the formula
$$ \E_N [\Gamma_N] = \tr_{\gH ^N} \left[ H_N \Gamma_N \right].$$ 
The equilibrium state of the system at temperature $T$ is obtained by minimizing the free-energy functional
\begin{equation}\label{eq:app B free ener func}
\F_N [\Gamma_N] := \E_N [\Gamma_N] + T \tr\left[ \Gamma_N \log \Gamma_N \right] 
\end{equation}
amongst mixed states. The minimizer is the Gibbs state
\begin{equation}\label{eq:app B quant Gibbs}
\Gamma_N = \displaystyle \frac{\exp \left( - T ^{-1} H_N \right)}{\tr\left[ \exp \left( - T ^{-1} H_N \right)\right]}. 
\end{equation}
The associated minimum free-energy is obtained from the partition function (normalisation factor in~\eqref{eq:app B quant Gibbs}) as follows:
\begin{equation}\label{eq:app B free ener} 
F_N = \inf \left\{ \F_N [\Gamma_N], \Gamma_N \in \cS (\gH ^N) \right\} = -T \log \tr\left[ \exp \left( - \frac{1}{T} H_N \right)\right]. 
\end{equation}
We shall be interested in the limit of these objects in the limit
\begin{equation}\label{eq:app B lim}
N\to \infty, \quad T = tN, \quad t \mbox{ fixed } 
\end{equation}
which happens to be the good regime to obtain an interesting limit problem. We will in fact obtain a classical free-energy functional that we now define.

\medskip 

Since $\gH$ is finite-dimensional, one may define $du$, the normalized Lebesgue measure on its unit sphere $S\gH$. The limiting objects will be de Finetti measures, hence probability measures $\mu$ on $S\gH$, and more precisely functions of $L ^1 (S\gH,du)$. We introduce for these objects a classical free-energy functional
\begin{equation}\label{eq:app B free ener func class}
\Fcl [\mu] = \int_{S\gH} \EH [u] \mu(u) du + t  \int_{S\gH} \mu(u) \log \left(\mu (u)\right) du
\end{equation}
and denote $\Fcle$ its infimum amongst positive normalized $L ^1$ functions. It is attained by the Gibbs measure
\begin{equation}\label{eq:app B Gibbs class}
\mucl = \frac{\exp\left( -t^{-1} \EH [u] \right)}{\int_{S\gH} \exp\left( -t^{-1} \EH [u] \right) du} 
\end{equation}
and we have
$$ \Fcle = -t \log \left( \int_{S\gH} \exp\left( -\frac{1}{t} \EH [u] \right) du \right)$$
Here $\EH[u]$ is the Hartree functional
\begin{equation}\label{eq:app B Hartree}
 \EH [u] = \frac{1}{N} \bral u ^{\otimes N}, H_N u ^{\otimes N} \ketr_{\gH ^N} = \bral u,h u \ketr_{\gH} + \half \bral u\otimes u, w u \otimes u \ketr_{\gH ^2}.  
\end{equation}
The theorem we shall prove, due to Gottlieb~\cite{Gottlieb-05} (see also~\cite{GotSch-09,JulGotMarPol-13} and, independently~\cite[Chapter 3]{Knowles-thesis}) is of a semi-classical nature since it provides a link between the quantum and classical Gibbs states,~\eqref{eq:app B quant Gibbs} and~\eqref{eq:app B Gibbs class}:

\begin{theorem}[\textbf{Mean-field/large temperature limit in finite dimension}]\label{thm:app B}\mbox{}\\
In the limit~\eqref{eq:app B lim}, we have
\begin{equation}\label{eq:app B lim ener}
F_N = - T \log \dim\left( \gH_s ^N \right)  + N \Fcle + O(d).
\end{equation}
Moreover, denoting $\gamma_N ^{(n)}$ the $n$-th reduced density matrix of the Gibbs state~\eqref{eq:app B quant Gibbs}, 
\begin{equation}\label{eq:app B lim mat}
\gamma_N ^{(n)} \to \int_{S\gH} | u ^{\otimes n} \rangle \langle u ^{\otimes n} |\mucl(u) du 
\end{equation}
strongly in the trace-class norm of $\gS ^1(\gH ^n)$.
\end{theorem}

\begin{remark}[Mean-field/large temperature limit]\label{rem:app B thm}\mbox{}\\
A few comments:
\begin{enumerate}
\item One should understand this theorem as saying that essentially, in the limit under consideration,
$$ \Gamma_N \approx \int_{S\gH} |u ^{\otimes N} \rangle \langle u ^{\otimes N} | \, \mucl(u) du.$$
The Gibbs state is thus close to a superposition of Hartree states. The notions of reduced density matrices and de Finetti measures provide the appropriate manner to make this rigorous. We will see that the de Finetti measure (lower symbol) associated to $\Gamma_N$ by the methods of Chapter~\ref{sec:deF finite dim} converges to $\mucl (u) du$.
\item Note that the first term in the energy expansion~\eqref{eq:app B lim ener} diverges very rapidly, see~\eqref{eq:dim boson}. The classical free-energy only appears as a correction. In view of the dependence on $d$ of this first term, it is clear that the approach in this appendix cannot be adapted easily to an infinite dimensional setting.
\item Our method of proof differs from that of~\cite{Gottlieb-05}. We shall exploit more fully the semi-classical nature of the problem by using the Berezin-Lieb inequalities introduced in~\cite{Berezin-72,Lieb-73b,Simon-80}. The method we present~\cite{LewRou-unpu13} owes a lot to the seminal paper~\cite{Lieb-73b} and is reminiscent of some aspects of~\cite{LieSeiYng-05}.  
\item It will be crucial for the proof that the lower symbol of $\Gamma_N$ is an approximate de Finetti measure for $\Gamma_N$. This will allow us to apply the first Berezin-Lieb inequality to obtain a lower bound to the entropy term. A new interest of the constructions of Chapter~\ref{sec:deF finite dim} is thus apparent here. We use not only the estimate of Theorem~\ref{thm:DeFinetti quant} but also the particular form of the constructed measure.  
\end{enumerate}\hfill\qed

\end{remark}

\subsection{Berezin-Lieb inequalities}\label{sec:app B Berezin}\mbox{}\\\vspace{-0.4cm}

We recall the resolution of the identity~\eqref{eq:Schur} over $\gH ^N$ given by Schur's lemma. We thus have for each state $\Gamma_N \in \cS (\gH_s ^N)$ a lower symbol defined as
$$ \mu_N (u) = \dim \left( \gH ^N_s\right)\tr \left[\Gamma_N |u ^ {\otimes N}\rangle \langle u ^ {\otimes N} | \right].$$
The first Berezin-Lieb inequality is the following statement:

\begin{lemma}[\textbf{First Berezin-Lieb inequality}]\label{lem:Ber Lie 1}\mbox{}\\
Let $\Gamma_N \in \cS (\gH ^N_s)$ have lower symbol $\mu_N$ and $f:\R^+ \to \R$ be a convex function. We have 
\begin{equation}\label{eq:app B Ber Lie 1}
\tr \left[ f(\Gamma_N) \right] \geq \dim\left( \gH_s ^N \right)\int_{S\gH} f\left(\frac{\mu_N}{\dim\left( \gH_s ^N \right)} \right) du. 
\end{equation}
\end{lemma}

The second Berezin-Lieb inequality applies to states having a positive upper symbol (see Section~\ref{sec:CKMR heur}). One may in fact show that every state has an upper symbol, but it is in general not a positive measure.

\begin{lemma}[\textbf{Second Berezin-Lieb inequality}]\label{lem:Ber Lie 2}\mbox{}\\
Let $\Gamma_N \in \cS (\gH ^N_s)$ have upper symbol $\mu_N \geq 0$,
\begin{equation}\label{eq:app B symb sup}
\Gamma_N = \int_{u\in S \gH}  |u ^{\otimes N} \rangle \langle u ^{\otimes N}| \mu_N (u) du 
\end{equation}
and $f:\R ^+ \to \R$ be a convex function. We have
\begin{equation}\label{eq:app B Ber Lie 2}
\tr \left[ f(\Gamma_N) \right] \leq \dim\left( \gH_s ^N \right) \int_{S\gH} f \left(\frac{\mu_N}{\dim\left( \gH_s ^N \right)} \right) du. 
\end{equation}
\end{lemma}

\begin{proof}[Proof of Lemmas~\ref{lem:Ber Lie 1} and~\ref{lem:Ber Lie 2}]
We follow~\cite{Simon-80}. Since $\Gamma_N$ is a state we decompose it under the form 
$$ \Gamma_N = \sum_{k=1} ^{\infty} \lambda_N ^k |V_N ^k \rangle \langle V_N ^k|$$
with $V_N ^k \in \gH ^N_s$ normalized and $\sum_k \lambda_N ^k = 1$. We denote 
$$ \mu_N ^k (u)=  \left| \bral V_N ^k , u ^{\otimes N} \ketr \right| ^2$$
and by~\eqref{eq:Schur} we have 
\begin{equation}\label{eq:app B sum 1}
 \dim\left( \gH_s ^N \right) \int_{S\gH} \mu_N ^k (u) du  = \bral V_N ^k , V_N ^k \ketr =  1.
\end{equation}
On the other hand, since $(V_N ^k)_{k}$ is a basis of $\gH ^N _s$, for all $u\in S\gH$
\begin{equation}\label{eq:app B sum 2}
 \sum_k \mu_N ^k (u) = \sum_k \left| \bral V_N ^k , u ^{\otimes N} \ketr \right| ^2 =1.
\end{equation}

\medskip

\noindent \emph{First inequality.} Here $\mu_N$ is the lower symbol of $\Gamma_N$ and we have
$$ \mu_N (u) = \dim\left( \gH_s ^N \right)\sum_k \lambda_N ^k \mu_N ^k (u)$$
and thus
$$ \dim\left( \gH_s ^N \right) \int_{S\gH} f\left(\frac{\mu_N}{\dim\left( \gH_s ^N \right)} \right) du \leq \dim\left( \gH_s ^N \right) \int_{S\gH}  \sum_k f \left(\lambda_N ^k \right) \mu_N ^k (u) du$$
by Jensen's inequality and~\eqref{eq:app B sum 2}. Next 
$$ \dim\left( \gH_s ^N \right) \int_{S\gH}  \sum_k f \left(\lambda_N ^k \right) \mu_N ^k (u) du = \sum_k f \left(\lambda_N ^k \right) = \tr \left[f (\Gamma_N)\right]$$
using~\eqref{eq:app B sum 1}.

\medskip

\noindent \emph{Second inequality.} Here $\Gamma_N$ and $\mu_N$ are related via~\eqref{eq:app B symb sup}. We write
\begin{align*}
 \tr\left[ f(\Gamma_N) \right] = \sum_k f\left(\lambda_N ^k\right) &= \sum_k f \left( \langle V_N ^k, \Gamma_N V_N ^k \rangle \right) \\
 &= \sum_k f \left( \int_{S\gH} \mu_N (u) \mu_N ^k (u) du\right) 
 \\&\leq \sum_k \dim\left( \gH_s ^N \right) \int_{S\gH} f\left(\frac{\mu_N(u)}{\dim\left( \gH_s ^N \right)} \right) \mu_N ^k (u) du\\
 &= \dim\left( \gH_s ^N \right) \int_{S\gH} f\left(\frac{\mu_N (u)}{\dim\left( \gH_s ^N \right)} \right)  du
\end{align*}
using Jensen's inequality and~\eqref{eq:app B sum 1} to prove the inequality in the third line and then~\eqref{eq:app B sum 2} to conclude.
\end{proof}

We have here presented a specific version of these general inequalities. It is clear that the proof applies more generally to any self-adjoint operator on a separable Hilbert space having a coherent state decomposition of the form~\eqref{eq:Schur}. In the next section these inequalities will be used to deal with the entropy term by taking $f(x) = x \log x$. This will complete the treatment of the energy using Theorem~\ref{thm:DeFinetti quant} and make the link with the discussion of Chapter~\ref{sec:deF finite dim}.

\subsection{Proof of Theorem~\ref{thm:app B}}\label{sec:app B preuve}\mbox{}\\\vspace{-0.4cm}

\noindent\textbf{Upper bound.} We take as a trial state
$$ \Gamma_N ^{\rm test} := \int_{S\gH} |u ^{\otimes N} \rangle \langle u ^{\otimes N} |\: \mucl(u) du. $$
The energy being linear in the density matrix we have
$$ \E_N \left[ \Gamma_N ^{\rm test} \right] = \int_{S\gH} \E_N \left [|u ^{\otimes N} \rangle \langle u ^{\otimes N} |\right]\mucl(u) du = N \int_{S\gH} \EH [u] \mucl(u) du.$$
For the entropy term we use the second Berezin-Lieb inequality, Lemma~\ref{lem:Ber Lie 2}, with $f(x) = x \log x$. This gives
\begin{align*}
\tr\left[ \Gamma_N ^{\rm test} \log \Gamma_N  ^{\rm test}\right] &\leq \dim\left( \gH_s ^N \right) \int_{S\gH} \frac{\mucl (u)}{\dim\left( \gH_s ^N \right)} \log\left( \frac{\mucl (u)}{\dim\left( \gH_s ^N \right)}\right) du 
\\&= - \log \dim\left( \gH_s ^N \right) + \int_{S\gH} \mucl (u) \log\left( \mucl (u) \right) du.
 \end{align*}

Summing these estimates we obtain 
$$ F_N \leq \F_N \left[ \Gamma_N ^{\rm test} \right] \leq -T \log \dim\left( \gH_s ^N \right) + N \Fcle$$
since $\mucl$ minimizes $\Fcl$.

\medskip

\noindent\textbf{Lower bound.} For the energy we use the reduced density matrices as usual to write
$$ \E_N [\Gamma_N] = N \tr_{\gH} \left[ h \gamma_N ^{(1)}\right] + \frac{N}{2}\tr_{\gH ^2} \left[ w \gamma_N ^{(2)}\right].$$
Denoting
$$ \mu_N (u) = \dim \left(\gH ^N_s \right) \bral u ^{\otimes N}, \Gamma_N u ^{\otimes N}\ketr$$
the lower symbol of $\Gamma_N$, we recall that it has been proved in Chapter~\ref{sec:deF finite dim} that 
\begin{align*}
\tr_{\gH} \left| \gamma_N ^{(1)} - \int_{S\gH} |u\rangle \langle u| \mu_N (u) du \right| &\leq C_1 \frac{d}{N}\\ 
\tr_{\gH^2} \left| \gamma_N ^{(2)} - \int_{S\gH} |u ^{\otimes 2}\rangle \langle u ^{\otimes 2}| \mu_N (u) du \right| &\leq C_2 \frac{d}{N}.
\end{align*}
Since we work in finite dimension, $h$ and $w$ are bounded operators and it follows that
\begin{align*}
 \E_N [\Gamma_N] &\geq N \int_{S\gH} \tr_{\gH} \left[ h |u\rangle \langle u|\right] \mu_N (u) du + \frac{N}{2} \int_{S\gH} \tr_{\gH ^2} \left[ w |u ^{\otimes 2}\rangle \langle u ^{\otimes 2}|\right] \mu_N (u) du - C d\\
 &= N \int_{S\gH} \EH[u] \mu_N (u) du - Cd.
\end{align*}
To estimate the entropy we use the first Berezin-Lieb inequality, Lemma~\ref{lem:Ber Lie 1}, with $f(x) = x\log x$. This yields
\begin{align*}
\tr\left[ \Gamma_N \log \Gamma_N \right] &\geq \dim\left( \gH_s ^N \right) \int_{S\gH} \frac{\mu_N (u)}{\dim\left( \gH_s ^N \right)} \log\left( \frac{\mu_N (u)}{\dim\left( \gH_s ^N \right)}\right) du \nonumber
\\&= - \log \dim\left( \gH_s ^N \right) + \int_{S\gH} \mu_N (u) \log\left( \mu_N (u) \right) du.
\end{align*}
There only remains to sum these estimates to deduce
\begin{align*}
 F_N &= \F_N [\Gamma_N] \geq - T \log \dim\left( \gH_s ^N \right) + N \Fcl[\mu_N] -Cd \\
 &\geq - T \log \dim\left( \gH_s ^N \right) + N \Fcle -Cd 
\end{align*}
since
$$\int_{S\gH} \mu_N (u) du =1$$ 
by definition.
 
\medskip

\noindent\emph{Convergence of reduced density matrices.} The lower symbol $\mu_N (u) du$ is a probability measure on the compact space $S\gH$. We extract a converging subsequence
$$ \mu_N (u) du \to \mu (du) \in \PP(S\gH)$$
and it follows from the results of Chapter~\ref{sec:deF finite dim} that, for all $n\geq 0$, along a subsequence, 
\begin{equation}\label{eq:app B preuve 1}
 \gamma_N ^{(n)} \to \int_{S\gH} | u ^{\otimes n} \rangle \langle u ^{\otimes n}| d\mu(u). 
\end{equation}
Combining the previous estimates gives
\begin{equation}\label{eq:app B preuve}
 \Fcle \geq \Fcl[\mu_N] - C\frac{d}{N}. 
\end{equation}
To pass to the liminf $N\to \infty$ in this inequality, the energy term is dealt with as in the preceding chapters. For the entropy term we note that since $du$ is normalized
$$ \int_{S\gH} \mu_N \log \mu_N du \geq 0$$
since this quantity may be interpreted as the relative entropy of $\mu_N$ with respect to the constant function $1$. Using Fatou's lemma we thus deduce from~\eqref{eq:app B preuve} that
$$ \Fcle \geq \Fcl [\mu]$$
and thus $d\mu (u)= \mucl(u) du$ by uniqueness of the minimizer of $\Fcl$. Uniqueness of the limit also guarantees that the whole sequence converges and there only remains to go back to~\eqref{eq:app B preuve 1} to conclude the proof. \hfill\qed
%
%
%

\end{document}